%% file: AWF_COMP.tex
    \documentclass[12pt,a4paper]{article}
    \usepackage[cp1251]{inputenc}
    \usepackage[english]{babel}
    \usepackage{amsmath,amssymb,amsthm}
    \usepackage{algorithm}
\usepackage[noend]{algpseudocode}

\DeclareMathOperator{\Clo}{Clo}

\DeclareMathOperator{\size}{Size}

\DeclareMathOperator{\proj}{pr}

\DeclareMathOperator{\Var}{Var}

\DeclareMathOperator{\VPol}{UnPol}
\DeclareMathOperator{\Expanded}{ExpCov}
\DeclareMathOperator{\ExpShort}{Coverings}
\DeclareMathOperator{\ConOne}{Con}
\DeclareMathOperator{\PCCon}{ConPC}
\DeclareMathOperator{\LinCon}{ConLin}
\DeclareMathOperator{\ConPC}{ConPC}
\DeclareMathOperator{\BD}{BirthDate}

\DeclareMathOperator{\ConLin}{ConLin}
\DeclareMathOperator{\Opt}{Opt}

\DeclareMathOperator{\Congruences}{Con}
\DeclareMathOperator{\LinkedCon}{LinkedCon}

\DeclareMathOperator{\CSP}{CSP}
\DeclareMathOperator{\QCSP}{QCSP}
\DeclareMathOperator{\PCSP}{PCSP}
\DeclareMathOperator{\SurjCSP}{SurjCSP}
\DeclareMathOperator{\VCSP}{VCSP}
\DeclareMathOperator{\CSPWNU}{CSP-WNU}
\DeclareMathOperator{\Kids}{KIDs}

\DeclareMathOperator{\MinVar}{MinVar}

\newcommand{\Solve}{\mbox{\textsc{Solve}}}
\newcommand{\AnswerOrReduce}{\mbox{\textsc{AnswerOrReduce}}}
\newcommand{\CheckCycleConsistency}{\mbox{\textsc{CheckCycleConsistency}}}
\newcommand{\CheckIrreducibility}{\mbox{\textsc{CheckIrreducibility}}}

\newcommand{\CheckWeakerInstance}{\mbox{\textsc{CheckWeakerInstance}}}
\newcommand{\SolveLinearCase}{\mbox{\textsc{SolveLinearCase}}}
\newcommand{\RemoveTriv}{\mbox{\textsc{RemoveTrivialities}}}
\newcommand{\FactorizeInstance}{\mbox{\textsc{FactorizeInstance}}}
\newcommand{\Reduce}{\mbox{\textsc{Reduce}}}
\newcommand{\SolveLinearSystem}{\mbox{\textsc{SolveLinearSystem}}}
\newcommand{\WeakenConstraint}{\mbox{\textsc{WeakenConstraint}}}
\newcommand{\CheckAllTuples}{\mbox{\textsc{CheckAllTuples}}}
\newcommand{\SolveNonlinked}{\mbox{\textsc{SolveNonlinked}}}
\newcommand{\FindOneEquationLinked}{\mbox{\textsc{FindOneEquationLinked}}}
\newcommand{\FindOneEquationNonlinked}{\mbox{\textsc{FindOneEquationNonlinked}}}
\newcommand{\FindEquationsForNonlinked}{\mbox{\textsc{FindEquationsNonlinked}}}
\newcommand{\WeakenEveryConstraint}{\mbox{\textsc{WeakenEveryConstraint}}}

\newtheorem*{THMKeyConjunctionMain}{Theorem~\ref{KeyConjunctionMain}}
\newtheorem*{THMFindPerfectConstraint}{Theorem~\ref{FindPerfectConstraint}}
\newtheorem*{THMCannotLooseSolution}{Theorem~\ref{CannotLooseSolution}}
\newtheorem*{THMParPropertyMain}{Theorem~\ref{ParPropertyMain}}
\newtheorem*{THMParPropertyForSubcontraint}{Theorem~\ref{ParPropertyForSubcontraint}}

\newcommand{\cover}[1]{
{{#1}^{*}}
}

\renewcommand{\le}{\leqslant}
\renewcommand{\ge}{\geqslant}

\theoremstyle{definition}

\theoremstyle{plain}
\newtheorem{thm}{Theorem}[section]
\newtheorem{conj}{Conjecture}
\newtheorem{problem}{Problem}

\newtheorem{cons}{Corollary}[thm]
\newtheorem{lem}[thm]{Lemma}
\newtheorem{remark}{Remark}
\newtheorem{conslem}{Corollary}[thm]

\newtheorem*{THMmainLinearStep}{Theorem~\ref{LinearStep}}
\newtheorem*{thmAbsorptionCenterStep}{Theorem~\ref{AbsorptionCenterStep}}
\newtheorem*{thmPCStepThm}{Theorem~\ref{PCStepThm}}

\newtheorem*{AbsImpliesConsCorollary}{Corollary~\ref{AbsImpliesCons}}
\newtheorem*{CenterImpliesConsCorollary}{Corollary~\ref{CenterImpliesCons}}
\newtheorem*{PCImpliesCorollary}{Corollary~\ref{PCImplies}}
\newtheorem*{LinearImpliesCorollary}{Corollary~\ref{LinearImplies}}
\newtheorem*{CenterLessThanThreeCorollary}{Corollary~\ref{CenterLessThanThree}}
\newtheorem*{PCLessThanThreeCorollary}{Corollary~\ref{PCLessThanThree}}

\begin{document}

\title{A Proof of the CSP Dichotomy Conjecture}

\author{Dmitriy Zhuk\\
Department of Mechanics and Mathematics \\
Lomonosov Moscow State University\\
Moscow, Russia
}
\date{}
\maketitle


\input{abstract}

\input{preamble}

\input{ZFourExample}
\input{definitions_short}

\input{algorithm_pc}

\input{proof_ismvl}

\input{Definitions}
\input{AbsCenterPCLinearProperty}

\section{Proof of the Auxiliary Statements}\label{AuxStatements}
\input{auxstatements}

\input{Finiteness2}
\section{Proof of the Main Theorems}\label{MainProofs}
\input{Main}

\section*{Acknowledgement}

I want to dedicate this paper to my father 
Nikolay Zhuk who was a great chemist 
and a great person.
He explained me trigonometric functions
when I was in primary school, 
he gave me a programmable calculator on which I wrote 
my first program, and
he taught me to love science. 

I am very grateful to Zarathustra Brady whose comments and remarks allowed to fill many gaps in the original proof and
to significantly improve the text.
Also, I want to thank my colleagues and friends for very fruitful discussions, especially Andrei Bulatov,
Marcin Kozik, Libor Barto, Ross Willard,  Jakub Opr\v{s}al, Jakub Bul\'in, Valeriy Kudryavtsev, Alexey Galatenko, Stanislav Moiseev, and Grigoriy Bokov.
Additionally, I want to thank 
the referees of the paper whose suggestions helped
me to rewrite the whole paper.
They deserve to be the authors of this paper.

The author has received funding from the European Research Council (ERC) under the European Unions Horizon 2020 research and innovation programme (grant
agreement No 771005) and from Russian Foundation for Basic Research (grant 19-01-00200).

\bibliographystyle{plain}
\bibliography{refs}


\end{document}

%% file: abstract.tex
\begin{abstract}
Many natural combinatorial problems can be expressed as constraint satisfaction problems.
This class of problems is known to be NP-complete in general, but certain restrictions on the
form of the constraints can ensure tractability.
The standard way to parameterize interesting subclasses of the constraint satisfaction
problem is via finite constraint languages. The main
problem is to classify those subclasses that are solvable in polynomial
time and those that are NP-complete.
It was conjectured that if a constraint language has a weak near-unanimity polymorphism
then the corresponding constraint satisfaction problem is tractable,
otherwise it is NP-complete.

In the paper we present an algorithm that solves Constraint Satisfaction Problem
in polynomial time for constraint languages having a weak near unanimity polymorphism,
which proves the remaining part of the conjecture.
\end{abstract}

%% file: preamble.tex
\section{Introduction}%

The \emph{Constraint Satisfaction Problem (CSP)} is the problem of deciding whether there is an assignment to a set of variables
subject to some specified constraints. 
Formally, the \emph{Constraint Satisfaction Problem} is defined as a triple $\langle \mathbf{X} , \mathbf{D} , \mathbf{C} \rangle$,
where
\begin{itemize}
\item
$\mathbf{X}=\{x_1,\ldots ,x_n\}$ is a set of variables,
\item
$\mathbf{D}=\{D_{1},\ldots ,D_{n}\}$ is a set of the respective domains,
\item
$\mathbf{C}=\{C_{1},\ldots ,C_{m}\}$ is a set of constraints,
\end{itemize}
where
each variable $x_{i}$ can take on values in the nonempty domain $D_{i}$,
every \emph{constraint} $C_{j}\in \mathbf{C}$ is a pair
$(t_{j},\rho_{j})$ where
$t_{j}$ is a tuple of variables of length $m_{j}$, called the \emph{constraint scope},
and $\rho_{j}$ is an $m_{j}$-ary relation on the corresponding domains,
called the \emph{constraint relation}.

The question is whether there exists \emph{a solution} to
$\langle \mathbf{X} , \mathbf{D} , \mathbf{C} \rangle$,
that is a mapping that assigns a value from $D_{i}$ to every variable $x_{i}$
such that
for each constraint $C_{j}$ the image of the constraint scope is a member of the constraint relation.

In this paper we consider only CSP over finite domains.
The general CSP is known to be NP-complete \cite{Num26, Num30}; however, certain restrictions
on the allowed form of constraints involved may ensure tractability (solvability in polynomial time)
\cite{Num4,Num20,Num22,Num23,CSPconjecture,BulatovAboutCSP}.
Below we provide a formalization to this idea.

To simplify the formulation of the main result we assume that the domain of every variable is a finite set $A$. 
Later we will assume that the domain of every variable is
a unary relation from the constraint language $\Gamma$
(see below).
By $R_{A}$ we denote the set of all finitary relations on $A$,
that is, subsets of $A^{m}$ for some $m$.
Thus, all the constraint relations are from $R_{A}$.

For a set of relations $\Gamma\subseteq R_{A}$ by $\CSP(\Gamma)$
we denote the Constraint Satisfaction Problem where all
the constraint relations are from $\Gamma$.
The set $\Gamma$ is called \emph{a constraint language}.
Another way to formalize the Constraint Satisfaction Problem is
via conjunctive formulas.
Every $h$-ary relation on $A$ can be viewed as
a predicate, that is, a
mapping $A^{h}\rightarrow \{0,1\}$.
Suppose $\Gamma\subseteq R_{A}$, then $\CSP(\Gamma)$ is the following decision problem:
given a formula
$$\rho_{1}(v_{1,1},\ldots,v_{1,n_{1}})
\wedge
\dots
\wedge
\rho_{s}(v_{s,1},\ldots,v_{s,n_{s}}),$$
where $\rho_{1},\dots,\rho_{s}\in \Gamma$, 
and $v_{i,j}\in \{x_{1},\dots,x_{n}\}$ for every $i,j$;
decide whether this formula is satisfiable.

It is well known that many combinatorial problems can be expressed as $\CSP(\Gamma)$
for some constraint language $\Gamma$.
Moreover, for some sets $\Gamma$ the corresponding decision problem can be solved in polynomial time;
while for others it is NP-complete.
It was conjectured that
$\CSP(\Gamma)$ is either in P, or NP-complete \cite{FederVardi}.

\begin{conj}\label{FederVardiConj}
Suppose $\Gamma\subseteq R_{A}$ is a finite set of relations.
Then $\CSP(\Gamma)$ is either solvable in polynomial time, or $NP$-complete.
\end{conj}

We say that an operation
$f\colon A^{n}\to A$ \emph{preserves} the relation $\rho\in R_{A}$ of arity $m$
if for any tuples $(a_{1,1},\ldots,a_{1,m}),\dots,(a_{n,1},\ldots,a_{n,m})\in \rho$
the tuple
$(f(a_{1,1},\ldots,a_{n,1}),\ldots,f(a_{1,m},\ldots,a_{n,m}))$ is in $\rho$.
We say that an operation \emph{preserves a set of relations $\Gamma$} if it preserves every relation in $\Gamma$.
A mapping $f:A\to A$ is called \emph{an endomorphism} of $\Gamma$ if it preserves $\Gamma$. 

\begin{thm}\label{coreHasSameComplexity}\cite{jeavons1998algebraic}
Suppose $\Gamma\subseteq R_{A}$.
If $f$ is an endomorphism of $\Gamma$, 
then $\CSP(\Gamma)$ is polynomially reducible to $\CSP(f(\Gamma))$ and vice versa,
where $f(\Gamma)$ is a constraint language with domain $f(A)$
defined by $f(\Gamma) = \{f(\rho)\colon \rho\in \Gamma\}$.
\end{thm}

A constraint language is \emph{a core} if every endomorphism of $\Gamma$ is a bijection.
It is not hard to show that if $f$ is an endomorphism of $\Gamma$ with minimal range,
then $f(\Gamma)$ is a core. Another important fact is that we can add all
singleton unary relations to a core constraint language
without increasing the complexity of its $\CSP$.
By $\sigma_{=a}$ we denote the unary relation $\{a\}$.

\begin{thm}\label{addingIdempotency}\cite{CSPconjecture}
Let $\Gamma\subseteq R_{A}$ be a core constraint language, and
$\Gamma' = \Gamma\cup \{\sigma_{=a}\mid a\in A\}$.
Then $\CSP(\Gamma')$ is polynomially reducible to $\CSP(\Gamma)$.
\end{thm}
Therefore, to prove Conjecture~\ref{FederVardiConj} it is sufficient to consider only the case when
$\Gamma$ contains all unary singleton relations.
In other words, all the predicates $x = a$, where $a\in A$, are in the constraint language $\Gamma$.

In \cite{Schaefer} Schaefer
classified all tractable constraint languages over two-element
domain. In \cite{BulatovForThree} Bulatov generalized the result for three-element domain.
His dichotomy theorem was formulated in terms of a $G$-set.
Later, the dichotomy conjecture was formulated in several different forms (see \cite{CSPconjecture}).

The result of Mckenzie and Mar{\'o}ti~\cite{miklos} allows us to formulate the dichotomy conjecture in the following nice way.
An operation $f$ on a set $A$ is called \emph{a weak near-unanimity operation (WNU)} if it satisfies 
$f(y,x,\ldots,x) = f(x,y,x,\ldots,x) = \dots = f(x,x,\ldots,x,y)$
for all $x,y\in A$.
An operation $f$ is called \emph{idempotent} if $f(x,x,\ldots,x) = x$ for all $x\in A$.

\begin{conj}\label{mainconj}
Suppose $\Gamma\subseteq R_{A}$ is a finite set of relations.
Then $\CSP(\Gamma)$ can be solved in polynomial time if there exists a WNU
preserving $\Gamma$;
$\CSP(\Gamma)$ is NP-complete otherwise.
\end{conj}

It is not hard to see that
the existence of a WNU preserving $\Gamma$ is equivalent to the existence of
a WNU preserving a core of $\Gamma$,
and also equivalent to
the existence of an idempotent WNU preserving the core.
Hence, Theorems~\ref{coreHasSameComplexity} and \ref{addingIdempotency}
imply that it is sufficient to prove Conjecture~\ref{mainconj}
for a core and an idempotent WNU.

One direction of this conjecture follows from \cite{miklos}.
\begin{thm}\label{MiklosMckenzie}\cite{miklos}
Suppose $\Gamma\subseteq R_{A}$ and 
$\{\sigma_{=a}\mid a\in A\} \subseteq \Gamma$.
If there exists no WNU
preserving $\Gamma$, then $\CSP(\Gamma)$ is NP-complete.
\end{thm}


The dichotomy conjecture was proved for many special cases:
for CSPs over undirected graphs \cite{hell1990complexity},
for CSPs over digraphs with no sources or sinks \cite{barto2009csp},
for constraint languages containing all unary relations~\cite{bulatov2003conservative},
and many others.
More information about the algebraic approach to CSP can be found in \cite{bartopolymorphisms}.

In this paper we present an algorithm that
solves $\CSP(\Gamma)$ in polynomial time if $\Gamma$ is preserved by an idempotent WNU,
and therefore prove the dichotomy conjecture.

\begin{thm}
Suppose $\Gamma\subseteq R_{A}$ is a finite set of relations.
Then $\CSP(\Gamma)$ can be solved in polynomial time if there exists a WNU
preserving $\Gamma$;
$\CSP(\Gamma)$ is NP-complete otherwise.
\end{thm}

Another proof of the dichotomy conjecture was announced by Andrei Bulatov~\cite{BulatovProofCSP,BulatovProofCSPFOCS}.
Even though both algorithms appeared at the same time, they are significantly different. 
Bulatov's algorithm uses full strength of the few subpowers algorithm \cite{idziak2010tractability}, 
uses Maroti's trick for trees on top of Mal'tsev \cite{maroti2011tree}, 
while this one just checks some local consistency and solves linear equations over prime fields. 
Also Bulatov's algorithm works for infinite constraint languages, which is not the case for the algorithm presented in this paper.
But its slight modification works even for infinite constraint languages \cite{zhuk2018modification}.

The paper is organized as follows.
In Section~\ref{ZFourExample} we explain the algorithm informally and 
give an example showing how the algorithm works
for a system of linear equations in $\mathbb Z_{4}$.
In Section~\ref{Definition} we give main definitions,
in Section~\ref{Algorithm}
we give a formal description of the algorithm
showing a pseudocode for most functions
and explain the meaning of every function.
In Section~\ref{CorretnessSection} we formulate all theorems 
that are necessary to prove the correctness of the algorithm.
Then, we prove that on every algebra (domain) with
a WNU operation 
there exists a subuniverse of one of four types, which is the main ingredient of the algorithm. 
Additionally, in this section we prove that 
some functions of the algorithm work properly 
and the algorithm actually works in polynomial time.

In Section~\ref{DefinitionSection} we give the remaining definitions.
The proof of the main theorems is divided into three sections.
In Section~\ref{AbsCenterPCLinear} we study  properties of
subuniverses of each of four types (absorbing, central, PC, and linear subuniverses).
In Section~\ref{AuxStatements} we prove all the auxiliary statements, and in the last section we prove the main theorems of this paper formulated in Section~\ref{CorretnessSection}.

In Section~\ref{ConclusionsSection}, 
we discuss open questions
and consequences of this result.
In particular, we consider generalizations 
of the CSP such as Valued CSP, Infinite Domain CSP, Quantified CSP, Promise CSP and so on.

%% file: ZFourExample.tex
\section{Outline of the algorithm}\label{ZFourExample} 

In this section we give an informal description of the algorithm and show how it 
works for a system of linear equations in $\mathbb Z_{4}$.
The algorithm is based on the following three ingredients:

\begin{itemize}
\item 
Each domain has either one of three kinds of proper strong subsets 
(absorbing, central, polynomially complete) or an equivalence relation modulo which the domain is essentially a product of prime fields
(Theorem \ref{NextReduction}).

\item If a sufficient level of consistency (cycle consistency + irreducibility - see Section \ref{CSPInstancesDef})  
is enforced, then we do not lose all the solutions when we reduce the domain to a proper strong subset 
(that is, if the original instance has a solution, then the reduced instance has a solution as well), which is guaranteed by 
Theorems \ref{AbsorptionCenterStep} and \ref{PCStepThm}.

\item If we cannot reduce the domain in such a way, we are left with an instance 
whose each domain has an equivalence relation modulo which it is a product of prime fields, and all relations are affine subspaces. 
Now we have:

$A$ = the set of all solutions of the instance factorized by the equivalences;

$B$ = the set of all solutions of the factorized instance (where all domains and relations are factorized).

Both $A$ and $B$ are affine subspaces, $A \subseteq B$. 
We would like to know whether $A$ is empty, what we can efficiently compute is $B$ (using Gaussian elimination).
The algorithm gradually makes $B$ smaller (of smaller dimension), while maintaining the property $A \subseteq B$.

First, for some solution from $B$ we check whether $A$ has the same solution, which can be done 
by a recursive call of the algorithm for smaller domains.
If $A$ has it then we are done. 
If $A$ has not then $A\neq B$.
In this case we can make (see Theorem \ref{LinearStep})
the instance weaker  maintaining the property $A'\subsetneq B$ (here $A'$ is $A$ for the weaker instance) until the moment when 
\begin{enumerate}
\item  $A'$ is a subspace of $B$ of codimension one, 
\item or $A=A'=\varnothing$, 
\item or the obtained instance is not linked (it splits into several instances on smaller domains, hence 
$A'$ can be calculated using recursion).
\end{enumerate}
In (1) and (2) $A'$
can be computed by linearly many recursive calls of the algorithm for smaller domains.
In fact, $A'$ can be defined by a linear equation
$c_1x_1 + \dots + c_h x_h = c_{0}$ in a prime field $\mathbb Z_{p}$. 
Then the coefficients $c_0,c_{1},\dots,c_{h}$ can be learned 
(up to a multiplicative constant) by $(p\cdot h+1)$ queries of the form ``$(a_1,\ldots,a_h) \in A'$?''
(see Subsection~\ref{FindingLinearEquationSection} for more details).
To check each query we just need to call
the algorithm recursively for the smaller domains that are the equivalence classes 
corresponding to $a_1,\ldots,a_h$.

We update $B=A'$, return back to the original instance  and continue tightening $B$.
We eventually stop when $B = A$, which gives us the answer to our question:
if $B \neq \varnothing$ then the original instance has a solution, 
if $B = \varnothing$ then it has no solutions.
\end{itemize}

We demonstrate the work of the algorithm on a system of linear equations in $\mathbb Z_{4}$:
\begin{equation}\label{OriginalEquation}
\left\{
\begin{aligned} 
x_{1}+2x_{2}+x_{3}+x_{4}&=2 \\
2x_{1}+x_{2}+x_{3}+x_{4}&=2\\
x_{1}+x_{2} &= 2\\
x_{1}+x_{2} +2x_{4}&= 2
\end{aligned}
\right.
\end{equation}

All the relations (equations) are invariants of the WNU $x_{1}+\dots+x_{5} \;(mod\ 4)$, 
therefore, this system of equations is 
an instance of $\CSP(\Gamma)$, where $\Gamma$ 
is the set of all relations  of arity at most $4$ preserved by the WNU.
Hence, we can apply the algorithm.

First, $\mathbb Z_{4}$ does not have a proper strong subset, which is why 
for every domain there should be
an equivalence relation modulo which it is just a product of prime fields.
In our example it is the modulo 2 equivalence relation.

We factorize our instance modulo 2, and obtain a system of linear equations in $\mathbb Z_{2}$,
where $x_{i}' = x_{i}\;(mod \ 2)$ for every $i$.
\begin{equation}\label{ZTwoEquation}
\left\{
\begin{aligned} 
x_{1}'+x_{3}'+x_{4}'&=0 \\
x_{2}'+x_{3}'+x_{4}'&=0\\
x_{1}'+x_{2}' &= 0\\
x_{1}'+x_{2}' &= 0
\end{aligned}
\right.
\end{equation}

Using Gaussian elimination we solve this system of equations in a field, 
choose independent variables $x_{1}'$ and $x_{3}'$, and write the general solution (the set $B$ in the informal description):
$x_{1}' = x_{1}', x_{2}' = x_{1}', x_{3}' = x_{3}', x_{4}' = x_{1}' + x_{3}'.$

We choose any solution from $B$. Let it be $(0,0,0,0)$ for $x_{1}' = x_{3}' = 0$.
Then we check whether (\ref{OriginalEquation}) has a solution corresponding to $(0,0,0,0)$ by 
restricting every domain to the set $\{0,2\}$ ($x_{i}\mod 2 = 0$).
We recursively call the algorithm for smaller domain and find out that 
(\ref{OriginalEquation}) has no solutions 
inside $\{0,2\}$.
This means that $(0,0,0,0)$ does not belong to $A$ from the informal description, 
therefore $A\subsetneq B$.

Then we try to make the instance weaker so that $A'\subsetneq B$, where $A'$ is 
the intersection of $B$ with the set of all solutions of the new instance factorized by the equivalences.
%
Let us remove the last equation from~(\ref{OriginalEquation})
to obtain a new solution set $A'$.
\begin{equation}\label{SimplifiedEquation}
\left\{
\begin{aligned} 
x_{1}+2x_{2}+x_{3}+x_{4}&=2 \\
2x_{1}+x_{2}+x_{3}+x_{4}&=2\\
x_{1}+x_{2} &= 2
\end{aligned}
\right.
\end{equation}

Again, by solving an instance on the 2-element domain $\{0,2\}$ we find out that 
(\ref{SimplifiedEquation}) has no solutions
corresponding to $(0,0,0,0)$. Therefore, 
we have $A'\subsetneq B$.

We need to check that if we remove one more equation from (\ref{SimplifiedEquation}),
then we get $A' = B$.
Thus, for every weaker instance we need to check that 
for any $a_{1},a_{3}\in \mathbb Z_{2}$ there exists a solution corresponding to $(x_{1}',x_{3}') = (a_{1},a_{3})$.
Since $A'$ is an affine subspace, it is sufficient to check this 
for $(x_{1}',x_{3}') = (0,0)$, $(x_{1}',x_{3}') = (1,0)$, and $(x_{1}',x_{3}') = (0,1)$, 
i.e. for $h+1$ tuples, where $h$ is the dimension of $B$.
Again, to check a concrete solution from $B$ we recursively call the algorithm for 2-element domains. 

Since (\ref{SimplifiedEquation}) is linked, Theorem \ref{LinearStep} guarantees that 
the dimension of $A'$ equals the dimension of $B$ minus one or $A'$ is empty. 
Hence, we need exactly one equation to describe
all pairs $(a_{1},a_{3})$ such that
(\ref{SimplifiedEquation}) has a solution corresponding to $(x_{1}',x_{3}') = (a_{1},a_{3})$.
Let the equation be $c_{1}x_{1}'+c_{3}x_{3}'=c_{0}$. We need to find $c_{1},c_{3},$ and $c_{0}$.
Recursively calling the algorithm for smaller domains, 
we find out that (\ref{SimplifiedEquation}) has a solution $(3,3,0,1)$ corresponding to
$(x_{1}',x_{3}')=(1,0)$ (the solution $(1,1,0,1)$ from $B$) but does not have a solution corresponding to
$(x_{1}',x_{3}')=(0,1)$ (the solution $(0,0,1,1)$ from $B$).
We have 
\begin{equation*}
\left\{
\begin{aligned} 
c_{1}\cdot 0 + c_{3}\cdot 0 &\neq c_{0} \\
c_{1}\cdot 1 + c_{3}\cdot 0 &= c_{0} \\
c_{1}\cdot 0 + c_{3}\cdot 1 &\neq c_{0} 
\end{aligned}
\right.,
\end{equation*}
which implies that $c_{1} = 1$, $c_{3} = 0$, $c_{0} = 1$, and
the equation we are looking for is $x_{1}'=1$.
Thus, we found $A'$.

We add this equation to (\ref{ZTwoEquation}) (update $B= A'$)
and solve the new system of linear equations in $\mathbb Z_{2}$.
\begin{equation}\label{ZTwoEquationNew}
\left\{
\begin{aligned} 
x_{1}'+x_{3}'+x_{4}'&=0 \\
x_{2}'+x_{3}'+x_{4}'&=0\\
x_{1}'+x_{2}' &= 0\\
x_{1}'+x_{2}' &= 0\\
x_{1}' &= 1
\end{aligned}
\right.
\end{equation}
The general solution of this system (the new set $B$) is
$x_{1}'=1$, $x_{2}'=1$, $x_{3}' = x_{3}'$, $x_{4}' = x_{3}'+1$,
where $x_{3}'$ is an independent variable.
Thus, we decreased the dimension of the solution set $B$ by 1 and we still have the property that $A\subseteq B$.
We go back to (\ref{OriginalEquation}), and check whether
it has a solution corresponding to $x_{3}' = 0$ (the solution $(1,1,0,1)$ from $B$).
Again, by solving an instance on the 2-element domain we find out that 
$(1,1,0,1)\notin A$. Therefore $A\subsetneq B$.

The remaining part of the procedure looks trivial but we want to follow the algorithm till the end to make it clear.
Again, we try to make the instance weaker so that $A'\subsetneq B$. 
Let us remove the third equation from~(\ref{OriginalEquation}).
\begin{equation}\label{SimplifiedEquation2}
\left\{
\begin{aligned} 
x_{1}+2x_{2}+x_{3}+x_{4}&=2 \\
2x_{1}+x_{2}+x_{3}+x_{4}&=2\\
x_{1}+x_{2} +2x_{4}&= 2
\end{aligned}
\right.
\end{equation}
By solving this instance on smaller domains we find out that
(\ref{SimplifiedEquation2}) has no solutions
corresponding to $x_{3}'=0$ 
(the solution $(1,1,0,1)$ of $B$). Therefore, we obtained a new 
set $A'\subsetneq B$.

Then we try to remove one more equation from (\ref{SimplifiedEquation2}) maintaining the property 
$A'\subsetneq B$.
We check for every weaker instance that
for any $a_{3}\in \mathbb Z_{2}$ there exists a solution corresponding to $x_{3}' = a_{3}$.

Again, the instance (\ref{SimplifiedEquation2}) is linked, and by Theorem \ref{LinearStep} 
we need exactly one equation to describe
all elements $a_{3}$ such that
(\ref{SimplifiedEquation2}) has a solution corresponding to $x_{3}' = a_{3}$.
Let the equation be $c_{3}x_{3}'=c_{0}$. 
We already checked that it does not hold for $x_{3}'=0$.
By solving an instance on 2-element domains we find out that 
(\ref{SimplifiedEquation2}) has a solution $(3,3,1,0)$ corresponding to the solution 
$(1,1,1,0)$ from $B$ and $x_{3}' = 1$. Thus we have 
\begin{equation*}
\left\{
\begin{aligned} 
c_{3}\cdot 0 &\neq c_{0} \\
c_{3}\cdot 1 &= c_{0} 
\end{aligned}
\right.,
\end{equation*}
which implies $c_{3}=1$, $c_{0} = 1$, and the equation we are looking for is $x_{3}'=1$ (we calculated $A'$).

We add this equation to (\ref{ZTwoEquationNew}) 
(update $B = A'$)
and solve the new system of linear equations in $\mathbb Z_{2}$.
\begin{equation}\label{ZTwoEquationNewNew}
\left\{
\begin{aligned} 
x_{1}'+x_{3}'+x_{4}'&=0 \\
x_{2}'+x_{3}'+x_{4}'&=0\\
x_{1}'+x_{2}' &= 0\\
x_{1}'+x_{2}' &= 0\\
x_{1}' &= 1\\
x_{3}' &= 1
\end{aligned}
\right.
\end{equation}
The only solution of this system is 
$(x_{1}',x_2',x_{3}',x_{4}') = (1,1,1,0)$.
Thus, we decreased the dimension of the solution set $B$ to 0 and we still have the property that $A\subseteq B$.
It remains to check whether the original system (\ref{OriginalEquation}) has a solution 
corresponding to the solution $(1,1,1,0)$ of $B$.
Again, by solving an instance on the 2-element domain we find 
a solution $(3,3,1,0)$ of the original instance. 
Therefore, $(1,1,1,0)\in A$ and we finally reached the condition $A=B$.


%% file: definitions_short.tex
\section{Definitions}\label{Definition}
A set of operations is called \emph{a clone} if it is closed under composition and contains all projections.
For a set of operations $M$ by $\Clo(M)$ we denote the clone generated by $M$.

An idempotent WNU $w$ is called \emph{special} if $x \circ (x \circ y) = x \circ y$, where
$x \circ y = w(x,\dots,x,y)$.
It is not hard to show that for any idempotent WNU $w$ on a finite set there exists a special WNU $w'\in\Clo(w)$
(see Lemma 4.7 in \cite{miklos}).

A relation $\rho \subseteq A_{1}\times\dots\times A_{n}$ is called \emph{subdirect} if
for every $i$ the projection of $\rho$ onto the $i$-th coordinate is $A_{i}$.
For a relation $\rho$ by $\proj_{i_1,\ldots,i_{s}}(\rho)$
we denote the projection of $\rho$ onto the coordinates
$i_1,\ldots,i_{s}$.

\subsection{Algebras}
\emph{An algebra} is a pair $\mathbf{A}:=(A;F)$, where $A$ is a finite set, called \emph{universe},
and $F$ is a family of operations on $A$, called \emph{basic operations of $\mathbf{A}$}.
In the paper we always assume that we have a special WNU
$w$ preserving all constraint relations.
Therefore, every domain $D$, which is from the constraint language, can be viewed as an algebra $(D;w)$.
By $\Clo(\mathbf{A})$ we denote the clone generated by all basic operations of $\mathbf{A}$.



An equivalence relation $\sigma$ on the universe of an algebra $\mathbf{A}$ is called \emph{a congruence}
if it is preserved by every operation of the algebra.
A congruence (an equivalence relation) is called \emph{proper}, if it is not equal 
to the full relation $A\times A$.
A subuniverse is called \emph{nontrivial} if it is proper and nonempty.
We use standard universal algebraic notions of term operation, subalgebra,  factor algebra, product of algebras,
see~\cite{bergman2011universal}.
We say that a subalgebra $\mathbf{R} = (R;F_R)$ is
\emph{a subdirect subalgebra} of $\mathbf{A}\times \mathbf{B}$
if $R$ is a subdirect relation in $A\times B$.


\subsection{Polynomially complete algebras}

An algebra $(A;F_{A})$ is called \emph{polynomially complete (PC)}
if the clone generated by $F_{A}$ and all constants on $A$ is the clone of all operations on $A$
(see \cite{istinger1979characterization,lausch2000algebra}).


\subsection{Linear algebra}
An idempotent finite algebra $(A;w_{A})$ is called \emph{linear}
(similar to affine in  \cite{freese1987commutator})
if
it is isomorphic to $(\mathbb{Z}_{p_1}\times\dots\times \mathbb{Z}_{p_s};x_1+\ldots+x_m)$
for prime numbers $p_{1},\ldots,p_{s}$.
Since 
$\mathbf{A}/(\sigma\cap\tau)$ is always isomorphic to 
a subalgebra of 
$\mathbf{A}/\sigma\times
\mathbf{A}/\tau$, 
and since linear algebras are closed under 
products and subalgebras 
by Corollary~\ref{LinearAlgebrasAreClosed},
for every idempotent finite algebra $(B;w_{B})$ there exists a least congruence $\sigma$, called
\emph{the minimal linear congruence}, such that
$(B;w_{B})/\sigma$ is linear.

\subsection{Absorption}
%
Let $B$ be a (probably empty) subuniverse of $\mathbf{A}=(A;F_{A})$.
We say that \emph{$B$ absorbs $\mathbf{A}$}
if there exists $t\in \Clo(\mathbf{A})$ such that
$t(B,B,\dots,B,A,B,\dots,B) \subseteq B$ for any position of $A$.
In this case we also say that \emph{$B$ is an absorbing subuniverse of $\mathbf A$
with a term operation $t$}.
If the operation $t$ can be chosen binary then we say that
$B$ is a binary absorbing subuniverse of $\mathbf A$.
For more information about absorption and its connection with CSP see \cite{barto2017absorption}.

\subsection{Center}
Suppose $\mathbf{A} = (A;w_{A})$ is a finite algebra with a special WNU operation.
$C\subseteq A$ is called a \emph{center}
if there exists an algebra $\mathbf{B} = (B;w_{B})$ with a special WNU operation of the same arity and
a subdirect subalgebra $(R;w_{R})$ of $\mathbf{A}\times\mathbf{B}$ such that
there is no nontrivial binary absorbing subuniverse in $\mathbf{B}$ and
$C = \{a\in A\mid \forall b\in B\colon (a,b)\in R\}.$
This notion was motivated by central relations defining 
maximal clones on finite sets (see section 5.2.5 in \cite{lau}) 
and it is very similar to ternary absorption (see Corollary~\ref{ternaryAbsorption}).

\subsection{CSP instance}\label{CSPInstancesDef}

An instance of the constraint satisfaction problem 
is called \emph{a CSP instance}.
Sometimes we use the same letter for a CSP instance and for the set of all constraints of this instance.
For a variable $z$ by $D_{z}$ we denote the domain of the variable $z$.


We say that $z_{1}-C_{1}-z_{2}-\dots - C_{l-1}-z_{l}$ is
\emph{a path} in a CSP instance $\Theta$ if $z_{i},z_{i+1}$ are in the scope of $C_{i}$ for every $i$.
We say that \emph{a path $z_{1}-C_{1}-z_{2}-\dots- C_{l-1}-z_{l}$  connects $b$ and $c$}
if there exists $a_{i}\in D_{z_{i}}$ for every $i$
such that
$a_{1} = b$, $a_{l} = c$, and
the projection of $C_{i}$ onto $z_{i}, z_{i+1}$
contains the tuple $(a_{i},a_{i+1})$.

A CSP instance is called \emph{1-consistent} if every constraint of the instance is subdirect.
A CSP instance is called \emph{cycle-consistent} if
it is 1-consistent and 
for every variable $z$ and $a\in D_{z}$
any path starting and ending with $z$ in $\Theta$ 
connects $a$ and $a$.
Other types of local consistency and 
its connection with the complexity of CSP
are considered in \cite{kozik2016weak}.
A CSP instance $\Theta$ is called \emph{linked}
if for every variable $z$ occurring in the scope of a constraint of $\Theta$ and every $a,b\in D_{z}$
there exists a path starting and ending with $z$ in $\Theta$ that connects $a$ and $b$.

Suppose $\mathbf{X'}\subseteq\mathbf{X}$.
Then we can define a projection of $\Theta$ onto $\mathbf{X'}$,
that is a CSP instance where variables are elements of $\mathbf{X'}$ and constraints are projections of the constraints of $\Theta$ onto the intersection of their scopes with $\mathbf{X'}$,
ignoring any constraint whose scope
does not intersect $\mathbf{X'}$.
We say that an instance $\Theta$ is \emph{fragmented}
if the set of variables $\mathbf X$ can be divided into 2 disjoint sets $\mathbf{X_1}$ and
$\mathbf{X_2}$ such that
each of them contains a variable from the scope of a constraint of $\Theta$, 
and 
the constraint scope of any constraint of $\Theta$
either has variables only from $\mathbf{X_1}$, or only from $\mathbf{X_2}$.
Thus, if an instance is fragmented, then it can be divided into several nontrivial instances.

A CSP instance $\Theta$ is called \emph{irreducible} if
any instance $\Theta'=(\mathbf X',\mathbf D',\mathbf C')$ 
such that  
$\mathbf X'\subseteq\mathbf X$, $D_{x}'=D_{x}$ for every 
$x\in\mathbf X'$,
and 
every constraint of $\Theta'$ is a projection of 
a constraint from $\Theta$ on some set of variables
is fragmented, or linked, or its solution set is subdirect.



We say that a constraint $C_{1}= ((y_{1},\ldots,y_{t});\rho_{1})$ is \emph{weaker or equivalent to}
a constraint $C_{2}= ((z_{1},\ldots,z_s);\rho_{2})$
if $\{y_{1},\ldots,y_{t}\}\subseteq \{z_{1},\ldots,z_s\}$
and 
$C_{2}$ implies $C_{1}$.
In other words, the second condition says that
the solution set to $\Theta_{1}:=(\{z_{1},\dots,z_{s}\},(D_{z_{1}},\ldots,D_{z_{s}}),C_{1})$
contains the solution set to 
$\Theta_{2}:=(\{z_{1},\dots,z_{s}\},(D_{z_{1}},\ldots,D_{z_{s}}),C_{2})$.
We say that 
$C_{1}$ is \emph{weaker than} 
$C_{2}$ 
if $C_{1}$ is weaker or equivalent to $C_{2}$
but $C_{1}$ does not imply $C_{2}$.

The following remark justifies weakening constraints of the instance in the algorithm
(this remark follows from Lemma~\ref{ExpandedConsistencyLemma}).
\begin{remark}
Suppose $\Theta = \langle\mathbf{X};\mathbf{D};\mathbf{C}\rangle$ and $\Theta' = \langle\mathbf{X'};\mathbf{D'};\mathbf{C'}\rangle$ are CSP instances
such that 
$\mathbf{X'}\subseteq \mathbf{X}$,
$D_{x}'=D_{x}$ for every $x\in \mathbf{X'}$,
and every constraint of $\Theta'$ is weaker or equivalent to a constraint of $\Theta$. If $\Theta$ is cycle-consistent and irreducible,
then so is $\Theta'$.
\end{remark}

%

We say that a variable $y_{i}$ of the constraint $((y_{1},\ldots,y_{t});\rho)$ is \emph{dummy} if 
$\rho$ does not depend on its $i$-th variable.

\begin{remark}
Adding a dummy variable to a constraint 
and removing of a dummy variable do not affect the property of being 
cycle-consistent and irreducible.
\end{remark}

Let $D_{i}'\subseteq D_{i}$ for every $i$.
A constraint $C$ of $\Theta$ is called \emph{crucial in $(D_{1}',\ldots,D_{n}')$}
if it has no dummy variables, $\Theta$ has no solutions in $(D_{1}',\ldots,D_{n}')$ but
the replacement of $C\in\Theta$ by all
weaker constraints gives an instance with a solution in $(D_{1}',\ldots,D_{n}')$.
A CSP instance $\Theta$ is called \emph{crucial in $(D_{1}',\ldots,D_{n}')$} if
it has at least one constraint and 
every constraint of $\Theta$ is crucial in $(D_{1}',\ldots,D_{n}')$.

\begin{remark}\label{GetCrucialInstance}
Suppose $\Theta$ has no solutions in $(D_{1}',\ldots,D_{n}')$.
We can replace each constraint by
its projection onto its non-dummy variables. Then we iteratively replace every constraint
by all weaker constraints having no dummy variables until it is crucial.
Finally, we get
a CSP instance that is crucial in $(D_{1}',\ldots,D_{n}')$.
\end{remark}




%% file: algorithm_pc.tex
\newcommand{\CheckTuple}{\mbox{\textsc{CheckTuple}}}
\newcommand{\type}{\mbox{\textsc{type}}}

\newcommand{\Break}{\State \textbf{break} }

\newcommand{\Output}{\mbox{Output}}
\newcommand{\Changed}{\mbox{Changed}}
\newcommand{\calC}{\mathbf{C}}
\newcommand{\calD}{\mathbf{D}}
\newcommand{\X}{\mathbf{X}}

\section{Algorithm}\label{Algorithm}
\subsection{Main part}\label{AlgorithmMainPart}

Suppose we have a constraint language $\Gamma_{0}$ that is preserved by an idempotent WNU operation.
As it was mentioned before,
$\Gamma_{0}$ is also preserved by a special WNU operation $w$.
Let $k_{0}$ be the maximal arity of the relations in $\Gamma_{0}$.
By $\Gamma$ we denote the set of all relations of arity at most $k_{0}$
that are preserved by $w$.
Obviously, $\Gamma_{0}\subseteq \Gamma$, therefore every instance of $\CSP(\Gamma_{0})$ is an instance of $\CSP(\Gamma)$.

In this section we provide an algorithm that solves $\CSP(\Gamma)$ in polynomial time.
Suppose we have a CSP instance $\Theta = \langle \mathbf{X} , \mathbf{D} , \mathbf{C} \rangle$,
where
$\mathbf{X}=\{x_1,\ldots ,x_n\}$ is a set of variables,
$\mathbf{D}=\{D_{1},\ldots ,D_{n}\}$ is a set of the respective domains,
$\mathbf{C}=\{C_{1},\ldots ,C_{q}\}$ is a set of constraints.
Let the arity of the WNU $w$ be equal to $m$.

The main part of the algorithm (function \Solve)
is an iterative loop; in each pass through the loop, the algorithm calls a subroutine $\AnswerOrReduce$ whose
job is to find a reduction of a domain or to terminate with the final answer.
The reduction returned by the function should satisfy the following property: if $\Theta$ has a solution, 
then it has a solution after the reduction.
If the reduction was found then we apply the function 
$\Reduce$, which takes an instance $\Theta=(\X,\calD,\calC)$ and a domain set 
$\mathbf{D'} = (D_1',\ldots,D_{n}')$, and returns a new instance 
$(\X,\mathbf{D'},\mathbf{C'})$, where 
$\mathbf{C'}=\{((x_{i_{1}},\dots,x_{i_{s}}),\rho\cap (D_{i_{1}}'\times\dots \times D_{i_{s}}'))\mid
    ((x_{i_{1}},\dots,x_{i_{s}}),\rho)\in\calC\}$.

\begin{algorithm}
\begin{algorithmic}[1]
\Function{\Solve}{$\Theta$}
  \State{\textbf{Input:} CSP($\Gamma$) instance $\Theta=(\X,\calD,\calC)$,
$\X=(x_1,\ldots,x_n)$, $\calD = (D_1,\ldots,D_n)$}
\Repeat
  \State{$\Output := \AnswerOrReduce(\Theta)$}
  \If{$\Output = \mbox{``Solution"}$} \Return{``Solution"}
  \EndIf
  \If{$\Output = \mbox{``No solution"}$} \Return{``No solution"}
  \EndIf
  \If{$\Output = (x_i,U)$} \Comment{$\varnothing \ne U\subset D_i$}
    \State{$\Theta:= \Reduce(\Theta,(D_1,\dots,D_{i-1},U,D_{i+1},\dots,D_{n}))$}
    \Comment{Set $D_{i}=U$}
  \EndIf
\Until{Done}
\EndFunction
\end{algorithmic}
\end{algorithm}

%

The function $\AnswerOrReduce$  
(see the pseudocode)
checks different types of consistency such as cycle-consistency and irreducibility, and reduce a domain if the instance is not consistent. 
If it is consistent, then either it reduces a domain to a 
proper strong subset, or 
it uses $\SolveLinearCase$ to 
solve the remaining case.


First, the function $\AnswerOrReduce$
checks whether the instance $\Theta$ is cycle-consistent (function $\CheckCycleConsistency$). 
If it is not cycle-consistent then either some domain can be reduced, or the instance has no solutions. In both cases we terminate the function and return the result. If it is 
cycle-consistent then we go on.

If the size of every domain is one it returns that a solution was found.

Then we check whether the instance is irreducible (function \CheckIrreducibility).  
If it is not irreducible then we return how to reduce some domain or return that there is no solutions, otherwise we go on.

\begin{algorithm}
\begin{algorithmic}[1]
\Function{AnswerOrReduce}{$\Theta$}
\State{\textbf{Input:} CSP($\Gamma$) instance $\Theta=(\X,\calD,\calC)$,
$\X=(x_1,\ldots,x_n)$, $\calD = (D_1,\ldots,D_n)$}

  \State{$\Output := \CheckCycleConsistency(\Theta)$}
  \If{$\Output \neq \mbox{``Ok"}$} \Return{$\Output$}
  \EndIf
    \If{$|D_{i}|=1$ for every $i$}
        \Return{``Solution"}
    \EndIf
  \State{$\Output := \CheckIrreducibility(\Theta)$}
  \If{$\Output \neq \mbox{``Ok"}$} \Return{$\Output$}
  \EndIf
  
  \State{$\Output := \CheckWeakerInstance(\Theta)$}
  \If{$\Output \neq \mbox{``Ok"}$} \Return{$\Output$}
  \EndIf

  \If{$B_{i}$ is a nontrivial binary absorbing subuniverse of $D_{i}$}  \Return{$(x_{i},B_{i})$}
  \EndIf

\If{$C_{i}$ is a nontrivial center of $D_{i}$}  \Return{$(x_{i},C_{i})$}
  \EndIf
  
\If{$\sigma$ is a proper congruence on $D_{i}$ and $(D_{i};w)/\sigma$ is polynomially complete} 

Choose an equivalence class $E$ of $\sigma$ 

\Return{$(x_{i},E)$}
\EndIf
\Return{$\SolveLinearCase(\Theta)$}

\EndFunction
\end{algorithmic}
\end{algorithm}

After that we check a different type of consistency (function $\CheckWeakerInstance$).
We make a copy of $\Theta$, and simultaneously replace every constraint by all weaker constraints without dummy variables. Recursively calling the algorithm, we check that the obtained instance has a solution
with $x_{i}=b$ for every $i\in\{1,2,\ldots,n\}$ and $b\in D_{i}$.
If not, reduce $D_{i}$ to the projection onto $x_{i}$ of the solution set of the obtained instance. Otherwise, go on.

By Theorem~\ref{AbsorptionCenterStep} we 
cannot pass from an
instance having solutions to an instance having no solutions
when reduce a domain to a
nontrivial binary absorbing subuniverse or 
to a nontrivial center. 
Thus, if $D_{i}$ has a nontrivial binary absorbing subuniverse $B_{i}\subsetneq D_{i}$ for some $i$, then we reduce $D_{i}$ to $B_{i}$,
Similarly, if $D_{i}$ has a nontrivial center $C_{i}\subsetneq D_{i}$ for some $i$, then we reduce $D_{i}$ to $C_{i}$

By Theorem~\ref{PCStepThm} 
we 
cannot pass from an
instance having solutions to an instance having no solutions
when reduce a domain to an equivalence class of 
a proper congruence $\sigma$ such that $(D_{i};w)/\sigma$ is polynomially complete.
Thus, if such a congruence on $D_{i}$ exists, we reduce $D_{i}$ to
its equivalence class.

By Theorem~\ref{NextReduction},
it remains to consider the case when
on every domain $D_{i}$
of size greater than 1
there exists a proper congruence $\sigma$
such that 
$(D_{i};w)/\sigma$ is isomorphic to 
$(\mathbb Z_{p};x_1+\dots+x_{m})$ for some $p$.
In this case the problem is solved by the function $\SolveLinearCase$, which  will be described in the next subsection. 
A detailed description and a pseudocode for the functions 
$\CheckCycleConsistency$, $\CheckIrreducibility$, and $\CheckWeakerInstance$
will be given in Subsection \ref{AlgorithmTechnicalities}

\subsection{Linear case}\label{AlgorithmLinearCase}

In this section we define the function $\SolveLinearCase$ (see the pseudocode).
For every $i$ let 
$\sigma_{i}$ be the minimal linear congruence 
on $D_{i}$, which is the smallest congruence $\sigma$ 
such that $(D_{i};w)/\sigma$ is linear.
Then $(D_{i};w)/\sigma_{i}$ is isomorphic to 
$(\mathbb Z_{p_{1}}\times \dots \times \mathbb Z_{p_{l}};x_{1}+\dots+x_{m})$
for prime numbers $p_{1},\ldots,p_{l}$.
Recall that we apply the function $\SolveLinearCase$ 
only if $\sigma_{i}$ is proper for every $i$ such that $|D_{i}|>1$.
We will show that modulo these congruences the instance can be viewed as a system of linear equations in fields.

We denote $D_{i}/\sigma_{i}$ by $L_{i}$ and
define a new CSP instance $\Theta_{L}$ with domains $L_{1},\ldots,L_{n}$ as follows.
To every constraint $((x_{i_1},\ldots,x_{i_s});\rho)\in \Theta$
we assign a constraint
$((x_{i_1}',\ldots,x_{i_s}');\rho')$,
where $\rho'\subseteq L_{i_{1}}\times\dots\times L_{i_{s}}$
and $(E_{1},\ldots,E_{s})\in\rho'\Leftrightarrow
(E_{1}\times\dots\times E_{s})\cap\rho\neq\varnothing.$
The constraints of $\Theta_{L}$ are all constraints that are assigned to the constraints of $\Theta$.
The function generating the instance $\Theta_{L}$ from $\Theta$ is called
$\FactorizeInstance$ in the pseudocode.
Note that $\Theta_{L}$ is a CSP instance but not necessarily an instance in the constraint language $\Gamma$.

Since each $L_{i}$ is isomorphic to some
$\mathbb Z_{m_{1}}\times \dots\times\mathbb Z_{m_s}$, we may define a natural bijective mapping 
$\psi:\mathbb Z_{p_{1}}\times\dots\times \mathbb Z_{p_r}\to
L_{1}\times\dots\times L_{n}$, and 
assign a variable $z_{i}$ to every $\mathbb Z_{p_{i}}$.
Since every relation on $\mathbb Z_{p_{1}}\times \dots \times \mathbb Z_{p_{r}}$ preserved by $x_{1}+\ldots+x_{m}$ is known (see Lemma~\ref{LinearAlgebrasFact}) to be
a conjunction of linear equations,
the instance $\Theta_{L}$ can be viewed as a system of linear equations over $z_{1},\ldots,z_{r}$.
Note that 
every equation is an equation in $\mathbb Z_{p}$ but $p$ can be different for different equations,
and only variables with the same domain $\mathbb Z_{p}$ may appear in one equation.

\begin{algorithm}
\begin{algorithmic}[1]
\Function{SolveLinearCase}{$\Theta$}
\State{\textbf{Input:} CSP($\Gamma$) instance $\Theta=(\X,\calD,\calC)$,
$\X=(x_1,\ldots,x_n)$, $\calD = (D_1,\ldots,D_n)$}

\State{$\Theta_{L} := \FactorizeInstance(\Theta)$}
\State{$Eq := \varnothing$} \Comment{The equations we add to $\Theta_{L}$}
\Repeat
  \State{$\phi := \SolveLinearSystem(\Theta_{L}\cup Eq)$}
  
  \Comment{$\phi(\mathbb Z_{q_{1}}\times\dots\times\mathbb Z_{q_{k}})$ is the solution set of $\Theta_{L}\cup Eq$}
  \If{$\phi  = \varnothing$}  
    \Return{``No solution"}
  \EndIf    
    \If{$\Solve(\Reduce(\Theta, \phi(0,0,\ldots,0)))  = \mbox{``Solution"}$}
    \Return{``Solution"}
    \ElsIf{k=0} 
        \Return{``No solution"} \Comment{$\Theta_{L}$ has just one solution}
    \EndIf

\State{$\Theta':= \RemoveTriv(\Theta)$}
\Repeat \Comment{Try to weaken $\Theta'$}
    \State{$\Changed:= false$}
    \For{$C\in \Theta'$} 
        \State{$\Omega:= \RemoveTriv(\WeakenConstraint(\Theta',C))$}
        \If{$\neg\CheckAllTuples(\Omega,\phi)$}
            \State{$\Theta':=\Omega$}
            \State{$\Changed:= true$}
            \Break
        \EndIf
    \EndFor
\Until{$\neg\Changed$} \Comment{$\Theta'$ cannot be weakened anymore}
  \If{$\Theta'$ is not linked} 
    \State{$Eq := Eq\cup\FindEquationsForNonlinked(\Theta')$}
  \Else 
  \State{$Eq := Eq\cup\{\FindOneEquationLinked(\Theta',\phi)\}$}
\EndIf
\Until{Done}
\EndFunction
\end{algorithmic}
\end{algorithm}

As it was described in Section \ref{ZFourExample}, 
we consider the set $A$, which is the solution set of $\Theta$ factorized by the congruences $\sigma_{1},\ldots,\sigma_{n}$, 
and the set $B$, which is the solution set of $\Theta_{L}$.
We know that $A\subseteq B$ and we want to check whether $A$ is empty.
We iteratively add new equations to the set $\Theta_{L}$ maintaining the 
property that $A\subseteq B$, and therefore reduce the dimension of $B$.
We start with the empty set of equations $Eq$ (line 4 of the pseudocode).


Then we apply the function $\SolveLinearSystem$ that solves the system of linear equations $\Theta_{L}\cup Eq$ 
using Gaussian elimination.
If the system has no solutions then $\Theta$ has no solutions and we are done.
Otherwise, we choose independent variables
$y_{1},\ldots,y_{k}$, then 
the general solution (the set $B$) can be written as an 
affine mapping $\phi\colon\mathbb Z_{q_{1}}\times \dots \times \mathbb Z_{q_{k}}\to 
L_{1}\times\dots\times L_{n}$.
Denote ${Z} = \mathbb Z_{q_{1}}\times \dots \times \mathbb Z_{q_{k}}$,
then any solution of $\Theta_{L}\cup Eq$ can be obtained
as $\phi(a_{1},\ldots,a_{k})$ for some $(a_{1},\ldots,a_{k})\in {Z}$.

Note that for any tuple $(a_{1},\ldots,a_{k})\in {Z}$
we can check recursively whether $\Theta$ has a solution
in $\phi(a_{1},\ldots,a_{k})$ (i.e. whether 
$\phi(a_{1},\ldots,a_{k})\in A$).
To do this, we just need to 
reduce the domains to the solution (function \Reduce)
and solve an easier CSP instance (on smaller domains).
Similarly, we can check whether $\Theta$ has a solution in
$\phi(a_{1},\ldots,a_{k})$
for every $(a_{1},\ldots,a_{k})\in \mathbb {Z}$ (i.e. whether $A=B$).
Since $A$ and $B$ are subuniverses 
of $L_{1}\times\dots\times L_{n}$ (almost subspaces), we just need to check the
existence of a solution in $\phi(0,\ldots,0)$ and
$\phi(0,\ldots,0,1,0,\ldots,0)$  for any position of $1$.
See the pseudocode of the function $\CheckAllTuples$ for the last procedure.

\begin{algorithm}
\begin{algorithmic}[1]
\Function{\CheckAllTuples}{$\Theta$, $\phi$}
\State{\textbf{Input:} CSP($\Gamma$) instance $\Theta$, a solution of a linear system of equations $\phi$}
    \If{$\Solve(\Reduce(\Theta,\phi(0,\ldots,0))) = \mbox{``No solution"}$}
        \Return{$false$}
    \EndIf
   \For{$i=1,2,\ldots,k$}
        \State{$t := (\underbrace{0,\ldots,0,1}_{i},0,\ldots,0)$}
        \If{$\Solve(\Reduce(\Theta,\phi(t))) = \mbox{``No solution"}$}
            \Return{$false$}
        \EndIf
    \EndFor
    \Return {$true$};
\EndFunction
\end{algorithmic}
\end{algorithm}

Let us go back to the function $\SolveLinearCase$.
After solving the linear system we check whether there exists a solution of $\Theta$ corresponding to the solution $\phi(0,0,\ldots,0)$ of $\Theta_{L}\cup Eq$. 
If $k=0$, i.e. 
$\Theta_{L}\cup Eq$ has only one solution, then we denote this solution by $\phi(0,0,\ldots,0)$.
If $\Theta$ has a solution in $\phi(0,\ldots,0)$, then it remains to return the result ``Solution''. 
If it has no solutions and $k=0$ then return the result ``No solution''. 

At this point (line 10 of the pseudocode of 
$\SolveLinearCase$), we have the property that 
the set $B$ is of dimension at least 1, and $A\neq B$
since we found a solution $\phi(0,\ldots,0)$ of the system of linear equations 
without the corresponding solution of $\Theta$.

Then we iteratively remove from $\Theta$ all
constraints that are weaker than some other constraints of $\Theta$, 
remove all constraints without non-dummy variables, 
and replace every constraint by its projection onto non-dummy variables.
This procedure we denote by the function $\RemoveTriv$.
In the pseudocode of  
$\SolveLinearCase$ we denote the obtained instance by $\Theta'$.

Then we try to make the constraints of $\Theta'$ weaker maintaining the property 
that $A'\neq B$, where $A'$ is the solution set of $\Theta'$ factorized by 
the congruences $\sigma_{1},\ldots,\sigma_{n}$.
Precisely, we choose a constraint $C$, replace it by all weaker constraints 
without dummy variables (function $\WeakenConstraint$),
apply $\RemoveTriv$, 
and check using the function $\CheckAllTuples$ whether 
$A' = B$. If not, then we replace $\Theta'$ by the new weaker instance.

Suppose we cannot make any constraint weaker maintaining the property 
$A'\neq B$.
Then $\Theta'$ has no solutions in $\phi(b_{1},\ldots,b_{k})$
for some $(b_{1},\ldots,b_{k})\in {Z}$,
but if we replace any constraint $C\in\Theta'$ by all weaker constraints,
then we get an instance that has a solution
in $\phi(a_{1},\ldots,a_{k})$
for every $(a_{1},\ldots,a_{k})\in {Z}$.
Therefore, 
$\Theta'$ is crucial in
$\phi(b_{1},\ldots,b_{k})$.
Note that by Lemma~\ref{ExpandedConsistencyLemma} the instance $\Theta'$ is still cycle-consistent and irreducible.
Also, $\Theta'$ is not fragmented because it is 
crucial.

Then, in line 20 of the function $\SolveLinearCase$ we have two options.

If $\Theta'$ is not linked then using the function $\FindEquationsForNonlinked$ we calculate its solution set factorized by the congruences (the set $A'$). 
This solution set can be defined by a set of linear equations, 
which we add to $Eq$ and therefore replace $B$ by $A'\cap B$. 
Thus, we made $B$ smaller and we still have the property $A\subseteq B$, 
since $A'$ is the factorized solution set of the instance $\Theta'$, which is weaker than 
$\Theta$.

If $\Theta'$ is linked then by Theorem~\ref{LinearStep}
either $A' = \varnothing$, or 
the dimension of $A'$ is equal to the dimension of $B$ minus 1, 
which allows us to find a new linear equation by polynomially many 
queries ``Does there exist a solution of $\Theta'$ in 
$\phi(a_{1},\ldots,a_{k})$?''.
We calculate this new equation 
by the function $\FindOneEquationLinked$, which will be defined in the next section 
as well as the function $\FindEquationsForNonlinked$. 
Note that the new equation can be ``$0=1$'' if $A'=\varnothing$.

After new equations found, we go back to line 6 of the function 
$\SolveLinearCase$ and solve a system of linear equations again.
Since every time we reduce the dimension of $B$ by at least one, 
the procedure will stop in at most $r$ steps.

\subsection{Finding linear equations}\label{FindingLinearEquationSection}

In this section we define 
the functions 
$\FindOneEquationLinked$, $\FindOneEquationNonlinked$, and  $\FindEquationsForNonlinked$, which 
allow us to find new equations defining the set $A'$.

\begin{algorithm}
\begin{algorithmic}[1]
\Function{\FindOneEquationLinked}{$\Theta,\phi$}
\State{\textbf{Input:} CSP($\Gamma$) instance $\Theta$, a solution of a system of linear equations $\phi$}
    \State{$t := \varnothing$}     \Comment{We search for a tuple $t$ outside of the solution set}    
    \If{$\Solve(\Reduce(\Theta,\phi(0,\ldots,0))) = \mbox{``No solution"}$}
        \State{$t:= (0,\ldots,0)$}
    \Else
        \For{$i=1,2,\ldots,k$}
            \State{$t' := (\underbrace{0,\ldots,0,1}_{i},0,\ldots,0)$}
            \If{$\Solve(\Reduce(\Theta,\phi(t'))) = \mbox{``No solution"}$}
                \State{$t := t'$}
                \Break
            \EndIf
        \EndFor
    \EndIf
    \If{$t = \varnothing$} \Return{``$0=0$''} 
    \EndIf
    \For{$i=1,2,\ldots,k$}
        \State{$b_{i}:=0$}
        \For{$a\in \mathbb Z_{q_{i}}\setminus\{t(i)\}$}
            \State{$t' := t$}
            \State{$t'(i):= a$}
            \If{$\Solve(\Reduce(\Theta,\phi(t'))) = \mbox{``Solution"}$}         \State{$b_{i} := 1/(a-t(i))$}
            \EndIf
        \EndFor
    \EndFor    
    \Return {``$b_{1}(y_{1}-t(1)) +\dots + b_{k}(y_{k}-t(k)) = 1$''}
\EndFunction
\end{algorithmic}
\end{algorithm}

First, we explain how 
the function $\FindOneEquationLinked$ works.
Suppose $V$ is an affine subspace of $\mathbb Z_{p}^{k}$ of dimension $k-1$, thus $V$ is the solution set of a linear equation
$c_1y_1 + \dots + c_k y_k = c_{0}$. Then the coefficients $c_0,c_{1},\dots,c_{k}$ can be learned (up to a multiplicative constant) by $(p\cdot k+1)$ queries of the form ``$(a_1,\ldots,a_k) \in V$?'' as follows.
First, we need at most $(k+1)$ queries to find a tuple $(t_{1},\ldots,t_{k})\notin V$. To do this we 
just check all tuples with 0s and at most one 1 (lines 4-11
of the pseudocode).
Then, to find this equation it is sufficient to check for every $a$ and every $i$
whether the tuple $(t_{1},\ldots,t_{i-1},a,t_{i+1},\ldots,t_{k})$ satisfies this equation (lines 13-19 of the pseudocode).
Here the query is performed by 
the reduction of all domains to the corresponding solution 
(the function $\Reduce$) and a recursive call of the main function 
$\Solve$.

As we said before, we may define a natural bijective mapping $\psi:\mathbb Z_{p_{1}}\times\dots\times \mathbb Z_{p_r}\to
L_{1}\times\dots\times L_{n}$, and 
assume that all relations from $\Theta_{L}$ and $Eq$ are 
systems of linear equations over $z_{1},\ldots,z_{r}$.
Below we explain how the function 
$\FindEquationsForNonlinked$ calculates the solution set of $\Theta'$ factorized by the congruences (the set $A'$) if $\Theta'$ is not linked. It describes the solution set by linear equations over $z_{1},\ldots,z_{r}$.

\begin{algorithm}
\begin{algorithmic}[1]
\Function{\FindEquationsForNonlinked}{$\Theta$}
\State{\textbf{Input:} CSP($\Gamma$) instance $\Theta$}
    \State{$I := \{1\}$}     \Comment{$I$ is the set of independent variables}
    \State{$E := \varnothing$}     \Comment{We start with an empty set of equations}    
   \For{$j=1,2,\ldots,r$}
    \State{$e := \FindOneEquationNonlinked(\Theta,I\cup\{j\})$}     
    \If{$e = \mbox{``$0=0$"}$}  \Comment{$j$-th variable is independent}
        \State{$I := I\cup \{j\}$}
    \ElsIf{$e = \mbox{``$0=1$"}$}
        \Return{$\mbox{``No solution"}$}
    \Else
        \State{$E := E\cup e$}     \Comment{Add the equation we found}    
    \EndIf
    \EndFor
    \Return {$E$}
\EndFunction
\end{algorithmic}
\end{algorithm}

We start with an empty set of equations $E$ and 
claim that the first variable is independent, by $I$ we denote the set of independent variables
(see the pseudocode of  $\FindEquationsForNonlinked$).
Assume that we already found all the equations 
over $z_{1},\ldots,z_{j-1}$, i.e. we described the projection 
of $A'$ onto $z_{1},\ldots,z_{j-1}$.

Then the projection of $A'$ onto 
the independent variables and the $j$-th variable 
is either full or of codimension 1.
Thus, we can 
learn this equation by queries of the form 
``Does there exist
$v\in A'$ such that 
$\proj_{I\cup\{j\}}(v) = (a_{1},\ldots,a_{h})$?''
in the same way as we did in $\FindOneEquationLinked$, 
but now we use $\FindOneEquationNonlinked$.
The only difference in these functions is how we check a query: in $\FindEquationsForNonlinked$
we use the function $\CheckTuple$ instead of 
$\Reduce$ and $\Solve$ (see the pseudocode). 

If the new equation was found and this equation is not trivial then 
we add this equation to $E$ and claim that $z_{j}$ is not independent.
If the equation we found is ``$0=0$'' then 
add $x_{j}$ to the set of independent variables and go to the next variable.

\begin{algorithm}
\begin{algorithmic}[1]
\Function{\FindOneEquationNonlinked}{$\Theta,I$}
\State{\textbf{Input:} CSP($\Gamma$) instance $\Theta$, 
$I = \{i_1,\ldots,i_{h}\}$ a set of variables}
    \State{$t := \varnothing$}     \Comment{We search for a tuple $t$ outside of the solution set}            
    \If{$\neg\CheckTuple(\Theta,I,(0,\ldots,0))$} 
        \State{$t:= (0,\ldots,0)$}
    \Else
        \For{$j=1,2,\ldots,h$}
            \State{$t' := (\underbrace{0,\ldots,0,1}_{j},0,\ldots,0)$}
            \If{$\neg\CheckTuple(\Theta,I,t')$} 
                \State{$t := t'$}
                \Break
            \EndIf
        \EndFor
    \EndIf
    \If{$t = \varnothing$} \Return{``$0=0$''} 
    \EndIf
    \For{$j=1,2,\ldots,h$}
        \State{$b_{j}:=0$}
        \For{$a\in \mathbb Z_{p_{i_{j}}}\setminus\{t(j)\}$}
            \State{$t' := t$}
            \State{$t'(j):= a$}
            \If{$\CheckTuple(\Theta,I,t')$} 
                \State{$b_{j} := 1/(a-t(j))$}
            \EndIf
        \EndFor
    \EndFor    
    \Return {``$b_{1}(z_{i_{1}}-t(1)) +\dots + b_{h}(z_{i_{h}}-t(r)) = 1$''}
\EndFunction
\end{algorithmic}
\end{algorithm}

It remains to explain how the function $\CheckTuple$ works.
As an input it takes an instance $\Theta$, a set of 
variables $I$, and a tuple $t$ of length $|I|$.
The restriction of the variables from $I$ to the 
tuple $t$ implies the restrictions $L_{1}',\ldots,L_{n}'$ of the domains 
$L_{1},\ldots,L_{n}$.
Put $D_{i}' = \bigcup\limits_{E\in L_{i}'} E$
for every $i$. 
Then we add unary constraints $x_{i}\in D_{i}'$ to $\Theta$ and solve the obtained instance by the function 
$\SolveNonlinked$, which works only for non-linked instances and will be defined 
in the next section.

\begin{algorithm}
\begin{algorithmic}[1]
\Function{\CheckTuple}{$\Theta,I,t$}
\State{\textbf{Input:} CSP($\Gamma$) instance $\Theta$, 
$I$ a subset of variables, $t$ a tuple of length $|I|$}
    \State{$R:=\{\alpha\in \mathbb Z_{p_{1}}\times\dots\times\mathbb Z_{p_{r}} \mid \proj_{I}(\alpha) = t\}$}    
    \Comment{We don't really calculate $R$}
    \For{$i=1,2,\ldots,n$}
        \State{$D_{i}':=\bigcup_{E\in\proj_{i}(\psi(R))} E$} \Comment{We calculate $D_{i}'$}
    \EndFor
    \If{$\SolveNonlinked(\Theta\wedge(x_{1}\in D_{1}')\wedge\dots\wedge(x_{n}\in D_{n}')) = \mbox{``Solution''}$}
        \Return true
    \Else 
        ~\Return false
    \EndIf
\EndFunction
\end{algorithmic}
\end{algorithm}

\subsection{Remaining functions}\label{AlgorithmTechnicalities}

In this subsection we define 
the functions 
$\CheckCycleConsistency$, 
$\CheckIrreducibility$,
and
$\CheckWeakerInstance$
which were used in Subsection~\ref{AlgorithmMainPart}, 
and 
function 
$\SolveNonlinked$ from Subsection~\ref{FindingLinearEquationSection}.

First, we define the function 
$\CheckCycleConsistency$. To check cycle-consistency it is sufficient to use constraint propagation providing a variant of (2,3)-consistency (see the pseudocode).
First, for every pair of variables $(x_{i},x_{j})$ we consider the intersections of projections of all constraints onto these variables.
The corresponding relation we denote by $\rho_{i,j}$.
Then, for every $i,j,k\in\{1,2,\ldots,n\}$
we replace
$\rho_{i,j}$ by $\rho_{i,j}'$
where $\rho_{i,j}'(x,y) = \exists  z \; \rho_{i,j}(x,y)\wedge \rho_{i,k}(x,z)\wedge \rho_{k,j}(z,y).$

We repeat this procedure while we can change some $\rho_{i,j}$.
If in the end we get a relation $\rho_{i,j}$ that is not subdirect in $D_{i}\times D_{j}$, 
then we can either reduce $D_{i}$ or $D_{j}$, or,
 if $\rho_{i,j}$ is empty, state that there are no solutions.
If every relation $\rho_{i,j}$ is subdirect in
$D_{i}\times D_{j}$, then we claim (see Lemma \ref{ProofCycleConsistencyFunction}) that the original CSP instance is cycle-consistent.

\begin{algorithm}
\begin{algorithmic}[1]
\Function{\CheckCycleConsistency}{$\Theta$}
\State{\textbf{Input:} CSP($\Gamma$) instance $\Theta$}
    \For{$i,j\in\{1,2,\ldots,n\}$} 
    \Comment{Calculate binary projections $\rho_{i,j}$}
        \State{$\rho_{i,j} := D_{i}\times D_{j}$}
        \For{$C\in\Theta$}        
            \State{$\rho_{i,j} :=\rho_{i,j}\cap \proj_{x_i,x_j} C$}
            \Comment{$\proj_{x_i,x_j} C$ is the projection 
            of $C$ onto $x_{i},x_{j}$}
        \EndFor
    \EndFor
    \Repeat \Comment{Propagate constraints to reduce  $\rho_{i,j}$}
        \State{$\Changed:= false$}
        \For{$i,j,k\in\{1,2,\ldots,n\}$} 
            \State{$\rho_{i,j}'(x,y) := \exists z\;
            \rho_{i,j}(x,y)\wedge \rho_{i,k}(x,z)\wedge 
            \rho_{k,j}(z,y)$}
            \If{$\rho_{i,j} \neq \rho_{i,j}'$}
                \State{$\rho_{i,j}:=\rho_{i,j}'$}
                \State{$\Changed:= true$}               
            \EndIf
        \EndFor
    \Until{$\neg\Changed$} \Comment{We cannot reduce $\rho_{i,j}$ anymore}
    \For{$i,j\in\{1,2,\ldots,n\}$} 
        \If{$\rho_{i,j}=\varnothing$}
            \Return{``No solution"}
        \EndIf
        \If{$\proj_{1}(\rho_{i,j})\neq D_{i}$}
            \Return{$(x_{i},\proj_{1}(\rho_{i,j}))$}
        \EndIf
        \If{$\proj_{2}(\rho_{i,j})\neq D_{j}$}
            \Return{$(x_{j},\proj_{2}(\rho_{i,j}))$}
        \EndIf

    \EndFor
    \Return{\mbox{``Ok"}}
\EndFunction
\end{algorithmic}
\end{algorithm}

Let us explain how 
$\CheckIrreducibility$ works.
For every $k\in\{1,2,\ldots,n\}$ and every maximal
congruence $\sigma_{k}$ on $D_{k}$ we do the following.
We start with the partition $\sigma_{k}$ of the $k$-th variable, 
so we put $I=\{k\}$ (line 4 of the pseudocode), which is the set of variables with a partition.
Then we try to extend the partition of $D_{k}$ to other domains.
We choose a constraint having $x_{k}$ in the scope, choose another variable 
$x_{j}$, 
and consider the projection of $C$ onto $x_{k},x_{j}$, which we denote by $\delta$.
Since $\sigma_{k}$ is maximal, we may have two possibilities:
either all equivalence classes of $\sigma_{k}$ are connected in $\delta$, 
or none of the equivalence classes are connected in $\delta$.
In the second case the partition of $D_{k}$ generates a partition of $D_{j}$ with the same number of classes, and we add $j$ to $I$ (lines 10-15 of the pseudocode).

We continue this procedure while we can add new variables to $I$. 
As a result we get a set $I$ and
a partition of $D_{i}$ for every $i\in I$.
Put $\mathbf{X'} = \{x_{i}\mid i\in I\}$.
Then, the projection of $\Theta$ onto $\mathbf{X'}$ 
can be split into several instances on smaller domains, 
and each of them can be solved using recursion.
Thus, we can check whether the solution set of the projection of the instance onto $\mathbf{X'}$ is subdirect or empty.
If it is empty then we state that there are no solutions.
If it is not subdirect, then we can reduce the corresponding domain.
If it is subdirect, then we
go to the next $k\in\{1,2,\ldots,n\}$ and the next maximal
congruence $\sigma_{k}$ on $D_{k}$, and repeat the procedure.
If for all $k$ and all maximal congruences 
the solution set of the obtained instance is subdirect, 
then the instance is irreducible (see Lemma~\ref{CheckIrreducibilityCorrectness}).

\begin{algorithm}
\begin{algorithmic}[1]
\Function{\CheckIrreducibility}{$\Theta$}
\State{\textbf{Input:} CSP($\Gamma$) instance $\Theta$}
    \For{$k=1,\ldots,n$}
        \For{$\sigma_{k}=\{E_{k}^{1},\ldots,E_{k}^{t}\}$ is a maximal congruence on $D_{k}$}
            \State{$I:=\{k\}$}
            \Repeat
                \State{$\Changed:=false$}
                \For{$C\in\Theta$, $i\in I$, $j\notin I$ such that $x_{i}$ and $x_{j}$ are in the scope of $C$}
                    \State{$\delta:=\proj_{x_{i},x_{j}} C$} 
                    \Comment{$\proj_{x_{i},x_{j}} C$ is the projection of $C$ onto $x_{i},x_{j}$}
                    \For{$u=1,2,\dots,t$}  
                        \Comment{Calculate the partition on $D_{j}$}
                            \State{$E_{j}^{u}:= \{b\in D_{j}\mid \exists a\in E_{i}^{u}: (a,b)\in \delta\}$}
                    \EndFor                    
                    \If{$E_{j}^{1},\dots,E_{j}^{t}$ are disjoint}
                        \State{$I:=I\cup\{j\}$}
                        \State{$\Changed:=true$}
                        \Break
                    \EndIf
                \EndFor
            \Until{$\neg\Changed$}
            \For{$i\in I$}
                \State{$D_{i}':=\varnothing$}
                \For{$a\in D_{i}$}
                    \State{Choose $u$ such that $a\in E_{i}^{u}$}
                    \For{$j=1,2,\ldots,n$}
                        \If{$j=i$}
                            \State{$E_{j} := \{a\}$}
                        \ElsIf{$j\in I$}
                            \State{$E_{j} := E_{j}^{u}$}
                        \Else
                            \State{$E_{j} := D_{j}$}
                        \EndIf                        
                    \EndFor
                    \State{$\mathbf{X'} := \{x_{i}\mid i\in I\}$}                        
                    \If{$\Solve(\proj_{\mathbf{X'}}(\Reduce(\Theta,(E_{1},\ldots,E_{n}))))  = \mbox{``Solution"}$}
                        \State{$D_{i}':=D_{i}'\cup\{a\}$}
                    \EndIf
                \EndFor
                \If{$D_{i}'=\varnothing$}
                    \Return{``No solution"}
                \ElsIf{$D_{i}'\neq D_{i}$} 
                    \Return{$(x_{i},D_{i}')$} 
                \EndIf
            \EndFor
        \EndFor
    \EndFor
    \Return{\mbox{``Ok"}}
\EndFunction
\end{algorithmic}
\end{algorithm}

Define the function
$\CheckWeakerInstance$, which checks that 
if we simultaneously weaken every constraint then the solution set of the obtained instance is subdirect.
Thus, we weaken every constraint of $\Theta$ 
(function $\WeakenEveryConstraint$ in the pseudocode), 
that is, we make a copy of $\Theta$, and replace each constraint by all weaker constraints without dummy variables. Recursively calling the algorithm, check that the obtained instance has a solution
with $x_{i}=b$ for every $i\in\{1,2,\ldots,n\}$ and $b\in D_{i}$.
If not, reduce $D_{i}$ to the projection onto $x_{i}$ of the solution set of the obtained instance. Otherwise, go on.

\begin{algorithm}
\begin{algorithmic}[1]
\Function{\CheckWeakerInstance}{$\Theta$}
  \State{\textbf{Input:} CSP($\Gamma$) instance $\Theta$}
    \State{$\Theta'= \WeakenEveryConstraint(\Theta)$}
    \For{$i=1,\ldots,n$}
        \State{$D_{i}':=\varnothing$}
        \For{$a\in D_{i}$}
            \State{$\Output := \Solve(\Reduce(\Theta',(D_{1},\dots,D_{i-1},
            \{a\},D_{i+1},\dots,D_{n})))$}
            \If{$\Output = \mbox{``Solution"}$}
                \State{$D_{i}':=D_{i}'\cup \{a\}$}
            \EndIf
        \EndFor
        \If{$D_{i}'=\varnothing$}
            \Return{``No solution"}
        \ElsIf{$D_{i}'\neq D_{i}$} 
            \Return{$(x_{i},D_{i}')$} 
        \EndIf
    \EndFor
    \Return{``Ok"}    
\EndFunction
\end{algorithmic}
\end{algorithm}

It remains to define the function
$\SolveNonlinked$, which solves an instance that is not linked and not fragmented (see the pseudocode).
Such an instance can be split into several instances on smaller domains.
First, we consider the set $\mathbf X'$ of all variables appearing
in the constraints of the instance
and take the projection of the instance onto $\mathbf{X'}$.
Then we consider each linked component, that is, elements that can be connected by a path in the instance.
Since the instance is cycle-consistent, 
the division into linked components 
defines a congruence on every domain 
(see Lemma~\ref{LinkedConIsCon}),
and each block of this congruence is a subuniverse 
of the domain.
Thus, each linked component can be viewed as a CSP instance in a constraint language $\Gamma$ on smaller domains, 
which can be solved using the recursion.
If at least one of them has a solution, then the 
original instance has a solution. 

\begin{algorithm}
\begin{algorithmic}[1]
\Function{\SolveNonlinked}{$\Theta$}
  \State{\textbf{Input:} CSP($\Gamma$) instance $\Theta$}
    \State{$\mathbf{X'}:=\Var(\Theta)$}
    \Comment{Choose variables that appear in $\Theta$}
    \State{$\Theta':=\proj_{\mathbf{X'}}(\Theta)$}
    \Comment{Remove variables that never occur}
    \For{a linked component $(D_{1}',\dots,D_{n'}')$ of $\Theta'$} 
        \If{$\Solve(\Reduce(\Theta',(D_{1}',\dots,D_{n'}'))) = \mbox{``Solution"}$}
            \Return{``Solution"}
        \EndIf
    \EndFor
    \Return{``No solution"}    
\EndFunction
\end{algorithmic}
\end{algorithm}


%% file: proof_ismvl.tex
\section{Correctness of the Algorithm}\label{CorretnessSection}
\subsection{Rosenberg completeness theorem}\label{RosenbergSection}
The main idea of the algorithm is based on a beautiful result
obtained by Ivo Rosenberg in 1970, who found all maximal clones on a finite set.
Applying this result to the clone generated by a WNU together with all constant operations,
we can show that every algebra with a WNU operation has a nontrivial binary absorbing subuniverse, or
a nontrivial center, or it is polynomially complete or linear modulo some proper congruence.

\begin{thm}\label{NextReduction}
Suppose $\mathbf{A} = (A;w)$ is a finite algebra, where $w$ is a special WNU of arity $m$.
Then one of the following conditions holds:
\begin{enumerate}
\item there exists a nontrivial binary absorbing subuniverse $B\subsetneq A$,
\item there exists a nontrivial center $C\subsetneq A$,
\item there exists a proper congruence $\sigma$ on $A$ such that
$(A;w)/\sigma$ is polynomially complete,
\item there exists a proper congruence $\sigma$ on $A$ such that
$(A;w)/\sigma$ is isomorphic to $(\mathbb Z_{p};x_{1}+\dots +x_{m})$.
\end{enumerate}
\end{thm}
\begin{proof}
Let us prove this statement by induction on the size of $A$.
If we have a nontrivial binary absorbing subuniverse in $A$ then there is nothing to prove.
Assume that $A$ has no nontrivial binary absorbing subuniverse.
Let $M$ be the clone generated by $w$ and all constant operations on $A$.
If $M$ is the clone of all operations, then
$(A;w)$ is polynomially complete.

Otherwise, by Rosenberg's Theorem \cite{rosmax}, $M$ belongs to one of the following maximal clones. 
\begin{enumerate}
\item Maximal clone of monotone operations, that is, the clone of operations preserving a partial order
relation with the greatest and the least element;

\item Maximal clone of autodual operations, that is, the clone of operations preserving the graph of a permutation of a prime order
without a fixed element;

\item Maximal clone defined by an equivalence relation;

\item Maximal clone of quasi-linear operations;

\item Maximal clone defined by a central relation;

\item Maximal clone defined by an $h$-regularly generated (or $h$-universal) relation.
\end{enumerate}

Let us consider all the cases.
\begin{enumerate}
\item 
As we assumed, there is no nontrivial binary absorbing subuniverse on $A$.
Hence, the least element of the partial order can be viewed as a center by letting $\mathbf B = \mathbf A$ 
and using the partial order relation as a subdirect subuniverse of $\mathbf{A}\times \mathbf{B}$ (the least element is connected 
with all other elements in the partial order relation).
Thus, we have a nontrivial center in $A$.
\item Constants are not autodual operations. This case cannot happen.
\item Let $\delta$ be a maximal congruence on $\mathbf{A}$.
We consider a factor algebra
$(A;w)/\delta$ and apply the inductive assumption.
\begin{enumerate}
\item
If $\mathbf{A}/\delta$ has a binary absorbing subuniverse
$B'\subseteq A/\delta$, then 
$\bigcup_{E\in B'}E$ is a binary absorbing subuniverse of $A$ with the same term operation.
\item If $\mathbf{A}/\delta$ has a nontrivial center $C'\subseteq A/\delta$ witnessed by a subdirect relation 
$R'\subseteq A/\delta\times B$, then
$\bigcup_{E\in C'}E$ is a nontrivial center of $A$ 
witnessed by 
$R = \bigcup_{(E,b)\in R'} E\times \{b\}$.

\item Suppose $(\mathbf{A}/\delta)/\sigma $ is polynomially complete.
Since $\delta$ is a maximal congruence, $\sigma$ is the equality relation
and $\mathbf{A}/\delta$ is polynomially complete.
\item Suppose $(\mathbf{A}/\delta)/\sigma $ is isomorphic to $(\mathbb Z_{p};x_{1}+\dots +x_{m})$.
Since $\delta$ is a maximal congruence, $\sigma$ is the equality relation
and $\mathbf{A}/\delta$ is isomorphic to $(\mathbb Z_{p};x_{1}+\dots +x_{m})$.
\end{enumerate}
\item By Lemma 6.4 from \cite{KeyRelations},
we know that $w(x_{1},\ldots,x_{m}) = x_{1}+\dots +x_{m}$, where $+$ is the operation in an abelian group.
We assume that $\mathbf{A}$ has no nontrivial congruences, otherwise we refer to case (3).
Then the algebra $\mathbf{A}$ is simple and isomorphic to $(\mathbb Z_{p};x_{1}+\dots +x_{m})$ for a prime number $p$.
\item Let $\rho$ be a central relation of arity $k$ preserved by $w$.
It is not hard to see that
the existence of a nontrivial binary absorbing subuniverse on
$\underbrace{\mathbf{A}\times\dots\times\mathbf{A}}_{k-1}$
implies the existence of a nontrivial binary absorbing subuniverse on $\mathbf{A}$
(see Lemma~\ref{GenBinAbToBinAb}).
Since there is no nontrivial binary absorbing subuniverse 
on $\mathbf{A}$ and 
the relation $\rho$ contains all tuples 
$(b_{1},\ldots,b_{k})$ such that $b_{1}$ is from the center of $\rho$, the center of $\rho$ is a center of $A$
by letting $\mathbf B = \underbrace{\mathbf{A}\times\dots\times\mathbf{A}}_{k-1}$.
\item By Corollary 5.10 from \cite{KeyRelations} this case cannot happen.
\end{enumerate}
\end{proof}

\subsection{The algorithm is polynomial}

\begin{lem}\label{RecursionDepth}
The depth of the recursion in the algorithm is less than $|A|+|\Gamma|$.
\end{lem}

\begin{proof}
We use the recursion in the functions
$\SolveLinearCase,
\FindOneEquationLinked,$
$\CheckAllTuples, 
\CheckIrreducibility,
\CheckWeakerInstance,
\SolveNonlinked.
$
In each
of them but 
$\CheckWeakerInstance$
we reduce all domains of size greater than 1
before using the recursion
and we never increase the domain.
Therefore, every path in the recursion tree 
contains at most $|A|$ calls of 
the function $\Solve$ in the above functions.

Let us consider the function
$\CheckWeakerInstance$.
First, we introduce a partial order on the set of relations in $\Gamma$.
We say that $\rho_{1}\le\rho_{2}$ if one of the following conditions hold
\begin{enumerate}
\item the arity of $\rho_{1}$ is less than the arity of $\rho_{2}$.
\item the arities of $\rho_{1}$ and $\rho_{2}$
are equal,
$\proj_{i}(\rho_{1})\subseteq \proj_{i}(\rho_{2})$ for every $i$,
$\proj_{j}(\rho_{1})\neq \proj_{j}(\rho_{2})$ for some $j$.
\item the arities of $\rho_{1}$ and $\rho_{2}$
are equal,
$\proj_{i}(\rho_{1})= \proj_{i}(\rho_{2})$ for every $i$,
and $\rho_{1}\supseteq \rho_{2}$.
\end{enumerate}
We can check that in the algorithm 
we never make any relation bigger,
and every time 
we use recursion in $\CheckWeakerInstance$
we make every constraint relation strictly smaller.
Since our constraint language $\Gamma$ is finite,
every path in the recursion tree 
contains at most $|\Gamma|$ calls of 
the function $\Solve$ in $\CheckWeakerInstance$.
Therefore the depth of
the recursion tree is bounded by 
$|A|+|\Gamma|$.
\end{proof}

\begin{cons}
The algorithm is polynomial.
\end{cons}
\begin{proof}
Since the depth of the recursive tree 
is bounded by $|A|+|\Gamma|$, 
it remains to show that each loop in each function is polynomial.

In the function $\Solve$ we 
go through the loop at most 
$n\cdot |A|$ times, which is polynomially many.

In the function 
$\SolveLinearCase$ 
we go through the external \textbf{repeat} loop 
at most $r$ times, where $r$ is the dimension 
of 
$L_{1}\times\dots\times L_{n}$. 
Therefore, $r$ is bounded 
by $|A|\cdot n$.
We go through the inner \textbf{repeat} loop 
at most 
$|\Gamma|\cdot N$ times, 
where $N$ is the number of constraints of the instance.

In the function 
$\CheckCycleConsistency$
we go through the \textbf{repeat} loop 
at most $|\Gamma|\cdot n^{2}$ times, 
because every time we change at least one relation $\rho_{i,j}$, which is from $\Gamma$, 
and we have $n^{2}$ of them.

In the function 
$\CheckIrreducibility$
we go through the \textbf{repeat} loop 
at most $n$ times, since we always add an element to $I$.

All other loops are \textbf{for} loops,
and polynomial bounds for them 
follow from the description of the algorithm.
Therefore, the algorithm is polynomial.
\end{proof}
\subsection{Correctness of the auxiliary functions}

\begin{lem}\label{ProofCycleConsistencyFunction}
If the function 
$\CheckCycleConsistency$ returns 
\mbox{``Ok"} then the instance is cycle-consistent, 
if it returns \mbox{``No solution"}
then the instance has no solutions, 
if it returns $(x_{i},D)$
then any solution of the instance 
has $x_{i}\in D$.
\end{lem}
\begin{proof}
Assume that the function returned 
\mbox{``Ok"}. 
Since every relation $\rho_{i,j}$ in the end of the algorithm is 
subdirect, 
the instance is 1-consistent.
Consider a path
$x_{i_1}-C_{1}-x_{i_2}-\dots-x_{i_{l-1}}-C_{l-1}-
x_{i_l}$
starting and ending 
with $x_{i_{1}}=x_{i_l}$. 
Since the projection 
of $C_{j}$ onto $x_{i_{j}},x_{i_{j+1}}$ 
contains $\rho_{i_{j},i_{j+1}}$ for every $j$, 
to show that the instance is cycle-consistent, 
it is sufficient to prove 
that the formula
$$\delta(x_{i_1}) = \exists x_{i_{2}}\dots\exists 
x_{i_{l-1}}\;\rho_{i_{1},i_{2}}(x_{i_1},x_{i_2})\wedge\dots
\wedge\rho_{i_{l-1},i_{l}}(x_{i_{l-1}},x_{i_l})
$$
defines $D_{i_1}$.
This follows from the fact that 
we terminated the function when for all $i,j,k$ 
$$\rho_{i,j}(x,y) = \exists z\;            \rho_{i,j}(x,y)\wedge \rho_{i,k}(x,z)\wedge \rho_{k,j}(z,y).$$

The remaining part follows from the fact that 
all the constraints 
$\rho_{i,j}(x_{i},x_{j})$ were derived from the original constraints,
and therefore they should hold for any solution.
\end{proof}

\begin{lem}\label{CheckIrreducibilityCorrectness}
If the function 
$\CheckIrreducibility$ returns 
\mbox{``Ok"} then the instance is irreducible, 
if it returns \mbox{``No solution"}
then the instance has no solutions, 
if it returns $(x_{i},D)$
then any solution of the instance 
has $x_{i}\in D$.
\end{lem}
\begin{proof}
Assume that $\CheckIrreducibility$ returned \mbox{``Ok"}
but the instance is not irreducible.
Then, there exists an instance 
$\Theta'$ such that every constraint of $\Theta'$ is a projection of a constraint from 
the original instance $\Theta$ on some set of variables,
and $\Theta'$ is not fragmented, not linked, and its solution set is not subdirect.
Let $\mathbf{X'}$ be the set of all variables occurring in $\Theta'$. 
Choose a variable $x_{k}\in \mathbf{X'}$. 
If we consider the set of all pairs $(a,b)\in D_{k}^{2}$ such that $a$ and $b$ 
can be connected by a path in $\Theta'$ then we get a congruence (see Lemma~\ref{LinkedConIsCon}).
Since $\Theta'$ is not linked, there should be a maximal congruence 
$\sigma_{k}$ containing the congruence. 
This congruence was chosen in the line 4
of the pseudocode.

Since $\Theta'$ is not fragmented, there exists a path in $\Theta'$ from $x_{k}$ to any other variable 
from $\mathbf{X'}$. Following this path we can always define a partition 
on the next variable using the partition on the previous one.
Since every constraint of $\Theta'$ is a projection 
of a constraint from $\Theta$, we could define the same partitions on $\Theta$ 
(see the pseudocode of the function).
We just need to show that 
on every domain $D_{i}$ we can generate a unique partition using $\sigma_{k}$ 
(the order in which we add elements to $I$ and the way how we choose constraints is not important).
Consider two paths from $x_{k}$ to $x_{i}$ defining two partitions.
We glue together the beginnings of these paths and get a path from $x_{i}$ to $x_{i}$
connecting these partitions.
Since the instance is cycle-consistent, these partitions should be equal.
Thus, we showed that starting from the congruence $\sigma_{k}$ (in the pseudocode)
we get a unique partition on every variable $x_{i}\in\mathbf{X'}$. 
Therefore, we actually checked in the algorithm that 
the solution set of $\Theta'$ is subdirect, which gives us a contradiction.
Hence, $\Theta$ is irreducible.

The remaining part follows from the fact that 
$D_{i}'$ is the set of all possible evaluations of $x_{i}$ in  solutions of a weaker instance.
\end{proof}

\subsection{Main theorems without a proof}


To explain the correctness of the algorithm in Section~\ref{Algorithm}
we used the following main facts, which  
will be proved in Section~\ref{MainProofs}.

\begin{thm}\label{AbsorptionCenterStep}
Suppose $\Theta$ is a cycle-consistent irreducible CSP instance, and 
$B$ is a nontrivial binary absorbing subuniverse or a nontrivial center of $D_{i}$.
Then $\Theta$ has a solution if and only if
$\Theta$ has a solution with $x_{i}\in B$.
\end{thm}


\begin{thm}\label{PCStepThm}
Suppose $\Theta$ is a cycle-consistent irreducible CSP instance,
there does not exist a nontrivial binary absorbing subuniverse or a nontrivial center on $D_{j}$
for every $j$,
$(D_{i};w)/\sigma$ is a polynomially complete algebra, 
and
$E$ is an equivalence class of $\sigma$.
Then $\Theta$ has a solution if and only if
$\Theta$ has a solution with $x_{i}\in E$.
\end{thm}



\begin{thm}\label{LinearStep}
Suppose the following conditions hold:
\begin{enumerate}
\item $\Theta$ is a linked cycle-consistent irreducible CSP instance with domain set
$(D_{1},\ldots,D_{n})$;
\item there does not exist a nontrivial binary absorbing subuniverse or a nontrivial center on $D_{j}$ for every $j$;
\item if we replace every constraint of $\Theta$ by all weaker constraints then the obtained instance
has a solution with $x_{i} = b$ for every $i$ and $b\in D_{i}$ (the obtained instance has a subdirect solution set);
\item $L_{i} = D_{i}/\sigma_{i}$ for every $i$, where $\sigma_{i}$ is the minimal linear congruence on $D_{i}$;
\item $\phi:\mathbb Z_{q_{1}}\times \dots \times \mathbb Z_{q_{k}}
\to L_{1}\times\dots\times L_{n}$ is a homomorphism,
where $q_{1},\dots,q_{k}$ are prime numbers;
\item if we replace any constraint of 
$\Theta$ by all weaker constraints then for every $(a_{1},\ldots,a_{k})\in \mathbb Z_{q_{1}}\times \dots \times \mathbb Z_{q_{k}}$ 
there exists a solution of the obtained instance in 
$\phi(a_{1},\ldots,a_{k})$.
\end{enumerate}
Then 
$\{(a_{1},\dots,a_{k})\mid \Theta \text{ has a solution in }\phi(a_1,\dots,a_{k})\}$ is
either empty, or is full, or is an affine subspace of $\mathbb Z_{q_{1}}\times \dots \times \mathbb Z_{q_{k}}$ of codimension 1 (the solution set of a single linear equation).
\end{thm}




%% file: Definitions.tex
\section{The Remaining Definitions}\label{DefinitionSection}

\subsection{Variety of algebras}
We consider the variety of all algebras $\mathbf A = (A;w)$
such that $w$ is a special WNU operation of arity $m$.
As it was mentioned in Section~\ref{Definition}
every domain $D$ will be viewed as a finite algebra $(D;w)$ from this variety.
Note that in the remainder of this paper any claim or assumption ``$\rho$ is a relation''
 should be understood as 
``$\rho$ is a subalgebra of
$\mathbf A_{1}\times\dots\times \mathbf A_{n}$'' for the corresponding finite algebras
$\mathbf A_{1},\ldots,\mathbf A_{n}$ from this variety.

\subsection{Additional notations}

For a relation 
$\rho\subseteq A_{1}\times\dots\times A_{n}$ and 
a congruence $\sigma$ on $A_{i}$, 
we say that the $i$-th variable of the relation $\rho$ is 
\emph{stable under $\sigma$}
if $(a_{1},\ldots,a_{n})\in\rho$ and $(a_{i},b_{i})\in\sigma$
imply
$(a_{1},\ldots,a_{i-1},b_{i},a_{i+1},\ldots,a_{n})\in\rho$.
%
We say that a relation is \emph{stable under} $\sigma$ if every variable of this relation is stable under $\sigma$.


We say that a congruence $\sigma$ is \emph{irreducible} if
it is proper and it cannot be represented as an intersection of other binary relations $\delta_{1},\ldots,\delta_{s}$ stable under $\sigma$.
For an irreducible congruence $\sigma$ on a set $A$
by $\cover{\sigma}$ we denote the minimal binary relation $\delta\supsetneq \sigma$ stable under $\sigma$.

For a relation $\rho$ by $\ConOne(\rho,i)$
we denote the binary relation $\sigma(y,y')$ defined by
$$\exists x_{1}\dots\exists x_{i-1}\exists x_{i+1}\dots\exists x_{n}\;\rho(x_{1},\ldots,x_{i-1},y,x_{i+1},\ldots,x_{n})\wedge
\rho(x_{1},\ldots,x_{i-1},y',x_{i+1},\ldots,x_{n}).$$
For a constraint $C = \rho(x_{1},\ldots,x_{n})$
by $\ConOne(C,x_{i})$ we denote $\ConOne(\rho,i)$.
For a set of constraints $\Omega$ by
$\Congruences(\Omega,x)$ we denote
the set $\{\ConOne(C,x)\mid C\in \Omega\}$.

A congruence $\sigma$ on $\mathbf{A}$ is called 
\emph{a PC congruence} if 
$\mathbf{A}/\sigma$ is a PC algebra without a nontrivial binary absorbing subuniverse or center.
For an algebra $\mathbf A$ by $\PCCon(\mathbf A)$ we denote the intersection of
all PC congruences.
A subuniverse $A'\subseteq A$ is called a \emph{PC subuniverse}
if $A' = E_{1}\cap\dots\cap E_{s}$,
where each $E_{i}$ is an equivalence class of a PC congruence.
Note that a PC subuniverse can be empty or full. 

A congruence $\sigma$ on $\mathbf{A}$ is called 
\emph{linear} if 
$\mathbf{A}/\sigma$ is a linear algebra.
For an algebra $\mathbf A$ by $\LinCon(\mathbf A)$ we denote the minimal linear congruence.
A subuniverse of $\mathbf A$ is called a \emph{linear subuniverse}
if it is stable under $\LinCon(\mathbf A)$.
Note that we could not define a PC subuniverse in the same way because not every subuniverse stable under $\PCCon(\mathbf A)$ is a PC subuniverse of $\mathbf A$
(see Subsection \ref{PCSubsection}).

A subuniverse $B\subseteq A$ is called \emph{a one-of-four subuniverse}
if it is a binary absorbing subuniverse, 
a center, a PC subuniverse, or a linear subuniverse. 
We say that $B$ is a one-of-four subuniverse of \emph{absorbing type}, 
\emph{central type}, \emph{PC type}, or \emph{linear type}, respectively.
A subuniverse of type $\mathcal T$
is called \emph{minimal} if it is 
a minimal nontrivial 
subuniverse of this type.
Note that 
a minimal PC/linear subuniverse is 
a block of $\PCCon(\mathbf A)$/$\ConLin(\mathbf A)$.

\subsection{pp-formula, subconstraint, coverings}

Every variable $x$ appearing in the paper has its domain, which we denote by $D_{x}$.
In the paper we usually identify a CSP instance and a set of constraints.
For an instance $\Omega$ by $\Var(\Omega)$ we denote the set of all variables occurring in 
constraints of $\Omega$ (the set of all variables $\mathbf{X}$ is not important, all the properties
of the instance depend only on the variables that actually occur in the instance).
For an instance $\Omega$ and two sets of variables
$x_{1},\ldots,x_{n}$ and $y_{1},\ldots,y_{n}$
by $\Omega_{x_{1},\ldots,x_{n}}^{y_{1},\ldots,y_{n}}$ we denote the instance
obtained from $\Omega$ by replacement of every variable $x_{i}$ by $y_{i}$.

Sometimes we write an instance $\{C_{1},\ldots,C_{n}\}$ 
as a conjunctive formula $C_{1}\wedge\dots\wedge C_{n}$.
We say that an instance is \emph{a tree-formula} if there is no a path
$z_{1}-C_{1}-z_{2}-\dots -z_{l-1}-C_{l-1}-z_{l}$
such that $l\ge 3$, $z_{1} = z_{l}$, and all the constraints $C_{1},\ldots,C_{l-1}$ are different.

An expression 
$\exists y_{1}\dots\exists y_{s}\; (C_{1}\wedge \dots\wedge C_{n})$
is called \emph{a positive primitive formula (pp-formula)}.
To simplify, 
we use a notation $\Omega(x_{1},\ldots,x_{n})$
to write the pp-formula
$\exists y_1\dots \exists y_s \Omega$, 
where $\Omega$ is an instance (or a conjunction of constraints)
and $y_1,\ldots,y_s$ are all variables occurring in $\Omega$ except for
$x_{1},\dots,x_{n}$.
Then, we say that a pp-formula $\Omega(x_{1},\ldots,x_{n})$ defines a relation $\rho$
if $\rho(x_{1},\ldots,x_{n}) = \exists y_{1}\dots\exists y_{s}\; \Omega$.
Sometimes, if it is convenient, we write 
$\Omega(x_{1},\dots,x_{n})$ meaning the relation defined by 
the pp-formula.
A pp-formula $\Omega(x_{1},\ldots,x_{n})$ is called a \emph{subconstraint of $\Theta$}
if $\Omega\subseteq \Theta$, and $\Omega$ and $\Theta\setminus\Omega$ do not have common variables except for $x_{1},\ldots,x_{n}$.
Note that all relations that can be defined by a pp-formula 
are preserved by the WNU (see \cite{geiger1968closed,bond1,bond2}).

For a formula $\Omega$ by $\ExpShort(\Omega)$ we denote the set of all formulas $\Omega'$
such that there exists a mapping $S:\Var(\Omega')\to\Var(\Omega)$
satisfying the following conditions:
\begin{enumerate}
\item the domain of any variable $x$ from $\Omega'$ is equal to
the domain of $S(x)$ in $\Omega$;
\item for every constraint $((x_{1},\ldots,x_{n});\rho)$ of $\Omega'$, 
$((S(x_{1}),\ldots,S(x_{n}));\rho)$
is a constraint of $\Omega$;
\item
if a variable $x$ appears in both $\Omega$ and $\Omega'$ then $S(x) = x$.
\end{enumerate}

Similarly, by $\Expanded(\Omega)$ (\emph{Expanded Coverings}) we denote the set of all formulas $\Omega'$
such that there exists a mapping $S:\Var(\Omega')\to\Var(\Omega)$
satisfying the following conditions:
\begin{enumerate}
\item the domain of any variable $x$ from $\Omega'$ is equal to
the domain of $S(x)$ in $\Omega$;
\item for every constraint $((x_{1},\ldots,x_{n});\rho)$ of $\Omega'$
either the variables $S(x_{1}),\ldots,S(x_{n})$ are different and the constraint $((S(x_{1}),\ldots,S(x_{n}));\rho)$ is weaker or equivalent to some constraint of $\Omega$,
or $S(x_{1}) = \dots = S(x_{n})$ and $\{(a,a,\ldots,a)\mid a\in D_{x_{1}}\}\subseteq\rho$;
\item
if a variable $x$ appears in both $\Omega$ and $\Omega'$ then $S(x) = x$.
\end{enumerate}

For a variable $x$ we say that $S(x)$ is \emph{the parent of x}.

The following easy facts about coverings 
can be derived from the definition.
\begin{enumerate}
    \item every time we replace some constraints by weaker constraints we get an expanded covering of the original instance;
    \item any solution of the original instance can be naturally expanded to a solution of a covering (expanded covering);
    \item suppose $\Omega$ is a covering (expanded covering) of a 1-consistent instance and $\Omega$ is a tree-formula,
    then the solution set of
    $\Omega$ is subdirect;
    \item the union (union of all constraints) of two coverings (expanded coverings)  is 
    also 
    a covering (expanded covering);
    \item a covering (expanded covering) of a covering (expanded covering) is a covering (expanded covering).
\end{enumerate}

Another important property is formulated in the following lemma.

\begin{lem}\label{ExpandedConsistencyLemma}
Suppose $\Theta$ is a cycle-consistent irreducible CSP instance
and $\Theta'\in\Expanded(\Theta)$.
Then $\Theta'$ is cycle-consistent and irreducible.
\end{lem}

\begin{proof}
Let us prove that  $\Theta'$ is cycle-consistent. Consider a path in $\Theta'$ starting and 
ending with $z$.
Since $\Theta'$ is an expanded covering, 
for every constraint of $\Theta'$ either there exists a corresponding constraint in $\Theta$, or 
this constraint is reflexive (contains all tuples $(a,a,\ldots,a)$). 
Thus, to transform the path in $\Theta'$
to a path in $\Theta$ it is sufficient to replace every variable $x$ in the path by 
$S(x)$ (from the definition of expanded coverings), 
remove all reflexive constraints, and replace the remaining constraints by the corresponding constraints from $\Theta$.
Since $\Theta$ is cycle-consistent, 
the obtained path connects $a$ with $a$ for any $a\in D_{z}$. Since constraints in the path in $\Theta'$
are weaker or equivalent to constraints 
in the path in $\Theta$ and relations we removed are reflexive,
the path in $\Theta'$ also connects $a$ with $a$ 
for every $a\in D_{z}$.

Let us show that $\Theta'$ is irreducible.
Assume the converse, then there exists an instance $\Omega'$ consisting of projections of constraints from $\Theta'$
that is not linked, not fragmented, and its solution set is not subdirect.
By $\Omega$ we denote the set of corresponding projections of constraints from $\Theta$ corresponding to the constraints of $\Omega'$ (we ignore reflexive constraints from $\Omega'$).
To be more accurate, suppose a constraint $C''\in\Omega'$ 
is equal to  $\proj_{\mathbf X}(C')$ 
for a constraint $C'\in\Theta'$ and a set of variable 
$\mathbf X$, and 
$C'$ is weaker or equivalent to a constraint $C\in\Theta$.
Then we add the constraint $\proj_{S(\mathbf X)}(C)$ to $\Omega$.

Let us show that $\Omega$ is not linked.
Assume the contrary.
For any path in $\Omega$ connecting elements $a$ and $b$ of $D_{x}$
we can build a path connecting $a$ and $b$ in $\Omega'$ in the following way.
We replace every constraint of $\Omega$ by the corresponding constraint of $\Omega'$,
and glue them with any path in $\Omega'$ starting and ending with the corresponding variables having the same parent.
Since $\Omega'$ is not fragmented, we can always do this.
Since $\Omega$ is cycle-consistent, the obtained path connects $a$ and $b$ in $\Omega'$.
Thus, $\Omega$ is not linked.
Any solution of $\Omega$ can be naturally extended to 
a solution of $\Omega'$, 
hence the solution set of 
$\Omega$ cannot be subdirect. 
Since $\Omega'$ is not fragmented, 
$\Omega$ is also not fragmented.
Thus, 
$\Omega$ is not linked, not fragmented, and its solution set is not subdirect,
which contradicts the fact that $\Theta$ is irreducible.
\end{proof}

For an instance $\Theta$ and its variable $x$
by $\LinkedCon(\Theta,x)$ we denote the binary relation on the set $D_{x}$ defined as follows:
$(a,b)\in \LinkedCon(\Theta,x)$ if there exists a path in $\Theta$ that connects $a$ and $b$.

\begin{lem}\label{LinkedConIsCon}
Suppose $\Theta$ is a cycle-consistent
CSP instance,
$x\in \Var(\Theta)$. 
Then there exists a path in $\Theta$ 
connecting all pairs $(a,b)\in \LinkedCon(\Theta,x)$
and $\LinkedCon(\Theta,x)$ is a congruence.
\end{lem}
\begin{proof}
Since the instance is cycle-consistent,
gluing all the paths starting and ending at $x$ we can build a path 
connecting all pairs $(a,b)\in \LinkedCon(\Theta,x)$.
The set of all pairs $(a,b)$ 
connected by this path can be defined by a pp-formula, 
therefore it is an invariant relation, which is also reflexive (by cycle-consistency) 
and transitive (we can glue paths).
\end{proof}

\subsection{Critical, key relations, and parallelogram property}
\label{DefinitionRectangularitySubsection}

We say that a relation $\rho$ \emph{has the parallelogram property}
if any permutation of its variables gives a relation $\rho'$ satisfying
$$\forall \alpha_{1},\beta_{1},\alpha_2,\beta_2\colon (\alpha_{1}\beta_2,\beta_1\alpha_2,\beta_1\beta_2\in\rho'
\Rightarrow \alpha_1\alpha_2\in\rho').$$
Note that the parallelogram property plays an important role in universal algebra (see \cite{agnes} for more details).

We say that \emph{the $i$-th variable of a relation $\rho$ is rectangular},
if for every $(a_{i},b_{i})\in\ConOne(\rho,i)$ and $(a_{1},\ldots,a_{n})\in\rho$
we have $(a_{1},\ldots,a_{i-1},b_{i},a_{i+1},\ldots,a_{n})\in\rho$.
We say that a relation is \emph{rectangular} if all of its variables are rectangular.
The following facts can be easily seen:
if the $i$-th variable of a subdirect relation $\rho$ is rectangular then $\ConOne(\rho,i)$ is a congruence;
if a relation has the parallelogram property then it is rectangular.

A relation $\rho\subseteq A_{1}\times\dots\times A_{n}$ is called \emph{essential} if it cannot be represented as a conjunction of relations with smaller arities.
It is easy to see that
any relation $\rho$ can be represented as a conjunction of essential relations that are projections of $\rho$ on some sets of variables (See Lemma 4.2 in \cite{MinimalClones}).

A relation $\rho\subseteq A_{1}\times\dots\times A_{n}$ is called \emph{critical}
if it cannot be represented as an intersection of other subalgebras
of $\mathbf A_{1}\times\dots\times \mathbf A_{n}$
and it has no dummy variables 
This notion was introduced in \cite{agnes} but appeared in \cite{mvlsc,mybook} by the name maximal.
For a critical relation $\rho$ the minimal relation $\rho'$ (a subalgebra of $\mathbf A_{1}\times\dots\times \mathbf A_{n}$) such that
$\rho'\supsetneq\rho$ is called \emph{the cover of $\rho$}.

Suppose $\rho\subseteq A_{1}\times\dots\times A_{h}$.
A tuple $\Psi =(\psi_1,\psi_2,\ldots,\psi_h)$, where
$\psi_i:A_{i}\to A_{i}$, is called a \emph{unary vector-function}.
We say that $\Psi$ \emph{preserves} $\rho$ if
$\Psi\left(\begin{smallmatrix}
a_1\\
a_2\\
\vdots\\
a_h
\end{smallmatrix}\right):=
\left(\begin{smallmatrix}
\psi_1(a_1)\\
\psi_2(a_2)\\
\vdots\\
\psi_h(a_h)
\end{smallmatrix}\right)\in \rho$
for every $\left(\begin{smallmatrix}
a_1\\
a_2\\
\vdots\\
a_h
\end{smallmatrix}\right)\in \rho$.
We say that $\rho$ is \emph{a key relation} if
there exists a tuple $\beta\in (A_{1}\times\dots\times A_{h})\setminus \rho$ such that for every
$\alpha\in (A_{1}\times\dots\times A_{h})\setminus \rho$ there exists
a vector-function $\Psi$ which
preserves $\rho$ and gives $\Psi(\alpha) = \beta$.
A tuple $\beta$ is called a \emph{key tuple} for $\rho$.
The notion key relation was introduced in 
\cite{KeyRelations}, where such relations were characterized 
for all algebras having a WNU term operation.

A constraint is called \emph{critical/essential/key} if the constraint relation
is critical/essential/key.
The notions 
critical, crucial, essential, and key relation are related to each other,
namely, we can observe:
\begin{enumerate}
    \item if $C$ is a constraint in a CSP instance and $C$ is crucial in some $(D_{1},\dots,D_{n})$ then
the constraint relation of $C$ is critical;
\item every critical relation of arity greater than 1 is essential;
\item every critical relation of arity greater than 1 is a key relation (see Lemma 2.4 in \cite{KeyRelations}).
\end{enumerate}

The notions essential, critical, and key relations
(see \cite{KeyRelations} for their comparison)
proved their efficiency in clone theory and 
universal algebra (see \cite{mvlsc,mybook,MinimalClones,VardiProblem,dm_post,agnes}).
Instead of considering all relations we consider 
only relations with one of these properties,
and this is still the general case because 
any relation can be represented as a conjunction of 
essential/key/critical relations. 
For instance, we can always assume that all constraint relations 
are critical.

\subsection{Reductions}

Suppose the domain set of an instance $\Theta$ is $D = (D_{1},\ldots,D_{n})$.
A domain set $D' = (D_{1}',\ldots,D_{n}')$ is called \emph{a reduction of $\Theta$} if
$D_{i}'$ is a subuniverse of $D_{i}$ for every $i$.
Note that to avoid unnecessary bold font starting at this subsection we do not 
use it for domain sets. 
Thus, every time we write $D$ without a subscript we mean 
a domain set or a reduction. 
Note that any reduction of $\Theta$ can be naturally extended 
to a covering (expanded covering) of $\Theta$, 
thus we assume that any reduction is automatically defined on 
any covering (expanded covering).

A reduction $D' = (D_{1}',\ldots,D_{n}')$ is called \emph{1-consistent}
if the instance obtained after reduction of every domain is 1-consistent.

We say that $D'$ is  \emph{an absorbing reduction}, 
if there exists a term operation
$t$ such that $D_{i}'$ is a binary absorbing subuniverse of $D_{i}$ with the term operation $t$ for every $i$. 
We say that $D'$ is  \emph{a central reduction}, 
if 
$D_{i}'$ is a center of $D_{i}$ for every $i$. 
We say that $D'$ is  \emph{a PC/linear reduction},
if 
$D_{i}'$ is a PC/linear subuniverse of $D_{i}$ and $D_{i}$ does not have a nontrivial binary absorbing subuniverse or a nontrivial center for every $i$.
Additionally, we say that $D'$ is \emph{a minimal central/PC/linear reduction} if 
$D'$ is a minimal center/PC/linear subuniverse of $D_{i}$ for every $i$. 
We say that $D'$ is \emph{a minimal absorbing reduction} for a term operation $t$ 
if $D'$ is a minimal absorbing subuniverse of $D_{i}$  with $t$ for every $i$.  


A reduction is called \emph{nonlinear} if
it is an absorbing, central, or PC reduction.
A reduction $D'$ is called \emph{one-of-four reduction} if it is an absorbing, central, PC, or linear reduction such that $D'\neq D$.



We usually denote reductions by $D^{(j)}$ for some $j$ (or by $D^{(\top)}$).
In this case by $C^{(j)}$  we denote the constraint obtained after the reduction of the constraint $C$.
Similarly, by $\Theta^{(j)}$ we denote the instance obtained after the reduction of every constraint of $\Theta$.
For a relation $\rho$ by $\rho^{(j)}$ we denote the relation $\rho$ restricted to the corresponding domains of $D^{(j)}$.
Sometimes we write $(a_{1},\ldots,a_{n})\in D^{(j)}$ meaning that
every $a_{i}$ belongs to the corresponding $D_{x}^{(j)}$.

A \emph{strategy} for a CSP instance $\Theta$ with a domain set $D$ is
a sequence of reductions
$D^{(0)},\ldots,D^{(s)}$,
where $D^{(j)} = (D_{1}^{(j)},\ldots,D_{n}^{(j)})$,
such that
$D^{(0)} = D$ and $D^{(j)}$ is a one-of-four 1-consistent reduction of $\Theta^{(j-1)}$
for every $j\ge 1$.
A strategy is called \emph{minimal} if every reduction in the sequence is minimal.


\subsection{Bridges}

Suppose $\sigma_{1}$ and $\sigma_{2}$ are congruences on $D_{1}$ and $D_{2}$, respectively.
A relation $\rho\subseteq D_{1}^{2}\times D_{2}^{2}$ is called \emph{a bridge} from $\sigma_{1}$ to $\sigma_{2}$ if
the first two variables of $\rho$ are stable under $\sigma_{1}$,
the last two variables of $\rho$ are stable under $\sigma_{2}$,
$\proj_{1,2}(\rho) \supsetneq \sigma_{1}$,
$\proj_{3,4}(\rho) \supsetneq \sigma_{2}$,
and
$(a_{1},a_2,a_{3},a_{4})\in \rho$ implies
$$(a_1,a_2)\in \sigma_{1}\Leftrightarrow (a_3,a_4)\in \sigma_{2}.$$

An example of a bridge 
is the 
relation 
$\rho=\{(a_{1},a_{2},a_{3},a_{4})\mid
a_{1},a_{2},a_{3},a_{4}\in \mathbb Z_{4}:
a_{1}-a_{2} = 2 a_{3} - 2 a_{4}\}$.
We can check that 
$\rho$ is a bridge from 
the equality relation (0-congruence) 
and $(mod\;2)$ equivalence relation.
For example, we have
$\proj_{1,2} \rho$ 
is $(mod\;2)$-equivalence relation,
$\proj_{3,4} \rho$ 
is full relation.

The notion of a bridge is strongly related to other notions in Universal Algebra and Tame Congruence Theory
such as similarity and centralizers
(see \cite{RossSlides} for the detailed comparison).

For a bridge $\rho$ by $\widetilde{\rho}$ we denote 
the binary relation defined by 
$\widetilde{\rho}(x,y) = \rho(x,x,y,y)$.

The following lemma shows how
we can compose bridges.

\begin{lem}\label{BridgeComposition}
Suppose $\sigma_{1}$, $\sigma_{2}$, $\sigma_{3}$ are irreducible congruences, 
$\rho_{1}$ is a bridge  from $\sigma_{1}$ to $\sigma_{2}$,
$\rho_{2}$ is a bridge from $\sigma_{2}$ to $\sigma_{3}$.
Then the formula
$$\rho(x_1,x_2,z_{1},z_{2}) = \exists y_{1}\exists y_{2}\; \rho_{1}(x_{1},x_{2},y_{1},y_{2})\wedge \rho_{2}(y_{1},y_{2},z_{1},z_{2})$$
defines a bridge from $\sigma_{1}$
to $\sigma_{3}$.
Moreover, 
$\widetilde{\rho} = 
\widetilde{\rho_{1}}\circ\widetilde{\rho_{2}}$.
\end{lem}
\begin{proof}
Stability of the first two variables 
under $\sigma_{1}$ and 
of the last two variables under
$\sigma_{3}$ follows from the definition.

Let us prove that 
$\proj_{1,2}(\rho)\supsetneq \sigma_{1}$
(the inclusion $\proj_{3,4}(\rho)\supsetneq \sigma_{3}$ can be proved in the same way).
By the definition, for every 
$a$ there exists 
$b$ such that 
$(a,a,b,b)\in \rho_{1}$, 
and for every $b$ there exists $c$
such that 
$(b,b,c,c)\in\rho_{2}$.
Then $(a,a,c,c)\in\rho$, 
and since the first two variables of $\rho_{1}$ 
are stable under $\sigma_{1}$ 
we obtain $\proj_{1,2}(\rho)\supseteq \sigma_{1}$.
Since $\sigma_2$ is irreducible, 
$\proj_{3,4}(\rho_{1})\supseteq \sigma_{2}^{*}$
and 
$\proj_{1,2}(\rho_{2})\supseteq \sigma_{2}^{*}$.
Choose 
$(b_{1},b_{2})\in \sigma_{2}^{*}$, 
then there exist 
$a_1,a_2, c_{1},c_{2}$ such that 
$(a_{1},a_{2},b_{1},b_{2})\in\rho_{1}$ 
and 
$(b_{1},b_{2},c_{1},c_{2})\in\rho_{2}$.
Then $(a_1,a_2,c_{1},c_2)\in\rho$, which means 
that $\proj_{1,2}(\rho)\supsetneq \sigma_{1}$.

Suppose 
$(a_1,a_2,c_{1},c_{2})\in\rho$.
If $(a_{1},a_{2})\in\sigma_{1}$ then,
since $\rho_{1}$ is a bridge,
the corresponding values of 
$y_{1}$ and $y_2$ are equivalent modulo 
$\sigma_{2}$. 
Since $\rho_{2}$ is a bridge we obtain that 
$c_{1}$ and $c_{2}$ are equivalent modulo 
$\sigma_{3}$.

The equation $\widetilde{\rho} = 
\widetilde{\rho_{1}}\circ\widetilde{\rho_{2}}$
follows directly from the definition of $\rho$.
\end{proof}

A bridge $\rho\subseteq D^{4}$ is called \emph{reflexive} if
$(a,a,a,a)\in \rho$ for every $a\in D$.

We say that two congruences $\sigma_{1}$ and $\sigma_{2}$ on a set $D$ are \emph{adjacent}
if there exists a reflexive bridge from $\sigma_{1}$ to $\sigma_{2}$.

%

\begin{remark}
Since we can always put
$\rho(x_{1},x_{2},x_{3},x_{4}) = \sigma(x_{1},x_{3})\wedge \sigma (x_{2},x_{4})$,
any proper congruence $\sigma$ is adjacent with itself.
\end{remark}

A reflexive bridge $\rho$ from an irreducible congruence $\sigma_{1}$ to an irreducible congruence $\sigma_{2}$ is called \emph{optimal} if
there does not exist a reflexive bridge $\rho'$ from $\sigma_{1}$ to $\sigma_{2}$
such that $\widetilde{\rho'}
\supsetneq\widetilde{\rho}$.
Suppose $\rho$ is a reflexive bridge from 
$\sigma_{1}$ to $\sigma_{2}$.
then we can build 
a new bridge 
$$\rho'(x_1,x_2,y_1,y_2)
=\exists x_1'\exists x_2'\exists y_1'\exists y_2' \left[\rho(x_1,x_2,y_1',y_2')\wedge
\rho(x_1',x_2',y_1',y_2')\wedge 
\rho(x_1',x_2',y_1,y_2)\right]
$$
from $\sigma_1$ to $\sigma_{2}$
such that $\widetilde{\rho'} = 
\widetilde{\rho}\circ \widetilde{\rho}^{-1}\circ\widetilde{\rho}$.
Note that because of the reflexivity, 
$\widetilde \rho$ contains the equality relation.
Thus, if $\rho$ is optimal, then $\widetilde{\rho}$ is a congruence.
For an irreducible congruence $\sigma$ by $\Opt(\sigma)$ we denote the congruence $\widetilde{\rho}$ for an optimal bridge $\rho$ from $\sigma$ to $\sigma$.
Since we can compose two reflexive bridges, $\Opt(\sigma)$ is unique and therefore well-defined.
For a set of irreducible congruences $\mathfrak C$ put $\Opt(\mathfrak C) = \{\Opt(\sigma)\mid\sigma\in \mathfrak C\}$.

\begin{lem}\label{OptimalForAdjacent}
Suppose 
$\sigma_{1}$ and $\sigma_{2}$ 
are irreducible adjacent congruences.
Then 
$\Opt(\sigma_{1}) = \Opt(\sigma_{2})$.
\end{lem}
\begin{proof}
Let $\rho_{1}$ be an optimal bridge from 
$\sigma_{1}$ to $\sigma_{1}$,
$\rho_{2}$ be an optimal bridge from 
$\sigma_{2}$ to $\sigma_{2}$,
and $\rho$ be a reflexive bridge from $\sigma_{1}$ to 
$\sigma_{2}$.

Assume that 
$\Opt(\sigma_{2})\not\subseteq\Opt(\sigma_{1})$, 
that is $\widetilde\rho_{2}\not\subseteq\widetilde\rho_{1}$.
Using Lemma~\ref{BridgeComposition}, 
we compose 
bridges 
$\rho_{1}$, $\rho$,$\rho_{2}$, and $\rho$
(in this order)
to obtain a reflexive bridge $\rho_{1}'$ from $\sigma_{1}$ to $\sigma_{1}$.
Since $\widetilde\rho_{1}'\supseteq \widetilde\rho_{1}\cup\widetilde\rho_{2}$, 
we get a contradiction with the fact that 
$\rho_{1}$ is optimal.
\end{proof}



We say that two rectangular constraints $C_{1}$ and $C_{2}$ are \emph{adjacent} in a common variable $x$ if
$\ConOne(C_{1},x)$ and $\ConOne(C_{2},x)$ are adjacent.
A formula is called \emph{connected} if
every constraint in the formula is critical and rectangular, and
the graph, whose vertexes are constraints 
and edges are adjacent constraints,
is connected.
Note that this connectedness is not related 
to the paths from one variable to another 
connecting two elements. 
Recall that 
if for every $a,b$  
there exists a path that connects $a$ and $b$, 
then the instance is called linked
(see Section \ref{CSPInstancesDef}).

It can be shown (see Corollary~\ref{PathInConnectedComponent}) that
every two constraints with a common variable in a connected instance are adjacent.

%% file: AbsCenterPCLinearProperty.tex
\section{Absorption, Center, PC Congruence, and Linear Congruence}\label{AbsCenterPCLinear}
\subsection{Binary Absorption}

\begin{lem}\label{AbsImplies}\cite{DecidingAbsorption}
Suppose $\rho$ is defined by a pp-formula $\Omega(x_{1},\ldots,x_{n})$
and
$\Omega'$ is obtained from $\Omega$ by replacement of some constraint relations $\sigma_{1},\ldots,\sigma_{s}$
by constraint relations $\sigma_{1}',\ldots,\sigma_{s}'$ such that $\sigma_{i}'$ absorbs $\sigma_{i}$ with a term operation $t$
for every $i$.
Then the relation defined by $\Omega'(x_{1},\ldots,x_{n})$ absorbs $\rho$ with the term operation $t$.
\end{lem}

\begin{conslem}\label{AbsorptionQuotient}
Suppose $\theta$ is a congruence of $A$.
\begin{enumerate}
    \item If $B$ is an absorbing subuniverse of $A$, 
    then $\{b/\theta\mid b\in B\}$ is an absorbing subuniverse of $A/\theta$
    with the same term.
        \item If $A$ has no nontrivial (binary) absorbing subuniverse, 
    then neither does $A/\theta$.
\end{enumerate}
\end{conslem}

\begin{conslem}\label{AbsImpliesCons}
Suppose $\rho \subseteq A_{1}\times\dots\times A_{n}$ is a relation such that
$\proj_1 (\rho) = A_{1}$ 
and 
$C = \proj_{1}((C_{1}\times\dots \times C_{n})\cap\rho)$,
where $C_{i}$ is an absorbing subuniverse in $A_{i}$ with a term $t$ for every $i$.
Then $C$ is an absorbing subuniverse in $A_{1}$ with the term $t$.
\end{conslem}

\begin{proof}
It is not hard to see that the sets $C$ and 
$A_{1}$ can be defined by
the following pp-formulas
$$(x_1\in C) = \exists 
x_{2}\dots\exists x_{n}\; 
\left[(x_{1}\in C_{1})\wedge 
\dots\wedge 
(x_{n}\in C_{n})\wedge 
\rho(x_{1},\ldots,x_{n})\right],$$
$$(x_1\in A_1)= \exists 
x_{2}\dots\exists x_{n}\; 
\left[(x_{1}\in A_{1})\wedge 
\dots\wedge 
(x_{n}\in A_{n})\wedge 
\rho(x_{1},\ldots,x_{n})\right].$$
It remains to apply Lemma~\ref{AbsImplies}.
\end{proof}


\begin{lem}\label{AbsorbingEquality}
Suppose $\kappa_{A}\subseteq A\times A$ is the equality relation,
$\sigma\supseteq \kappa_{A}$,
and 
$\omega$ is a nontrivial binary absorbing subuniverse in $\sigma$.
Then $\omega\cap\kappa_{A} \neq \varnothing$.
\end{lem}
\begin{proof}
We prove the lemma by induction on the size of $A$.
Suppose $\omega$ absorbs $\sigma$ with a binary absorbing term operation $f$.

Assume that there exists a nontrivial binary absorbing subuniverse $B\subsetneq A$ with the absorbing operation $f$.
For any $(b_{1},b_{2})\in\omega$ and $b\in B$ we have
$(f(b_{1},b),f(b_{2},b))\in \omega\cap (B\times B)$.
Then by Lemma~\ref{AbsImplies},  
$\omega\cap (B\times B)$ is a nontrivial 
absorbing subuniverse in 
$\sigma\cap (B\times B)$, and we can restrict
$\sigma$ and $\omega$ to $B$ and apply the inductive assumption.

Thus, we assume that there does not exist
a nontrivial binary absorbing subuniverse $B\subsetneq A$ with the absorbing operation $f$.
By Lemma~\ref{AbsImplies},
$\proj_{1}(\omega)$ and $\proj_{2}(\omega)$ binary absorb $A$,
then $\proj_{1}(\omega)=\proj_{2}(\omega)=A$.
Now, the statement of the lemma 
could be derived from \cite[Theorem 6]{barto2012near}
but we will finish the argument because it is simple.

For every $b\in A$ we consider 
$A_{b}= \{a\mid (a,b)\in \sigma\}$
and
$C_{b}= \{a\mid (a,b)\in \omega\}$.
Since $\proj_{2}(\omega)=A$, $C_{b}\neq\varnothing$ for every $b$.
By Lemma~\ref{AbsImplies} $C_{b}$ is a binary absorbing subuniverse in $A_{b}$ with $f$.
Therefore 
$A_{b}\neq A$ or
$A_{b}=C_{b} = A$.
In the latter case we have $(b,b)\in\omega$, 
which completes this case.

Assume that $A_{b}\neq A$ for some $b$.
Since $\sigma\supseteq \kappa_{A}$, 
we have $b\in A_{b}$ and 
$(A_{b}\times A_{b})\cap \omega
\supseteq (C_{b}\times\{b\})\cap \omega
\neq \varnothing$.
Then we restrict
$\sigma$ and $\omega$ to $A_{b}$ and apply the inductive assumption.
\end{proof}

\begin{lem}\label{GenBinAbToBinAb}
Suppose $\rho$ is a nontrivial absorbing subuniverse of $A_{1}\times \dots\times A_{n}$.
Then for some $i$ there exists a nontrivial absorbing subuniverse $B_{i}$ in $A_{i}$ with the same term.
\end{lem}
\begin{proof}
We prove this lemma by induction on the arity of $\rho$.
If the projection of $\rho$ onto the first coordinate is not $A_{1}$ then
by Lemma~\ref{AbsImplies} this projection is an absorbing subuniverse with the same term.
Otherwise, we choose any element $a\in A_{1}$
such that $\rho$ does not contain all tuples starting with $a$,
and consider $\rho' = \{(a_2,\ldots,a_{n})\mid (a,a_2,\ldots,a_n)\in \rho\}$, which, by Lemma~\ref{AbsImplies}, is
a nontrivial absorbing subuniverse in $A_{2}\times\dots \times A_{n}$ with the same term.
It remains to apply the inductive assumption.
\end{proof}

A relation $\rho\subseteq A^{n}$ is called \emph{$C$-essential} if
$\rho\cap(C^{i-1}\times A\times C^{n-i})\neq \varnothing$ for every $i$
but $\rho\cap C^{n}=\varnothing$.
A relation 
$\rho\subseteq A_{1}\times\dots\times A_{n}$ is called 
\emph{$(C_{1},\dots,C_{n})$-essential}
if
$\rho\cap
(C_{1}\times\dots\times C_{i-1}\times A_{i} \times C_{i+1}
\times\dots\times C_{n})\neq\varnothing$
for every $i$ 
but 
$\rho\cap
(C_{1}\times\dots\times C_{n})=\varnothing$.

\begin{lem}\label{NoEssential}\cite{DecidingAbsorption}
Suppose $C$ is a subuniverse of $A$.
Then $C$ absorbs $A$ with an operation of arity $n$ if and only if
there does not exist a $C$-essential relation $\rho\subseteq A^{n}$.
\end{lem}

\begin{lem}\label{AbsLessThanThree}
Suppose $D^{(1)}$ is an absorbing reduction
of a CSP instance $\Theta$ and
a relation $\rho\subseteq
D_{i_1}\times\dots\times D_{i_n}$
is subdirect,
where $D_{i_1}, \dots,D_{i_n}$ are domains 
of variables from $\Theta$.
Then $\rho^{(1)}$ is not empty.
\end{lem}
\begin{proof}
It is sufficient to apply the binary absorbing term operation $t$ to all the tuples 
of $\rho$ using term
$t(x_1,t(x_2,t(x_3,\dots,t(x_{s-1},x_{s}))))$,
where $s=|\rho|$.
The resulting tuple will be from 
$\rho^{(1)}$, which means that 
$\rho^{(1)}$ is not empty.
\end{proof}

\subsection{Center}
%

\begin{lem}\label{CenterImplies}
Suppose $\rho$ is defined by a pp-formula $\Omega(x_{1},\ldots,x_{n})$ and 
$\Omega'$ is obtained from $\Omega$ by replacement of some constraint relations $\sigma_{1},\ldots,\sigma_{s}$
by constraint relations $\sigma_{1}',\ldots,\sigma_{s}'$ such that $\sigma_{i}'$ is a center of $\sigma_{i}$ for every $i$.
Then the relation defined by $\Omega'(x_{1},\ldots,x_{n})$ is a center of~$\rho$.
\end{lem}

\begin{proof}
Suppose $\Omega'(x_{1},\ldots,x_{n})$ defines a relation $\rho'$.
Suppose $\mathbf B_{i}$ and $R_{i}$ are the corresponding algebra and binary relation such that
$\sigma_{i}' = \{c\mid \forall b\in B_{i}\colon (c,b)\in R_{i}\}$.
Let $|B_{i}| = n_{i}$ for every $i$.
Let $\Upsilon$ be obtained from $\Omega$ by replacement of every constraint
$\sigma_{i}(y_{1},\ldots,y_{t})$ by
$$R_{i}((y_{1},\ldots,y_{t}),z_{i,1})\wedge \dots\wedge R_{i}((y_{1},\ldots,y_{t}),z_{i,n_{i}}).$$
Suppose
$\Upsilon((x_1,\ldots,x_{n}),(z_{1,1},\dots,z_{s,n_{s}}))$ defines a relation $R$.
It is not hard to see that
$\rho' = \{c\mid \forall b\in (B_{1}^{n_{1}}\times\dots\times B_{s}^{n_{s}})\colon (c,b)\in R\}$.
By Lemma~\ref{GenBinAbToBinAb}, there is no nontrivial binary absorbing subuniverse on $B_{1}^{n_{1}}\times\dots\times B_{s}^{n_{s}}$.
This proves that $\rho'$ is a center of $\rho$.
\end{proof}

\begin{conslem}\label{CenterQuotient}
Suppose $\theta$ is a congruence of $A$
\begin{enumerate}
    \item If $B$ is a center of $A$, 
    then $\{b/\theta\mid b\in B\}$ is a center of $A/\theta$.
        \item If $A$ has no nontrivial center, 
    then neither does $A/\theta$.
\end{enumerate}
\end{conslem}

\begin{conslem}\label{CenterImpliesCons}
Suppose $\rho \subseteq A_{1}\times\dots\times A_{n}$ is a relation such that
$\proj_1 (\rho) = A_{1}$ and
$C = \proj_{1}((C_{1}\times\dots \times C_{n})\cap\rho)$,
where $C_{i}$ is a center in $A_{i}$ for every $i$.
Then $C$ is a center in $A_{1}$.
\end{conslem}

\begin{conslem}\label{CenterProduct}
Suppose $C_{i}$ is a center of $D_{i}$ for every $i$.
Then
$C_{1}\times\dots\times C_{n}$ is a center of
$D_{1}\times\dots\times D_{n}$.
\end{conslem}

\begin{conslem}\label{CenterIntersection}
Suppose $C_{1}$ and $C_{2}$ are centers of $D$.
Then $C_{1}\cap C_{2}$ is a center of $D$.
\end{conslem}

\begin{lem}\label{GenCenterToCenter}
Suppose $\rho$ is a nontrivial center of $A_{1}\times \dots\times A_{n}$.
Then for some $i$ there exists a nontrivial center $C_{i}$ of $A_{i}$.
\end{lem}
\begin{proof}
We prove by induction on the arity of $\rho$.
If the projection of $\rho$ onto the first coordinate is not $A_{1}$ then
by Lemma~\ref{CenterImplies} this projection is a center.

Otherwise, we choose any element $a\in A_{1}$
such that $\rho$ does not contain all tuples starting with $a$.
Then we consider $\rho' = \{(a_2,\ldots,a_{n})\mid (a,a_2,\ldots,a_n)\in \rho\}$, which, by Lemma~\ref{CenterImplies}, is
a nontrivial center of $A_{2}\times\dots \times A_{n}$ .
It remains to apply the inductive assumption.
\end{proof}

In the proof of the following two lemmas we assume that
a center $C$ is defined by
$C = \{a\in A\mid \forall b\in B\colon (a,b)\in R\}$
for a subalgebra $R$ of $\mathbf A\times \mathbf B$.
For an element $a\in A$ we put $a^{+}=\{b\mid (a,b)\in R\}$.
Also, we introduce a quasi-order on elements of $A$.
We say that $y_1\le y_2$ if $y_1^{+}\subseteq y_{2}^{+}$,
and $y_1\sim y_2$ if $y_1^{+}= y_{2}^{+}$.
Note that if $b_{1},b_{2},\ldots,b_{m}\ge c$, 
then 
$w(b_{1}^{+},\dots,b_{m}^{+})\supseteq 
w(c^{+},\dots,c^{+})\supseteq c^{+}$,
and therefore 
$w(b_1,\ldots,b_m)\ge c$.

\begin{lem}\label{wnuofcentralelements}
Suppose $(c_{1},\ldots,c_{m})\in A^{m}$, 
$c_{i}\in C$ for every $i\neq j$, 
and $c_{j}\notin C$.
Then $w(c_1,\ldots,c_{m})>c_{j}$. 
\end{lem}
\begin{proof}
Assume the contrary, then $w(c_1,\ldots,c_{m})\sim c_{j}$
and 
$w(\underbrace{B,\ldots,B}_{i-1},c_{j}^{+} ,\underbrace{B,\ldots,B}_{m-i})
\subseteq c_{j}^{+}$.
This is enough to imply 
that $c_{j}^{+}$ is a binary absorbing subuniverse 
with the term
$x\circ y = w(x,x,\ldots,x,y)$.
In fact, if $b_{1}\in B$ and $b_2 \in c_{j}^{+}$,
then we can write $b_1 \circ b_2 = w(b_{1},\ldots,b_{1},b_{2},b_{1},\ldots,b_{1})$
with $b_2$ in the $j$-th spot;
if $b_{1}\in c_{j}^{+}$ and $b_2 \in B$,
then we can write $b_1 \circ b_2 = w(b_{1},\ldots,b_{1},b_{2},b_{1},\ldots,b_{1})$
with one of the $b_1$'s in the $j$-th spot.
In both cases we obtain $b_{1}\circ b_{2}\in c_{j}^{+}$. Contradiction.

\end{proof}

\begin{lem}\label{AbsorptionFromSequence}
Suppose $w$ is a special WNU of arity $m$,
$C$ is a nontrivial center in $A$,
$\delta\subseteq A^{s}$ is $C$-essential.
Then $s<m^{|A|}$.
\end{lem}
\begin{proof}

Choose $\alpha_1,\dots,\alpha_{s}\in\delta$
such that 
$\alpha_{i}\in C^{i-1}\times A \times C^{s-i}$
for every $i$.
We start with the matrix $M_1$ whose columns are tuples
$\alpha_{1},\ldots,\alpha_{s}$.
Then we build 
a matrix 
$M_{2}$ whose columns are tuples 
$w(\alpha_{1},\ldots,\alpha_{m})$, 
$w(\alpha_{m+1},\ldots,\alpha_{2m})$,
$w(\alpha_{2m+1},\ldots,\alpha_{3m}),\ldots.$
Then we apply the WNU $w$ to the corresponding columns 
of the previous matrix to define a new matrix $M_{3}$.
We continue this way until we get a matrix with less than $m$ columns.
Note that the next matrix has $m$ times less columns than the previous one.
It is not hard to see that every row of every matrix 
has at most one element that is not from the center.
Moreover, by Lemma~\ref{wnuofcentralelements}, 
the noncentral element in the $i$-th row of the $(j+1)$-th matrix 
is greater than
the noncentral 
element in the $i$-th row of the $j$-th matrix.
This means that the $|A|$-th matrix, if it exists,
has only central elements, which contradicts our assumptions.
Hence, it does not exist and $s< m^{|A|}$. 
\end{proof}

Combining this result with
Lemma~\ref{NoEssential}, we obtain the following corollary.

\begin{conslem}\label{centerImpliesAbsorption}
Suppose $C$ is a center of $A$. Then
$C$ is an absorbing subuniverse of~$A$.
\end{conslem}

The following lemma is a stronger version of an original lemma suggested by Marcin Kozik.

\begin{lem}\label{IncreaseArity}
Suppose $C_{1}\subseteq A_{1}$ and $C_{2}\subseteq A_{2}$ are centers,
$B$ is a subuniverse of $D$,
and a relation $\rho\subseteq A_{1}\times D^{l}\times A_{2}$
is 
$(C_{1},B,\dots,B,C_{2})$-essential.
Then there exists 
a  relation $\rho'\subseteq A_{1}\times D^{2l}\times A_{1}$
that is 
$(C_{1},B,\dots,B,C_{1})$-essential.
\end{lem}

\begin{proof}
Assume that $\rho$ is a minimal relation 
(with respect to inclusion) 
that is $(C_{1},B,\dots,B,C_{2})$-essential.
Put $E = \proj_{l+2}(\rho\cap (C_{1}\times B^{l}\times A_{2}))$.
Since $\rho$ is minimal, for any $b\in E$ the algebra generated by $\{b\}\cup C_{2}$ contains $\proj_{l+2}(\rho)$ 
(otherwise we would restrict 
the $(l+2)$-th variable of $\rho$ to this algebra).
Fix $b\in E$.

Let $\sigma$ be the subalgebra of $A_{2}\times A_{2}$ generated by
$\{b\}\times C_{2}\cup C_{2} \times C_{2}\cup C_{2}\times \{b\}$.
Since our algebras are idempotent,
for any $c\in \proj_{l+2}(\rho)$
we have 
$\{c\}\times C_{2}\subseteq \sigma$.
Put
$$\rho'(x,y_1,\ldots,y_l,y_{1}',\ldots,y_{l}',x') = \exists z\exists z' \;\rho(x,y_1,\ldots,y_{l},z)\wedge
\rho(x',y_{1}',\ldots,y_{l}',z')\wedge\sigma(z,z').$$
Let us show that $\rho'$ 
is $(C_{1},B,\dots,B,C_{1})$-essential.
Since 
$\rho$ is 
$(C_{1},B,\dots,B,C_{2})$-essential, 
for any $i\in\{1,\ldots,l+1\}$ there exists 
a tuple 
$(a_{1},\ldots,a_{l+2})$ 
such that 
only its $i$-th element 
is not from the corresponding set 
of $(C_{1},B,\dots,B,C_{2})$.
Since $b\in E$,
there exists $c_{1},\ldots,c_{l+1}$ 
such that 
$(c_{1},\ldots,c_{l+1},b)\in 
\rho\cap (C_{1}\times B^{l}\times A_{2})$.
Then 
$(a_{1},\dots,a_{l+1},
c_{2},\ldots,c_{l+1},c_{1})\in \rho'$ (it is sufficient 
to put $z= a_{l+2}$ and $z' = b$).
Thus, for any 
$i\in\{1,\ldots,l+1\}$
we build a tuple from $\rho'$ such that 
only its $i$-th element 
is not from the corresponding set 
of $(C_{1},B,\dots,B,C_{1})$.
In the same way we can build such a tuple
for each
$i\in\{l+2,\ldots,2l+2\}$.

To prove that $\rho'$ 
is $(C_{1},B,\dots,B,C_{1})$-essential
it remains to show that 
$(C_{1}\times B^{2l}\times C_{1})\cap \rho' =\varnothing$.
Assume the converse, let 
a tuple from the intersection be obtained by sending $z$ to $d$ and $z'$ to $d'$.
Clearly, $d,d'\in E$ and $\{e\in A_{2}\mid (e,d')\in \sigma\}\supseteq \{d\}\cup C_{2}$,
therefore $\{e\in A_{2}\mid (e,d')\in\sigma\} \supseteq \proj_{l+2}(\rho)$.
Hence,  $\{e\in A_{2}\mid (b,e)\in\sigma\}\supseteq \{d'\}\cup C_{2}$
and $\{e\in A_{2}\mid (b,e)\in\sigma\} \supseteq \proj_{l+2}(\rho)$.

Thus, $(b,b)\in\sigma$ and there exists an $n$-ary term $t$ such that
$$t(b,b,\ldots,b,c_{1},\ldots,c_{i}) = b,\;\;\;
t(c_{1}',\ldots,c_{j}',b,b,\ldots,b) = b,$$
where $i+j\ge n$ and $c_{1},\ldots,c_{i},c_{1}',\ldots,c_{j}'\in C_{2}$.
Suppose $R\subseteq A_{2}\times G$ is a binary relation
from the definition of the center $C_{2}$,
$b^{+} = \{a\mid (b,a)\in R\}$.
Since $t$ preserves $R$, we have
$$t(b^{+},b^{+},\ldots,b^{+},\underbrace{G,\ldots,G}_{i}) \subseteq b^{+},\;\;\;
t(\underbrace{G,\ldots,G}_{j},b^{+},b^{+},\ldots,b^{+}) \subseteq b^{+},$$
and therefore $b^{+}$ absorbs $G$ with the binary term $t(\underbrace{x,\ldots,x}_{j},y,\ldots,y)$.
This contradiction completes the proof.
\end{proof}

\begin{conslem}\label{AlmostEssTuple}
Suppose $C_{1}\subseteq A_{1}$ and $C_{2}\subseteq A_{2}$ are centers and 
$B\subseteq D$ is an absorbing subuniverse.
Then there does not exist 
$(C_{1},B,C_{2})$-essential relation 
$\rho\subseteq A_{1}\times D\times A_{2}$.
\end{conslem}

\begin{proof}
Assume that such a relation $\rho$ exists.
Iteratively applying Lemma~\ref{IncreaseArity}
to $\rho$
we can obtain 
a $(C_{1},B,\dots,B,C_{1})$-essential 
relation $\rho_{l}\subseteq A_{1}\times D^{l}\times A_{1}$ 
for $l= 2,4,8,\dots$.
If we restrict the first  and the last variables 
of $\rho_{l}$ to $C_{1}$ and consider the projection onto the
remaining variables we get a $B$-essential relation
of arity $l$.
Since we can make $l$ as large as we need, 
we get a contradiction with 
Lemma~\ref{NoEssential} and the fact that $B$ is an absorbing subuniverse.
\end{proof}

\begin{conslem}\label{ternaryAbsorption}
Suppose $C$ is a center of $A$. Then
$C$ is a ternary absorbing subuniverse of~$A$.
\end{conslem}
\begin{proof}
Assume that $C$ is not a ternary absorbing subuniverse 
then by Lemma~\ref{NoEssential}, there exists a $C$-essential relation of arity 3.
By Corollary~\ref{centerImpliesAbsorption}, 
$C$ is an absorbing subuniverse of $A$,
then by 
Corollary~\ref{AlmostEssTuple} such a relation cannot exist.
\end{proof}

\begin{conslem}\label{CenterLessThanThree}
Suppose
$C_{i}$ is a center of 
$A_{i}$ for $i\in\{1,2,\dots,k\}$
and 
$k\ge 3$.
Then there does not exist a 
$(C_{1},\dots,C_{k})$-essential 
relation 
$\rho\subseteq A_{1}\times \dots \times A_{k}$.
\end{conslem}
\begin{proof}
If such a relation $\rho$ exists then 
restricting all but the first three variables of $\rho$ to the corresponding centers and projecting the result onto the first three variables we obtain 
$(C_{1},C_{2},C_{3})$-essential relation, 
which cannot exists by 
Corollary~\ref{AlmostEssTuple}.
\end{proof}

\subsection{PC Subuniverse}\label{PCSubsection}

\begin{lem}\label{ReflexivePCRelations}
Suppose
$A$ is a PC algebra and 
$\rho\subseteq A^{n}$ is a relation containing all
the constant tuples $(a,\dots,a)$. 
Then $\rho$ can be represented as a conjunction of
binary relations of the form $x_i = x_j$.\end{lem}
\begin{proof}
All constant operations preserve $\rho$, 
and together with the constant operations 
the algebra $A$ generates all operations on the set $A$.
Then $\rho$ is preserved by all operations on $A$,
and therefore, $\rho$ is diagonal (see Theorem 2.9.3 from \cite{lau})
and it can be represented
as a conjunction of
binary relations of the form $x_i = x_j$.
\end{proof}

\begin{lem}\label{FindCenter}
Suppose $\rho\subseteq A\times B$ is a subdirect relation
and
$A$ is a PC algebra. Then either
for every $b\in B$ there exists a unique $a\in A$ such that $(a,b)\in \rho$,
or there exists $b\in B$ such that $(a,b)\in \rho$ for every $a\in A$.
\end{lem}

\begin{proof}
Put
$\sigma_{l}(x_1,x_2,\ldots,x_{l}) = \exists y \;\rho(x_{1},y)
\wedge\dots\wedge \rho(x_{l},y).$
It is not hard to see that 
$\sigma_{l}$ contains all constant tuples.
Therefore, 
Lemma~\ref{ReflexivePCRelations} implies that $\sigma_{2}$ is either full, or the equality relation.

If $\sigma_{2}$ is the equality relation, then
for every $b\in B$ there exists a unique $a\in A$ such that $(a,b)\in \rho$.

Suppose $\sigma_2$ is full.
Then we consider the minimal $l$, 
if it exists, such that $\sigma_{l}$ is not full.
Since $\sigma_{l-1}$ is full, 
the relation $\sigma_{l}$ contains all tuples 
whose elements are not different.
Then 
Lemma~\ref{ReflexivePCRelations}
implies that $\sigma_{l}$ is a full 
relation, which means that
$\sigma_{l}$ is a full relation for every $l$.
Substituting 
$l = |A|$ and $\{x_1,\dots,x_{l}\} = A$ in the definition of 
$\sigma_{l}$ we obtain that 
there exists $b$ such that $(a,b)\in \rho$ for every $a\in A$.
\end{proof}

\begin{lem}\label{PCRelationsLem}
Suppose $\rho\subseteq A_1\times\dots \times A_{n}$ is a subdirect relation,
$A_{i}$ is a PC algebra for every $i\in\{2,\ldots,n\}$,
and there is no nontrivial 
binary absorbing subuniverse or nontrivial center on $A_{i}$ for every $i\in\{1,\ldots,n\}$.
Then $\rho$ can be represented as a conjunction of binary relations $\delta_{1},\ldots,\delta_{k}$ such that
$\ConOne(\delta_{l},j)$ is the equality relation whenever the domain of the $j$-th variable of $\delta_{l}$ is a PC algebra.
\end{lem}

This lemma 
says that 
the relation $\rho$ can be represented by 
constraints 
from the first coordinate to an $i$-th coordinate such that 
the $i$-th coordinate is uniquely determined by the first 
(also we can define the corresponding PC congruence on the first coordinate 
using this relation)
and by bijective binary constraints between pairs of coordinates
other than first.
Also, it says that in a subdirect product of PC algebras without a nontrivial binary absorbing subuniverse or center (even $A_{1}$ is a PC algebra) we can choose some essential coordinates which can have any value, each other coordinate is uniquely determined by exactly one of them (in a bijective way).

\begin{proof}
We prove by induction on the 
arity of $\rho$. 
If $\rho$ is binary,
Lemma~\ref{FindCenter} implies that 
there exists a nontrivial binary absorbing subuniverse on $A_{2}$,
or there exists  a nontrivial center on $A_{1}$ witnessed by $\rho$, 
or the second coordinate of $\rho$ is uniquely determined by the first, 
or $\rho$ is full. First two conditions contradict our assumptions, 
the last two conditions are what we need.

Assume that $\rho$ is not essential, then it can be represented as a conjunction of essential relations satisfying the same properties.
By the inductive assumption, each of them
can be represented as a conjunction of binary relations. It remains to join these binary relations to complete the proof for this case.

Assume that $\rho$ is essential.
The projection of $\rho$ onto any proper set of variables gives a relation of a smaller arity satisfying the same properties.
By the inductive assumption, the relation of a smaller arity can be represented as a conjunction of binary relations
$\delta_{1},\ldots,\delta_{k}$ such that
$\ConOne(\delta_{l},j)$ is the equality relation whenever the domain of the $j$-th variable of $\delta_{l}$ is a PC algebra.
In each relation 
$\delta_{i}$ 
one variable (let it be the $u$-th variable of $\rho$) is uniquely determined 
by another,
and therefore 
the relation 
$\rho$ can be represented as 
a conjunction of $\delta_{i}$ 
and the projection of 
$\rho$ onto all variables but $u$-th, 
which cannot happen
with an essential relation.
Therefore, 
each projection of $\rho$ onto any proper set of variables
is a full relation.

Let us consider the relation $\rho\subseteq (A_{1}\times\dots\times A_{n-1})\times A_{n}$ as a binary relation.
By Lemma~\ref{FindCenter} we have one of the following two situations.

Case 1: 
there exist $b_{1},\ldots,b_{n-1}$ such that $(b_{1},\ldots,b_{n-1},a)\in \rho$ for every $a\in A_{n}$.
We consider the maximal $s$ such that
$\rho(b_{1},\ldots,b_{s},x_{s+1},\ldots,x_{n})$ is not a full relation.
It is easy to see that $s\le n-2$ and $s$ exists.
Let $R(x_{s+1},\ldots,x_{n}) =\rho(b_{1},\ldots,b_{s},x_{s+1},\ldots,x_{n})$.
Since the projection of $\rho$ onto any proper subset of variables is full, $R$ is a subdirect relation.
By Lemma~\ref{GenBinAbToBinAb}, 
there is no 
nontrivial 
binary absorbing subuniverse 
on 
$A_{s+2}\times\dots\times A_{n}$,
then we get a nontrivial center $C$ on $A_{s+1}$
defined
by
$C = \{a_{s+1}\in A_{s+1}\mid \forall a_{s+2}\dots\forall a_{n} \colon
(a_{s+1},a_{s+2},\ldots,a_{n})\in R\}$
and 
witnessed by $R$.

Case 2: for every $a_{1},\ldots,a_{n-1}$ there exists a unique $b$ such that
$(a_{1},\ldots,a_{n-1},b)\in \rho$.
We can show in the same way that 
for any $(a_{1},a_{3},\ldots,a_{n})$
there exists a unique $b$ such that
$(a_{1},b,a_{3}\ldots,a_{n})\in\rho$.
Let us consider the relation 
$\zeta$ defined by
\begin{align*}\zeta(z_{1},z_{2},z_{3},z_{4}) =
\exists x_{1} \exists x_{2}\dots\exists x_{n-1} \exists x_{1}' \exists x_{2}'\;
&\rho(x_{1},x_{2},x_{3},\ldots,x_{n-1},z_{1})\wedge\\
\rho(x_{1},x_{2}',x_{3},\ldots,x_{n-1},z_{2})\wedge
&\rho(x_{1}',x_{2},x_{3},\ldots,x_{n-1},z_{3})\wedge
\rho(x_{1}',x_{2}',x_{3},\ldots,x_{n-1},z_{4}).
\end{align*}
Since 
any projection of $\rho$ onto any proper subset of variables is a full relation, 
any projection of 
$\zeta$ onto 3 variables is a full relation.
Since  $\rho$ is subdirect,
$\zeta$ contains all constant tuples.
Then
Lemma~\ref{ReflexivePCRelations}
implies that 
$\zeta$ is a full relation.
Suppose $a\neq b$ and 
$(a,a,a,b)\in\zeta$
witnessed by 
$x_{1},\ldots,x_{n-1},x_{1}',x_{2}'$.
Since $z_{1}=z_{2}=a$,
we have $x_{2} = x_{2}'$ and 
therefore $z_{3}=z_{4}$, that is $a=b$. 
Contradiction.
\end{proof}


\begin{conslem}\label{PCProperties}
Suppose $\sigma_{1},\ldots,\sigma_{k}$ 
are all PC congruences on $A$.
Put $A_{i} = A/\sigma_{i}$, 
and define 
$\psi:A\to A_{1}\times \dots\times A_{k}$
by $\psi(a) = (a/\sigma_{1},\dots,a/\sigma_{k})$. Then 
\begin{enumerate}
    \item $\psi$ is surjective, hence 
    $A/\PCCon(A)\cong A_{1}\times\dots\times A_{k}$;
    \item the PC subuniverses are the sets of the form 
    $\psi^{-1}(S)$, where $S\subseteq A_{1}\times \dots\times A_{k}$
    is a relation definable by unary constraints of the form 
    $x_{j} = a_{j}$;
    \item for each nonempty PC subuniverse $B$ of $A$ there is a congruence 
    $\theta$ of $A$ such that $B$ is an equivalence class of $\theta$ 
    and 
    $A/\theta$ is isomorphic to a product of PC algebras
    having no nontrivial binary absorbing subuniverse or center.
\end{enumerate}
\end{conslem}

\begin{proof}

Consider 
the image 
$\psi(A)$, which is a subdirect subuniverse 
of $A_{1}\times\dots\times A_{k}$.
By Lemma~\ref{PCRelationsLem},
this relation can be represented as a conjunction of binary relations whose one coordinate uniquely determines another (in a bijective way).
This means that congruences $\sigma_{i}$ corresponding to 
these coordinates should be equal, which contradicts the definition.
Then $\psi(A)$ is a full relation and $\psi$ is surjective.

Claim (2) follows directly from the definition of 
a PC subuniverse.

To prove (3) consider the intersection 
of all congruences whose equivalence classes we intersected to define
the PC subuniverse.
Then, in the same way as in (1) we can prove 
the isomorphism.
\end{proof}
\begin{conslem}\label{PCImplies}
Suppose $\rho \subseteq A_{1}\times\dots\times A_{n}$ is a subdirect relation,
there is no nontrivial binary absorbing subuniverse or 
nontrivial center on $A_{1}$, and 
$C = \proj_{1}((C_{1}\times\dots \times C_{n})\cap\rho)$,
where $C_{i}$ is a PC subuniverse in $A_{i}$ for every $i$.
Then $C$ is a PC subuniverse in $A_{1}$.
\end{conslem}

\begin{proof}
By the previous corollary
for every $i$ we choose
PC algebras $A_{i,1},\dots,A_{i,k_{i}}$ 
and a mapping $\psi_{i}:A_{i}\to A_{i,1}\times\dots\times A_{i,k_{i}}$
such that 
$A_{i}/\PCCon(A_{i})\cong A_{i,1}\times \dots\times A_{i,k_{i}}$.
Define 
$\phi:A_{1}\times\dots\times A_{n}\to
A_{1}\times \prod_{i,j} A_{i,j}$ by 
$\phi(a_1,\dots,a_{n}) = (a_1,\psi_{1}(a_1),\ldots,\psi_{n}(a_n))$.
Let 
$\gamma = C_{1}\times\dots\times C_{n}$,
$\rho' = \phi(\rho)$, $\gamma'= \phi(\gamma)$.
We can check that 
$\proj_{1}(\rho\cap\gamma) = \proj_{1}(\rho'\cap\gamma')$, 
then it is sufficient to show that 
$\proj_{1}(\rho'\cap\gamma')$ is a PC subuniverse of 
$A_{1}$. 

Since $\rho'$ is subdirect, 
by Lemma~\ref{PCRelationsLem}
it can be represented 
by binary constraints from the first coordinate
to an $i$-th coordinate such that 
the $i$-th coordinate is uniquely determined by the first,
and by bijective binary constraints between pairs of coordinates
other than first.
The relation $\gamma'$ can be represented by 
constraints of the form 
$x_{i,j} = a_{i,j}$ and 
canonical constraints 
saying that the $j$-th element of $\psi_{1}(x_{1})$ 
is equal to $x_{1,j}$.
To calculate 
$\proj_{1}(\rho'\cap \gamma')$ 
we join constraints of these two representations.
Let us explain how any constraint 
from $x_{1}$ to $x_{i,j}$ in this representation looks like.
There exists a congruence $\sigma$ on $A_{1}$ 
such that $A_{1}/\sigma$ is a PC algebra isomorphic to 
$A_{i,j}$, then the constraint  
assigns to all elements of each equivalence class of $\sigma$ 
the corresponding element of $A_{i,j}$.
All other constraints of this representations are 
of the form 
$x_{i,j} = a_{i,j}$
or bijective constraints between two coordinates.
This implies that 
$\proj_{1}(\rho'\cap\gamma')$ is an intersection of equivalence classes
of PC congruences, that is, $\proj_{1}(\rho'\cap\gamma')$ is a PC subuniverse.
\end{proof}
\begin{conslem}\label{PCLessThanThree}
Suppose
$C_{i}$ is a PC subuniverse of 
$A_{i}$ for $i\in\{1,2,\dots,n\}$
and 
$n\ge 3$.
Then there does not exist a 
subdirect $(C_{1},\dots,C_{n})$-essential 
relation 
$\rho\subseteq A_{1}\times \dots \times A_{n}$.
\end{conslem}

\begin{proof}
Assume 
that such a relation  $\rho$ exists. 
By Corollary~\ref{PCProperties}
for every $i$ we choose
PC algebras $A_{i,1},\dots,A_{i,k_{i}}$ 
and a mapping $\psi_{i}:A\to A_{i,1}\times \dots\times A_{i,k_{i}}$
such that 
$A_{i}/\PCCon(A_{i})\cong A_{i,1}\times \dots\times A_{i,k_{i}}$.
Define 
$\phi:A_{1}\times\dots\times A_{n}\to
\prod_{i,j} A_{i,j}$ by 
$\phi(a_1,\dots,a_{n}) = 
(\psi_{1}(a_1),\ldots,\psi_{n}(a_n))$.
Let 
$\gamma_{i} = C_{1}\times\dots\times
C_{i-1}\times A_{i}\times C_{i+1}\times\dots\times C_{n}$
for every $i$
and 
$\gamma = 
C_{1}\times\dots\times C_{n}$.
Put 
$\rho' = \phi(\rho)$, $\gamma'= \phi(\gamma)$,
and 
$\gamma_{i}'= \phi(\gamma_{i})$ for every $i$.

Since $\rho'$ is subdirect, by Lemma~\ref{PCRelationsLem}
$\rho'$ can be represented 
by bijective binary constraints between pairs.
By Corollary~\ref{PCProperties}
$\gamma'$ can be represented by 
constraints of the form 
$x_{i,j} = a_{i,j}$.
If $\rho\cap\gamma=\varnothing$, 
then 
$\rho'\cap\gamma'=\varnothing$, 
which can only happen 
if two unary constraints defining 
$\gamma'$ assign contradictory values to 
variables with respect to the binary 
constraints defining $\rho'$.
Since $k\ge 3$, we can choose
$l$ 
such that 
$\gamma_{l}'$ includes the two contradictory
unary constraints.
Then $\rho'\cap \gamma_{l}'=\varnothing$
and $\rho\cap \gamma_{l}=\varnothing$, 
which gives a contradiction.
\end{proof}

\begin{lem}\label{BiggerThanPC}
Suppose 
$\sigma\supseteq \sigma_{1}\cap \dots\cap \sigma_{n}$,
where $\sigma_{1},\dots,\sigma_{n}$ are PC congruences on $D$ 
and $\sigma$ is a proper congruence on $D$.
Then there exists $I\subseteq \{1,2,\dots,n\}$
such that 
$\sigma = \bigcap_{i\in I}\sigma_{i}$.
\end{lem}
\begin{proof}
Consider a $2n$-ary relation 
$R\subseteq D/\sigma_{1}\times\dots\times D/\sigma_{n}
\times D/\sigma_{1}\times\dots\times D/\sigma_{n}$
consisting of all 
tuples 
$(a/\sigma_{1},\dots,a/\sigma_{n},b/\sigma_{1},\dots,b/\sigma_{n})$, where 
$(a,b)\in \sigma$.
By Lemma~\ref{PCRelationsLem}, 
$R$ can be represented as a conjunction of binary 
bijective
relations.
Since 
$(a/\sigma_{1},\dots,a/\sigma_{n},a/\sigma_{1},\dots,a/\sigma_{n})\in R$ for every $a\in D$, 
we conclude that all these binary relations are equalities.
This implies that 
$\sigma = \bigcap_{i\in I}\sigma_{i}$
for some $I\subseteq \{1,2,\dots,n\}$.
\end{proof}


\begin{lem}\label{NoAbsCenterInPCAlgebra}
For every $D$ the algebra $D/\ConPC(D)$
has no nontrivial binary absorbing subuniverse 
or center.
\end{lem}
\begin{proof}
By Corollary~\ref{PCProperties},
$D/\PCCon(D)\cong A_{1}\times\dots\times A_{k}$, 
where $A_{i}$ is a PC algebra without a nontrivial binary absorbing subuniverse or center.
Then Lemmas~\ref{GenBinAbToBinAb} and \ref{GenCenterToCenter}
imply that there cannot be a nontrivial binary absorbing subuniverse or center on $D$.
\end{proof}

\begin{lem}\label{PCCongruenceOnProduct}
Suppose 
$\sigma$ is a PC congruence on $A_{1}\times A_{2}$, 
there is no nontrivial binary absorbing subuniverse or center on $A_{1}$ and $A_{2}$.
Then there exist $i\in\{1,2\}$ and a PC congruence $\sigma_{i}$ on $A_{i}$ 
such that 
$\sigma = 
\{(\alpha,\beta)\mid (\proj_{i}(\alpha), \proj_{i}(\beta))\in\sigma_{i}\}$.
\end{lem}
\begin{proof}
First, consider 
$S\subseteq A_{1}\times (A_{1}\times A_{2})/\sigma$ 
consisting of all pairs 
$(a_1,(a_1,a_2)/\sigma)$ such that $a_1\in A_{1}$, $a_{2}\in A_{2}$.
Since there is no nontrivial binary absorbing subuniverse or center on $A_{1}$,  Lemma~\ref{PCRelationsLem} implies that either 
$S$ is a full relation, or 
$\ConOne(S,2)$ is the equality relation.
In the latter case
the congruence $\sigma$ depends only on the first coordinate, 
that is, 
there exists a PC congruence 
$\sigma_{1}$ on $A_{1}$ 
such that 
$\sigma = 
\{(\alpha,\beta)\mid (\proj_{1}(\alpha), \proj_{1}(\beta))\in\sigma_{1}\}$,
which completes this case.

Thus, we assume that $S$ is a full relation and 
for every $a_1\in A_1$ and every equivalence class
$E$ of $\sigma$ there exists $a_{2}\in A_{2}$ such 
that $(a_{1},a_{2})\in E$.
In the same way we assume that 
for every $a_2\in A_2$ and every equivalence class
$E$ of $\sigma$ there exists $a_{1}\in A_{1}$ such 
that $(a_{1},a_{2})\in E$.

Choose an element $c_1\in A_{1}$.
By $\sigma_{2}$ we denote the congruence 
$\{(a_{2},a_{2}')\mid 
((c_{1},a_{2}),(c_{1},a_{2}'))\in\sigma\}$.
As it follows from the above assumptions, 
$A_{2}/\sigma_{2}\cong (A_{1}\times A_{2})/\sigma$.
Consider the ternary relation
$\rho\subseteq A_{1}\times (A_{1}\times A_{2})/\sigma\times 
A_{2}/\sigma_{2}$ consisting of 
all the tuples 
$(a_1,(a_{1},a_{2})/\sigma,a_{2}/\sigma_{2})$,
where $a_{1}\in A_{1}$ and $a_{2}\in A_{2}$.
As we already know, the projection of $\rho$ onto any two coordinates 
is a full relation. Then Lemma~\ref{PCRelationsLem}
implies that $\rho$ is a full relation, 
which contradicts the fact that 
$(c_{1},(c_{1},a_{2})/\sigma,b_{2}/\sigma_{2})\notin\rho$ 
for any  $(a_{2},b_{2})\notin \sigma_{2}$.
\end{proof}

\begin{conslem}\label{PCCongruenceOnProductGen}
Suppose 
$\sigma$ is a PC congruence on $A_{1}\times A_{2}\times\dots\times A_{n}$, 
there is no nontrivial binary absorbing subuniverse or center on $A_{i}$ for every $i$.
Then there exist $i\in\{1,2,\dots,n\}$ and a PC congruence $\sigma_{i}$ on $A_{i}$ 
such that 
$\sigma = 
\{(\alpha,\beta)\mid (\proj_{i}(\alpha), \proj_{i}(\beta))\in\sigma_{i}\}$.
\end{conslem}

\begin{proof}
We prove this corollary by induction on $n$.
For $n=2$ it follows from Lemma~\ref{PCCongruenceOnProduct}.
By Lemmas~\ref{GenBinAbToBinAb}, \ref{GenCenterToCenter},
there is no nontrivial binary absorbing subuniverse or center on 
$A_{2}\times\dots\times A_{n}$.
We apply 
Lemma~\ref{PCCongruenceOnProduct} to 
$A_{1}\times(A_{2}\times\dots\times A_{n})$ 
to get a PC congruence on 
$A_{1}$ or on $A_{2}\times\dots\times A_{n}$.
In the latter case we apply the inductive assumption to complete the proof.
\end{proof}

\begin{lem}\label{PCSubuniverseOnProduct}
Suppose 
$B$ is a PC subuniverse on 
$A_{1}\times\dots\times A_{n}$, 
and there is no nontrivial binary absorbing subuniverse or center on $A_{i}$
for every $i$.
Then there exists 
a PC subuniverse 
$B_{i}$ on $A_{i}$ for every $i$ such that 
$B = B_{1}\times\dots\times B_{n}$.
\end{lem}
\begin{proof}
Assume that 
$B=E_{1}\cap\dots\cap E_{t}$, 
where $E_{i}$ is an equivalence class of a PC congruence $\sigma_{i}$ on $A_{1}\times\dots\times A_{n}$ for every $i$.
By Corollary~\ref{PCCongruenceOnProductGen},
for every $i$ there exists 
$s_{i}$ and a congruence $\sigma_{i}'$ on $A_{s_{i}}$ such that 
$\sigma_{i} = \{(\alpha,\beta)\mid (\proj_{s_{i}}(\alpha), \proj_{s_{i}}(\beta))\in\sigma_{i}'\}$.
Then there exists an equivalence class $E_{i}'$ of $\sigma_{i}'$
such that 
$E_{i} = 
A_{1}\times\dots\times A_{s_{i}-1}\times 
E_{i}'\times A_{s_{i}+1}\times\dots\times A_{n}$.
Hence, the intersection $E_{1}\cap\dots\cap E_{t}$ 
is equal to $B_{1}\times\dots\times B_{n}$
for PC subuniverses $B_{1},\dots,B_{n}$. 
\end{proof}

\begin{lem}\label{CenterProvidesBinaryAbsorptionForPC}
Suppose
$\rho\subseteq A\times B$ is a subdirect relation,
$A$ is a PC algebra without nontrivial binary absorbing subuniverse or center, and 
$C = \{b\in B\mid \forall a\in A\colon (a,b)\in \rho\}$.
Then $C$ binary absorbs $B$.
\end{lem}

\begin{proof}
Suppose $A = \{a_{1},\ldots,a_{k}\}$.
Let us consider the matrix $M$
whose rows are the tuples
$(\underbrace{a,a,\ldots,a}_{k+1},b,a_{1},\ldots,a_{k})$
and $(b,a_{1},\ldots,a_{k},\underbrace{a,a,\ldots,a}_{k+1})$
for all $a,b\in A$.
The $2k+2$ columns of this matrix we denote by
$\alpha_{1},\ldots,\alpha_{2k+2}$.
By $\beta$ we denote the tuple of length $2k^2$
such that the $i$-th element of $\beta$ equals $b$ from the corresponding row.
By Lemma~\ref{PCRelationsLem},
the relation
generated by $\alpha_{1},\ldots,\alpha_{2k+2}$
is a full relation.
Hence, there exists
a term operation $f$
such that
$f(\alpha_{1},\ldots,\alpha_{2k+2})=\beta$.
Let us show that $C$ absorbs $B$ with the term operation defined by
$h(x,y)=f(\underbrace{x,\ldots,x}_{k+1},y,\ldots,y)$.
Suppose $d\in B$, $c\in C$.
Assume that $h(d,c)=e\notin C$.
Choose elements $a,a'\in A$ such that
$(a,e)\notin\rho$ and
$(a',d)\in\rho$.
Consider the row
$(a',\ldots,a',a,a_{1},\ldots,a_{k})$ from the matrix.
We know that
$f$ returns $a$ on this tuple
and
$f(\underbrace{d,\ldots,d}_{k+1},c,\ldots,c) = e$, which contradicts the fact that
$f$ preserves $\rho$. Thus, $h(d,c)\in C$.

In the same way we can prove that $h(c,d)\in C$
for every $d\in B$, $c\in C$.
\end{proof}

\begin{lem}\label{IdentificationDoesNotReducePC}
Suppose
$\rho\subseteq A\times B\times B$ is a subdirect relation,
$A$ is a PC algebra without a nontrivial binary absorbing subuniverse or center,
and 
for every $b\in B$ there exists $a\in A$ such that
$(a,b,b)\in \rho$.
Then for every $a\in A$ there exists $b\in B$ such that
$(a,b,b)\in\rho$.
\end{lem}
\begin{proof}
We prove the lemma by induction on the size of $B$.

By Lemma~\ref{FindCenter}, only two situations are possible: either
there exist $c_{1},c_{2}\in B$ such that
$(a,c_{1},c_{2})\in \rho$ for every $a\in A$,
or
for each $(b_{1},b_{2})\in \proj_{2,3}(\rho)$ there exists a unique $a\in A$
such that $(a,b_{1},b_{2})\in \rho$.

Case 1. There exist $c_{1},c_{2}\in B$ such that
$(a,c_{1},c_{2})\in \rho$ for every $a\in A$.
Put $D = \{(b,c)\mid \forall a\in A\colon (a,b,c)\in\rho\}$.
By Lemma~\ref{CenterProvidesBinaryAbsorptionForPC}, $D$ is a binary absorbing subuniverse in the projection
of $\rho$ onto the last two variables.
By Lemma~\ref{AbsorbingEquality}, there exists $(b,b)\in D$. This completes this case.

Case 2. For each $(b_{1},b_{2})\in \proj_{2,3}(\rho)$ there exists a unique $a\in A$
such that $(a,b_{1},b_{2})\in \rho$.
Let $\delta_{1}$ be the projection of $\rho$ onto the first two variables.
By Lemma~\ref{FindCenter} we have one of two situations.

Case 2A. For every $b\in B$ there exists a unique $a$ such that
$(a,b)\in\delta_{1}$. 
Since $\rho$ is subdirect, 
for every $a$ there exists 
$(a,b,b')\in\rho$, which implies that $(a,b,b)\in\rho$ and completes this case.

Case 2B. There exists an element $b$ such that $(a,b)\in\delta_{1}$ for every $a\in A$.
Consider the relation
$\delta_{2}(x,y_{2}) = \rho(x,b,y_{2})$.
If $\proj_{2}(\delta_{2})\neq B$, then we restrict the last two variables of $\rho$ to $\proj_{2}(\delta_{2})$
and apply the inductive assumption.
Assume that $\proj_{2}(\delta_{2}) = B$.
By the definition of the second case 
we know that
for every $c\in B$ there exists a unique $a$ such that
$(a,c)\in\delta_{2}$.
Then $\sigma=\ConOne(\delta_2,2)$ 
is a proper congruence such that
$B/\sigma\cong A$.
If $\sigma$ is the equality relation, then $B\cong A$,
and, by Lemma~\ref{PCRelationsLem},
$\rho$ can be represented by binary bijective constraints.
If the first coordinate of $\rho$ is uniquely defined by the second or the third, then it is equivalent to the case 2A, which we already considered.
If the first coordinate of $\rho$ does not depend on the others, 
then the claim is trivial.


If $\sigma$ is not the equality relation, then
we consider the relation $\rho'$ obtained from $\rho$ by factorization of the last two variables by $\sigma$, 
that is, 
$\rho'\subseteq A\times B/\sigma\times B/\sigma$
contains all tuples 
$(a,b/\sigma,b'/\sigma)$
such that $(a,b,b')\in\rho$.
By the inductive assumption for any $a\in A$ there exists $E\in B/\sigma$ such that
$(a,E,E)\in\rho'$.
By Lemma~\ref{FindCenter}, we have one of the following situations.
Case 1. There exists $E\in B/\sigma$ such that
for every $a\in A$ we have $(a,E,E)\in\rho'$.
Then we restrict the last two variables of $\rho$ to $E$ and apply the inductive assumption.
Case 2. For every $E\in B/\sigma$ there exists a unique $a\in A$
such that $(a,E,E)\in\rho'$.
In this case for any $a\in A$ we
choose $E$ such that $(a,E,E) \in \rho'$.
By the uniqueness of $a$ we have
$(a,b,b)\in \rho$ for any $b\in E$, which completes the proof.
\end{proof}


\subsection{Linear Subuniverse}

We have the following well-known fact from linear algebra 
\cite{greub2012linear}.
\begin{lem}\label{LinearAlgebrasFact}
Suppose 
$\rho\subseteq  (\mathbb Z_{p_{1}})^{n_{1}}\times\dots\times
(\mathbb Z_{p_{k}})^{n_{1}}$, 
where 
$p_{1},\ldots,p_{k}$ are 
distinct prime numbers dividing 
$m-1$
and 
$\mathbb Z_{p_{i}} = (\mathbb Z_{p_{i}};x_{1}+\dots+x_{m})$
for every $i$.
Then 
$\rho = L_{1}\times\dots\times 
L_{k}$ where each 
$L_{i}$ is an affine subspace 
of ${(\mathbb Z_{p_{i}})}^{n_{i}}$.
\end{lem}

\begin{conslem}\label{LinearAlgebrasAreClosed}
The set of linear algebras is closed under 
taking subalgebras, quotients, 
and finite products.
\end{conslem}

\begin{lem}\label{NoAbsCenterPCInLinearAlgebra}
A linear algebra has no nontrivial absorbing subuniverse, nontrivial center, or nontrivial PC subuniverse.
\end{lem}

\begin{proof}
Let us prove that 
a linear algebra $A$ has no nontrivial 
absorbing subuniverse, 
which by Corollary~\ref{centerImpliesAbsorption} 
implies that 
$A$ has no nontrivial 
center. 
By Lemma~\ref{GenBinAbToBinAb},
it is sufficient to show 
that $\mathbb Z_{p}$ has no nontrivial 
absorbing subuniverse.
Every term operation in $\mathbb Z_{p}$ 
can be represented 
as 
$a_{1}x_{1}+\dots+a_{l}x_{l}$, 
and for each $a_{l}\neq 0$
fixing all variables but $x_{l}$ 
to some values gives a bijective mapping, 
which means that this term cannot witness
an absorption.

Since linear algebras (by Corollary~\ref{LinearAlgebrasAreClosed}) are closed under quotients, to prove that it does not have a nontrivial PC subuniverse 
it is sufficient to prove that a linear algebra $A$ cannot be 
a PC algebra.
Assume that $A$ is isomorphic to 
$\mathbb Z_{p_{1}}\times\dots \times\mathbb 
Z_{p_{k}}$ for prime numbers 
$p_{1},\dots,p_{k}$, 
and $\psi:A\to Z_{p_{1}}$ 
is the canonical mapping.
Let $\rho$ be the set of 
all tuples $(a,b,c,d)$ such that 
$\psi(a)+\psi(b) = \psi(c) +\psi(d)$.
We can check that $\rho$ is preserved by 
$w$ and all constants but not all operations, 
therefore $A$ cannot be polynomially complete.
\end{proof}

\begin{lem}\label{EqualNumberOfElements}
Suppose $\rho\subseteq A_{1}\times A_2$ is a subdirect relation,
$A_{2}$ is a linear algebra, and there is no nontrivial binary absorbing subuniverse on $A_1$.
Then for all $a,b\in A_1$ we have
$$|\{c\mid (a,c)\in \rho\}| = |\{c\mid (b,c)\in \rho\}|.$$
\end{lem}

\begin{proof}
Assume the contrary, then we choose
all elements $a$ with the maximal $|\{c\mid (a,c)\in \rho\}|.$
Denote the set of such elements by $C$.

Since $w(a_{1},\ldots,a_{i-1},x,a_{i+1},\ldots,a_{m})$ is a bijection on $A_{2}$
for every $a_{1},\ldots,a_{m}\in A_{2}$,
we have
$w(A_{1},\ldots,A_{1},C,A_{1},\ldots,A_{1})\subseteq C$.
Hence $w(x,\ldots,x,y)$ is a binary absorbing operation and $C$ is a binary absorbing subuniverse.
\end{proof}

\begin{lem}\label{LinearSpecialWNU}
Suppose $A$ is a linear algebra.
Then $w(a,b,\dots,b) = a$ for every $a,b\in A$.
\end{lem}
\begin{proof}
Suppose 
$A\cong \mathbb Z_{p_{1}}\times\dots\times \mathbb Z_{p_{k}}$.
Since the WNU $w$ is special and idempotent, 
each $p_{i}$ divides $m-1$.     
Therefore, $w(a,b,\ldots,b) = a$
for every $a,b\in A$.
\end{proof}

\begin{lem}\label{RelWithLinearPart}
Suppose $\rho\subseteq A_{1}\times A_2$ is a subdirect relation,
$A_{2}$ is a linear algebra,
and there is no nontrivial binary absorbing subuniverse on $A_1$.
Then $\rho$ has the parallelogram property.
\end{lem}

\begin{proof}
First, we define a relation $\sigma_{k}$ for every $k\ge 2$ by
$$\sigma_{k}(y_1,\ldots,y_k) = \exists x\; \rho(x,y_1)\wedge\dots\wedge\rho(x,y_k).$$

By Lemma~\ref{LinearSpecialWNU},
$w(a,b,\ldots,b,b) = a$ 
and 
$w(b,b,\ldots,b,c) = c$
for any $a,b,c\in A_{2}$, 
therefore
$(a,b),(b,c),(b,b)\in\sigma_{2}$ 
implies 
$(a,c)\in\sigma_{2}$,
Since $\sigma_2$ is reflexive and symmetric, 
it is a congruence.

Let us show by induction on $k$ that $\sigma_{k}(y_{1},\ldots,y_{k}) =
\bigwedge_{i=2}^{k}\sigma_{2}(y_{1},y_{i})$.
For $k=2$ it is obvious. 
Consider a
tuple 
$(a_{1},\ldots,a_{k})$ such that 
$(a_{i},a_{j})\in\sigma_{2}$ for any $i,j$.
By the inductive assumption for $k-1$
we have 
$(a_{1},a_{1},a_{3},\ldots,a_{k}), 
(a_{1},a_{2},a_{1},a_{4},\ldots,a_{k}),
(a_{1},a_{1},a_{1},a_{4},\ldots,a_{k})
\in\sigma_{k}$.
If we apply the term operation 
$g(x,y,z) = w(x,y,z,\ldots,z)$ 
to these three tuples (in the same order) 
we obtain $(a_{1},\ldots,a_{k})$, 
which means that 
$(a_{1},\ldots,a_{k})\in \sigma_{k}$.
Thus 
$\sigma_{k}(y_{1},\ldots,y_{k}) =
\bigwedge_{i=2}^{k}\sigma_{2}(y_{1},y_{i})$ for every $k$.

Substituting 
$\{y_{1},\dots,y_{k}\} = E$ 
in the definition of $\sigma_{k}$ for an equivalence 
class $E$ of $\sigma_{2}$ 
we derive that there exists $c\in A_{1}$ such
that $(c,d)\in\rho$ for any $d\in E$.
Then it follows from Lemma~\ref{EqualNumberOfElements}, that $\rho$ has 
the parallelogram property.
\end{proof}


\begin{conslem}\label{LinearImplies}
Suppose $\rho \subseteq A_{1}\times\dots\times A_{n}$ is a relation such that
$\proj_1 (\rho) = A_{1}$,
there is no nontrivial binary absorbing subuniverse on $A_{1}$,
and
$C = \proj_{1}((C_{1}\times\dots \times C_{n})\cap\rho)$,
where $C_{i}$ is a linear subuniverse of $A_{i}$ for every $i$.
Then $C$ is a linear subuniverse of $A_{1}$.
\end{conslem} 

\begin{proof}
Let
$\psi:A_{1}\times \dots\times A_{n}
\to A_{1}\times A_{2}/\ConLin(A_{2})\times 
\dots\times 
A_{n}/\ConLin(A_{n})$
be a natural homomorphism.
Put
$\gamma = C_{1}\times\dots\times C_{n}$, 
$\rho' = \psi(\rho)$, 
$\gamma' = \psi(\gamma)$.
We can check that 
$\proj_1 (\rho\cap \gamma) = \proj_1 (\rho'\cap \gamma')$.
The relation $\rho'$ can be viewed as 
a subdirect subalgebra of $A_{1}\times B$, where 
$B = \proj_{2,\ldots,n}(\rho')$ 
is a linear algebra by Corollary~\ref{LinearAlgebrasAreClosed}.
Then $D = \proj_{2,\ldots,n}(\gamma')\cap B$ 
can be viewed as a subalgebra of $B$.
We need to show that 
$\proj_{1}(\rho'\cap (C_{1}\times D))$ is a linear subuniverse of $A_{1}$.
By Lemma~\ref{RelWithLinearPart}, 
the binary relation $\rho'$ has the parallelogram property, 
then $\rho'$ induces 
an isomoprhism 
$A_{1}/\sigma_{1}\cong B/\sigma_{2}$, 
where 
$\sigma_{1} = \ConOne(\rho',1)$, 
$\sigma_{2} = \ConOne(\rho',2)$.
Note that 
$\sigma_{1}$ is a linear congruence since
$B$ is a linear algebra.
Hence, $D_1 = \{a\in A_{1}\mid 
\exists d\in D\colon (a,d)\in \rho'\}$ is stable under $\sigma_{1}$ (and under $\ConLin(A_{1})$).
Then $\proj_1 (\rho\cap \gamma) = \proj_1 (\rho'\cap \gamma') = C_{1}\cap D_1$ is stable under $\ConLin(A_{1})$, which completes the proof.
\end{proof}

\subsection{Common properties}

In this subsection we list some properties that are common 
for all types of one-of-four subuniverses. 

\begin{lem}\label{PCBrel}
Suppose 
$R\subseteq D_{1}\times\dots\times D_{n}$ is a subdirect relation, 
$B_{i}$ is a one-of-four subuniverse of $D_{i}$ 
of type $\mathcal T$
for every $i\in\{1,\dots,n\}$;
if $\mathcal T$ is the absorbing type 
then the absorbing subuniverses are witnessed by the same term 
operation.
Then 
$R\cap (B_{1}\times\dots\times B_{n})$ is a one-of-four subuniverse of $R$
of type $\mathcal T$.
\end{lem}

\begin{proof}
If $\mathcal T$ is the absorbing type, 
then the statement follows from 
Lemma~\ref{AbsImplies}, 
if $\mathcal T$ is the central type, 
then the statement follows from 
Lemma~\ref{CenterImplies}.

Suppose $\mathcal T$ is the linear type, 
then put $\sigma_{i}=\ConLin(D_{i})$ 
for each $i\in\{1,\dots,n\}$.
First, extend every $\sigma_{i}$ naturally on 
$D = D_{1}\times\dots\times D_{n}$ and 
denote the obtained congruence $\sigma_{i}'$  
so that 
$D/\sigma_{i}'\cong D_{i}/\sigma_{i}$.
Since 
linear algebras 
are closed under taking subalgebras and quotients
(Corollary~\ref{LinearAlgebrasAreClosed}), 
$\sigma = \sigma_{1}'\cap\dots\cap \sigma_{1}'$ is a linear congruence
and $B_{1}\times\dots\times B_{n}$ is stable under this congruence.
Therefore, $\sigma\cap (R\times R)$ is a linear congruence and 
$R\cap (B_{1}\times\dots\times B_{n})$ is stable under it.
This completes this case.

It remains to consider the case when $\mathcal T$ is the PC type.
Let $\delta_{1},\dots,\delta_{t}$ be the set of all PC congruences on 
$D_{1},\ldots,D_{n}$ we need to define
$B_{1},\dots,B_{n}$. 
For every $i\in\{1,2,\dots,t\}$ by $\delta_{i}'$ we denote 
$\delta_{i}$ naturally extended on 
$D = D_{1}\times\dots\times D_{n}$, 
by $E_{i}$ we denote the equivalence class 
of $\delta_{i}'$ containing $B_{1}\times\dots\times B_{n}$.
Since $R$ is subdirect, 
$R/\delta_{i}'\cong D/\delta_{i}'$
and $R/\delta_{i}'$ is a PC algebra without a nontrivial binary absorbing subuniverse or center.
Since 
$R\cap (B_{1}\times\dots\times B_{n}) = 
R\cap (E_{1}\cap \dots\cap E_{t})$, 
the set $R\cap (B_{1}\times\dots\times B_{n})$ is a PC subuniverse of $R$.
\end{proof}

\begin{lem}\label{FactorByStableCongruence}
Suppose 
$\sigma$ is a congruence on $D$, 
$B$ is a one-of-four subuniverse of $D$ stable under $\sigma$.
Then 
$\{b/\sigma\mid b\in B\}$ is a one-of-four subuniverse of $D/\sigma$ of the same type as $B$.
\end{lem}
\begin{proof}
For a binary subuniverse and a center it follows from 
Corollaries~\ref{AbsorptionQuotient} and \ref{CenterQuotient},
respectively.
Suppose $B$ is a linear subuniverse.
Let $\delta$ be the minimal congruence containing both 
$\sigma$ and $\ConLin(D)$.
By Corollary~\ref{LinearAlgebrasAreClosed}, 
$D/\delta$ is a linear algebra.
Since $B$ is stable under $\delta$, 
$\{b/\sigma\mid b\in B\}$ is a linear subuniverse of $D/\sigma$.

It remains to consider the case when 
$B$ is a PC subuniverse of $D$, that is, 
$B=E_{1}\cap \dots \cap E_{s}$, where 
$E_{i}$ is an equivalence class of 
a PC congruence $\sigma_{i}$ for every $i$.
Let $\delta$ be the minimal congruence containing 
$\sigma$ and $\sigma_{1}\cap\dots\cap\sigma_{s}$.
By Lemma~\ref{BiggerThanPC}, 
$\delta$ is an intersection of 
PC congruences 
$\delta_{1},\ldots,\delta_{t}$.
Since $B$ is stable under $\delta$
and $B$ is an equivalence class
of $\sigma_{1}\cap\dots\cap\sigma_{s}$,
$B$ is an equivalence class of $\delta$.
Hence, 
$\{b/\sigma\mid b\in B\}$ is an intersection of the equivalence 
classes of congruences on $D/\sigma$ corresponding to 
$\delta_{1},\dots,\delta_{t}$.
\end{proof}

The
following corollaries (proved earlier)
state that if we restrict all coordinates 
of a relation to one-of-four subuniverses 
of type $\mathcal T$
then we restrict its projection onto the first coordinate 
to a subuniverse of type $\mathcal T$.
The only difference is that for PC subuniverse we require
the relation to be subdirect and without nontrivial binary absorbing subuniverse  
or center on every coordinate, 
and for linear subuniverse the first coordinate should be without a
nontrivial binary absorbing subuniverse.

\begin{AbsImpliesConsCorollary}
Suppose $\rho \subseteq A_{1}\times\dots\times A_{n}$ is a relation such that
$\proj_1 (\rho) = A_{1}$ 
and 
$C = \proj_{1}((C_{1}\times\dots \times C_{n})\cap\rho)$,
where $C_{i}$ is an absorbing subuniverse in $A_{i}$ with a term $t$ for every $i$.
Then $C$ is an absorbing subuniverse in $A_{1}$ with the term $t$.
\end{AbsImpliesConsCorollary}

\begin{CenterImpliesConsCorollary}
Suppose $\rho \subseteq A_{1}\times\dots\times A_{n}$ is a relation such that
$\proj_1 (\rho) = A_{1}$ and
$C = \proj_{1}((C_{1}\times\dots \times C_{n})\cap\rho)$,
where $C_{i}$ is a center in $A_{i}$ for every $i$.
Then $C$ is a center in $A_{1}$.
\end{CenterImpliesConsCorollary}

\begin{PCImpliesCorollary}
Suppose $\rho \subseteq A_{1}\times\dots\times A_{n}$ is a subdirect relation,
there is no nontrivial binary absorbing subuniverse or 
nontrivial center on $A_{1}$, and 
$C = \proj_{1}((C_{1}\times\dots \times C_{n})\cap\rho)$,
where $C_{i}$ is a PC subuniverse in $A_{i}$ for every $i$.
Then $C$ is a PC subuniverse in $A_{1}$.
\end{PCImpliesCorollary}

\begin{LinearImpliesCorollary}
Suppose $\rho \subseteq A_{1}\times\dots\times A_{n}$ is a relation such that
$\proj_1 (\rho) = A_{1}$,
there is no nontrivial binary absorbing subuniverse on $A_{1}$,
and
$C = \proj_{1}((C_{1}\times\dots \times C_{n})\cap\rho)$,
where $C_{i}$ is a linear subuniverse of $A_{i}$ for every $i$.
Then $C$ is a linear subuniverse of $A_{1}$.
\end{LinearImpliesCorollary}

Another common property is
that we cannot have $(C_{1},\dots,C_{k})$-essential relation of arity greater than 2 
if $C_{1},\ldots,C_{k}$ are subuniverses of a fixed type (any but linear).
Note that for PC subuniverses we additionally require 
the relation to be subdirect.
From these claims it can be derived that for the nonlinear case (see Corollary~\ref{BoundedWidthCase}) 
it is sufficient to check cycle consistency (all calculations are on binary relations)
to guarantee a solution.

\begin{lem}\label{BinAbsLessThanTwoCorollary}
Suppose
$C_{i}$ is a nontrivial binary absorbing subuniverse of
$A_{i}$ with a term $t$ for $i\in\{1,2,\dots,k\}$, 
$k\ge 2$.
Then there does not exist a 
$(C_{1},\dots,C_{k})$-essential 
relation 
$\rho\subseteq A_{1}\times \dots \times A_{k}$.
\end{lem}
\begin{proof}
Assume that such a relation exists.
To get a contradiction it is sufficient to apply term $t$ to 
a tuple from 
$A_{1}\times C_{2}\times\dots\times C_{k}$ and 
a tuple from 
$C_{1}\times \dots\times C_{k-1}\times A_{k}$.
\end{proof}

The following two corollaries were proved earlier.

\begin{CenterLessThanThreeCorollary}
Suppose
$C_{i}$ is a center of 
$A_{i}$ for $i\in\{1,2,\dots,k\}$, 
$k\ge 3$.
Then there does not exist a 
$(C_{1},\dots,C_{k})$-essential 
relation 
$\rho\subseteq A_{1}\times \dots \times A_{k}$.
\end{CenterLessThanThreeCorollary}

\begin{PCLessThanThreeCorollary}
Suppose
$\rho\subseteq A_{1}\times \dots \times A_{n}$ is a subdirect relation,
$n\ge 3$,
$C_{i}$ is a PC subuniverse in $A_{i}$.
There does not exist
a $(C_{1},\dots,C_{n})$-essential relation.
\end{PCLessThanThreeCorollary}

\subsection{Interaction}

Here we explain how one-of-four subuniverses of 
different types interact with each other.

\begin{lem}\label{PCBsubNonPC}
Suppose 
$B_1$ is a binary absorbing, central, or linear subuniverse of $D$,
$B_{2}$ is a subuniverse of $D$.
Then $B_{1}\cap B_{2}$ is a binary absorbing, central, or linear subuniverse of $B_{2}$,respectively.
\end{lem}

\begin{proof}
If $B_{1}$ is a binary absorbing subuniverse or a center, 
then the claim follows from 
Lemmas~\ref{AbsImplies} and \ref{CenterImplies}, respectively.
If $B_{1}$ is a linear subuniverse,
then 
by Corollary~\ref{LinearAlgebrasAreClosed}
$B_{2}/\ConLin(D)$ is a linear algebra, 
hence $B_{1}\cap B_{2}$ is a linear subuniverse of $B_{2}$.
\end{proof}

\begin{lem}\label{IntersectionOfTwoSubuniverses}
Suppose $B_{1}$ and $B_{2}$ 
are nonempty one-of-four subuniverses of $D$, 
$B_{1}\cap B_{2} = \varnothing$.
Then $B_{1}$ and $B_{2}$ are subuniverses of the same type.
\end{lem}

\begin{proof}
Assume the converse. Consider all possible cases.

Case 1. $B_{1}$ is a linear subuniverse, $B_{2}$ is a binary absorbing subuniverse.
By Corollary~\ref{AbsorptionQuotient}
$\{b/\ConLin(D)\mid b\in B_{2}\}$ 
is a binary absorbing subuniverse on $D/\ConLin(D)$.
By Lemma~\ref{NoAbsCenterPCInLinearAlgebra} 
this subuniverse should be trivial, which contradicts the 
fact that $B_{1}\cap B_{2}=\varnothing$ and $B_{1}$ is stable 
under $\ConLin(D)$.

Case 2. $B_{1}$ is a linear subuniverse, $B_{2}$ is a center.
By Corollary~\ref{CenterQuotient}
$\{b/\ConLin(D)\mid b\in B_{2}\}$ 
is a center of $D/\ConLin(D)$.
By Lemma~\ref{NoAbsCenterPCInLinearAlgebra} 
this subuniverse should be trivial, which contradicts the 
fact that $B_{1}\cap B_{2}=\varnothing$ and $B_{1}$ is stable 
under $\ConLin(D)$.

Case 3. $B_{1}$ is a linear subuniverse, $B_{2}$ is a PC subuniverse.
Let 
$S\subseteq (D/\ConLin(D))\times D$
consist of all the tuples $(c/\ConLin(D),c)$, 
where $c\in D$.
By Lemma~\ref{NoAbsCenterPCInLinearAlgebra}, there is no 
nontrivial 
binary absorbing subuniverse or center on $D/\ConLin(D)$.
Hence, by Corollary~\ref{PCImplies}, 
the restriction of  
the second variable to 
$B_{2}$ 
implies the restriction of 
the first variable to a PC subuniverse.
Since $B_{1}\cap B_{2}=\varnothing$, 
this restriction is nontrivial.
Thus, there exists a nontrivial PC subuniverse on $D/\ConLin(D)$,
which contradicts Lemma~\ref{NoAbsCenterPCInLinearAlgebra}.

Case 4. $B_{1}$ is a PC subuniverse, $B_{2}$ is a binary absorbing subuniverse.
By Corollary~\ref{AbsorptionQuotient}
the set $\{b/\ConPC(D)\mid b\in B_{2}\}$ 
is a binary absorbing subuniverse of $D/\ConPC(D)$.
By Lemma~\ref{NoAbsCenterInPCAlgebra} 
this subuniverse should be trivial, which contradicts the 
fact that $B_{1}\cap B_{2}=\varnothing$ and $B_{1}$ is a 
PC subuniverse.

Case 5. $B_{1}$ is a PC subuniverse, $B_{2}$ is a center.
By Corollary~\ref{CenterQuotient}
$\{b/\ConPC(D)\mid b\in B_{2}\}$ 
is a center of $D/\ConPC(D)$.
By Lemma~\ref{NoAbsCenterInPCAlgebra} 
this subuniverse should be trivial, which contradicts the 
fact that $B_{1}\cap B_{2}=\varnothing$ and $B_{1}$ is a 
PC subuniverse.

Case 6. $B_{1}$ is a binary absorbing subuniverse, 
$B_{2}$ is a center.
Suppose $R\subseteq D\times G$ is the binary relation from the
definition of the center $B_{2}$, and 
denote $b^{+} = \{a\mid (b,a)\in R\}$ for every $b\in D$.
We prove this case by induction on the size of $D$.
Assume that $b_1^{+}\neq b_2^{+}$ for some 
$b_{1},b_{2}\in B_{1}$.
Choose an element $c\in b_1^{+}\setminus b_{2}^{+}$ 
(or in $b_2^{+}\setminus b_{1}^{+}$).
Put $D' = \{a\mid (a,c)\in R\}$.
Note that 
$D'\subsetneq D$, $D'\cap B_{1}\neq \varnothing$,
$D'\cap B_{2} = B_{2}$. Thus, we obtain 
subuniverses $B_{1}\cap D'$ and $B_{2}$ 
of a smaller set $D'$
that are a binary absorbing subuniverse 
 and a center
(by Lemma~\ref{PCBsubNonPC}), respectively.
It remains to apply the inductive assumption 
to $B_{1}\cap D'$ and $B_{2}$.
Let us consider the case when $b_{1}^{+} = b_{2}^{+}$
for any $b_{1},b_{2}\in B_{1}$.
Since $B_{1}\cap B_{2}=\varnothing$, 
$b^{+}\neq G$ for every $b\in B_{1}$.
Let $f$ be the binary absorbing operation.
Choose $b\in B_{1}$ and $e\in B_{2}$.
Then 
$f(b,e)= b_{1}\in B_{1}$ and 
$f(e, b)= b_{2}\in B_{1}$, which means that 
$f(b^{+},G)\subseteq b_{1}^{+} = b^{+}$,
$f(G,b^{+})\subseteq b_{2}^{+} = b^{+}$.
This contradicts the definition of a center, 
saying that there is no nontrivial binary absorbing subuniverse on $G$.
\end{proof}

\begin{thm}\label{PCBsub}
Suppose 
$B_1$ and $B_{2}$ are one-of-four subuniverses of $D$ of types $\mathcal T_{1}$
and $\mathcal T_{2}$, respectively.
Then $B_{1}\cap B_{2}$ is a one-of-four subuniverse of $B_{2}$ of type 
$\mathcal T_{1}$.
\end{thm}

\begin{proof}
If $B_{1}$ is not a PC subuniverse, 
then the claim follows from Lemma~$\ref{PCBsubNonPC}$.
Assume that $B_{1}$ is a PC subuniverse of $D$.

Let $\sigma_{1},\ldots,\sigma_{t}$ be the set of all 
PC congruences on $D$.
Assume that $B_{2}$ is not a 
PC subuniverse.
Every equivalence class $E$
of $\sigma_{i}$ is a PC subuniverse.
Then Lemma~\ref{IntersectionOfTwoSubuniverses}
implies that $E$ has a nonempty intersection with $B_{2}$.
Therefore $B_{2}/\sigma_{i}\cong D/\sigma_{i}$ 
and $\sigma_{i}\cap (B_{2}\times B_{2})$ is 
a PC congruence on $B_{2}$ for every $i$.
Hence, $B_{1}\cap B_{2}$ is a PC subuniverse of $B_{2}$, which completes this case.

If $B_{2}$ is also a PC subuniverse,
then by Corollary~\ref{PCProperties}, 
$B_{1}\cap B_{2}$ is a PC subuniverse of $B_{2}$.
\end{proof}

\begin{lem}\label{SequencesOfSubuniverses}
Suppose 
$D = A_{0} = B_{0}$, 
$s\ge 1$, $t\ge 0$,
$A_{i}$ is a one-of-four subuniverse of $A_{i-1}$ 
for every $i\in\{1,\dots,s\}$,
and 
$B_{i}$ is a one-of-four subuniverse of $B_{i-1}$ 
for every $i\in\{1,\dots,t\}$.
Then 
$A_{s}\cap B_{t}$ is 
a one-of-four subuniverse of 
$A_{s-1}\cap B_{t}$ of the same type as $A_{s}$.
\end{lem}
\begin{proof}
We prove this lemma by induction on $s+t$.
Let $A_{s}$ be a one-of-four subuniverse of 
$A_{s-1}$ of type $\mathcal T$.
For $t=0$ the claim follows from the statement.
Assume that $t\ge 1$.
By the inductive assumption, 
$A_{s-1}\cap B_{t}$ 
and 
$A_{s}\cap B_{t-1}$ are one-of-four subuniverses
of $A_{s-1}\cap B_{t-1}$, 
and the second of them is of type $\mathcal T$.
Then by Theorem~\ref{PCBsub}, 
their intersection 
$A_{s}\cap B_{t}$ is a one-of four subuniverse 
of 
$A_{s-1}\cap B_{t}$ of type $\mathcal T$.
\end{proof}

\begin{lem}\label{ReductionAndProjectionGivesOneOfFour}
Suppose $R\subseteq A_0\times B_{0}$ 
is a subdirect relation, 
$B_{i}$ is a one-of-four subuniverse of 
$B_{i-1}$ for every $i\in\{1,\dots,t\}$,
$A_{1}$ is a one-of-four subuniverse of $A_{0}$.
Then 
$\proj_{2}(R\cap (A_{1}\times B_{t}))$
is a one-of-four subuniverse of 
$\proj_{2}(R\cap (A_{1}\times B_{t-1}))$
of the same type as $B_{t}$.
\end{lem}

\begin{proof}
By Lemma~\ref{PCBrel},
$R\cap (A_{0}\times B_{i})$ is a one-of-four subuniverse of 
$R\cap (A_{0}\times B_{i-1})$
of the same type as $B_{i}$,
and 
$R\cap (A_{1}\times B_{0})$
is a one-of-four subuniverse of 
$R$.
By Lemma~\ref{SequencesOfSubuniverses}, 
$R\cap (A_{1}\times B_{t})$
is a one-of-four subuniverse 
of 
$R\cap (A_{1}\times B_{t-1})$
of the same type as $B_{t}$.
Let $\sigma$ be the congruence on $R\cap (A_{1}\times B_{0})$
such that two elements are equivalent whenever 
there projections onto the second coordinate are equal.
Then $R\cap (A_{1}\times B_{t})$ is stable under $\sigma$ for every $i$.
By Lemma~\ref{FactorByStableCongruence},
$\proj_{2}(R\cap (A_{1}\times B_{t}))$
is a one-of-four subuniverse of 
$\proj_{2}(R\cap (A_{1}\times B_{t-1}))$ of the same type as $B_{t}$.
\end{proof}

\begin{thm}\label{PCBint}
Suppose 
$B_{1},\dots,B_{n}$ are one-of-four subuniverses of $D$, 
and 
$B_{1}\cap\dots\cap B_{n} = \varnothing$.
Then there exists $I\subseteq\{1,\dots,n\}$ 
with $\bigcap_{i\in I}B_{i} = \varnothing$ satisfying
one of the following conditions:
\begin{enumerate}
    \item $|I|\le 2$ and all subuniverses $B_{i}$, where $i\in I$,
    are of the same type;
    \item $B_{i}$ is a linear subuniverse for every $i\in I$;
    \item $B_{i}$ is a binary absorbing subuniverse for every $i\in I$.
\end{enumerate}
\end{thm}

\begin{proof}
Let us prove by induction on $n$. For $n=1$ it is trivial.
For $n=2$ it follows from 
Lemma~\ref{IntersectionOfTwoSubuniverses}.
If $\bigcap_{i\in I}B_{i} = \varnothing$ for some 
$I\subsetneq\{1,2,\dots,n\}$, then 
applying the inductive assumption to 
$\bigcap_{i\in I}B_{i}$ we obtain the required property.
Thus, we assume that 
if we remove one one-of-four subuniverse
from the intersection
$B_{1}\cap\dots\cap B_{n}$ we get a nonempty set.

Let us show that all subuniverses should be of the same type.
Put 
$C_{i} = B_{i}\cap B_{n}$
for every $i\in\{1,2,\dots,n-1\}$.
By Lemma~\ref{PCBsub}, 
$C_{i}$ is a one-of-four subuniverse of $B_{n}$ of the same type as $B_{i}$.
Applying the inductive assumption to 
$C_{1}\cap \dots\cap C_{n-1}=\varnothing$, 
we derive that 
$C_{1},\dots,C_{n-1}$ are of the same type, hence 
$B_{1},\dots,B_{n-1}$ are of the same type.
Similarly we can show that 
$B_{2},\dots,B_{n}$ are of the same type, and 
therefore, since $n\ge 3$,  all of them are of the same types.

Assume that all subuniverses 
$B_{1},\dots,B_{n}$ are centers or PC subuniverses.
Let $R$ be the $n$-ary relation 
consisting of all tuples $(a,a,\dots,a)$.
Then
$R$ is a $(B_{1},\dots,B_{n})$-essential relation, 
which contradicts Corollary~\ref{CenterLessThanThree}
for centers and Corollary~\ref{PCLessThanThree}
for PC subuniverses.
\end{proof}

%% file: auxstatements.tex
\subsection{One-of-four reductions}

\begin{lem}\label{nonPCReductionImpliesSubuniverse}
Suppose $D^{(1)}$ is a one-of-four reduction for 
an instance $\Theta$ of type $\mathcal T$,
which is not the PC type.
Then $\Theta^{(1)}(z)$ is a one-of-four subuniverse of $\Theta(z)$ of type $\mathcal T$ for every varaible $z$. 
\end{lem}

\begin{proof}
Let 
$\Var(\Theta) = \{x_{1},\dots,x_{t}\}$ and 
$\Theta(x_{1},\dots,x_{t})$ define the relation $R$.
By Lemma~\ref{PCBsubNonPC}, 
$D_{x_{i}}^{(1)}\cap \proj_{i}(R)$ is a one-of-four 
subuniverse of $\proj_{i}(R)$ of type $\mathcal T$ for every $i$. 
Considering $R$ as a subdirect relation on smaller domains and applying 
Corollaries~\ref{AbsImpliesCons}, \ref{CenterImpliesCons},
and \ref{LinearImplies} we conclude that 
$\Theta^{(1)}(z)$ is a one-of-four subuniverse of $\Theta(z)$ of type
$\mathcal T$.
\end{proof}

\begin{lem}\label{PCReductionImpliesSubuniverse}
Suppose $D^{(1)}$ is a 
PC reduction for 
a 1-consistent instance $\Theta$,
for every variable $y$ appearing 
at least twice in $\Theta$
the pp-formula
$\Theta(y)$ defines $D_{y}$,
and
$\Theta(z)$ defines $D_{z}$ for a variable $z$.
Then
$\Theta^{(1)}(z)$ is a PC subuniverse of $D_{z}$. 
\end{lem}
\begin{proof}
First, we rename the variables in $\Theta$ so that 
every variable occurs just once and denote the obtained instance 
by $\Theta_{0}$.
Then we identify variables back to obtain the original instance 
step by step.
Thus, we get 
a sequence 
$\Theta_{0},\Theta_{1},\Theta_{2},\dots,
\Theta_{s}$ such that 
$\Theta_{i+1}$ is obtained from 
$\Theta_{i}$ by identifying of two variables
and $\Theta_{s} = \Theta$.
Let us show by induction on $i$
that 
for every variable $z$ 
the set $\Theta_{i}(z)\cap D_{z}^{(1)}$
is a PC subuniverse of 
$\Theta_{i}(z)$.
For $i=0$ it follows from the fact 
that 
$\Theta$ is 1-consistent, and therefore, $\Theta_{0}(z)$ defines the full $D_{z}$.

Assume that 
$\Theta_{i+1}$ is obtained from $\Theta_{i}$
by identifying of $y$ and $y'$,
and the variable in $\Theta$ corresponding to $y$ and $y'$ is $y$.
We know that for every variable $z$ appearing 
at least twice in $\Theta$,  $\Theta(z)$ defines $D_{z}$.
Hence 
$\Theta_{i+1}(y)$ also defines $D_{y}$.
Thus, we just need to show that 
for any variable 
$z$ different from $y$ and $y'$ 
the set
$\Theta_{i+1}(z)\cap D_{z}^{(1)}$ is 
a PC subuniverse of 
$\Theta_{i+1}(z)$.
By the inductive assumption 
$\Theta_{i}(z)\cap D_{z}^{(1)}$ is 
a PC subuniverse of $\Theta_{i}(z)$.
Then 
$\Theta_{i}(z)\cap D_{z}^{(1)} = 
E_{1}\cap\dots\cap E_{t}$, where 
$E_{j}$ is an equivalence class of a PC 
congruence $\sigma_{j}$ on $\Theta_{i}(z)$
for every $j$.
Let $S\subseteq \Theta_{i}(z)/\sigma_{j}\times D_{y}\times D_{y}$
be the relation consisting of all tuples
$(a/\sigma_{j},b,b')$
such that  $\Theta_{i}$ has a 
solution with 
$z=a$, $y=b$, $y'=b'$.
Since the variable $y$ appears at least twice in $\Theta$, 
$\Theta(y)$ defines a full relation. 
Hence, 
the relation $S$ is subdirect and 
for every $b\in D_{y}$ there exists $E$ 
such that $(E,b,b)\in S$.
Lemma~\ref{IdentificationDoesNotReducePC}
implies that 
for every 
equivalence class $E$ of $\sigma_{j}$ there exists 
$b$ such that 
$\Theta_{i}$ has a solution with $z\in E$ and 
$y=y'=b$,
which means that 
there exists a solution of 
$\Theta_{i+1}$ with $z\in E$.
Therefore, 
$\Theta_{i}(z)/\sigma_{j}\cong \Theta_{i+1}(z)/\sigma_{j}$, 
which implies that 
$\Theta_{i+1}(z)\cap D_{z}^{(1)}$ is a PC subuniverse of 
$\Theta_{i+1}(z)$. This completes the inductive step.

Since $\Theta=\Theta_{s}$, we proved that 
$\Theta(z)\cap D_{z}^{(1)}$ is a PC subuniverse of
$\Theta(z)$ for every variable $z$ of $\Theta$.

Suppose 
$\Var(\Theta) = \{x_{1},\dots,x_{t}\}$,
$\Theta(x_{1},\dots,x_{t})$ 
defines a relation $R$.
Then $R$ can be viewed as a subdirect relation 
if we reduce the domain of every variable $x_{i}$ to 
$\Theta(x_{i})$. 
By Corollary~\ref{PCImplies}, 
for any variable $z$ with $\Theta(z) = D_{z}$
we obtain that 
$\Theta^{(1)}(z)$ is a PC subuniverse 
of $D_{z}$.
\end{proof}

\begin{lem}\label{nonPCReductionForFormulas}
Suppose $D^{(1)}$ is a minimal absorbing, central, or linear reduction for an instance $\Theta$, and 
$\Theta(x_{1},\ldots,x_{n})$ defines a full relation.
Then 
$\Theta^{(1)}(x_{1},\ldots,x_{n})$ defines a full relation or an empty relation.
\end{lem}

\begin{proof}
If $\Theta^{(1)}(x_{1},\ldots,x_{n})$ defines an empty relation, 
then there is nothing to prove. 
Assume that $\Theta^{(1)}(x_{1},\ldots,x_{n})$ is not empty.

We prove by induction on $n$.
For $n=1$ by Lemma~\ref{nonPCReductionImpliesSubuniverse}
$\Theta^{(1)}(x_{1})$ is a subuniverse of 
$\Theta(x_{1})$ of the corresponding type.
By the minimality of the reduction $D^{(1)}$ 
the pp-formula $\Theta^{(1)}(x_{1})$ 
defines $D_{x_{1}}^{(1)}$.

Let us prove the induction step.
For each $i\in\{1,\dots,n-1\}$ choose 
$a_{i}\in D_{x_{i}}^{(1)}$.
By the inductive assumption,
$\Theta^{(1)}(x_{1},\dots,x_{n-1})$ defines a full 
relation, hence 
there exists a solution of 
$\Theta^{(1)}$ 
having 
$x_{i} = a_{i}$ for every 
$i\in\{1,\dots,n-1\}$.

Add the constraint 
$x_{i}= a_{i}$
to $\Theta$ for every $i\in\{1,\dots,n-1\}$
and denote the obtained instance by $\Omega$.
By the condition of this lemma 
$\Omega(x_{n})$ defines $D_{x_{n}}$.
By Lemma~\ref{nonPCReductionImpliesSubuniverse},
$\Omega^{(1)}(x_{n})$ defines a one-of-four subuniverse of $D_{x_{n}}$ of the
corresponding type, which by the minimality of the reduction 
$D^{(1)}$ implies that 
$\Omega^{(1)}(x_{n})$ defines $D_{x_{n}}^{(1)}$.
Since we chose $a_{1},\dots,a_{n-1}$ arbitrary, 
this means that $\Theta^{(1)}(x_{1},\ldots,x_{n})$ 
defines a full relation.
\end{proof}



\begin{lem}\label{PCReductionForFormulas}
Suppose 
$D^{(1)}$ is a minimal PC reduction for a 1-consistent instance $\Theta$,
for every variable $y$ appearing 
at least twice in $\Theta$
the pp-formula
$\Theta(y)$ defines $D_{y}$, and 
$\Theta(x_{1},\ldots,x_{n})$ defines a full relation.
Then 
$\Theta^{(1)}(x_{1},\ldots,x_{n})$ defines a full relation or an empty relation.
\end{lem}

\begin{proof}
If $\Theta^{(1)}(x_{1},\ldots,x_{n})$ defines an empty relation, 
then there is nothing to prove. 
Assume that $\Theta^{(1)}(x_{1},\ldots,x_{n})$ is not empty.

First, we join variables
$x_{1},\dots,x_{n}$ into one variable $X$
with domain $D_{x_{1}}\times\dots\times D_{x_{n}}$. 
We replace $x_{1},\dots,x_{n}$ by $X$ and change all constraints containing one of the variables 
$x_{1},\ldots,x_{n}$ correspondingly.
The obtained instance we denote by $\Omega$.
Since $\Theta(x_{1},\ldots,x_{n})$ defines a full relation, 
the instance $\Omega$ is 1-consistent.

Second, we define a reduction $D^{(1)}$ on the domain of 
the new variable $X$ by 
$D^{(1)}_{X} = D^{(1)}_{x_{1}}\times\dots\times D^{(1)}_{x_{n}}$.
Let us show that this is a PC 
reduction.
By Lemma~\ref{PCBrel}, 
$D^{(1)}_{X}$ is a PC subuniverse of $D_{X}$. 
By Lemmas~\ref{GenBinAbToBinAb} and \ref{GenCenterToCenter}, 
there is no nontrivial binary absorbing subuniverse or center on $D_{X}$.
Thus, $D^{(1)}$ is a PC reduction for $\Omega$. 
By Lemma~\ref{PCReductionImpliesSubuniverse}, 
$\Omega^{(1)}(X)$ is a PC subuniverse of 
$D_{X}$. 
By Lemma~\ref{PCSubuniverseOnProduct}, 
$\Omega^{(1)}(X) = 
B_{1}\times\dots\times B_{n}$, 
where $B_{i}$ is a PC subuniverse of $D_{x_{i}}$ 
for every $i$.
By the minimality of $D^{(1)}$ on $\Theta$ we obtain that 
$B_{i} = D^{(1)}_{x_{i}}$. Hence, 
$\Theta^{(1)}(x_{1},\ldots,x_{n})$ defines a full relation.
\end{proof}

\begin{lem}\label{ProperReductionPreservesSubdirectness}
Suppose $D^{(1)}$ is a one-of-four minimal reduction of an instance $\Theta$,
$\rho(x_{1},\ldots,x_{n})$ is a subdirect constraint  of $\Theta$,
and
$\rho^{(1)}$ is not empty.
Then
$\rho^{(1)}$ is subdirect.
\end{lem}
\begin{proof}
We need to show 
that 
$\proj_{i}(\rho\cap (D_{x_{1}}^{(1)}\times \dots\times D_{x_{n}}^{(1)})) = 
D_{x_{i}}^{(1)}$.
By Corollaries~\ref{AbsImpliesCons}, \ref{CenterImpliesCons}, \ref{PCImplies}, \ref{LinearImplies},
$B_{i} = \proj_{i}(\rho\cap (D_{x_{1}}^{(1)}\times \dots\times D_{x_{n}}^{(1)}))$ 
is a one-of-four subuniverse of $D_{x_{i}}$ of the same type.
Since $\rho^{(1)}$ is not empty, $B_{i}$ is not empty. 
Since $D_{x_{i}}^{(1)}$ is a minimal subuniverse of this type, we have $B_{i} =D_{x_{i}}^{(1)}$.
\end{proof}

\begin{lem}\label{ProperReductionPreservesCycleConAndIrreducability}
Suppose $D^{(1)}$ is a one-of-four minimal reduction for
a cycle-consistent irreducible CSP instance $\Theta$,
and $\Theta^{(1)}$ has a solution.
Then $\Theta^{(1)}$ is cycle-consistent and irreducible.
\end{lem}
\begin{proof}

Consider a path $P$ in 
$\Theta$
starting and ending with one variable $x$.
By $\Omega$ we denote its
covering 
$z_{1}-Q_{1}-z_{2}-\dots-Q_{l-1}-z_{l}$ 
(which is also a covering of $\Theta$)
that is obtained from $P$ by renaming the variables so that every variable except for $z_{2},\ldots,z_{l-1}$ occurs just once,
$z_{2},\ldots,z_{l-1}$ occur twice.
Thus, $z_{1}$ and $z_{l}$ are different but 
$S(z_{1}) = S(z_{l}) = x$ in the definition 
of the covering.
By $\Omega'$ we denote the formula obtained from $\Omega$ by substituting $z_{1}$ for $z_{l}$.


First, we prove that $P$ connects $a$ with $a$ in $\Theta^{(1)}$ for every $a\in D_{x}^{(1)}$.
Since $\Theta$ is cycle-consistent, $\Omega'(z_{1})$ defines $D_{x}$.
Since $\Theta^{(1)}$ has a solution, $\Omega'^{(1)}(z_{1})$
defines a nonempty relation.
By Lemmas~\ref{nonPCReductionForFormulas} and
\ref{PCReductionForFormulas},
$\Omega'^{(1)}(z_{1})$ defines $D_{x}^{(1)}$,
which means that $P$ connects $a$ with $a$ in $\Theta^{(1)}$ for every $a\in D_{x}^{(1)}$. Hence, $\Theta^{(1)}$ is cycle-consistent.

Assume that $P$ connects any two elements of $D_{x}$, which means
that $\Omega(z_{1},z_{l})$ defines a full relation.
Since $\Theta^{(1)}$ has a solution, $\Omega^{(1)}(z_{1},z_{l})$
defines a nonempty relation.
By Lemmas~\ref{nonPCReductionForFormulas},
\ref{PCReductionForFormulas},
$\Omega^{(1)}(z_{1},z_{l})$ also defines a full relation, 
which means that $P$ connects any two elements of $D_{x}^{(1)}$ 
in $\Theta^{(1)}$.

Let us prove that $\Theta^{(1)}$ is irreducible.
Consider an instance $\Upsilon_{1}=\{C_{1}',\ldots,C_{s}'\}$ consisting of projections of constraints from $\Theta^{(1)}$
such that it is not fragmented and not linked.
Let $\Var(\Upsilon_{1}) = \{x_{1},\ldots,x_{n}\}$.
By the definition 
for each constraint $C_{i}'$ we can find a constraint $C_{i}\in\Theta$ such that 
$C_{i}'$ is a projection of $C_{i}^{(1)}$ onto some variables.
Let 
$\Upsilon_{2}$ consist of the 
projections of $C_{1},\dots,C_{s}$ onto the same variables as in 
$\Upsilon_{1}$,
and $\Upsilon\in\ExpShort(\Theta)$ is obtained from 
$\{C_{1},\ldots,C_{s}\}$ by renaming variables so that 
each variable except for $x_{1},\ldots,x_n$ appears just once.
Then the pp-formulas $\Upsilon_{1}(x_{1},\ldots,x_{n})$
and $\Upsilon^{(1)}(x_{1},\ldots,x_{n})$
define the same relation,
$\Upsilon_{2}(x_{1},\ldots,x_{n})$ 
and $\Upsilon(x_{1},\ldots,x_{n})$ 
define the same relation.
Since $\Upsilon_{1}$ is not fragmented, both $\Upsilon$ and $\Upsilon_{2}$ are not fragmented. 
Also, by Lemma~\ref{ExpandedConsistencyLemma}, 
both $\Upsilon$ and $\Upsilon_{2}$ are cycle-consistent 
and irreducible.

Assume that $\Upsilon_{2}$ is linked.
By Lemma~\ref{LinkedConIsCon} there exists a path that connects any two elements of $D_{x_{1}}$ in $\Upsilon_{2}$.
Then there exists a corresponding path within 
the variables $x_{1},\ldots,x_{n}$ of
$\Upsilon$ connecting any two elements of $D_{x_{1}}$. 
As we showed earlier this 
path, reduced to $D^{(1)}$,
also connects any two elements of $D_{x_{1}}^{(1)}$ in $\Upsilon^{(1)}$.
The same path can be used to connect any two elements of $D_{x_{1}}^{(1)}$ 
in $\Upsilon_{1}$, which contradicts our assumption that $\Upsilon_{1}$ is not linked.

Suppose $\Upsilon_{2}$ is not linked.
Since $\Upsilon$ is irreducible,
the solution set of $\Upsilon_{2}$ is subdirect.
Thus, for each variable $x_{i}$ 
(these are only variables appearing more than once in $\Upsilon$)
we have 
$\Upsilon(x_{i}) = \Upsilon_{2}(x_{i}) = D_{x_{i}}$.
Then by 
Lemmas~\ref{nonPCReductionForFormulas}, \ref{PCReductionForFormulas},
$\Upsilon^{(1)}(x_{i})$ defines $D_{x_{i}}^{(1)}$ or an empty set.
It cannot be empty because 
$\Theta^{(1)}$ has a solution, 
therefore 
we have 
$\Upsilon_{1}(x_{i})=\Upsilon^{(1)}(x_{i})=D_{x_{i}}^{(1)}$ for every $i$, 
and 
the solution set of $\Upsilon_{1}$ is subdirect,
which completes the proof.
\end{proof}

\begin{lem}\label{LinkedStayLinkedForAC}
Suppose 
$D^{(1)}$ is a minimal absorbing or central reduction 
for $\Theta$, 
the solution set of $\Theta$ is subdirect, 
$D_{x_{1}} = D_{x_{2}}$, $D^{(1)}_{x_{1}} = D^{(1)}_{x_{2}}$,
both $\Theta(x_1,x_2)$ and $\Theta^{(1)}(x_{1},x_{2})$ define reflexive symmetric relations,
and 
$\Theta(x_1,x_2)$ contains 
$(a,b)\in D_{x_1}^{(1)}\times D_{x_2}^{(1)}$.
Then 
$a$ and $b$ are linked in the relation defined by $\Theta^{(1)}(x_{1},x_{2})$.
\end{lem}
\begin{proof}
Let $\Var(\Theta) = 
\{x_{1},x_{2},y_{1},\dots,y_{t}\}$, 
$\Theta(x_{1},x_{2},y_{1},\dots,y_{t})$ define a relation $R$.
The relation $R$ 
can be viewed as  
a ternary relation  
$R\subseteq D_{x_{1}}\times D_{x_{2}}\times 
(D_{y_{1}}\times\dots\times D_{y_{t}})$.
By Lemmas~\ref{AbsImplies} and \ref{CenterImplies},
$G:= D^{(1)}_{y_{1}}\times\dots\times D^{(1)}_{y_{t}}$
is a one-of-four subuniverse of 
$D_{y_{1}}\times\dots\times D_{y_{t}}$
of the same type as the reduction $D^{(1)}$.
Let 
$$R'(Y,Y',Y'') =\exists x_{1}\exists x_{2}\;R(a,x_{1},Y)
\wedge 
R(x_{1},x_{2},Y')
\wedge 
R(x_{2},b,Y'')\wedge x_{1}\in D_{x_{1}}^{(1)}\wedge x_{2}\in D_{x_{1}}^{(1)}.$$

Since $\Theta(x_{1},x_{2})$ contains $(a,b)$ 
and $\Theta^{(1)}(x_{1},x_{2})$ 
defines a reflexive relation, 
there exist $B_{1},B_{1}'\in G$ such that 
$(B_1,B_{1}',B_{1}'')\in R'$ (put $x_{1}=x_{2}= a$).
Similarly,
there exist $B_{2}',B_{2}''\in G$ such that 
$(B_2,B_{2}',B_{2}'')\in R'$ (put $x_{1}=x_{2}= b$),
and
$B_{3},B_{3}''\in G$ such that 
$(B_3,B_{3}',B_{3}'')\in R'$ (put $x_{1}=a$, $x_{2}= b$).
By 
Lemma~\ref{BinAbsLessThanTwoCorollary} and Corollary~\ref{CenterLessThanThree}
$R'$ cannot be $G$-essential, which means that 
$R\cap (G\times G\times G)\neq \varnothing$.
Hence $a$ and $b$ are linked (by a path of length 3) in 
$\Theta^{(1)}(x_{1},x_{2})$.
\end{proof}

\begin{lem}\label{LinkedStayLinkedForPC}
Suppose 
$D^{(1)}$ is a minimal PC reduction 
for $\Theta$, 
the solution set of $\Theta$ is subdirect, 
$D_{x_{1}} = D_{x_{2}}$, $D^{(1)}_{x_{1}} = D^{(1)}_{x_{2}}$,
both $\Theta(x_1,x_2)$ and $\Theta^{(1)}(x_{1},x_{2})$ define reflexive symmetric relations,
and 
$\Theta(x_1,x_2)$ contains 
$(a,b)\in D_{x_1}^{(1)}\times D_{x_2}^{(1)}$.
Then 
$a$ and $b$ are linked in the relation defined by $\Theta^{(1)}(x_{1},x_{2})$.
\end{lem}

\begin{proof}
Suppose $y$ is a variable of $\Theta$ and 
$\sigma$ is a PC congruence on $D_{y}$.
Consider a 
relation $\rho\subseteq D_{x_{1}}\times D_{y}/\sigma$ 
consisting of all the tuples 
$(c,C)$ such that 
there exists a solution of $\Theta$ with 
$x_{1} = c$ and $y\in C$.
Since there is no nontrivial binary absorbing subuniverse or center on 
$D_{x_{1}}$, 
by Lemma~\ref{PCRelationsLem}, 
either $\rho$ is a full relation, or 
$\ConOne(\rho,2)$ is the equality relation.
In the first case it does not matter what value we substitute for 
the variable $x_{1}$ the variable $y$ can be at any equivalence class
of $\sigma$.
In the second case the equivalence class is uniquely determined by 
the variable $x_{1}$. 
Moreover, since $D_{x_{1}}^{(1)}$ is a minimal PC subuniverse, 
the equivalence class is the same 
for all elements of $D_{x_{1}}^{(1)}$.
Since $\Theta^{(1)}$ has a solution, 
this equivalence class is the class containing $D_{y}^{(1)}$.
Later, we will specify whether a PC congruence is 
of the \emph{first type} (from the first case) or of the \emph{second type} (from the second case). 

Let
$\Var(\Theta) = \{x_1,x_2,y_1,\dots,y_t\}$.
Let 
$R$ be the relation defined by 
$\Theta(x_{1},x_{2},y_{1},\ldots,y_{t})$.
By $\Upsilon$ 
denote the 
following formula
$$R(a,x_2,y_1,\dots,y_{t})
\wedge 
R(x_{1},x_2,y_1',\dots,y_{t}')
\wedge 
R(x_{1},x_{2}',z_1,\dots,z_{t})
\wedge 
R(b,x_2',z_{1}',\dots,z_{t}')
\wedge x_{1}\in D^{(1)}_{x_{1}}.$$

Consider a congruence $\sigma$ of the 
first type on the domain of any variable $y$ of $\Upsilon$.

Assume that 
$y\in \{x_{2},y_{1},\dots,y_{t},y_1',\dots,y_{t}'\}$.
It follows from the definition of the first type
that 
for any equivalence class $E$ of $\sigma$ 
there exists a solution 
of $\Upsilon$ 
such that 
$x_{1} = a$,
$x_{2}' = b$,
$y_{i} = y_{i}'$ for every $i$, 
and $y\in E$.

Similarly, assume that 
$y\in \{x_{2}',z_{1},\dots,z_{t},z_1',\dots,z_{t}'\}$.
For any equivalence class $E$ of $\sigma$ 
there exists a solution 
of $\Upsilon$ 
such that 
$x_{1} = x_{2} = b$,
$z_{i} = z_{i}'$ for every $i$, 
and $y\in E$.

Thus, we showed that in both cases
$\Upsilon(y)/\sigma\cong D_{y}/\sigma$.
Let $E$ be the equivalence class of $\sigma$ containing 
$D_{y}^{(1)}$.
By $\delta$ we denote the extension of $\sigma$ onto 
the solution set of $\Upsilon$, 
and by $E_{\sigma}$ we denote the equivalence class
of $\delta$ 
corresponding to $E$.
Since $\Upsilon(y)/\sigma\cong D_{y}/\sigma$ and 
$D_{y}/\sigma$ is a PC algebra without a nontrivial binary absorbing subuniverse or center,
$E_{\sigma}$ is a PC subuniverse of the solution set of 
$\Upsilon$.

Consider the intersection of 
$E_{\sigma}$ for all PC congruences $\sigma$ of the first type.
If this intersection is not empty, then 
there exists a solution of $\Upsilon$ such that 
any element of this solution is in the equivalence 
class containing $D_{y}^{(1)}$ for any PC congruence of the first type.
Since $a,b\in D_{x_{1}}^{(1)}$ and $x_{1}\in D^{(1)}_{x_{1}}$ in the definition of
$\Upsilon$, 
the same is true for any PC congruence of the second type.
Since $D_{y}^{(1)}$ is the intersection 
of all equivalence classes containing $D^{(1)}_{y}$ of all PC congruences for any variable $y$, 
the solution is in $D^{(1)}$, 
which means that 
$a$ and $b$ are linked in $\Theta^{(1)}(x_{1},x_{2})$.

Assume that the intersection of 
$E_{\sigma}$ for all PC congruences $\sigma$ of the first type is empty.
By Theorem~\ref{PCBint} there should be 
two congruences $\sigma$ and $\sigma'$ 
such that 
$E_{\sigma}\cap E_{\sigma'} = \varnothing.$
Let $y$ and $y'$ be the variables 
of $\Upsilon$ 
corresponding to $\sigma$ and $\sigma'$.
Consider several cases. 

Case 1. $y,y'\in\{x_{2},x_{2}',y_{1}',\dots,y_{t}',z_{1},\dots,z_{t},z_{1}',\dots,z_{t}'\}$.
Since 
$\Theta^{(1)}(x_{1},x_{2})$ defines a reflexive relation,
$\Upsilon$ has a solution 
with 
$x_{2} = x_{1} = x_{2}' = b$ 
and 
all the variables 
$y_{1}',\dots,y_{t}',z_{1},\dots,z_{t},z_{1}',\dots,z_{t}'$ 
are from $D^{(1)}$.
This contradicts the fact that 
$E_{\sigma}\cap E_{\sigma'} = \varnothing.$

Case 2. $y,y'\in\{x_{2},x_{2}',y_{1},\dots,y_{t},y_{1}',\dots,y_{t}',z_{1},\dots,z_{t}\}$.
Similarly, 
$\Upsilon$ has a solution 
with 
$x_{1} = x_{2} = x_{2}' = a$ 
and 
all the variables 
$y_{1},\dots,y_{t},y_{1}',\dots,y_{t}',z_{1},\dots,z_{t}$ 
are from $D^{(1)}$.
This contradicts the fact that 
$E_{\sigma}\cap E_{\sigma'} = \varnothing.$

Case 3.
$y\in \{y_{1},\dots,y_{t}\}$, 
$y'\in\{z_{1}',\dots,z_{t}'\}$.
Similarly, 
$\Upsilon$ has a solution 
with 
$x_{1} = x_{2} =a$, $x_{2}' = b$ 
and 
all the variables 
$y_{1},\dots,y_{t},y_{1}',\dots,y_{t}',z_{1}',\dots,z_{t}'$ 
are from $D^{(1)}$.
Again, this contradicts the fact that 
$E_{\sigma}\cap E_{\sigma'} = \varnothing.$
\end{proof}

\begin{lem}\label{LinkedStayLinked}
Suppose 
$D^{(1)}$ is a minimal nonlinear reduction 
for $\Theta$, 
the solution set of $\Theta$ is subdirect, 
$\Theta^{(1)}$ is not empty,
$\Theta(x_1,x_2)$ defines a relation containing 
$(a,b)\in D_{x_1}^{(1)}\times D_{x_2}^{(1)}$.
Then 
$a$ and $b$ are linked in the relation defined by $\Theta^{(1)}(x_{1},x_{2})$.
\end{lem}

\begin{proof}
By Lemma~\ref{ProperReductionPreservesSubdirectness},
the solution set of $\Theta^{(1)}$ is subdirect.
Let $\Var(\Theta) = \{x_{1},x_{2},y_{1},\dots,y_{t}\}$, 
$R$ be the relation defined by 
$\Theta(x_{1},x_{2},y_{1},\dots,y_{t})$.
Let $\Omega$ be the following instance 
$$R(x_{1},x_{2},y_{1},\dots,y_{t})
\wedge 
R(x_{1}',x_{2},y_{1}',\dots,y_{t}').$$
Since 
the solution sets of $\Theta$ and $\Theta^{(1)}$ 
are subdirect, 
the solution sets of $\Omega$ and $\Omega^{(1)}$ are 
also subdirect.
Also, 
there should be a solution of $\Theta^{(1)}$
with $x_{2} = b$. Let $x_{1} = b'$ in this solution.
Then $\Omega(x_{1},x_{1}')$ contains 
$(a,b')$.
Since 
both $\Omega(x_{1},x_{1}')$
and $\Omega^{(1)}(x_{1},x_{1}')$ define symmetric 
reflexive relations, 
Lemmas~\ref{LinkedStayLinkedForAC}
and \ref{LinkedStayLinkedForPC}
imply that 
$a$ and $b'$ are linked in $\Omega^{(1)}(x_{1},x_{1}')$. 
Since $(b',b)$ is in $\Theta^{(1)}(x_{1},x_{2})$, 
we derive that 
$a$ and $b$ are linked in $\Theta^{(1)}(x_{1},x_{2})$, 
which completes the proof.
\end{proof}

\subsection{Properties of $\ConOne(\rho,x)$} 

\begin{lem}\label{RectangularCriticalArityTwo}
Suppose 
$\rho$ is a critical rectangular relation of arity $n\ge 2$, 
$\rho'$ is the cover of $\rho$.
Then 
$\ConOne(\rho',1)\supsetneq\ConOne(\rho,1)$,
for $n>2$ we also have 
$\ConOne(\proj_{1,2}(\rho),1)\supsetneq\ConOne(\rho,1)$,
\end{lem}

\begin{proof}
For every 
$i\in\{1,2,\dots,n\}$
we define 
$\rho_{i}(x_{1},\ldots,x_{n}) =
\exists x_{i}' 
\rho(x_{1},\ldots,x_{i-1},x_{i}',x_{i+1},\dots,x_{n})$.
Since $\rho$ is critical, 
it has no dummy variables, therefore
$\rho\subsetneq\rho_{i}$ for every $i$.
Also
$\rho\subsetneq \bigcap_{i} \rho_{i}$.
Choose a tuple 
$(a_{1},\ldots,a_{n})\in \rho'\setminus \rho$.
Since $\rho'$ is the cover of $\rho$
we have $\rho'\subseteq\bigcap_{i} \rho_{i}$.
Since this tuple is in $\rho_{i}$, 
for every $i$ there is $b_{i}$
such that 
$(a_{1},\ldots,a_{i-1},b_{i},a_{i+1},\dots,a_{n})\in\rho$.
Then $(a_{1},b_{1})\in\ConOne(\rho',1)$, which means 
by the rectangularity of $\rho$ that 
$\ConOne(\rho',1)\supsetneq\ConOne(\rho,1)$.
For $n>2$ we have 
$(b_{1},a_2,\dots,a_{n}),
(a_{1},\dots,a_{n-1},b_{n})\in \rho$, hence
$(a_{1},b_{1})\in\ConOne(\proj_{1,2}(\rho),1)$
and therefore
$\ConOne(\proj_{1,2}(\rho),1)\supsetneq\ConOne(\rho,1)$.
\end{proof}

\begin{lem}\label{CriticalMeansIrreducible}
Suppose $\rho$ is a critical subdirect relation and the $i$-th variable of $\rho$ is rectangular.
Then $\ConOne(\rho,i)$ is an irreducible congruence.
\end{lem}

\begin{proof}
To simplify notations assume that $i=1$.
Put $\sigma=\ConOne(\rho,1)$.
As we mentioned in 
Section~\ref{DefinitionRectangularitySubsection}, 
$\sigma$ should be a congruence.
Assume that it is not an irreducible congruence.
Consider binary relations $\delta_{1},\ldots,\delta_{s}
\supsetneq\sigma$ stable under $\sigma$
such that $\delta_{1}\cap\dots\cap\delta_{s} = \sigma$.
Put
$$\rho_{j}(x_{1},\ldots,x_{n})=\exists x_{1}' \;\rho(x_{1}',x_{2},\ldots,x_{n})\wedge \delta_{j}(x_{1},x_{1}').$$
Consider a tuple $(x_1,\dots,x_{n})$ in the intersection of $\rho_{1},\ldots,\rho_{s}$.
Since $\delta_{j}$ is stable under $\sigma=\ConOne(\rho,1)$, 
we may assume that $x_{1}'$ takes the same value 
in the definition of every $\rho_{j}$.
Then $(x_{1},x_{1}')$ should be in $\delta_{j}$
for every $j$, 
which implies that $(x_{1},x_{1}')\in\sigma$
and 
$(x_{1},\dots,x_{n})\in\rho$.
Hence, the intersection of 
$\rho_{1},\dots,\rho_{s}$ gives $\rho$.
Since $\rho\subsetneq \rho_{j}$ for every $j$, 
this contradicts the fact that 
$\rho$ is critical.
\end{proof}

For a relation $\rho$ of arity $n$ by $\VPol(\rho)$ we denote the set of all unary
vector-functions preserving the relation $\rho$.


\begin{lem}\label{KeyDerivation}
Suppose 
a pp-formula $\Omega(x_{1},\ldots,x_{n})$
defines a relation $\rho$, 
$\alpha\in D_{x_{1}}\times\dots\times D_{x_{n}}$,
and $\rho' = \{f(\alpha)\mid f\in \VPol(\rho)\}$.
Then
there exists $\Omega'\in\ExpShort(\Omega)$
such that $\Omega'(x_{1},\ldots,x_{n})$ defines~$\rho'$.
\end{lem}
\begin{proof}
Suppose $\alpha = (a_{1},\ldots,a_{n})$.
We introduce new variables $x_{i}^{a}$ for every $i\in\{1,2,\ldots,n\}$ and $a\in D_{x_{i}}$.
By $\Upsilon$ we denote the following formula
$\bigwedge\limits_{(b_{1},\ldots,b_{n})\in\rho} \rho(x_{1}^{b_{1}},\ldots,x_{n}^{b_{n}}).$
This formula can be understood in the following way.
If we encode a unary vector function
by variables so that 
$f(b_{1},\ldots,b_{n}) = 
(x_1^{b_{1}},\dots,x_n^{b_{n}})$
for every $b_{1},\ldots,b_{n}$, 
then the formula says that 
the vector function preserves $\rho$.
Then 
$\rho'$ can defined by a pp-formula $\Upsilon(x_{1}^{a_{1}},\ldots,x_{n}^{a_{n}})$.
To obtain the formula $\Omega'$ it is sufficient to
replace
each $\rho(x_{1}^{b_{1}},\ldots,x_{n}^{b_{n}})$ by 
a copy of $\Omega$ (replacing 
$x_{1},\dots,x_{n}$ with 
$x_{1}^{b_{1}},\dots,x_{n}^{b_{n}}$)
and then replace 
$x_{1}^{a_{1}},\dots,x_{n}^{a_{n}}$ with 
$x_{1},\dots,x_{n}$.
\end{proof}

\begin{conslem}\label{MaximalMeansKey}
Suppose a pp-formula $\Omega(x_{1},\ldots,x_{n})$ 
defines a relation without a tuple $\alpha\in D_{x_{1}}\times\dots\times D_{x_{n}}$,
$\Sigma$ is the set of all relations defined by
$\Upsilon(x_{1},\ldots,x_{n})$ where $\Upsilon\in\ExpShort(\Omega)$,
and $\rho$ is an inclusion-maximal relation in $\Sigma$ without the tuple $\alpha$.
Then $\alpha$ is a key tuple for $\rho$.
\end{conslem}

\begin{proof}
For every tuple $\beta\notin\rho$
we consider $\rho_{\beta} := \{f(\beta)\mid f\in \VPol(\rho)\}$.
Since $f$ can be 
a constant mapping to a tuple from $\rho$ and an identity, we have $\rho_{\beta}\supsetneq \rho$ for every $\beta$.
By Lemma~\ref{KeyDerivation}, $\rho_{\beta}$ 
should be in $\Sigma$.
Since $\rho$ is inclusion-maximal, $\alpha\in\rho_{\beta}$.
Therefore, any $\beta$ can be mapped to $\alpha$ 
by a unary vector-function preserving $\rho$, 
which means that $\alpha$ is a key tuple for $\rho$.
\end{proof}

The next lemma shows that 
we can apply 
the operation
$\ConOne$ and a nonlinear reduction $D^{(1)}$ to a 
pp-formula 
$\Upsilon(x_{1},\dots,x_{n})$ 
in any order, the result will be the same.
For the linear reduction a slight modification of the statement
is required (see Lemma~\ref{AddLinearVariables}).

\begin{lem}\label{SameConOneForNonlinear}
Suppose $D^{(1)}$ is a minimal nonlinear reduction for an instance $\Upsilon$,
the solution set of $\Upsilon$ is subdirect,
and
$\Upsilon^{(1)}(x_{1},\ldots,x_{n})$
defines a subdirect rectangular relation.
Then for every $i$
$$
(\ConOne(\Upsilon(x_{1},\ldots,x_{n}),i))^{(1)}=
\ConOne(\Upsilon^{(1)}(x_{1},\ldots,x_{n}),i).
$$
\end{lem}

\begin{proof}
Put $\sigma_{0} = \ConOne(\Upsilon(x_{1},\ldots,x_{n}),i)$,
$\sigma_{1} = \ConOne(\Upsilon^{(1)}(x_{1},\ldots,x_{n}),i)$.
Let $\{x_{1},\ldots,x_{n},y_{1},\ldots,y_{s}\}$ be the set of all variables of $\Upsilon$.
Let 
$\Xi = \Upsilon\wedge \Upsilon_{x_{i},y_{1},\ldots,y_{s}}^{x_{i}',y_{1}',\ldots,y_{s}'}$.
We can check that 
$\sigma_{0}$
is defined by $\Xi(x_{i},x_{i}')$, 
and 
$\sigma_{1}$
is defined by $\Xi^{(1)}(x_{i},x_{i}')$.
Since 
$\Upsilon^{(1)}(x_{1},\ldots,x_{n})$
defines a rectangular relation, 
$\sigma_{1}$ is a congruence.
It follows from the definition that 
$\sigma_{0}^{(1)}\supseteq\sigma_{1}$.
Let us prove the backward inclusion.
Choose a pair $(a,b)\in \sigma_{0}^{(1)}$.
Since $\sigma_{0}$ is defined by 
$\Xi(x_{i},x_{i}')$, 
by Lemma~\ref{LinkedStayLinked},
$a$ and $b$ should be linked in 
$\Xi^{(1)}(x_{i},x_{i}')$.
Since $\sigma_{1}$ is a congruence, 
$a$ and $b$ can be linked only if $(a,b)\in\sigma_{1}$,
which means that $\sigma_{0}^{(1)} = \sigma_{1}$.
\end{proof}

\begin{lem}\label{AddLinearVariables}
Suppose $D^{(1)}$ is a minimal linear reduction for $\Upsilon$, 
$\Upsilon^{(1)}(x_{1},\ldots,x_{n})$ defines a subdirect rectangular relation,
$\Var(\Upsilon) = \{x_{1},\ldots,x_{n},v_{1},\ldots,v_{r}\}$,
and $\Omega = \Upsilon \wedge \bigwedge_{i=1}^{r} \sigma_{i}(v_{i},u_{i})$,
where $\sigma_{i}=\ConLin(D_{v_{i}})$. 
Then
$(\ConOne(\Omega(x_{1},\ldots,x_{n},u_{1},\ldots,u_{r}), j))^{(1)} =
\ConOne(\Upsilon^{(1)}(x_{1},\ldots,x_{n}),j))$ for every~$j$.
\end{lem}

\begin{proof}
Without loss of generality assume that $j=1$.
Since the reduction 
$D^{(1)}$ is minimal, we have the following inclusion
$$(\ConOne(\Omega(x_{1},\ldots,x_{n},u_{1},\ldots,u_{r}), 1))^{(1)}
  \supseteq
\ConOne(\Upsilon^{(1)}(x_{1},\ldots,x_{n}), 1).
$$
Let us  prove the backward inclusion.
Suppose 
$\Omega(x_{1},\ldots,x_{n},u_{1},\ldots,u_{r})$ and
$\Upsilon^{(1)}(x_{1},\ldots,x_{n})$
define the relations $\rho'$ and $\rho$ respectively.
Choose $a,b\in D_{x_{1}}^{(1)}$ such that $(a,b)\in \ConOne(\rho',1)$.
For some $\beta$ we have $a\beta,b\beta\in\rho'$.
Since
$\rho$ is subdirect,
there exist $\alpha_{a}$ and $\alpha_{b}$ in $D^{(1)}$ such that $a\alpha_{a},b\alpha_{b}\in\rho'$.
Since $w$ preserves $\rho'$, 
\begin{align*}w(a,a,\ldots,a) w(\alpha_{a},\beta,\ldots,\beta) \in\rho',\\
w(a,b,\ldots,b) w(\alpha_{a},\beta,\ldots,\beta) \in\rho',\\
w(b,\ldots,b,a) w(\alpha_{b},\beta,\ldots,\beta) \in\rho',\\
w(b,b,\ldots,b) w(\alpha_{b},\beta,\ldots,\beta) \in\rho'.
 \end{align*} 
By Lemma~\ref{LinearSpecialWNU}, $w(\alpha_{a},\beta,\ldots,\beta)$ and
$w(\alpha_{b},\beta,\ldots,\beta)$ belong to $D^{(1)}$.
Then,
for $c= w(a,b,\ldots,b)=w(b,\ldots,b,a)$ we have
$(a,c),(c,b)\in \ConOne(\rho,1)$.
Since $\rho$ is rectangular,
we have $(a,b)\in \ConOne(\rho,1)$.
\end{proof}

\subsection{Adding linear variable}

Below we formulate few statements from \cite{KeyRelations}
that will help us to prove the main property of a bridge.
This property will be the main ingredient of the proof 
of the fact that 
$A'$ from the informal description of the algorithm 
should be of codimension 1. 

A relation $\rho\subseteq A^{n}$ is called \emph{strongly rich}
if for every tuple
$(a_{1},\ldots,a_{n})$ and every $j\in \{1,\ldots,n\}$ there exists a unique $b\in A$
such that $(a_{1},\ldots,a_{j-1},b,a_{j+1},\ldots,a_n)\in\rho.$
We will need two statements from \cite{KeyRelations}.

Recall that for any bridge $\rho$ by $\widetilde{\rho}$ we denote 
the binary relation defined by 
$\widetilde{\rho}(x,y) = \rho(x,x,y,y)$.

\begin{thm}\label{StronglyRichRelationTHM}\cite{KeyRelations}
Suppose $\rho\subseteq A^{n}$ is a strongly rich relation
preserved by an idempotent WNU.
Then there exists an abelian group $(A;+)$
and bijective mappings
$\phi_1$, $\phi_2$, \ldots,$\phi_n: A\to A$ such that
\[\rho = \{(x_1,\ldots,x_n)\mid \phi_1(x_1)+\phi_2(x_2) + \ldots +\phi_n(x_n) = 0\}.\]
\end{thm}

\begin{lem}\label{LinearWNU}\cite{KeyRelations}
Suppose $(G;+)$ is a finite abelian group,
the relation $\sigma\subseteq G^{4}$ is defined by
$\sigma = \{(a_1,a_2,a_3,a_4)\mid a_1+a_2=a_3+a_4\}$,
and $\sigma$ is preserved by an idempotent WNU $f$.
Then $f(x_{1},\ldots,x_{n}) = t\cdot x_{1}+t\cdot x_2 + \ldots + t\cdot x_{n}$
for some $t\in \{1,2,3,\ldots\}$.
\end{lem}

\begin{thm}\label{LinkedBridgeThm}
Suppose $\sigma\subseteq A^{2}$ is a congruence,
 $\rho$ is a bridge from
$\sigma$ to $\sigma$ such that $\widetilde{\rho}$ is a full relation,
$\proj_{1,2}(\rho) = \omega$, $\omega$ is a minimal relation stable under $\sigma$ such that
$\omega\supsetneq \sigma$.
Then there exists a prime number $p$
and a relation $\zeta\subseteq A\times A\times \mathbb Z_{p}$
such that
$\proj_{1,2}\zeta = \omega$
and $(a_{1},a_{2},b)\in\zeta$ 
implies that 
$(a_{1},a_{2})\in \sigma\Leftrightarrow (b=0)$.
\end{thm}

\begin{proof}
Since the relations $\rho$ and $\omega$ are stable under $\sigma$,
we consider $A/\sigma$ instead of $A$ and assume that $\sigma$ is the equality relation.

Without loss of generality we assume that
$\rho(x_{1},x_{2},y_{1},y_{2}) = \rho(y_{1},y_{2},x_{1},x_{2})$
and
$(a,b,a,b)\in\rho$ for any $(a,b)\in\omega$.
Otherwise,
we consider the relation $\rho'$ instead of $\rho$, where
$$\rho'(x_{1},x_{2},y_{1},y_{2}) = \exists z_{1}\exists z_{2}\;
\rho(x_{1},x_{2},z_{1},z_{2}) \wedge \rho(y_{1},y_{2},z_{1},z_{2}).$$

We prove by induction on the size of $A$.
Assume that for some subuniverse $A'\subsetneq A$
we have $(A'\times A')\cap (\omega\setminus \sigma) \neq \varnothing$.
By $\sigma'$ we denote the equality relation on $A'$.
By $\omega'$ we denote a minimal relation such that
$\sigma'\subsetneq\omega'\subseteq (A'\times A')\cap \omega$.
Since $\proj_{1,2}(\rho\cap (\omega'\times\omega'))
=\omega'\supsetneq \sigma'$, 
the relation $\rho\cap (\omega'\times\omega')$
is a bridge from $\sigma'$ to $\sigma'$.
The inductive assumption for
$\rho\cap (\omega'\times\omega')$ implies that there exists a relation
$\zeta'\subseteq A'\times A'\times \mathbb Z_{p}$
such that
$(x_{1},x_{2},0)\in \zeta'\Leftrightarrow (x_{1},x_{2})\in\sigma'$
and $\proj_{1,2}(\zeta') = \omega'$.
Put $$\zeta(x_{1},x_{2},z) = \exists y_{1} \exists y_{2} \;\rho(x_{1},x_{2},y_{1},y_{2})\wedge \zeta'(y_1,y_2,z).$$
By the minimality of $\omega$, 
we have $\proj_{1,2}(\zeta)=\omega$.
The remaining property of $\zeta$ 
follows from the fact that 
$\rho$ is a bridge and the properties of $\zeta'$.

Thus, we may assume that for any subuniverse $A'\subsetneq A$
we have $(A'\times A')\cap (\omega\setminus\sigma)= \varnothing$.

Consider a pair $(a_{1},a_{2})\in \omega\setminus\sigma$.
Let $A' = \{a\mid (a_{1},a)\in\omega\}$.
Since $\omega\supsetneq \sigma$, 
we have $a_{1}\in A'$, and therefore
$(a_{1},a_{2})\in(A'\times A')\cap (\omega\setminus\sigma) \neq\varnothing$
and $A'=A$.
Thus, $\{a\mid (a_{1},a)\in\omega\} =\{a \mid (a,a_{2})\in\omega\} = A$.
Hence, any element connected in $\omega$ to some other element is connected to all elements. 
Therefore, $(a_{1},a),(a,a_{2})\in\omega$ for every $a\in A\setminus\{a_{1},a_{2}\}$,
which for $|A|>2$ implies that $\omega = A\times A$.

If $|A|=2$ and $\omega \neq A\times A$ then $\omega = \{(a,a),(a,b),(b,b)\}$ and $\rho$ is uniquely defined.
We know \cite{Post} that 
any clone on a 2-element domain containing an idempotent WNU operation 
contains majority operation, conjunction, disjunction, or 
minority operation. 
None of them preserve $\rho$, 
which contradicts our assumptions.

Thus, we proved that $\omega = A\times A$ and $A$ has no proper subuniverses of size at least 2.

Note the remaining part of the proof could also be 
derived from known facts of commutator theory.
In fact, it follows from the properties of $\rho$ 
that $\sigma$ (the equality) is an equivalence block of a congruence on $A^{2}$, 
which means that $A$ is Abelian. 
Using Abelianess for Taylor varieties (since we have a WNU), 
we could also define the required ternary relation 
$\zeta$ (see \cite{bergman2011universal} for more details).
Nevertheless, we do not want to introduce new algebraic notions,
and give a proof based on two claims from \cite{KeyRelations}.

Let us show that for any $a_{1},a_{2},a_{3}\in A$ there exists a unique $a_{4}$ 
such that $(a_{1},a_{2},a_{3},a_{4})\in \rho$. 
For every $a\in A$ put $\lambda_{a}(x_1,x_2) = \exists y_2 \rho(x_1,x_2,a,y_{2})$.
It is easy to see that $\sigma\subsetneq\lambda_{a}\subseteq\omega$.
Therefore $\lambda_{a}=\omega = A\times A$ for every $a$.
We consider the unary relation defined by $\delta(x) = \rho(a_{1},a_{2},a_{3},x)$.
By the above fact $\delta$ is not empty.
Since $\rho$ is a bridge, $\delta$ is not full.
If $\delta$ contains more than one element, then
we get a contradiction with the fact that there are no proper subuniverses
of size at least 2.

Then $\rho$ is a strongly rich relation.
By Theorem~\ref{StronglyRichRelationTHM}, there exist an Abelian group $(A;+)$ and bijective mappings
$\phi_1, \phi_2, \phi_3,\phi_4\colon A\to A$ such that
$$\rho = \{(a_1,a_2,b_1,b_2)\mid \phi_1(a_1)+\phi_2(a_2) + \phi_{3}(b_{1}) +\phi_4(b_2) = 0\}.$$
Without loss of generality
we can assume that 
$\phi_{1}(x) = x$.
We know that $(a,a,b,b)\in\rho$ for any $a,b\in A$,
then 
$\phi_{1}(x)+\phi_{2}(x)+\phi_{3}(0)+\phi_{4}(0)=0$, which means 
that $\phi_{2}(x) = -x - \phi_{3}(0)-\phi_{4}(0)$.
Since $(a,b,a,b)\in\rho$ for any $a,b\in A$,
we have 
$\phi_{1}(x)+\phi_{2}(0)+\phi_{3}(x)+\phi_{4}(0) = 0$, which means 
that 
$\phi_{3}(x) = -\phi_{1}(x) - \phi_{2}(0)-\phi_{4}(0)
=-x+\phi_{3}(0)$.
Similarly, 
since $\phi_{1}(0)+\phi_{2}(0)+\phi_{3}(x)+\phi_{4}(x)=0$,
we have $\phi_{4}(x) = x-\phi_{3}(0)-\phi_{2}(0)-\phi_{1}(0) =
x+\phi_{4}(0)$.
Substituting 
this into the definition of 
$\rho$ we obtain 
$$\rho = \{(a_1,a_2,b_1,b_2)\mid a_{1}-a_{2}-a_{3}+a_{4} = 0\}.$$


It follows from Lemma~\ref{LinearWNU}
that $w$ on $A$ is defined by $t(x_{1}+\ldots+x_{m})$, 
Since $w$ is special, $t\cdot (t-1)$ should be divided by the 
order of any element of $A$. By the idempotency, 
$t$ and the order of any element are coprime. 
Hence, $t-1$ should be divided by the order of any element and we may put $t=1$.
Therefore, the relation $\zeta\subseteq A\times A\times A$ defined by
$\zeta = \{(b_1,b_2,b_{3})\mid b_{1}-b_{2}+b_{3}=0\}$ is preserved by $w$.
If $(A;+)$ is not simple, then 
any equivalence class of a congruence is
a proper subuniverse of size at least 2,
which contradicts our assumption.
Therefore,
$(A;+)$ is a simple Abelian group.
\end{proof}

\begin{cons}\label{LinkedLink}
Suppose $\sigma\subseteq A^{2}$ is an irreducible congruence
and $\rho$ is a bridge from
$\sigma$ to $\sigma$ such that $\widetilde{\rho}$ is a full relation.
Then there exists a prime number $p$
and a relation $\zeta\subseteq A\times A\times \mathbb Z_{p}$
such that $\proj_{1,2}\zeta = \sigma^{*}$
and $(a_{1},a_{2},b)\in\zeta$ 
implies that 
$(a_{1},a_{2})\in \sigma\Leftrightarrow (b=0)$.
\end{cons}

\begin{lem}\label{ReflexiveBridgeProperty}
Suppose $\rho\subseteq A^{4}$ is an optimal bridge from $\sigma_{1}$ to $\sigma_{2}$, 
and $\sigma_{1}$ and $\sigma_{2}$ are different irreducible congruences.
Then $\widetilde \rho\supsetneq\sigma_{2}$.
\end{lem}
\begin{proof}
Since the first two variables 
are stable under $\sigma_{1}$ and 
the last two variables 
are stable under $\sigma_2$, 
we have 
$\sigma_{1}\subseteq \widetilde\rho$
and
$\sigma_{2}\subseteq \widetilde\rho$.
Assume that the lemma does not hold, then $\widetilde{\rho}=\sigma_{2}$.

Since $\sigma_{1}$ and $\sigma_{2}$ are different, 
we obtain $\sigma_{1}\subsetneq\sigma_{2}$,

First, we want 
$(a,d)$ to be from $\sigma_{2}$ for every $(a,b,c,d)\in \rho$.
Put $\rho_{1}(x_{1},x_{2},y_{1},y_{2}) = \rho(x_{1},x_{2},y_{1},y_{2})\wedge \sigma_{2}(x_{1},y_{2})$.
If $\rho_{1}$ is a bridge then we replace $\rho$ by $\rho_{1}$.
Assume that $\rho_{1}$ is not a bridge,
then for every $(a,b,c,d)\in \rho$ with $(a,d)\in \sigma_{2}$ we have $(a,b)\in\sigma_{1}$.
Put
$\rho_{2}(x_{1},x_{2},y_{1},y_{2}) = \exists z\; \rho(x_{1},x_{2},z,y_{1})\wedge \sigma_{2}(x_{1},y_{2})$.
Let us show that 
$\rho_{2}$ is a bridge.
Suppose $(x_{1},x_{2})\in\sigma_{1}$
and $(x_1,x_2,y_1,y_2)\in\rho_{2}$.
Then $(x_{1},x_{2},z,y_{1})\in\rho$. 
Since $\rho$ is a bridge from $\sigma_1$ to $\sigma_2$, this implies that
$(z,y_{1})\in\sigma_2$ and 
$(x_1,y_1)\in\widetilde\rho$.
Since $\widetilde \rho=\sigma_{2}$,
we have $(x_{1},y_{1})\in\sigma_{2}$ 
and therefore 
$(y_1,y_{2})\in\sigma_{2}$.
If $(y_{1},y_2)\in\sigma_{2}$
and $(x_1,x_2,y_1,y_2)\in\rho_{2}$,
then 
$(x_{1},y_{1})\in\sigma_{2}$ and by the above assumption 
we have $(x_{1},x_{2})\in \sigma_{1}$.
It remains to show 
that $\proj_{1,2}(\rho_{2})\supsetneq \sigma_{1}$. 
Consider any tuple
$(a_{1},a_{2},a_{3},a_{4})\in\rho$ such that 
$(a_{1},a_{2})\notin\sigma_{1}$,
then $(a_{1},a_{2},a_{4},a_{1})\in\rho_{2}$.
Thus, $\rho_{2}$ is a bridge with the required property, so 
we replace $\rho$ by $\rho_{2}$.

Second, we want 
$\proj_{1,2}(\rho)$ to be equal to $\sigma_{1}^{*}$, 
and $\proj_{3,4}(\rho)$ to be equal to $\sigma_{2}^{*}$.
To achieve this we replace $\rho$ by the relation defined by $\rho(x_{1},x_{2},y_{1},y_{2})\wedge \sigma_{1}^{*}(x_{1},x_{2})\wedge
\sigma_{2}^{*}(y_{1},y_{2})$,
which has the same properties.

Recall that 
\emph{a polynomial} is an operation that can be 
defined by a term over the basic operations of an algebra
and constant operations. In our case,
a polynomial is an operation defined by a term over the WNU $w$ and 
constants.
Let $D$ be a minimal subset (not necessarily a subuniverse) of $A$ such that 
\begin{enumerate}
    \item there exists a unary polynomial $h$ such that 
    $h(h(x)) = h(x)$ and $h(A) = D$, and
    \item $(\sigma_{2}^{*}\setminus \sigma_{2})\cap D^{2}\neq \varnothing.$
\end{enumerate}

Since constants preserve a reflexive bridge $\rho$
and congruences 
$\sigma_{1}$ and $\sigma_{2}$, 
the unary polynomial $h$ also preserves $\rho$, $\sigma_{1}$ and
$\sigma_{2}$.
It is not hard to see that
$h(w(x_{1},\ldots,x_{m}))$ is an idempotent WNU on $D$, 
then by $w^{D}$ we denote a special WNU on $D$ that can be derived from
the idempotent WNU on $D$. 
For any relation $\delta$, by $\delta^{D}$ we denote its restriction to $D$
(that is $h(\delta)$).
It is not hard to see that 
$\rho^{D}$ is a bridge from $\sigma_{1}^{D}$ to $\sigma_{2}^{D}$.

The idea of the proof is to define a bridge $\epsilon^{D}$ 
from $\sigma_{2}^{D}$ to $\sigma_{2}^{D}$ 
such that $\widetilde{\epsilon^{D}}\not\subseteq \sigma_{2}$.
Then we define a bridge $\epsilon$ from $\sigma_{2}$ to $\sigma_2$ having 
the same property and use this bridge to make $\widetilde{\rho}$ bigger, 
which gives us a contradiction because 
$\rho$ is optimal.

Consider $(b_{1},b_{2})\in(\sigma_{2}^{*})^{D}\setminus\sigma_{2}^{D}$
and the unary operation $g_{b_1}(x) = w^D(b_{1},\ldots,b_{1},x)$.
Since 
$w^D$ is a special WNU, 
$g_{b_1}(g_{b_1}(x)) = g_{b_1}(x)$. 
Let us show that $(\sigma_{2}^{*}\setminus \sigma_{2})\cap (D')^{2}\neq \varnothing$,
where $D' = g_{b_1}(D)$.

Since $\proj_{3,4}(\rho)= \sigma_{2}^{*}$, 
there are $a_1,a_2$ such that 
$(a_1,a_2,b_1,b_2)\in\rho$.
Since $(a_{1},b_{2})\in\sigma_{2}$, 
and $(b_{1},b_2)\notin\sigma_{2}$, 
we have $(a_1,b_1)\notin\sigma_2$.
Consider the relation 
$\delta(x,y)$
defined by 
$\exists x_1\exists x_2 \exists y_2
\sigma_{2}(x,x_1)\wedge \rho(x_1,x_2,y,y_2)$.
It follows from the definition 
that $\delta$ is stable under $\sigma_{2}$.
Also $(a_1,b_1)\in\delta$, therefore
$\delta\supsetneq\sigma_{2}$.
From irreducibility of $\sigma_2$ we obtain that $(b_1,b_2)\in \delta$.
Then by the definition of $\delta$ 
there exist
$(c_1,c_2,b_2,c_3)\in\rho$ such that 
$(b_1,c_{1})\in\sigma_{2}$.
Put 
$d_{i} = h(c_{i})$ for $i=1,2,3$.
Then 
$(d_1,d_2,b_2,d_3)\in\rho$ 
(we have $h(b_2)=b_2$).
Since $h$ preserves $\sigma_{2}$ 
and $h(b_1) = b_1$, 
we have $(d_1,b_1)\in\sigma_2$.
Therefore, 
$(d_1,b_2)\notin\sigma_2$. 
Since $\widetilde{\rho} = \sigma_2$, 
we have 
$(d_1,d_2)\notin\sigma_1$ and 
$(b_2,d_3)\notin\sigma_2$.
Let $E$ be the equivalence class of 
$\sigma_{2}^{D}$ containing $d_1$ and $d_2$
(they are in one class because
$\proj_{1,2}(\rho)=\sigma_{1}^{*}\subseteq \sigma_2$).
By $w'$ we denote $w^{D}$ restricted to $E$,
put $\rho' = \rho\cap (E^{2}\times D^{2})$ and
$\sigma_{1}' = \sigma_{1}\cap E^{2}$.

Since $(d_1,d_2)\in (\sigma_{1}^{*}\setminus \sigma_{1})\cap E^{2}$, 
we can find a minimal relation $\omega\subseteq\sigma_{1}^{*}\cap E^{2}$ stable under $\sigma_{1}'$
such that $\omega\supsetneq\sigma_{1}'$.
It is not hard to check that
the formula
$$\rho''(x_{1},x_{2},x_{1}',x_{2}') = \exists y_{1}\exists y_{2}\;
\rho'(x_{1},x_{2},y_{1},y_{2})\wedge\rho'(x_{1}',x_{2}',y_{1},y_{2})\wedge \omega(x_{1},x_{2})\wedge \omega(x_{1}',x_{2}')$$
defines a reflexive bridge $\rho''$ from $\sigma_{1}'$ to $\sigma_{1}'$.
Since $\widetilde{\rho}=\sigma_{2}$ and $E$ is an equivalence class of 
$\sigma_{2}^{D}$, 
we have 
$\widetilde{\rho''}=E^{2}$.
Then by Theorem~\ref{LinkedBridgeThm},
there exists a prime number $p$
and a relation $\zeta\subseteq E\times E\times \mathbb Z_{p}$
such that
$\proj_{1,2}(\zeta)=\omega$ 
and $(a_1,a_2,b)\in\zeta$ implies 
$(a_1,a_2)\in\sigma_{1}'\Leftrightarrow (b=0)$.
We want to show that
for each $(e_{1}, e_{2}) \in \omega\setminus\sigma_{1}'$
we have $(w^{D}(d_1,\ldots,d_{1},e_{1}),w^{D}(d_{1},\ldots,d_{1},e_{2}))\notin\sigma_{1}$.
In fact, choose $b\in \mathbb Z_{p}$ such that 
$(e_1,e_2,b)\in\zeta$. Note that $b\neq 0$
and $(d_1,d_1,0)\in\zeta$.
Applying $w'$ to 
$n$ tuples $(d_1,d_1,0)$ and one tuple $(e_1,e_2,b)$
we get, by Lemma~\ref{LinearSpecialWNU}, 
a tuple 
$(w^{D}(d_1,\ldots,d_{1},e_{1}),w^{D}(d_{1},\ldots,d_{1},e_{2}),b)\in\zeta$.
Since $b\neq 0$, we have $$(w^{D}(d_1,\ldots,d_{1},e_{1}),w^{D}(d_{1},\ldots,d_{1},e_{2}))\notin\sigma_{1}.$$
We can find $(e_3,e_4)\in\sigma_{2}^{*}\setminus \sigma_{2}$ such that 
$(e_1,e_2,e_3,e_4)\in\rho$.
Since $h$ preserves $\rho$, $w^{D}$ preserves $\rho^{D}$, 
and $\rho^{D}$ is a bridge,
we can derive that  
$(w^{D}(d_1,\ldots,d_1,h(e_3)),w^{D}(d_1,\ldots,d_1,h(e_4)))\notin\sigma_{2}$.
Since $(d_1,b_1)\in\sigma_{2}$
we also have 
$(w^{D}(b_1,\ldots,b_1,h(e_3)),w^{D}(b_1,\ldots,b_1,h(e_4)))\notin\sigma_{2}$.
Thus, we proved
that $(\sigma_{2}^{*}\setminus \sigma_{2})\cap (D')^{2}\neq \varnothing$,
where $D' = g_{b_1}(D)$ and 
$g_{b_1}(x) = w^D(b_{1},\ldots,b_{1},x)$.
If $D'\neq D$, 
then we found a smaller set $D'$ and 
the corresponding polynomial 
$g_{b_{1}}(h(x))$, which contradicts the minimality of $D$.

Thus, for any 
$(b_{1},b_{2})\in(\sigma_{2}^{*})^{D}\setminus\sigma_{2}^{D}$
we have 
$w^{D}(b_1,\ldots,b_1,x) = x$.
Let us show that 
$\ConOne(\rho^{D},1) = \sigma_{1}^{D}$
and $\ConOne(\rho^{D},3) = \sigma_{2}^{D}$.
Choose
$(a_{1},a_{2},a_{3},a_{4})\in\rho^{D}$.
We consider two cases.

Case 1:
$(a_{1},a_{2})\in\sigma_{1}$
and 
$(a_{3},a_{4})\in\sigma_{2}$.
Then
for any tuple 
$(a_{1}',a_{2},a_{3},a_{4})\in\rho^{D}$
we have $(a_{1},a_{1}')\in\sigma_{1}$.
Similarly,
for any tuple
$(a_{1},a_{2},a_{3}',a_{4})\in\rho^{D}$
we have $(a_{3},a_{3}')\in\sigma_{2}$.

Case 2: $(a_{1},a_{2})\notin\sigma_{1}$
and $(a_{3},a_{4})\notin\sigma_{2}$.
Since $(a_{1},a_{4})\in\sigma_{2}$
and
$(a_{1},a_{2})\in \sigma_{1}^{*}\subseteq \sigma_{2}$, we have
$(a_{1},a_{3}),(a_{2},a_{3}), (a_{3},a_{4})\in (\sigma_{2}^{*})^{D}\setminus\sigma_{2}^{D}$.
Also notice that $\sigma_{2}^{*}$ is symmetric.
Assume that 
$(a_{1}',a_{2},a_{3},a_{4})\in\rho^{D}$.
Since $w^{D}$ preserves $\rho^{D}$,
we have
$$(w^{D}(a_{1}',a_{1},\ldots,a_{1}),
w^{D}(a_{2},\ldots,a_{2},a_{1}),
w^{D}(a_{3},\ldots,a_{3},a_{1}),
w^{D}(a_{4},\ldots,a_{4},a_{1}))\in\rho^{D}.$$
As we showed before
this tuple equals
$(a_{1}',a_{1},a_{1},a_{1})$, 
which 
means that $(a_{1},a_{1}')\in\sigma_{1}$.
Similarly, if 
$(a_{1},a_{2},a_{3}',a_{4})\in\rho^{D}$
we consider a tuple
$$(w^{D}(a_{1},\ldots,a_{1},a_{3}),
w^{D}(a_{2},\ldots,a_{2},a_{3}),
w^{D}(a_{3}',a_{3},\ldots,a_{3},a_{3}),
w^{D}(a_{4},\ldots,a_{4},a_{3}))\in\rho^{D},$$
that equals
$(a_{3},a_{3},a_{3}',a_{3})$, which means
that $(a_{3},a_{3}')\in\sigma_{2}$.
Thus we proved that $\ConOne(\rho^{D},1) = \sigma_{1}^{D}$
and $\ConOne(\rho^{D},3) = \sigma_{2}^{D}$.

Consider $(a_{1},a_{2},b_{1},b_{2})\in\rho^{D}$ with $(b_{1},b_{2})\notin\sigma_{2}$ and the formula
$$\Theta = \rho(z,x_{1},x_{2},x_{3})
\wedge
\rho(z',x_{1},x_{2}',x_{3}')
\wedge \rho(z,x_{4},x_{5},x_{6})
\wedge \rho(z',x_{4},x_{5}',x_{6}').$$
Let $\epsilon$ be the relation defined by  $\Theta(x_{2},x_{2}',x_{5},x_{5}')$.
Since $h(h(x)) = h(x)$ and $h$ preserves $\rho$, 
$\epsilon^{D}$ is defined by the same formula 
but with 
$\rho^{D}$ instead of $\rho$ everywhere.
Let us prove that $\epsilon^{D}$ is a bridge
from 
$\sigma_{2}^{D}$ to 
$\sigma_{2}^{D}$.
Assume that 
$(x_{2},x_{2}')\in\sigma_{2}^{D}$.
Since $(x_{1},z),(x_{1},z')\in
\sigma_{1}^{*}\subseteq \sigma_{2}$,
we have $(z,z')\in\sigma_{2}^{D}$. 
Recall that 
$(a,d)\in\sigma_{2}$ whenever $(a,b,c,d)\in\rho$,
hence 
$(x_{3},x_{3}'),(x_{6},x_{6}')\in\sigma_{2}^{D}$.
Since  
$\ConOne(\rho^{D},1) = \sigma_{1}^{D}$,
we have $(z,z')\in\sigma_{1}$.
Since $\ConOne(\rho^{D},3) = \sigma_{2}^{D}$,
we have $(x_{5},x_{5}')\in\sigma_{2}$.
In the same way we can show that 
if $(x_{5},x_{5}')\in\sigma_{2}^{D}$, 
then $(x_{2},x_{2}')\in\sigma_{2}^{D}$.
Since all the variables of $\epsilon$
are stable under $\sigma_{2}$
and $\epsilon^{D}$ is reflexive, 
we have 
$\proj_{1,2}(\epsilon^{D})\supseteq \sigma_{2}^{D}$.
By sending
$(z,x_{1},x_{2},x_{3})$ to $(a_{1},a_{2},b_{1},b_{2})$,
$(z',x_{1},x_{2}',x_{3}')$ to $(a_{2},a_{2},a_{2},a_{2})$,
$(z,x_{4},x_{5},x_{6})$ to $(a_{1},a_{2},b_{1},b_{2})$,
$(z',x_{4},x_{5}',x_{6}')$ to $(a_{2},a_{2},a_{2},a_{2})$,
we show that $(b_{1},a_{2},b_{1},a_{2})\in \epsilon$.
Since $(a_1,a_2),(a_1,b_2)\in\sigma_2$ and $(b_{1},b_{2})\notin\sigma_2$,
we have $(b_1,a_2)\notin \sigma_2$,
and therefore $\proj_{1,2}(\epsilon^{D})\supsetneq \sigma_{2}^{D}$.
In the same way we can show that 
$\proj_{3,4}(\epsilon^{D})\supsetneq \sigma_{2}^{D}$.
Thus $\epsilon^{D}$ is a bridge.

Let us show that $\epsilon$ is a bridge
from $\sigma_{2}$ to $\sigma_{2}$.
Assume the contrary.
Then without loss of generality we assume that there exists
$(d_{0},d_{0},d_{1},d_{2})\in \epsilon$ such that
$(d_{1},d_{2})\notin\sigma_{2}$.
Put $\delta_{0}(y,z) = \exists x\;\epsilon(x,x,y,z)$.
The relation $\delta_{0}$ is stable under $\sigma_2$ and 
strictly larger than $\sigma_2$, hence 
$\delta_{0}\supseteq \sigma_{2}^{*}$ and 
$(b_{1},b_{2})\in\delta_{0}$.
Then there exists $d$ such that
$(d,d,b_{1},b_{2})\in\epsilon$,
which means that $(h(d),h(d),b_{1},b_{2})\in\epsilon^{D}$.
This contradicts the fact that 
$\epsilon^{D}$ is a bridge. 
Hence, $\epsilon$ is also a bridge.
By sending $(z,x_{1},x_{2},x_{2}',x_{3},x_{3}')$
to $(a_{1},a_{2},b_{1},b_{1},b_{2},b_{2})$
and
$(z',x_{4},x_{5},x_{5}',x_{6},x_{6}')$
to $(a_{1},a_{1},a_{1},a_{1},a_{1},a_{1})$
we can show that
$(b_{1},a_{1})\in\widetilde{\epsilon}$.
If we compose 
bridges 
$\rho$ and $\epsilon$,
then we get a bridge $\epsilon'$ 
from $\sigma_{1}$ to $\sigma_{2}$
containing $(b_{1},b_1,a_{1},a_1)$.
Hence
$\widetilde{\epsilon'}\supsetneq \widetilde{\rho}$, 
which contradicts the fact that $\rho$ is optimal.
\end{proof}

\subsection{Existence of a bridge}

In this subsection we show four ways to build a 
bridge:
from congruences with an additional property, 
from a rectangular relation, 
by composing bridges appearing in the instance,
and 
from a pp-formula.

\begin{lem}\label{MinimalsAdjacent}
Suppose
$\sigma, \sigma_{1}$, and $\sigma_{2}$ are congruences on $A$, 
$\sigma\cap\sigma_{1} = \sigma\cap\sigma_{2}$,
and $\sigma\setminus\sigma_{1}\neq\varnothing$.
Then $\sigma_{1}$ and $\sigma_{2}$ are adjacent.
\end{lem}

\begin{proof}
Let us define a relation $\rho$ by
$$\rho(x_{1},x_{2},y_{1},y_{2}) =
\exists z_{1}\exists z_{2}\;
\sigma_{1}(x_{1},z_{1})\wedge \sigma_{2}(z_{1},y_{1})\wedge
\sigma_{1}(x_{2},z_{2})\wedge \sigma_{2}(z_{2},y_{2})\wedge
\sigma(z_{1},z_{2}).$$
It is clear that 
the first two variables of $\rho$ are stable under
$\sigma_{1}$ and the last two variables 
are stable under $\sigma_{2}$.

Let us show that
for any $(a_{1},a_{2},a_{3},a_{4})\in\rho$ that
$(a_{1},a_{2})\in\sigma_{1}\Leftrightarrow(a_{3},a_{4})\in\sigma_{2}$.
In fact, if $(x_{1},x_{2})\in\sigma_{1}$, then 
$(z_{1},z_{2})\in\sigma_{1}$.
Since $\sigma\cap\sigma_{1} = \sigma\cap\sigma_{2}$, 
we have $(z_{1},z_{2})\in\sigma_{2}$.
Therefore, $(y_{1},y_{2})\in\sigma_{2}$.

Also $(a,a,a,a)\in\rho$ for any $a\in A$.
Choose $(a,b)\in\sigma\setminus\sigma_{1}$.
Then $(a,b,a,b)\in\rho$ (put $z_{1} = a$, $z_{2} = b$), which proves that $\rho$ is a reflexive bridge.
\end{proof}

\begin{lem}\label{OneLink}
Suppose $\rho\subseteq A_{1}\times\dots\times A_{n}$ is a subdirect relation,
the first and the last variables of $\rho$ are rectangular,
and there exist
$(b_{1},a_{2},\ldots,a_n),(a_{1},\ldots,a_{n-1},b_{n})\in\rho$
such that $(a_{1},a_{2},\ldots,a_n)\notin\rho$.
Then there exists a bridge $\delta$ from $\ConOne(\rho,1)$ to $\ConOne(\rho,n)$
such that $\widetilde{\delta} = \proj_{1,n}(\rho)$.
\end{lem}
\begin{proof}
The required bridge can be defined by
$$\delta(x_{1},x_{2},y_{1},y_{2}) =
\exists z_{2}\dots\exists z_{n-1}\;
\rho(x_{1},z_{2},\ldots,z_{n-1},y_{1})\wedge
\rho(x_{2},z_{2},\ldots,z_{n-1},y_{2}).$$
In fact, since the first and the last variables of $\rho$ are rectangular, 
we have $(x_{1},x_{2})\in\ConOne(\rho,1)$
if and only if
$(y_{1},y_{2})\in\ConOne(\rho,n)$.
It remains to notice that 
$(b_{1},a_{1},a_{n},b_{n})\in\delta$
and $(b_{1},a_{1})\notin\ConOne(\rho,1)$.
\end{proof}

Recall that by Lemma~\ref{CriticalMeansIrreducible}
$\ConOne(\rho,i)$ is an irreducible congruence 
for every critical subdirect rectangular relation $\rho$ and its coordinate $i$.

\begin{lem}\label{OneLinkCritical}
Suppose $\rho\subseteq A_{1}\times\dots\times A_{n}$ is a 
critical subdirect rectangular relation. 
Then
\begin{enumerate}
    \item there exists a bridge $\delta$ from $\ConOne(\rho,1)$ to $\ConOne(\rho,n)$
such that $\widetilde{\delta} = \proj_{1,n}(\rho)$. Moreover, if $n=2$ then 
$\ConOne(\widetilde{\delta},1) = \ConOne(\rho,1)$ and
$\ConOne(\widetilde{\delta},2) = \ConOne(\rho,n)$;
if $n>2$ then 
$\ConOne(\widetilde{\delta},1) \supsetneq \ConOne(\rho,1)$ and
$\ConOne(\widetilde{\delta},2) \supsetneq \ConOne(\rho,n)$.
    \item if $\Opt(\ConOne(\rho,n))\neq \ConOne(\rho,n)$,
    then there  exists a bridge $\delta$ from $\ConOne(\rho,1)$ to $\ConOne(\rho,n)$
such that $\widetilde{\delta}$ contains the projection of 
the cover of $\rho$ onto the first and the last coordinates.
\end{enumerate}
\end{lem}

\begin{proof}
Using the argument from Lemma~\ref{RectangularCriticalArityTwo} we find tuples 
$(b_{1},a_{2},\dots,a_{n})$
$(a_{1},\dots,a_{n-1},b_{n})$
satisfying the conditions of Lemma~\ref{OneLink}.
Then, to prove the first part it is sufficient to use the formula from Lemma~\ref{OneLink} to define a bridge $\delta$.
If $n=2$ then $\widetilde{\delta} = \rho$
and we have the required property.
If $n>2$ then by Lemma~\ref{RectangularCriticalArityTwo}
we have 
$\ConOne(\widetilde{\delta},1) = \ConOne(\proj_{1,n}(\rho),1)\supsetneq\ConOne(\rho,1)$
and 
$\ConOne(\widetilde{\delta},2) = \ConOne(\proj_{1,n}(\rho),2)\supsetneq\ConOne(\rho,n)$.

Let us prove the second part of the claim.
Let $\xi$
be an optimal bridge from $\ConOne(\rho,n)$
to $\ConOne(\rho,n)$. 
Define a bridge 
$\delta(x_{1},x_{2},y_{1},y_{2})$ 
by 
$$
\exists z_{2}\dots\exists z_{n-1}\exists u_1\exists u_2\;
\rho(x_{1},z_{2},\ldots,z_{n-1},u_{1})\wedge
\rho(x_{2},z_{2},\ldots,z_{n-1},u_{2})
\wedge 
\xi(u_1,u_2,y_1,y_2)
.$$
Note that $\delta$ is just a composition of the bridge 
constructed in Lemma~\ref{OneLink} and $\xi$.
Then we have $\widetilde{\delta}(x,y)=
\exists z_{2}\dots\exists z_{n-1}\exists u\;
\rho(x,z_{2},\ldots,z_{n-1},u)
\wedge 
\widetilde{\xi}(u,y).$

Put 
$\rho'(x_{1},\ldots,x_{n}) = 
\exists x_{n}'
\rho(x_{1},\dots,x_{n-1},x_{n}')\wedge 
\widetilde{\xi}(x_{n}',x_{n})$.
Since $\widetilde{\xi}\supsetneq\ConOne(\rho,n)$,
the relation $\rho'$ contains the cover of $\rho$.
Since $\proj_{1,n}(\rho')=\widetilde{\delta}$, 
$\widetilde{\delta}$ contains the projection of 
the cover of $\rho$ onto the first and the last coordinates,
which completes the proof.
\end{proof}

\begin{thm}\label{PathInConnectedComponentThm}
Suppose $\Theta$ is a cycle-consistent connected instance. 
Then for every constraints $C, C'$ with variables $x,x'$ there exists
a bridge $\delta$ from $\ConOne(C,x)$ to $\ConOne(C',x')$ such that
$\widetilde{\delta}$ contains all pairs of elements linked in $\Theta$.
Moreover, if $\ConOne(C'',x'')\neq \LinkedCon(\Theta,x'')$ for some constraint $C''\in\Theta$ and a variable $x''$, then
$\delta$ can be chosen so that
$\widetilde{\delta}$ contains all pairs of elements linked in $\Theta'$,
where $\Theta'$ is obtained from $\Theta$ by replacing every constraint relation by its cover.
\end{thm}
\begin{proof}

Since $C$ and $C'$ are connected, there exists a path
$z_{0} C_{1} z_{1}C_{2} z_{2}\dots C_{t-1} z_{t-1} C_{t} z_{t}$,
where
$z_{0}= x$, $z_{t} = x'$, $C_{1} = C$, $C_{t} = C'$,
$z_{i-1}\neq z_{i}$,
and $C_{i}$ and $C_{i+1}$ are adjacent in $z_{i}$ for every $i$.

By Lemma~\ref{CriticalMeansIrreducible}, every relation defined by $\ConOne(C_{0},x_{0})$
for some $C_{0}$ and $x_{0}$ is an irreducible congruence.
Suppose $\zeta_{i}$ is an optimal bridge from $\ConOne(C_{i},z_{i})$ to $\ConOne(C_{i+1},z_{i})$,
$\delta_{i}$ is a bridge from $\ConOne(C_{i},z_{i-1})$ to $\ConOne(C_{i},z_{i})$
from the first item of Lemma~\ref{OneLinkCritical}
for every $i$.
Then we compose all bridges together and define a new bridge
$\delta(u_{0},u_{0}',v_{t},v_{t}')$
from $\ConOne(C,x)$ to $\ConOne(C',x')$
 by
\begin{multline}\label{BridgeBuilding}
\exists u_{1}\exists u_{1}'\exists v_{1}\exists v_{1}' \dots \exists u_{t-1}\exists u_{t-1}'\exists v_{t-1}\exists v_{t-1}'\;
\delta_{1}(u_{0},u_{0}',v_{1},v_{1}')
\wedge
\\
\bigwedge_{i=1}^{t-1}
\left(
\zeta_{i}(v_{i},v_{i}',u_{i},u_{i}')
\wedge
\delta_{i+1}(u_{i},u_{i}',v_{i+1},v_{i+1}')
\right).\end{multline}
Since 
$\widetilde{\delta}$ can be defined as a composition 
of 
$\widetilde{\zeta}$'s and $\widetilde{\delta}$'s,
and 
$\widetilde{\zeta}$'s are reflexive, 
it follows that 
$\widetilde{\delta}$ contains all pairs of elements linked by this path.
Since $\Theta$ is cycle-consistent, if $x=x'$ then $\delta$ is a reflexive bridge from $\ConOne(C,x)$ to $\ConOne(C',x)$.
Thus we proved that any two constraints with a common variable are adjacent.

Using Lemma~\ref{LinkedConIsCon}, 
we can show that there exists a path in $\Theta$ starting at $x$ and ending at $x'$ that connects any pair of elements linked in $\Theta$.
Since any two constraints with a common variable are adjacent,
we can assume that the above path $z_{0} C_{1} z_{1}C_{2} z_{2}\dots C_{t-1} z_{t-1} C_{t} z_{t}$
connects any pair of elements linked in $\Theta$.
Again,
it follows that $\widetilde{\delta}$ contains all pairs of elements linked in $\Theta$.

To prove the remaining part of the theorem,
assume that $\ConOne(C'',x'')\neq \LinkedCon(\Theta,x'')$ for some constraint $C''\in\Theta$ and a variable $x''$.
First, observe that any bridge $\rho$ from $\sigma_{1}$ to $\sigma_{2}$ defined by the first item of Lemma~\ref{OneLinkCritical}
satisfies one of the following properties:
\begin{enumerate}
\item $\ConOne(\widetilde{\rho},1) = \sigma_{1}$ and
$\ConOne(\widetilde{\rho},2) = \sigma_{2}$,
\item
$\ConOne(\widetilde{\rho},1) \supsetneq \sigma_{1}$ and
$\ConOne(\widetilde{\rho},2) \supsetneq \sigma_{2}$.
\end{enumerate}
If $\sigma_1\neq\sigma_2$, by Lemma~\ref{ReflexiveBridgeProperty} 
an optimal bridge from $\sigma_1$ to $\sigma_{2}$ satisfies property (2)
.
If $\sigma_1=\sigma_2$, an optimal bridge from $\sigma_1$ to $\sigma_2$ obviously satisfies one of the two properties.
Thus, every bridge in (\ref{BridgeBuilding}) satisfies one of the above properties.

Let us show that
if a bridge $\rho$ from $\sigma_1$ to $\sigma_2$ satisfies property (1), 
then for all $(a_1,b_1),(a_2,b_2)\in\widetilde\rho$
we
have $(a_1,a_2)\in\sigma_1\Leftrightarrow(b_1,b_2)\in\sigma_2$.
In fact, if $(a_1,a_2)\in\sigma_1$ then 
since the first two variables of $\rho$ are stable under $\sigma_1$, we 
have $(a_1,b_2)\in\widetilde \rho$, hence 
$(b_1,b_2)\in\ConOne(\widetilde \rho,2)=\sigma_2$.
Now we want to show that 
if we compose bridges $\rho_{1},\dots,\rho_{s}$ together as in (\ref{BridgeBuilding})
and at least one of the bridges satisfies property (2)
then the obtained bridge satisfies property (2).
Let $\rho_{j}$ be a bridge from $\sigma_{j-1}$ to $\sigma_j$ for every $j$, then the composition $\rho$ is a bridge from $\sigma_{0}$ to $\sigma_{s}$.
Consider the first bridge $\rho_{i}$ in the sequence having property (2). 
Then $(a_{i-1},a_{i}),(b_{i-1},b_{i})\in\widetilde\rho_{i}$ for some
$(a_{i-1},b_{i-1})\notin\sigma_{i-1}$
and $a_{i} = b_{i}$. 
Choose $a_{j}$ and $b_{j}$ so that 
$(a_{j-1},a_{j}),(b_{j-1},b_{j})\in\widetilde\rho_{j}$ for every $j$,
and $a_{j} = b_{j}$ for every $j\ge i$.
Then $(a_0,b_{0})\in \ConOne(\widetilde \rho_1,1)\setminus \sigma_{0}$
and $(a_{0},a_{s}), (b_{0},b_{s})\in \widetilde\rho$.
Since $a_{s} = b_{s}$, we get 
$(a_0,b_0)\in\ConOne(\widetilde \rho,1)$
and $\ConOne(\widetilde \rho,1)\supsetneq \sigma_{0}$.
To prove that $\ConOne(\widetilde \rho,2)\supsetneq \sigma_{s}$
we consider the last bridge in the sequence having property (2) and do exactly the same.

By the first part of the theorem
$\Opt(\ConOne(C'',x''))\supsetneq \ConOne(C'',x'')$, hence an optimal 
bridge from $\ConOne(C'',x'')$ to $\ConOne(C'',x'')$ satisfies property (2).
We may assume that any path goes through the variable $x''$ and through the optimal 
bridge from $\ConOne(C'',x'')$
to $\ConOne(C'',x'')$, which
guarantees that every bridge we obtain satisfies property (2).
Consider a constraint $C_{0}\in \Theta$ and a variable $x_{0}$ in it.
Considering a path from $x_{0}$ to $x_{0}$ going through $x''$ we can build a reflexive bridge having property (2),
which means that 
$\Opt(\ConOne(C_{0},x_{0}))\supsetneq \ConOne(C_{0},x_{0})$
for any constraint $C_{0}\in \Theta$ and any variable $x_{0}$ in it.

To complete the proof, 
we replace
every $\delta_{i}$ in
(\ref{BridgeBuilding}) 
by the corresponding bridge $\delta_{i}'$ obtained using the second item 
of Lemma~\ref{OneLinkCritical},
and replace the path by the corresponding 
path connecting any pair of linked 
elements in $\Theta'$.
Since $\widetilde{\delta_{i}'}$ contains the projection of the cover of $C_{i}$ onto the variables $z_{i-1}$ and $z_{i}$,
$\widetilde{\delta}$ contains all pairs of elements linked in $\Theta'$.
\end{proof}

\begin{cons}\label{PathInConnectedComponent}
Suppose $\Theta$ is a cycle-consistent connected instance. 
Then for every constraints $C, C'$ with a common variable $x$ there exists
a bridge $\delta$ from $\ConOne(C,x)$ to $\ConOne(C',x)$ such that
$\widetilde{\delta}$ contains the relation $\LinkedCon(\Theta,x)$.
\end{cons}

\begin{lem}\label{SubconstraintConnectivity}
Suppose $D^{(1)}$ is a minimal one-of-four reduction for an instance $\Upsilon$,
the solution set of $\Upsilon$ is subdirect,
$\Upsilon^{(1)}(x_{1},\ldots,x_{n})$
defines a subdirect key rectangular relation $\rho$.
For $i=1,2$ the variable 
$x_{i}$ of 
every constraint of $\Upsilon$ containing $x_{i}$
is stable under an irreducible congruence $\sigma_{i}$,
and there exist 
tuples
$(a_{1},a_2,\dots,a_n),(b_{1},b_2,b_3\dots,b_n)\in \rho$, 
$(a_{1}',a_2,\dots,a_n),(b_{1},b_2',b_3\dots,b_n)\notin \rho$ 
such that 
$(a_1,a_1')\in\sigma_{1}^{*}\setminus \sigma_{1}$,
$(b_2,b_2')\in\sigma_{2}^{*}\setminus \sigma_{2}$.
Then there exists a bridge 
$\delta$ from 
$\sigma_{1}$ to $\sigma_{2}$ such that 
$\widetilde \delta$ contains $\Upsilon(x_{1},x_{2})$.
\end{lem}

\begin{proof}
If $D^{(1)}$ is a nonlinear reduction
then by $\omega$ we denote 
the relation defined by $\Upsilon(x_{1},\ldots,x_{n})$.
If $D^{(1)}$ is a linear reduction,
then by $\omega$ we denote 
the relation defined by
$\Omega(x_{1}\dots,x_n,u_{1},\ldots,u_{r})$,
where 
$\Var(\Upsilon) = \{x_{1},\ldots,x_{n},v_{1},\ldots,v_{r}\}$,
$\Omega=\Upsilon \wedge \bigwedge_{i=1}^{r} \delta_{i}(v_{i},u_{i})$ and
$\delta_{i}=\ConLin(D_{v_{i}})$
(see Lemma~\ref{AddLinearVariables}).

We know from Lemmas \ref{SameConOneForNonlinear} and 
\ref{AddLinearVariables}
that
$\ConOne(\omega,j)^{(1)} = \ConOne(\rho,j)$
for every $j\in\{1,2\}$.
Since 
$\rho$ is rectangular, 
we have
$(a_{1},a_{1}')\notin\ConOne(\rho,1)$,
and therefore $(a_{1},a_{1}')\notin\ConOne(\omega,1)^{(1)}$.
Since $x_{1}$ in every constraint containing $x_{1}$ is stable under $\sigma_{1}$, 
the relation $\ConOne(\omega,1)$
is stable under $\sigma_{1}$.
Therefore, $\ConOne(\omega,1)$
should be equal $\sigma_{1}$, 
since otherwise 
$\ConOne(\omega,1)\supseteq \sigma_{1}^{*}$, which 
contradicts 
$(a_{1},a_{1}')\notin\ConOne(\omega,1)^{(1)}$.
In the same way we can show that 
$\ConOne(\omega,2)=\sigma_2$.

Since $\rho$ is a key relation, 
there should be a key tuple $\beta\in (D_{x_{1}}^{(1)}\times\dots\times D_{x_{n}}^{(1)})\setminus \rho$ such that for every
$\alpha\in (D_{x_{1}}^{(1)}\times\dots\times D_{x_{n}}^{(1)})\setminus \rho$ 
there exists
a vector-function $\Psi$ which
preserves $\rho$ and gives $\Psi(\alpha) = \beta$.
First,
we put $\alpha_a = (a_{1}',a_2,\dots,a_n)$ and
apply the corresponding unary vector-function $\Psi_{a}$
to $(a_{1},a_2,\dots,a_n)$  to get a tuple $\beta_{a}$.
Second, we put $\alpha_{b} = (b_{1},b_2',b_{3},\dots,b_n)$ and
apply the corresponding unary vector-function $\Psi_{b}$
to $(b_{1},b_2,b_{3},\dots,b_n)$ to get a tuple $\beta_{b}$.
As a result we get two tuples $\beta_{a}$ and $\beta_{b}$ from $\rho$ 
that differ from the key tuple $\beta$ 
just in the first and second coordinates, respectively.
Since $\ConOne(\omega,j)^{(1)} = \ConOne(\rho,j)$
for every $j\in\{1,2\}$,
we have $\beta\notin\omega$ if $D^{(1)}$ is nonlinear, 
and $\beta\notin\proj_{1,\dots,n}(\omega^{(1)})$ if 
$D^{(1)}$ is linear.

Then by applying  
Lemma~\ref{OneLink} to $\omega$ and 
$\beta_{a},\beta_{b},\beta$ (if $D^{(1)}$ is a linear reduction
we extend these tuples),
we get a bridge $\delta$ from $\sigma_1$ to $\sigma_2$ such that 
$\widetilde\delta$ contains $\proj_{1,2}(\omega)$, which is equal to the relation defined by $\Upsilon(x_{1},x_{2})$.
\end{proof}

\subsection{Expanded coverings of crucial instances}

In this subsection we prove two properties 
of expanded coverings of crucial instances.

\begin{lem}\label{KeepCrucialConstraint}
Suppose $\Theta$ is a crucial instance in $D^{(1)}$, 
$\Theta'\in\Expanded(\Theta)$ via the map $S\colon \Var(\Theta')
\to\Var(\Theta)$, and $\Theta'$ has no solution in $D^{(1)}$.
Then for every constraint $C = \rho(x_1,\dots,x_n)$ 
in $\Theta$ there exists a constraint $C'$ in
$\Theta'$ whose image in $\Theta$ is $C$ 
(i.e., $C' = \rho(y_1,\dots,y_n)$ and $S(y_i) = x_i$ 
for $i = 1,2,\dots,n$).
\end{lem}

\begin{proof}
Let $\Theta''$ be obtained from $\Theta$ by replacing every variable $y$ by $S(y)$.
Obviously, $\Theta''$ still does not have a solution in $D^{(1)}$.
By the definition of expanded coverings every relation in the obtained instance 
is either unary (and full), or weaker or equivalent to a constraint from $\Theta'$.
Since $\Theta$ is crucial in $D^{(1)}$ and $\Theta''$ has no solutions in $D^{(1)}$, 
there should be a constraint $C'$ in $\Theta'$ such that its image 
$C''$ in $\Theta''$ is weaker or equivalent to $C$ 
but not weaker than $C$.
Since $\Theta$ is crucial, all variables of $C$ are not dummy. 
Since $C''$ cannot have more variables than $C$ we obtain that 
$C'' = C$, which means that 
$C' = \rho(y_{1},\dots,y_{n})$ and $S(y_{i}) = x_{i}$ for every $i\in\{1,2,\dots,n\}$.
\end{proof}

\begin{lem}\label{StayNotConnected}
Suppose $\Theta$ is a crucial instance in $D^{(1)}$, $\Theta'\in\Expanded(\Theta)$ has no solutions in $D^{(1)}$, 
every constraint relation of $\Theta$ is a critical rectangular relation, and $\Theta'$ is connected.
Then $\Theta$ is connected.
\end{lem}

\begin{proof}
Let $\Theta''$ be obtained from $\Theta'$ by replacing every variable $y$ by $S(y)$ from the definition of the expanded covering.

Let us show that any two constraints $C_{1}$ and $C_{2}$ with a common variable $x$ of $\Theta$ are adjacent.
By Lemma~\ref{KeepCrucialConstraint}, 
there exist constraints $C_{1}'$ and $C_{2}'$ of $\Theta'$ whose images in $\Theta$ are  $C_{1}$ and $C_{2}$.
Since $\Theta'$ is connected, the instance $\Theta''$ is also connected.
By Corollary~\ref{PathInConnectedComponent}
constraints $C_{1}$ and $C_{2}$ of $\Theta''$  are adjacent in $x$.
Therefore, $C_{1}$ and $C_{2}$ are adjacent in $\Theta$.
Thus, we proved that any two constraints of $\Theta$ with a common variable are adjacent.
Since $\Theta$ is crucial in $D^{(1)}$, it is not fragmented, which implies that $\Theta$ is connected.
\end{proof}

\subsection{Strategies}

\begin{thm}\label{PreviousReductions}
Suppose $D^{(0)},D^{(1)},\dots,D^{(s)}$ is a strategy for $\Omega$,
the solution set of $\Omega^{(i)}$ is subdirect
for every $i\in\{0,1,\ldots,s\}$, $j<s$,
$D^{(s+1)}$ is a one-of-four reduction,
at least one of the two reductions $D^{(j+1)}$, $D^{(s+1)}$ is nonlinear,
and $(\Omega^{(j)}(x_{1},\ldots,x_{n}))^{(s+1)}$ defines a  nonempty relation.
Then $(\Omega^{(j+1)}(x_{1},\ldots,x_{n}))^{(s+1)}$ defines a nonempty relation.
\end{thm}
\begin{proof}
Let $\Var(\Omega) = \{x_{1},\dots,x_{n},y_{1},\dots,y_{t}\}$,
$\Omega^{(j)}(x_{1},\dots,x_{n},y_{1},\dots,y_{t})$ define a relation $R$.
Let the reduction $D^{(j+1)}$ be of type $\mathcal T_{1}$,
the reduction $D^{(s+1)}$ be of type $\mathcal T_{2}$.

Assume that $D^{(s+1)}$ is an absorbing reduction.
Since $\Omega^{(s)}(x_{1},\ldots,x_{n},y_{1},\ldots,y_{t})$ defines a subdirect relation,
Lemma~\ref{AbsLessThanThree} implies that
$\Omega^{(s+1)}(x_{1},\ldots,x_{n},y_{1},\ldots,y_{t})$ defines a nonempty relation.
From now on we assume that $\mathcal T_{2}$ is not the
absorbing type.

For $i\in\{j,\dots,s\}$ and $k\in\{1,\dots,n\}$ put
\begin{align*}
B_{i}= R\cap& (D_{x_{1}}^{(i)}\times\dots\times D_{x_{n}}^{(i)}\times D_{y_{1}}^{(j)}\times\dots\times D_{y_{t}}^{(j)}),\\
B_{i}'= R\cap& (D_{x_{1}}^{(i)}\times\dots\times D_{x_{n}}^{(i)}\times D_{y_{1}}^{(j+1)}\times\dots\times D_{y_{t}}^{(j+1)}),\\
B^{k} = R\cap& (D_{x_{1}}^{(s)}\times\dots\times D_{x_{k-1}}^{(s)} \times D_{x_{k}}^{(s+1)}\times D_{x_{k+1}}^{(s)}\times \dots\times D_{x_{n}}^{(s)}\times D_{y_{1}}^{(j)}\times\dots\times D_{y_{t}}^{(j)}).
\end{align*}

By Lemma~\ref{PCBrel}
$B_{j+1}$ and $B_{j}'$ are one-of-four subuniverses of $R=B_{j}$ of type $\mathcal T_{1}$.
Similarly, 
$B_{i+1}$ is a one-of-four subuniverse of 
$B_{i}$ for every $i$ 
and $B^{k}$ is a one-of-four subuniverse of $B_{s}$  of type $\mathcal T_{2}$ for every $k$ (here we may need to reduce the domain of the last $t$ variables to achieve the subdirectness of $B_{i}$).

Let us show by induction on $i$ 
that $B_{i}'$ is a one-of-four subuniverse of $B_{i}$
of type $\mathcal T_{1}$.
For $i=j$ we already know this.
Assume that 
$B_{i}'$ is a one-of-four subuniverse of $B_{i}$.
By Lemma~\ref{PCBsub},
$B_{i+1}\cap B_{i}' = B_{i+1}'$ is a one-of-four subuniverse 
of $B_{i+1}$ of type $\mathcal T_{1}$.
Therefore, 
$B_{s}'$ is a one-of-four subuniverse of 
$B_{s}$ of type $\mathcal T_{1}$.

We need to prove that 
$B^{1}\cap\dots\cap B^{n}\cap B_{s}'\neq\varnothing$.
Since $(\Omega^{(j)}(x_{1},\ldots,x_{n}))^{(s+1)}$ defines a  nonempty relation,
$B^{1}\cap\dots\cap B^{n}\neq \varnothing$,
Since the solution set of 
$\Omega^{(s)}$ is subdirect, 
$B^{k}\cap B_{s}'\neq\varnothing$
for every $k\in\{1,\dots,n\}$.
Note that $B^{k}$ is of type $\mathcal T_{2}$ and 
$B_{s}'$ is of type $\mathcal T_{1}$, 
and they cannot be both linear.
Since $B^{k}$ is not a binary absorbing subuniverse,
Lemma~\ref{PCBint}
implies that
$B^{1}\cap\dots\cap B^{n}\cap B_{s}'\neq\varnothing$.
\end{proof}

\begin{cons}\label{PathStability}
Suppose $\Theta$ is a cycle-consistent
CSP instance, $D^{(0)},D^{(1)},\dots,D^{(s)}$ is a strategy for $\Theta$,
$\Upsilon\in\Expanded(\Theta)$ is a tree-formula, 
$x$ is a parent of $x_{1}$ and $x_{2}$,
and 
either (i) $B$ is a center of $D_{x}^{(s)}$,
or (ii) $B$ is a PC subuniverse of $D_{x}^{(s)}$
and $D_{y}^{(s)}$ has no nontrivial binary absorbing subuniverse or center for every $y$.
Then 
the pp-formula
$\Upsilon^{(s)}(x_{1},x_{2})$
defines a binary relation with a nonempty intersection with $B\times B$.
\end{cons}
\begin{proof}
Since every reduction in a strategy
is 1-consistent
and $\Upsilon$ is a tree-formula,
the solution set of $\Upsilon^{(i)}$ is
subdirect for every $i$.
If $B = D_{x}^{(s)}$ then the claim follows 
from the definition of a strategy (every reduction is 1-consistent).
Otherwise, let us define 
a reduction $D^{(s+1)}$ by 
$D_{x}^{(s+1)} = D_{x_1}^{(s+1)} = D_{x_2}^{(s+1)} = B$, 
$D_{y}^{(s+1)}=D_{y}^{(s)}$ for the remaining variables.
Thus, we have a nonlinear reduction $D^{(s+1)}$.
Since the instance $\Theta$ is cycle-consistent,
and $x$ is a parent of $x_{1}$ and $x_{2}$, 
$\Upsilon(x_{1},x_{2})$ defines a reflexive relation.
Hence,
$(\Upsilon(x_{1},x_{2}))^{(s+1)}$ defines a nonempty relation.
By Theorem~\ref{PreviousReductions},
we obtain that 
$(\Upsilon^{(1)}(x_{1},x_{2}))^{(s+1)}$ defines a nonempty relation.
Repeatedly applying Theorem~\ref{PreviousReductions},
we show that 
$(\Upsilon^{(2)}(x_{1},x_{2}))^{(s+1)},
(\Upsilon^{(3)}(x_{1},x_{2}))^{(s+1)},\dots
(\Upsilon^{(s)}(x_{1},x_{2}))^{(s+1)}$ 
define nonempty relations, which means that 
$\Upsilon^{(s)}(x_{1},x_{2})$
has a nonempty intersection with $B\times B$.
\end{proof}

\begin{lem}\label{FitInLinearSubuniverses}
Suppose $R\subseteq A_{0}\times B_{0}$ is a subdirect relation, and
\begin{enumerate}
\item $A_{0}\supseteq A_{1}\supseteq\dots\supseteq A_{s+1}$ and $A_{i+1}$ is a one-of-four subuniverse of $A_{i}$
for $i\in\{0,1,2,\dots,s\}$;
\item $B_{0}\supseteq B_{1}\supseteq\dots\supseteq B_{t+1}$ and $B_{i+1}$ is a one-of-four subuniverse of $B_{i}$
for $i\in\{0,1,2,\dots,t\}$;
\item $A_{s+1}$ and $B_{t+1}$
are linear subuniverses of $A_{s}$ and $B_{t}$, respectively;
\item there exist
$a\in A_{s+1}$, 
$b\in B_{s+1}$, 
$a'\in A_{0}$, $b'\in B_{0}$ such that 
$(a,b'),(a',b),(a',b')\in R$;
\item $R\cap (A_{s}\times B_{t})\neq \varnothing$.
\end{enumerate}

Then $R\cap (A_{s+1}\times B_{t+1})\neq \varnothing$.
\end{lem}
\begin{proof}
Denote $a'' = w(a,a',\dots,a')$, 
$b'' = w(b,b',\dots,b') = w(b',\dots,b',b)$.

We prove by induction on $s+t$.
Assume that $s+t=0$, which implies $s=t=0$.
By Lemma~\ref{LinearSpecialWNU}, 
we have 
$a''\in A_{1}$ and 
$b''\in B_{1}$.
Since $w$ preserves $R$, $(a'',b'')\in R$, which completes this case.

Let us prove the induction step.
Assume that $s+t>0$. 
Without loss of generality we assume that
$s>0$.
Put 
$$R'(x_{1},x_2,y_{1},y_{2}) = 
R(x_{1},y_{2})\wedge
R(y_{1},x_{2})\wedge 
R(y_{1},y_{2}).$$
Put 
$P_{i} = 
R'\cap (A_{i}\times B_{0}\times A_{0}\times B_{0})$, 
$Q_{i} = 
R'\cap (A_{0}\times B_{i}\times A_{0}\times B_{0})$,
$T = 
R'\cap (A_{0}\times B_{0}\times A_{1}\times B_{0})$.
Since $R'$ is subdirect, 
by Lemma~\ref{PCBrel}
$P_{i+1}$ is a one-of-four subuniverse of $P_{i}$
,
$Q_{i+1}$ is a one-of-four subuniverse of $Q_{i}$
for every $i$,
and $T$ is a one-of-four subuniverse of 
$R' = P_{0}=Q_{0}$.
We want to prove that 
$P_{s+1}\cap Q_{t+1}\cap T\neq \varnothing$.
Since 
$(a,b,a',b')\in P_{s+1}\cap Q_{t+1}$, 
Lemma~\ref{SequencesOfSubuniverses},
implies that 
$P_{s+1}\cap Q_{t}$
and 
$P_{s}\cap Q_{t+1}$ 
are one-of-four subuniverses of 
$P_{s}\cap Q_{t}$.
Since 
$R\cap (A_{s}\times B_{t})\neq \varnothing$, 
we have 
$T\cap P_{s}\cap Q_{t}\neq\varnothing$.  
Lemma~\ref{SequencesOfSubuniverses} implies
that 
$P_{0}\supseteq P_{1}\supseteq\dots\supseteq
P_{s}\supseteq P_{s}\cap Q_{1}\supseteq 
\dots\supseteq P_{s}\cap Q_{t}$ (here all inclusions mean one-of-four subuniverses),
which by the same lemma implies that 
$T\cap P_{s}\cap Q_{t}$ is a 
one-of-four subuniverse of 
$P_{s}\cap Q_{t}$.

If $A_{1}$ is not a linear subuniverse of 
$A_{0}$, then by Theorem~\ref{PCBint}
the intersection of 
one-of-four subuniverses 
$P_{s+1}\cap Q_{t}$,
$P_{s}\cap Q_{t+1}$, 
and $T\cap P_{s}\cap Q_{t}$
of different types cannot be empty, 
that is,
$P_{s+1}\cap Q_{t+1}\cap T\neq \varnothing$, 
which completes this case.

Assume that $A_{1}$ is linear.
Since
$(a,b,a',b'),
(a',b',a',b'),
(a',b',a,b')\in R'$
and $w$ preserves $R'$, 
we obtain
$(a,b,a',b'),
(a'',b'',a',b'),
(a'',b'',a'',b')\in R'$.
Note that by Lemma~\ref{LinearSpecialWNU}
$a''\in A_{1}$ and 
$b''\in B_{1}$.
We look at $R'$ as a binary relation 
$R'\subseteq (A_{0}\times B_{0})\times 
(A_{0}\times B_{0})$ 
to apply the inductive assumption.
Put
$\mathcal A_{i} = 
\proj_{1,2}(R')\cap (A_{i+1}\times B_{1})$
for $0\le i\le s$
and 
$\mathcal A_{i} = 
\proj_{1,2}(R')\cap (A_{s+1}\times B_{i-s+1})$
for $s+1\le i\le s+t$.
Combining Lemma~\ref{PCBrel} 
and Lemma~\ref{SequencesOfSubuniverses}
we derive that 
$\mathcal A_{i+1}$ is a one-of-four subuniverse of $\mathcal A_{i}$ for $i=0,1,\dots,s+t-1$.
Put 
$\mathcal B_{i} = \proj_{3,4}(R'\cap (A_{1}\times B_{1}\times A_{i}\times B_{0}))$
for $i=0,1$.
Combining Lemmas~\ref{PCBrel},
\ref{SequencesOfSubuniverses},
and \ref{ReductionAndProjectionGivesOneOfFour},
we derive that $\mathcal B_{1}$ is a linear subuniverse of 
$\mathcal B_{0}$.
Then we apply the inductive assumption to 
$R'\cap (\mathcal A_{0}\times \mathcal B_{0})$,
$\mathcal A_{0}\supseteq \mathcal A_{1}\supseteq\dots\supseteq \mathcal A_{s+t}$
and $\mathcal B_{0}\supseteq \mathcal B_{1}$,
and show that
$R'\cap (A_{s+1}\times B_{t+1}\times A_{1}\times B_{0})=
P_{s+1}\cap Q_{t+1}\cap T\neq \varnothing$.

Put 
$B_{i}' = \proj_{2}(R\cap (A_{1}\times B_{i}))$
for $i=0,1,\dots,s+1$.
By Lemma~\ref{ReductionAndProjectionGivesOneOfFour}, 
$B_{i+1}'$ is a one-of-four subuniverse of $B_{i}'$
for every $i$.
Applying the inductive assumption 
to $R\cap (A_{1}\times B_{0})$,
$A_{1}\supseteq A_{2}\supseteq\dots \supseteq A_{s+1}$, 
and 
$B_{0}'\supseteq B_{1}'\supseteq\dots \supseteq B_{t+1}'$, 
we obtain that $A_{s+1}\cap B_{t+1} \neq\varnothing$.
 \end{proof}

\begin{lem}\label{ParProperty}
Suppose $D^{(0)}, D^{(1)},\ldots,D^{(s)}$ is a strategy for
a subdirect constraint $\rho(x_{1},\ldots,x_{n})$,
$D^{(s+1)}$ is a linear reduction, and
\begin{align*}
(b_1,\ldots,b_{t},a_{t+1},\ldots,a_{n})&\in\rho,\\
(a_1,\ldots,a_{t},b_{t+1},\ldots,b_{n})&\in\rho,\\
(b_1,\ldots,b_{t},b_{t+1},\ldots,b_{n})&\in \rho,\\
(a_1,\ldots,a_{t},a_{t+1},\ldots,a_{n})&\in D^{(s+1)}.
\end{align*}
Then there exists
$(d_1,d_{2},\ldots,d_{n})\in \rho^{(s+1)}$.
\end{lem}
\begin{proof}
For $i=0,1,2,\dots,s+1$ put 
$$A_{i} = 
\proj_{1,2,\dots,t} 
(\rho
\cap (D_{x_{1}}^{(i)}\times\dots\times D_{x_{t}}^{(i)}
\times D_{x_{t+1}}\times\dots\times D_{x_{n}}),$$
$$B_{i} = 
\proj_{1,2,\dots,t} 
(\rho
\cap (D_{x_{1}}\times\dots\times D_{x_{t}}
\times D_{x_{t+1}}^{(i)}\times\dots\times D_{x_{n}}^{(i)}).$$
Since $\rho^{(i)}$ is subdirect for every $i\in\{0,1,\dots,s\}$,
Lemma~\ref{PCBrel} implies that
$A_{i+1}$ is a one-of-four subuniverse of 
$A_{i}$ 
and 
$B_{i+1}$ is a one-of-four subuniverse of 
$B_{i}$ for every $i\in\{0,1,\dots,s\}$.
Since 
$(a_{1},\ldots,a_{t})\in A_{s+1}$
and 
$(a_{t+1},\ldots,a_{n})\in B_{s+1}$,
Lemma~\ref{FitInLinearSubuniverses}
implies that
$A_{s+1}\cap B_{s+1}\neq \varnothing$, which completes the proof.
\end{proof}

%% file: Finiteness2.tex
\subsection{Growing population divides into colonies.}
In this section we prove a theorem that clarifies the inductive strategy used in the proof of Theorem \ref{FindPerfectConstraint}.
To simplify explanation we decided to avoid our usual terminology.
Instead, we argue in terms of organisms, reproduction, and friendship.

We consider a set $X$ whose elements we call \emph{organisms}.
At the moment 1 we had a set of organisms $X_{1}$.
At every moment some organisms give a birth to new organisms,
as a result we get a sequence of organisms $X_{1}\subseteq X_{2}\subseteq X_{3}\subseteq \dots,$
where 
$\bigcup_{i} X_{i} = X$, $X_{i}\subseteq X$,
 and $|X_{i}|<\infty$ for every $i$.
We assume that 
each organism from $X\setminus X_{1}$ 
has exactly one parent.

Every organism has a characteristic that we call \emph{strength}. Thus we have
a mapping $\xi:X\to \{1,2,\ldots, S\}$
that assigns a characteristic to every organism.
Also we have a binary reflexive symmetric relation $F$ on the set $X$, which we call \emph{friendship}.
For an organism $x$ by $\BD(x)$ we denote the minimal $i$ such that $x\in X_{i}$.
A sequence of organisms $x_{1},\ldots,x_{n}$ such that $x_{i}$ is a friend of $x_{i+1}$ for every $i$ is called \emph{a path}.

\begin{thm}\label{newcolonies}
Suppose $X_{1},X_{2},X_{3},\ldots$, $\xi$, and $F$ satisfy the following conditions:
\begin{enumerate}
\item
\textbf{A child is always weaker than its parent.}
If $y$ is the parent of $x$, then $\xi(y)>\xi(x)$.
\item
\textbf{Older friends are parents' friends.}
If $\BD(y)<\BD(x)$ and $x$ is a friend of $y$,
then the parent of $x$ is a friend of $y$ 
(or the parent of $x$ is $y$).
\item
\textbf{Only friends' kids can be friends.}
If $\BD(x)  = \BD(y)$ and $x$ is a friend of $y$,
then the parents of $x$ and $y$ are friends.
\item
\textbf{No one can have infinitely many friends.} $|\{y\in X\mid (x,y)\in F\}|<\infty$ for every $x\in X$.
\item
\textbf{Reproduction never stops.}
$|\bigcup_{i}X_{i}|=\infty$.
\end{enumerate}
Then there exists $N$ such that
$X_{N}$ can be divided into
two nonempty disjoint sets $X_{N}'$ and $X_{N}''$ such that 
there is no friendship between
$X_{N}'$ and $X_{N}''$.
\end{thm}

\begin{proof}
Assume the contrary. Then 
there exists a path between any two organisms.

For every moment 
$t$ 
and every organism $x$ 
by $x^{t}$ we denote 
the predecessor of $x$
from $X_{t}$
with the maximal $\BD$, that is 
a closest predecessor who already lived at the moment $t$.
For example 
$x^{t} = x$ for $t\ge \BD(x)$, 
and 
$x^{\BD(x)-1}$ is the parent of $x$.

Suppose we have a path of organisms 
$x_{1},\ldots,x_{n}$.
We claim that  
$x_{1}^{t},\ldots,x_{n}^{t}$ is also a path
for any $t$.
We will prove by induction starting with 
sufficiently large $t$ such that $x_{1},\ldots,x_{n}\in X_{t}$
and therefore
$(x_{1}^{t},\ldots,x_{n}^{t}) = (x_{1},\ldots,x_{n})$.
As the inductive step, 
we assume that this is a path for $t=t_{0}$
and show that this is a path for $t = t_{0}-1$.
The induction step follows from
hypotheses (2) and (3).
The path 
$x_{1}^{t},\ldots,x_{n}^{t}$ will be called 
\emph{a path at the moment $t$}.
Note that organisms of the path at the moment $t$ are 
not weaker than the corresponding organisms of the original path.

Choose a maximal strength $s$ such that
we have infinitely many organisms of strength $s$.
Then infinitely many of them have the same parent, hence, there exists a parent reproducing infinitely many times.

For every $x$ and a strength $s$
by $\Kids(x,s)$ we denote the set of all children $y$ of $x$ such that
there exists a path from $x$ to $y$ with all the organisms in the path stronger than $s$.
We consider the maximal $s_{0}$ such that
$\Kids(x,s_{0})$ is infinite for some organism $x$.
Since we can always put $s_{0}=0$ for 
a parent reproducing infinitely many times, 
$s_0$ exists.
Note that this implies that $x$ is stronger than $s_{0}+1$.


By $Y$ we denote the set of all organisms $y$ such that
there exists a path from $x$ to $y$
with all the organisms in the path stronger than $s_{0}+1$.
Note that $Y$ includes $x$.
Let us show that $Y$ is finite.
Assume the opposite.
Let $s$ be the maximal strength such that we have infinitely many organisms of this strength in $Y$.
Consider an organism $v$ from $Y$ 
with strength $s$ such that 
$\BD(v)>\BD(x)$ (we still have infinitely many of them).
Considering the path from $x$ to $v$ at the moment 
$\BD(v)-1$, we get a path from 
$x$ to the parent of $v$, which means 
that the parent of $v$ is also in $Y$.
Since parents are stronger than children
and we may have only finitely many organisms stronger than $s$ in $Y$,
we have only finitely many such parents in $Y$.
Therefore, there exists a parent $z\in Y$ with infinitely many children from $Y$.
Since we can glue a 
path from the parent to $x$ and 
a path from $x$ to its kid,
this implies that $\Kids(z,s_{0}+1)$ is infinite.
This contradicts the maximality of $s_{0}$ and proves that $Y$ is finite.

Let $t$ be the first moment such that 
$X_t$ contains all friends of friends of organisms from $Y$.
Consider an organism $y$ from
$\Kids(x,s_{0})$ with $\BD(y)>t$. 
Choose a path from $x$ to $y$ with all organisms stronger than $s_{0}$.
We consider the last organism $u$ in the path such that
$\BD(u)<\BD(y)$.
Taking the fragment of this path from $y$ to $u$ 
at the moment $\BD(y)-1$ 
we obtain a path from
$x$ to $u$ with all organisms but $u$ stronger than $s_{0}+1$.
This means that all the organisms but $u$ in this path are
from $Y$. 
Thus $u$ has a friend from $Y$, 
which means that all friends of $u$ were born before the moment 
$t$.
This contradicts the fact that
an organism next to $u$ in the original path from $x$ to $y$ 
was born after the moment $\BD(y)-1$.
\end{proof} 

%% file: Main.tex
\subsection{Existence of a next reduction}

The next Lemma has its roots in Theorem 20 from \cite{FederVardi},
where the authors proved that bounded width 1
is equivalent to tree duality. 

\begin{lem}\label{ConstraintPropagation}
Suppose $D^{(0)},D^{(1)},\dots,D^{(s)}$ is a strategy for a 1-consistent CSP instance  $\Theta$,
and $D^{(\top)}$ is a reduction of $\Theta^{(s)}$.
\begin{enumerate}
\item If there exists a 1-consistent reduction contained in
$D^{(\top)}$ and $D^{(s+1)}$ is maximal among such reductions,
then for every variable $y$ of $\Theta$ there exists a tree-formula $\Upsilon_{y}\in \ExpShort(\Theta)$
such that
$\Upsilon_{y}^{(\top)}(y)$ defines $D_{y}^{(s+1)}$.
\item Otherwise,
there exists a tree-formula $\Upsilon\in \ExpShort(\Theta)$
such that $\Upsilon^{(\top)}$ has no solutions.
\end{enumerate}
\end{lem}

\begin{proof}
The proof is based on the constraint propagation procedure.
We consider the instance $\Theta^{(s)}$.
We start with an empty set $\Upsilon_{y}$ (empty tree-formula) for every $y$, 
these tree-formulas define the reduction $D^{(\top)}$.

Then we introduce a recursive algorithm that gives a correct tree-formula $\Upsilon_{y}$ for every variable $y$.
If at some step the reduction defined by these tree-formulas is 1-consistent,
then we are done. Otherwise, we consider a constraint $C$ that breaks 1-consistency.
Then the current restrictions of the variables $z_{1},\ldots,z_{l}$ in the constraint $C= \rho(z_{1}\ldots,z_{l})$
imply a stronger restriction of some variable $z_{i}$ and the corresponding domain $D_{z_{i}}^{(s)}$.
We change the tree-formula $\Upsilon_{z_{i}}$ describing
the reduction of the variable $z_{i}$ in the following way
$\Upsilon_{z_{i}}:= C\wedge \Upsilon_{z_{1}}\wedge\dots\wedge\Upsilon_{z_{l}}$.

Note that we have to be careful with all the variables appearing in different $\Upsilon_{y}$ to avoid collisions.
Every time we join $\Upsilon_{u}$ and $\Upsilon_{v}$ we rename the variables so that
they do not have common variables.

Obviously, this procedure will eventually stop. 
If $\Upsilon_{y}^{(\top)}(y)$ defines an empty set for some 
$y$, then $\Upsilon_{y}$ can be taken as $\Upsilon$ to 
witness condition (2).
Otherwise, these tree-formulas define a 1-consistent 
reduction, which is a maximal 1-consistent reduction since it is defined by tree-formulas.
\end{proof}

\begin{thm}\label{NextReductionOne}
Suppose $D^{(0)},D^{(1)},\dots,D^{(s)}$ is a strategy for a cycle-consistent CSP instance  $\Theta$.
\begin{itemize}
\item If $D_{x}^{(s)}$ has a nontrivial binary absorbing subuniverse $B$ then there exists a 1-consistent absorbing reduction $D^{(s+1)}$ of $\Theta^{(s)}$
with $D_{x}^{(s+1)}\subseteq B$.
\item If $D_{x}^{(s)}$ has a nontrivial center $B$ then there exists a 1-consistent  central reduction $D^{(s+1)}$ of $\Theta^{(s)}$
with $D_{x}^{(s+1)}\subseteq B$.
\item If $D_{y}^{(s)}$ has no nontrivial binary absorbing subuniverse or center for every $y$
but there exists a nontrivial PC subuniverse $B$ in $D_{x}^{(s)}$ for some $x$, then there exists a 1-consistent PC reduction $D^{(s+1)}$
of $\Theta^{(s)}$ with $D_{x}^{(s+1)}\subseteq B$.
\end{itemize}
\end{thm}

\begin{proof}
Without loss of generality we assume that
$B$ is a minimal one-of-four subuniverse of this type.
Let us define a reduction 
$D^{(\top)}$ by 
$D_{x}^{(\top)} = B$ and 
$D_{y}^{(\top)}=D_{x}^{(s)}$ for $y\neq x$, 
and apply Lemma~\ref{ConstraintPropagation}.
We consider two cases corresponding to 
two cases of Lemma~\ref{ConstraintPropagation}.

Case 1. There exists a 1-consistent reduction $D^{(s+1)}$ of $\Theta^{(s)}$
such that 
$D_{y}^{(s+1)}$ is defined by $\Upsilon_{y}(y)$ for a tree-formula $\Upsilon_{y}$ for every variable $y$.
Let $R$ be the solution set of $\Upsilon_{y}^{(s)}$.
Since $\Upsilon_{y}$ is a tree-formula 
and $\Theta^{(s)}$ is 1-consistent, 
the solution set $R$ is subdirect.
Applying Corollaries~\ref{AbsImpliesCons}, \ref{CenterImpliesCons}, \ref{PCImplies} 
to $R$ 
we derive that $D_{y}^{(s+1)}$ is a one-of-four 
subuniverse of the corresponding type.

Case 2. There exists a tree-formula 
$\Upsilon\in\ExpShort(\Theta)$
such that $\Upsilon^{(\top)}$ has no solutions.
We consider the minimal set of variables $\{x_{1},\ldots,x_{k}\}$ from $\Upsilon$ whose parent is $x$
such that $\Upsilon^{(s)}(x_{1},\ldots,x_{k})$ does not have any tuple in $B^{k}$.
Since $\Theta^{(s)}$ is 1-consistent and $\Upsilon$ is a tree-formula, $k\ge 2$.
If $B$ is a binary absorbing subuniverse, then we get a contradiction with Lemma~\ref{AbsLessThanThree}.
For other cases with $k=2$ we get a contradiction from Corollary~\ref{PathStability}.
If $k\ge 3$ and $B$ is a center then we get a contradiction with Corollary~\ref{CenterLessThanThree}.
If $k\ge 3$ and $B$ is a PC subuniverse then we get a contradiction with Corollary~\ref{PCLessThanThree}.
\end{proof}

As a corollary we can derive that cycle-consistency is a sufficient condition to guarantee the existence of a solution of an instance whose domains avoid linear algebras (so called bounded width case). Note that this corollary follows from the result by Marcin Kozik in \cite{kozik2016weak}.

\begin{conslem}\label{BoundedWidthCase}
Suppose $\Theta$ is a cycle-consistent CSP instance, 
for every 
domain $D_{x}$ there is no 
$B\subseteq D_{x}$ and a congruence $\sigma$ on $B$ 
such that $B/\sigma$ is a nontrivial linear algebra.  
Then $\Theta$ has a solution.
\end{conslem}
\begin{proof}
We recursively build a strategy $D^{(0)},D^{(1)},\dots,D^{(s)}$.
We start with $s=0$. 
If every domain $D_{x}^{(s)}$ is of size 1, then we already have a solution because 
$\Theta^{(s)}$ is 1-consistent.
Otherwise, by Theorem~\ref{NextReduction} on every domain 
$D_{x}^{(s)}$ of size greater than 1 there exists a nontrivial one-of-four subuniverse. 
Note that this subuniverse cannot be linear because 
this contradicts the assumption that 
there is no 
$B\subseteq D_{x}$ and a congruence $\sigma$ on $B$ such that $B/\sigma$ is linear. 
If we found a binary absorbing subuniverse or a center, 
then by Theorem~\ref{NextReductionOne}
we can always find 
a next 1-consistent absorbing or central reduction $D^{(s+1)}$.
Otherwise, by the same theorem we can find 
a 1-consistent PC reduction.
Since the strategy cannot be infinite, we eventually stop with the instance whose variable domains are of size 1.
\end{proof}

\begin{thm}\label{NextReductionTwo}
Suppose $D^{(0)},D^{(1)},\dots,D^{(s)}$ is a strategy for a cycle-consistent CSP instance  $\Theta$,
and $D^{(\top)}$ is a nonlinear 1-consistent reduction
of $\Theta^{(s)}$.
Then there exists a 1-consistent minimal reduction $D^{(s+1)}$
of $\Theta^{(s)}$ of the same type such that
$D_{x}^{(s+1)}\subseteq D_{x}^{(\top)}$ for every variable $x$.
\end{thm}

\begin{proof}
Let the reduction $D^{(\top)}$ be of type $\mathcal T$.
Let us consider a minimal by inclusion 1-consistent reduction $D^{(s+1)}$ of $\Theta^{(s)}$ of type $\mathcal T$ such that
$D_{x}^{(s+1)}\subseteq D_{x}^{(\top)}$ for every variable $x$.

Assume that for some $z$ the domain $D_{z}^{(s+1)}$ is not a minimal one-of-four subuniverse of type $\mathcal T$.
Then choose a minimal one-of-four subuniverse $B$ of $D_{z}^{(s)}$ of this type contained in $D_{z}^{(s+1)}$.
We define a reduction $D^{(\bot)}$ of $\Theta^{(s)}$
by $D^{(\bot)}_{z} = B$, $D_{y}^{(\bot)} =  D_{y}^{(s+1)}$ if $y\neq z$, 
and apply Lemma~\ref{ConstraintPropagation}.
Since $D_{y}^{(s+1)}$ is a minimal by inclusion reduction,
there exists a tree-formula $\Upsilon\in\ExpShort(\Theta)$ such that
$\Upsilon^{(\bot)}$ has no solutions.
Again, we consider a minimal
set of variables $\{z_{1},\ldots,z_{k}\}$ from $\Upsilon$ whose parent is $z$
such that $\Upsilon^{(s+1)}(z_{1},\ldots,z_{k})$ does not have any tuple in $B^{k}$.
Since the reduction $D^{(s+1)}$ is 1-consistent, $B\subsetneq D_{z}^{(s+1)}$, and $\Upsilon$ is a tree-formula, we have $k\ge 2$.
If $D^{(\top)}$ is an absorbing or central reduction of $\Theta^{(s)}$, then it is also an absorbing
or central reduction of $\Theta^{(s+1)}$. Then we can get a contradiction just as we did
in the proof of Theorem~\ref{NextReductionOne} using Lemma~\ref{AbsLessThanThree}, Corollary~\ref{PathStability} or Corollary~\ref{CenterLessThanThree}.

It remains to consider the case when $B$ is a PC subuniverse.
Choose a minimal set of variables $y_{1},\ldots,y_{t}$ of  $\Upsilon$ different from $z_{1},\dots,z_{k}$ such that
$(\Upsilon^{(s)}(z_{1},\ldots,z_{k},y_{1},\ldots,y_{t}))^{(s+1)}$
does not have tuples with the first $k$ elements from $B$.
If $t=0$ and $k=2$ then 
$\Upsilon^{(s)}(z_{1},z_{2})$ has an empty intersection with $B\times B$, which contradicts Corollary~\ref{PathStability}.
If $t+k\ge 3$ then 
the relation defined by 
$\Upsilon^{(s)}(z_{1},\ldots,z_{k},y_{1},\ldots,y_{t})$
is $(B,\dots,B,D_{y_{1}}^{(s+1)},\dots,D_{y_{t}}^{(s+1)})$-essential relation, which contradicts Corollary~\ref{PCLessThanThree}.
\end{proof}

\begin{thm}\label{NextReductionThree}
Suppose $D^{(\top)}$ is a 1-consistent PC reduction for a cycle-consistent irreducible CSP instance  $\Theta$,
and $\Theta$ is not linked and not fragmented.
Then there exist 
a reduction $D^{(1)}$ of $\Theta$ 
and a minimal strategy
$D^{(1)},\ldots, D^{(s)}$ for $\Theta^{(1)}$
such that
the solution set of $\Theta^{(1)}$ is subdirect,
the reductions $D^{(2)}, \ldots, D^{(s)}$ are nonlinear,
$D_{x}^{(s)}\subseteq D_{x}^{(\top)}$ for every variable $x$.
\end{thm}

\begin{proof}
Since $\Theta$ is not linked,
there exists a maximal congruence $\sigma_{x}$ on $D_{x}$ for a variable $x$ of $\Theta$ such that
$\LinkedCon(\Theta,x)\subseteq \sigma_{x}$.
Choose an equivalence class $D_{x}^{(1)}$ of $\sigma_{x}$ with a nonempty intersection with $D_{x}^{(\top)}$.
For every variable $y$ by $D_{y}^{(1)}$ we denote the set of all elements
of $D_{y}$ linked to an element of $D_{x}^{(1)}$.
Note that for every $y$ there is a congruence $\sigma_{y}$ on $D_y$ such that 
$D_{x}/\sigma_x\cong D_{y}/\sigma_{y}$.
Then $D_{y}^{(1)}$ is an equivalence class of $\sigma_{y}$.
By Corollaries~\ref{AbsorptionQuotient} and \ref{CenterQuotient}, 
there is no nontrivial binary absorbing subuniverse or center on 
$D_{x}/\sigma_{x}$.
Then by Theorem~\ref{NextReduction}, 
$\sigma_{x}$ is either PC congruence, or linear congruence, 
which means that $D^{(1)}$ is a PC reduction or linear reduction.

Let us show that
$D_{y}^{(1)}\cap D_{y}^{(\top)}\neq\varnothing$
for every $y$. 
Since $\Theta$ is not fragmented, 
we may consider a path starting at $x$ and ending at $y$.
Since the reduction $D^{(\top)}$ is 1-consistent, 
this path connects an element of 
$D_{x}^{(1)}\cap D_{x}^{(\top)}$ 
with some element of 
$D_{y}^{(\top)}$, which is also in 
$D_{y}^{(1)}$.

Since $\Theta$ is irreducible, the solution set of $\Theta^{(1)}$ is subdirect.
We build the remaining part of the strategy in the following way.
Suppose we already have
$D^{(0)}, D^{(1)},\ldots, D^{(t)}$, where
the reductions $D^{(2)},\ldots,D^{(t)}$ are absorbing or central.
If there exists a nontrivial binary absorbing subuniverse or a nontrivial center on $D_{y}^{(t)}$ for some $y$,
then by Theorems~\ref{NextReductionOne}, \ref{NextReductionTwo} we can find the next
minimal 1-consistent absorbing or central reduction $D^{(t+1)}$.

Suppose there is no binary absorbing subuniverse or center on $D_{y}^{(t)}$ for every $y$.
Put $D_{y}^{(\bot)}= D_{y}^{(\top)}\cap D_{y}^{(t)}$ for every variable $y$.
By Lemma~\ref{SequencesOfSubuniverses}
$D_{y}^{(\bot)}$ is a PC subuniverse of 
$D_{y}^{(t)}$ for every variable $y$.
Hence, $D^{(\bot)}$ is a PC reduction of $\Theta^{(t)}$.

Then we apply Lemma~\ref{ConstraintPropagation} to find a 1-consistent reduction of $\Theta^{(t)}$ smaller than
$D^{(\bot)}$.
If we cannot find it, then there exists a tree-formula $\Upsilon$
such that $\Upsilon^{(\bot)}$ has no solutions.
Let $R$ be the solution set of $\Upsilon$. 
Note that $R^{(i)}$ is a subdirect relation for 
$i=0,1,\dots,t$ because $\Upsilon$ is a tree-formula and 
$D^{(i)}$ is a 1-consistent reduction.
By Lemma~\ref{PCBrel}, 
$R^{(\top)}$ is a PC subuniverse of $R$.
Since $D^{(\top)}$ is 1-consistent, 
the intersection
$R^{(1)}\cap R^{(\top)}$ is not empty.

Let us prove by induction on $i$ that
$R^{(i)}\cap R^{(\top)}$ is a nonempty PC subuniverse of $R^{(i)}$
for $i=1,2,\dots,t$.
By the inductive assumption, 
we assume that
$R^{(i-1)}\cap R^{(\top)}$ is a nonempty PC subuniverse of $R^{(i-1)}$
(for $i=1$ it follows from the definition).
By Lemma~\ref{PCBrel}, $R^{(i)}$ is a one-of-four subuniverse 
of $R^{(i-1)}$.
For $i\ge 2$ it is not a PC subuniverse, 
then by Theorem~\ref{PCBint}, 
the intersection of 
$R^{(i-1)}\cap R^{(\top)}$ and $R^{(i)}$, that is
$R^{(i)}\cap R^{(\top)}$, 
cannot be empty. For $i=1$ we already know that 
$R^{(1)}\cap R^{(\top)}\neq\varnothing$.
Applying Theorem~\ref{PCBsub}
to 
$R^{(i-1)}\cap R^{(\top)}\subseteq R^{(i-1)}$
and 
$R^{(i)}\subseteq R^{(i-1)}$
we derive that 
$R^{(i)}\cap R^{(\top)}$ is a nonempty PC subuniverse of $R^{(i)}$.
Thus, we proved that 
$R^{(t)}\cap R^{(\top)}$ is not empty, which contradicts 
the assumption about the tree-formula $\Upsilon$.

Hence, there exists a 1-consistent reduction $D^{(\triangle)}$ of 
$\Theta^{(t)}$ smaller than $D^{(\bot)}$
such that for every variable $y$
the new domain $D_{y}^{(\triangle)}$ can be defined by a tree-formula $\Upsilon_{y}^{(\bot)}$.
Since the solution set of $\Upsilon_{y}^{(t)}$ is subdirect, 
by Corollary~\ref{PCImplies},
the domain $D_{y}^{(\triangle)}$ is a PC subuniverse
of $D_{y}^{(t)}$. Hence $D^{(\triangle)}$ is a PC reduction of 
$\Theta^{(t)}$.
It remains to apply Theorem~\ref{NextReductionTwo} to find
a minimal PC reduction $D^{(t+1)}$ smaller than $D_{y}^{(\triangle)}$, put $s= t+1$, and finish the strategy.
\end{proof}

\subsection{Existence of a linked connected component}

In this subsection we prove that all constraints in a crucial instance have the parallelogram property,
show that we can always find a linked connected component with required properties,
and prove that we 
cannot pass from an
instance having solutions to an instance having no solutions
while applying a nonlinear reduction.

\begin{thm}\label{KeyConjunctionMain}
Suppose $D^{(1)}$ is a minimal 
1-consistent one-of-four reduction of a cycle-consistent irreducible CSP instance $\Theta$,
$\Omega(x_{1},\ldots,x_{n})$ is a subconstraint of $\Theta$,
the solution set of $\Omega^{(1)}$ is subdirect,
$\Theta\setminus\Omega$ has a solution in $D^{(1)}$,
and $\Theta$ has no solutions in $D^{(1)}$.
Then there exist instances
$\Upsilon_{1},\ldots,\Upsilon_{t}\in \ExpShort(\Omega)$
such that
$\Phi=(\Theta\setminus\Omega)\cup \Upsilon_{1}\cup\dots\cup\Upsilon_{t}$ has no solutions in $D^{(1)}$,
each $\Upsilon_{i}(x_{1},\ldots,x_{n})$ is a subconstraint 
of $\Phi$, 
and
$\Upsilon_{i}^{(1)}(x_{1},\ldots,x_{n})$ defines a subdirect key relation with the parallelogram property for every $i$.
\end{thm}

\begin{thm}\label{FindPerfectConstraint}
Suppose $D^{(1)}$ is a minimal 1-consistent one-of-four reduction of a cycle-consistent irreducible CSP instance $\Theta$,
$\Theta$ is crucial in $D^{(1)}$ and not connected.
Then there exists an instance $\Theta'\in\Expanded(\Theta)$ that is crucial in $D^{(1)}$
and contains a linked connected component whose solution set is not subdirect.
\end{thm}

\begin{thm}\label{CannotLooseSolution}
Suppose $D^{(1)}$  is a 1-consistent nonlinear reduction of a cycle-consistent irreducible CSP instance $\Theta$.
If $\Theta$ has a solution then it has a solution in $D^{(1)}$.
\end{thm}

\begin{thm}\label{ParPropertyMain}
Suppose $D^{(0)},\ldots,D^{(s)}$ is a minimal strategy for a cycle-consistent irreducible CSP instance
$\Theta$,
and a constraint $\rho(x_{1},\ldots,x_{n})$ 
of $\Theta$ is crucial in $D^{(s)}$.
Then $\rho$ is a critical relation with the parallelogram property.
\end{thm}

\begin{thm}\label{ParPropertyForSubcontraint}
Suppose $D^{(0)},\ldots,D^{(s)}$ is a minimal strategy for a cycle-consistent irreducible CSP instance
$\Theta$,
$\Upsilon(x_{1},\ldots,x_{n})$ is a subconstraint of $\Theta$,
the solution set of $\Upsilon^{(s)}$ is subdirect,
$k\in\{1,2,\dots,n-1\}$,
$\Var(\Upsilon) = \{x_{1},\ldots,x_{n},u_{1},\ldots,u_{t}\}$,
$$\Omega =
\Upsilon_{x_{1},\ldots,x_{k},u_{1},\ldots,u_{t}}^{y_{1},\ldots,y_{k},v_{1},\ldots,v_{t}}
\wedge
\Upsilon_{x_{k+1},\ldots,x_{n},u_{1},\ldots,u_{t}}^{y_{k+1},\ldots,y_{n},v_{t+1},\ldots,v_{2t}}
\wedge
\Upsilon_{x_{1},\ldots,x_{n},u_{1},\ldots,u_{t}}^{y_{1},\ldots,y_{n},v_{2t+1},\ldots,v_{3t}},$$
and $\Theta^{(s)}$ has no solutions.
Then $(\Theta\setminus\Upsilon)\cup\Omega$ has no solutions in $D^{(s)}$.
\end{thm}

To prove these theorems we need to introduce a partial order on domain sets.
To every domain set $D^{(\top)}$ we 
assign a tuple of integers
$\size(D^{(\top)}) = (|D_{1}|,|D_{2}|,\dots,|D_{s}|)$, 
where
$D_{1},D_{2},\ldots,D_{s}$ is the set of all different domains of $D^{(\top)}$
ordered by their size starting from the large one.
Then the lexicographic order on tuples of integers
induces a partial order on domain sets, 
that is 
we say that
$(a_{1},\ldots,a_{k})< (b_{1},\ldots,b_{l})$
if there exists $j\in\{1,2,\dots,\min(k+1,l)\}$
such that 
$a_{i} = b_{i}$ for every $i<j$, and 
$a_{j}<b_{j}$ or $j=k+1$.

It follows from the definition that $\le$ is transitive and
there does not exist an infinite descending chain of reductions.
Note that duplicating domains does not affect this partial order, 
that is why we do not make the size of a domain set larger if 
we consider an expanded covering.
At the same time, for every minimal (proper) one-of-four reduction 
$D^{(1)}$ of the instance with a domain set $D^{(0)}$ 
we have $\size(D^{(1)})<\size(D^{(0)})$.
Let us show this for a central reduction.
We replace every domain having a nontrivial center by 
a smaller domain and we do not change other domains.
Let $D_{y}^{(0)}$ be a domain of the maximal size having a nontrivial center.
Then $|D_{y}^{(0)}|$ will be replaced by smaller numbers in 
the sequence $\size(D^{(0)})$ 
making the sequence smaller.

We prove theorems of this subsection simultaneously by the induction on the size of the domain sets.
Let $D^{(\bot)}$ be a domain set.
Assume that Theorems~\ref{KeyConjunctionMain}, \ref{FindPerfectConstraint}, and \ref{CannotLooseSolution} hold
on instances $\Theta$ with a domain set $D^{(0)}$
if $\size(D^{(0)})< \size(D^{(\bot)})$, and
Theorems~\ref{ParPropertyMain} and \ref{ParPropertyForSubcontraint} hold
if $\size(D^{(s)})< \size(D^{(\bot)})$.
Let us prove Theorems~\ref{KeyConjunctionMain}, \ref{FindPerfectConstraint}, and \ref{CannotLooseSolution}
on instances $\Theta$ with a domain set $D^{(0)}$
if $\size(D^{(0)})=\size(D^{(\bot)})$,
and Theorems~\ref{ParPropertyMain} and \ref{ParPropertyForSubcontraint} for
$\size(D^{(s)})=\size(D^{(\bot)})$.

\begin{THMKeyConjunctionMain}
Suppose $D^{(1)}$ is a minimal 
1-consistent one-of-four reduction of a cycle-consistent irreducible CSP instance $\Theta$,
$\Omega(x_{1},\ldots,x_{n})$ is a subconstraint of $\Theta$,
the solution set of $\Omega^{(1)}$ is subdirect,
$\Theta\setminus\Omega$ has a solution in $D^{(1)}$,
and $\Theta$ has no solutions in $D^{(1)}$.
Then there exist instances
$\Upsilon_{1},\ldots,\Upsilon_{t}\in \ExpShort(\Omega)$
such that
$\Phi=(\Theta\setminus\Omega)\cup \Upsilon_{1}\cup\dots\cup\Upsilon_{t}$ has no solutions in $D^{(1)}$,
each $\Upsilon_{i}(x_{1},\ldots,x_{n})$ is a subconstraint 
of $\Phi$, 
and
$\Upsilon_{i}^{(1)}(x_{1},\ldots,x_{n})$ defines a subdirect key relation with the parallelogram property for every $i$.
\end{THMKeyConjunctionMain}

\begin{proof}
Let $\Sigma$ be the set of all relations defined by $\Upsilon^{(1)}(x_{1},\ldots,x_{n})$
where $\Upsilon\in\ExpShort(\Omega)$.
To every relation $\rho\in\Sigma$ we assign 
a constraint $((x_{1},\ldots,x_{n});\rho)$, 
which we denote by $C(\rho)$.
We can find 
$\Sigma_0\subseteq \Sigma$ such that
the instance $(\Theta^{(1)}\setminus\Omega^{(1)})\cup C(\Sigma_{0})$ has no solutions,
but if we replace any relation of $\Sigma_{0}$ by all bigger relations from $\Sigma$ (weaker in terms of constraints)
then we get an instance with a solution.

Let $\Sigma_{0} = \{\rho_{1},\ldots,\rho_{t}\}$.
For each $\rho_{i}$ and each 
$\alpha\notin\rho_{i}$ we consider an inclusion-maximal 
relation $\rho_{i,\alpha}\supseteq \rho_{i}$ from $\Sigma$ such that 
$\alpha\notin \rho_{i,\alpha}$.
Since $\rho_{i} = \bigcap_{\alpha\notin \rho_{i}} \rho_{i,\alpha}$, 
if $\rho_{i}\neq \rho_{i,\alpha}$ for each $\alpha$
then $\rho_{i}$ could be replace by bigger relations
that are still in $\Sigma$, which contradicts our assumptions.
Then 
for each $\rho_{i}$ 
there exists a tuple $\alpha_{i}$ such that $\rho_{i}$ is an inclusion-maximal relation without $\alpha_{i}$ in $\Sigma$.

By Corollary~\ref{MaximalMeansKey}, $\rho_{i}$ is a key relation for every $i$.
Therefore we get a sequence of instances $\Upsilon_{1},\ldots,\Upsilon_{t}\in\ExpShort(\Omega)$
such that 
$\Upsilon_{i}^{(1)}$ defines 
$\rho_{i}$ for every $i$.
Put 
$\Phi = (\Theta\setminus\Omega)\cup \Upsilon_{1}\cup\dots\cup\Upsilon_{t}$.
We choose variables in the instance so that
the only common variables of $\Upsilon_{1},\ldots,\Upsilon_{t}$ are $x_{1},\ldots,x_{n}$,
which guarantees that 
$\Upsilon_{i}(x_{1},\ldots,x_{n})$ is a subconstraint of
$\Phi$.

Since $\Phi$ is a covering of 
$\Theta$, by Lemma~\ref{ExpandedConsistencyLemma}, $\Phi$ is cycle-consistent and irreducible.
Assume that 
$\rho_{i}$ does not have the parallelogram property.
Without loss of generality we assume that 
the failing partition is 
$\{x_{1},\dots,x_{k}\}$,
$\{x_{k+1},\dots,x_{n}\}$.
Define the instance 
$\Omega_{i}$ from $\Upsilon_{i}$ 
using the construction from Theorem~\ref{ParPropertyForSubcontraint}.
Then the relation defined by 
$\Omega_{i}^{(1)}(x_{1},\dots,x_{n})$ is
bigger than $\rho_{i}$
and $\Omega_{i}\in\ExpShort(\Omega)$, which 
means that 
$(\Phi\setminus\Upsilon_{i})\cup\Omega_{i}$
has a solution in $D^{(1)}$
and contradicts the inductive assumption for Theorem~\ref{ParPropertyForSubcontraint}.
Hence, $\rho_{i}$ has the parallelogram property for every $i$.
\end{proof}

\input{Transformations.tex}

\begin{THMCannotLooseSolution}
Suppose $D^{(1)}$  is a 1-consistent nonlinear reduction of a cycle-consistent irreducible CSP instance $\Theta$.
If $\Theta$ has a solution then it has a solution in $D^{(1)}$.
\end{THMCannotLooseSolution}

\begin{proof}
Assume the contrary, 
that is, $\Theta$ has a solution but 
$\Theta^{(1)}$ has no solutions.
By Theorem~\ref{NextReductionTwo},
there exists a minimal 1-consistent 
nonlinear reduction 
such that $\Theta$ has no solutions in it.

First, we consider the set of all minimal 1-consistent nonlinear reductions of $\Theta$,
which we denote by $\mathfrak{R}$.
Then we consider an instance $\Theta'\in\Expanded(\Theta)$
with the minimal positive number of reductions $D^{(\vartriangle)}\in \mathfrak{R}$
such that $\Theta'$ has no solutions in $D^{(\vartriangle)}$.
Note that this transformation of $\Theta$ to $\Theta'$ can be omitted if $D^{(1)}$ is not a PC reduction.
Then we weaken the instance $\Theta'$ (replace any constraint by all weaker constraints) while we
still have a reduction $D^{(\vartriangle)}\in \mathfrak{R}$ such that $\Theta'$ has no solutions in $D^{(\vartriangle)}$.
After that we remove all dummy variables from constraints
and denote the obtained instance by $\Theta''$.
Note that $\Theta''$ is not fragmented (since it is crucial
in some $D^{(\vartriangle)}$), $\Theta''\in\Expanded(\Theta)$, and 
for any reduction $D^{(\vartriangle)}\in \mathfrak{R}$ the instance $\Theta''$ is either crucial in $D^{(\vartriangle)}$, or has a solution in $D^{(\vartriangle)}$.
The last property also holds for any expanded covering 
if $\Theta''$ which is crucial in some reduction $D^{(\vartriangle)}$.
Choose a reduction $D^{(\vartriangle)}$ from $\mathfrak{R}$
such that $\Theta''$ is crucial in it.

Assume that $\Theta''$ is not linked.
If $D^{(\vartriangle)}$ is a PC reduction, then
we apply Theorem~\ref{NextReductionThree} 
to find a reduction $D^{(1)}$ (it is a different 
reduction $D^{(1)}$)
and a strategy $D^{(1)},\dots,D^{(s)}$ 
for 
$\Theta''^{(1)}$
such that 
the solution set of $\Theta''^{(1)}$
is subdirect, 
the strategy has only nonlinear reductions,
$D_{y}^{(s)}\subseteq 
D_{y}^{(\vartriangle)}$ for every $y$.
Then 
$\Theta''^{(1)}$ is cycle-consistent and irreducible.
By the inductive assumption 
$\Theta''^{(2)}$ has a solution,
then by Lemma~\ref{ProperReductionPreservesCycleConAndIrreducability} $\Theta''^{(2)}$
is cycle-consistent and irreducible, 
by the inductive assumption 
$\Theta''^{(3)}$ has a solution,
and so on.
Thus we can prove that 
$\Theta''^{(s)}$ has a solution,
which means that 
$\Theta''^{(\vartriangle)}$ has a solution
and contradicts our assumption.

If $D^{(\vartriangle)}$ is an absorbing or central reduction,
then we choose a variable $x$ of $\Theta''$ and an element $c\in D_{x}^{(\vartriangle)}$,
and for every variable $y$ by $D_{y}^{(\top)}$ we denote the set of all elements of $D_{y}$ linked to $c$.
Since $\Theta''$ is irreducible, the solution set of $\Theta''^{(\top)}$ is subdirect.
Therefore, $\Theta''^{(\top)}$ is irreducible and cycle-consistent.
By Lemmas
\ref{AbsImplies},
\ref{CenterImplies} the reduction $D^{(\bot)}$, defined by $D_{y}^{(\bot)} = D_{y}^{(\top)}\cap D_{y}^{(\vartriangle)}$ for every variable $y$,
is an absorbing or central reduction for $\Theta''^{(\top)}$.
Since $D^{(\vartriangle)}$ is a 1-consistent reduction and $D^{(\top)}$ is just a linked component, 
the reduction $D^{(\bot)}$ is also 1-consistent.
By the inductive assumption, $\Theta''^{(\bot)}$ has a solution, which gives a contradiction.

Thus, we assume that $\Theta''$ is linked.
Recall that by the inductive assumption for Theorem~\ref{ParPropertyMain}, every constraint of
$\Theta''$ is critical and has the parallelogram property.
If $\Theta''$ is not connected, then by Theorem~\ref{FindPerfectConstraint}, there exists an instance $\Upsilon\in\Expanded(\Theta'')$
that is crucial in $D^{(\vartriangle)}$ and contains a linked connected subinstance $\Omega$.
If $\Theta''$ is connected, then $\Theta''$ is a linked connected component itself and we put $\Upsilon = \Omega = \Theta''$.
At the moment we have $\Upsilon\in\Expanded(\Theta'')$
that is crucial in $D^{(\vartriangle)}$ and 
a linked connected subinstance $\Omega$.

Let $x_{1}$ be the first variable in a constraint $C\in\Omega$.
By Lemma~\ref{CriticalMeansIrreducible}, $\ConOne(C,x_1)$ is irreducible.
By Corollary~\ref{PathInConnectedComponent},
there exists a bridge $\delta$
from $\ConOne(C,x_1)$
to $\ConOne(C,x_1)$ such that $\delta(x,x,y,y)$ is a full relation.
By Corollary~\ref{LinkedLink},
there exists a relation
$\zeta\subseteq D_{x_1}\times D_{x_1}\times \mathbb Z_{p}$
such that
$(y_{1},y_{2},0)\in \zeta\Leftrightarrow (y_{1},y_{2})\in\ConOne(C,x_1)$
and $\proj_{1,2}(\zeta) = \cover{\ConOne(C,x_1)}$.
Let us replace the variable $x_1$ of $C$ in $\Upsilon$ by $x_1'$
and add the constraint $\zeta(x_1,x_1',z)$.
The obtained instance we denote by $\Upsilon'$.
Let 
$\Var(\Upsilon) = 
\{x_{1},\ldots,x_{n}\}$,
$\Upsilon'(x_{1},\ldots,x_{n},z)$
define the relation $S$, which is 
the projection of the solution set of $\Upsilon'$
onto all variables but $x_1'$.
Let $C = R(x_1, x_{i_1},\ldots,x_{i_s})$, 
$R'(x_1, x_{i_1},\ldots,x_{i_s})
= \exists x_1' 
R(x_1', x_{i_1},\ldots,x_{i_s})\wedge 
(x_1,x_1')\in \cover{\ConOne(C,x_{1})}$.
The projection 
of $S$ onto the first $n$ variables 
is the solution set of the instance 
$\Upsilon$ whose constraint $C$
is replaced by the weaker constraint $R'(x_1, x_{i_1},\ldots,x_{i_s})$.
Since $\Upsilon$ is crucial in $D^{(\vartriangle)}$, 
the solution set 
$S$ contains a tuple whose first $n$ elements are from $D^{(\vartriangle)}$.
Moreover, the last element of all such tuples is not equal to 0, 
since otherwise this would imply that 
$\Upsilon$ has a solution in $D^{(\vartriangle)}$.

By the assumption, $\Theta$ has a solution, 
and therefore $\Upsilon$ has a solution,
which means that $\Upsilon'$ has a solution
with $z=0$
and, equivalently, 
$S$ has a tuple whose last element is $0$.
Since $\mathbb Z_{p}$ does not have proper subalgebras of size greater than 1, 
we have $\proj_{n+1}(S) = \mathbb Z_{p}$.

Let us show for $i\in\{1,2,\dots,n\}$ that 
$(\proj_{i}(S))^{(\vartriangle)}$ is 
a one-of-four subuniverse of $\proj_{i}(S)$
of the same type as $D^{(\vartriangle)}$.
For absorbing and central reductions 
it follows from Lemma~\ref{PCBsubNonPC}.
For the PC type we consider a PC congruence $\sigma$
on $D_{x_{i}}$.
By Theorems~\ref{NextReductionOne}, \ref{NextReductionTwo},
for every equivalence class $U$ of $\sigma$ 
there exists a minimal 1-consistent PC reduction $D^{(\triangledown)}\in \mathfrak{R}$ such that
$D_{x_{i}}^{(\triangledown)} \subseteq U$.
As we assumed earlier, 
for any reduction from $\mathfrak{R}$ the instance $\Upsilon$ is either crucial in it, or has a solution in it.
Therefore, $\Upsilon'$ has a solution in any reduction from $\mathfrak{R}$,
and
$\Upsilon'$ has a solution with $x_{i}\in U$.
Hence, $\sigma$ restricted to $\proj_{i}(S)$ 
is still a PC congruence.
Moreover, 
$(\proj_{i}(S))^{(\vartriangle)}$ 
is an intersection of equivalence classes 
of the corresponding PC congruences on $\proj_{i}(S)$.
Thus, we showed that 
$(\proj_{i}(S))^{(\vartriangle)}$ is 
a one-of-four subuniverse of $\proj_{i}(S)$
of the same type as $D^{(\vartriangle)}$.

By Lemma~\ref{PCBrel}, $S^{(\vartriangle)}$ is a 
nonlinear one-of-four subuniverse of $S$ (here we do not reduce the last variable).
Also, by Lemma~\ref{PCBrel}, the set of all tuples from $S$ whose last element is 0 is a linear subuniverse of $S$,
we denote this subuniverse by $S_{0}$.
By Lemma~\ref{IntersectionOfTwoSubuniverses},
the intersection $S^{(\vartriangle)}\cap S_{0}$ is not empty, which means that $\Upsilon$ has a solution in $D^{(\vartriangle)}$ and contradicts our assumptions.
\end{proof}

Note that Theorem~\ref{ParPropertyMain} could be derived from 
Theorem~\ref{ParPropertyForSubcontraint}, 
but we decided to keep the original proof 
of Theorem~\ref{ParPropertyMain}
because it demonstrates the idea for both theorems 
in an easier way.

\begin{THMParPropertyMain}
Suppose $D^{(0)},\ldots,D^{(s)}$ is a minimal strategy for a cycle-consistent irreducible CSP instance
$\Theta$,
and a constraint $\rho(x_{1},\ldots,x_{n})$ 
of $\Theta$ is crucial in $D^{(s)}$.
Then $\rho$ is a critical relation with the parallelogram property.
\end{THMParPropertyMain}

\begin{proof}
Since $\rho(x_{1},\ldots,x_{n})$ is crucial, $\rho$ is a critical relation.
Let $\Theta'$ be obtained from
$\Theta$ by
replacement of $\rho(x_{1},\ldots,x_{n})$ by
all weaker constraints.
Since $\Theta$ is crucial in $D^{(s)}$, 
$\Theta'$ has a solution in $D^{(s)}$.
By Lemma~\ref{ExpandedConsistencyLemma}, $\Theta'$ is cycle-consistent and irreducible.

Assume that $|D^{(s)}_{x}|=1$ for every variable $x$. Since the reduction $D^{(s)}$ is 1-consistent,
we get a solution, which contradicts the fact that
$\Theta$ has no solutions in $D^{(s)}$.

If we have a nontrivial binary absorbing subuniverse, or a nontrivial center, or a nontrivial PC subuniverse on some domain $D_{x}^{(s)}$, then
by Theorems~\ref{NextReductionOne},~\ref{NextReductionTwo},
there exists a minimal nonlinear 1-consistent reduction $D^{(s+1)}$ for $\Theta$.
As we explained before, $\size(D^{(s+1)})<\size(D^{(s)})$.

Then, by Lemma~\ref{ProperReductionPreservesCycleConAndIrreducability},
$\Theta'^{(s)}$ is cycle-consistent and irreducible.
By Theorem~\ref{CannotLooseSolution},
$\Theta'$ has a solution in $D^{(s+1)}$.
Hence, $\rho(x_{1},\ldots,x_{n})$ is crucial in $D^{(s+1)}$.
By the inductive assumption $\rho$ has the parallelogram property.

It remains to consider the case when $\ConLin(D_{x}^{(s)})$ is proper for every $x$ such that $|D_{x}^{(s)}|>1$.
Let $\alpha$ be a solution of $\Theta'$ in $D^{(s)}$.
Let the projection of $\alpha$ onto the variables $x_{1},\ldots,x_n$ be $(a_{1},\ldots,a_{n})$.

Assume that $\rho$ does not have the parallelogram property.
Without loss of generality we can assume that there exist $c_{1},\ldots,c_{n}$
and $d_{1},\ldots,d_{n}$
such that
\begin{align*}
(c_{1},\ldots,c_{k},c_{k+1},\ldots,c_{n})&\notin\rho,\\
(c_{1},\ldots,c_{k},d_{k+1},\ldots,d_{n})&\in\rho,\\
(d_{1},\ldots,d_{k},c_{k+1},\ldots,c_{n})&\in\rho,\\
(d_{1},\ldots,d_{k},d_{k+1},\ldots,d_{n})&\in\rho.
\end{align*}
Put
\begin{align*}
\rho'(x_{1},\ldots,x_{n}) =
\exists y_{1}\dots\exists y_{n}\;
\rho(x_{1},\ldots,x_{k},y_{k+1},\ldots,y_{n})
\wedge&\\
\rho(y_{1},\ldots,y_{k},x_{k+1},\ldots,x_{n})
\wedge
\rho(y_{1},\ldots,y_{k},y_{k+1},\ldots,y_{n}).
\end{align*}
Obviously, $\rho\subsetneq\rho'$
and $\rho'\in\Gamma$,
therefore $(a_{1},\ldots,a_{n})\in\rho'.$
Hence, there exist $b_{1},\ldots,b_{n}$ such that
\begin{align*}
(a_{1},\ldots,a_{k},b_{k+1},\ldots,b_{n})&\in\rho,\\
(b_{1},\ldots,b_{k},a_{k+1},\ldots,a_{n})&\in\rho,\\
(b_{1},\ldots,b_{k},b_{k+1},\ldots,b_{n})&\in\rho.
\end{align*}
By Lemma~\ref{ParProperty},
there exists a tuple
$(e_{1},\ldots,e_{n})\in \rho$
such that
$(a_{i},e_{i})\in\LinCon(D_{x_{i}}^{(s)})$  for every $i$.

Consider the minimal linear reduction $D^{(s+1)}$ of 
$\Theta^{(s)}$ such that 
$\alpha\in D^{(s+1)}$.
Then we have $(e_{1},\ldots,e_{n})\in \rho^{(s+1)}$,
and by Lemma~\ref{ProperReductionPreservesSubdirectness},
$D^{(s+1)}$ is a 1-consistent reduction of $\Theta^{(s)}$.
Since $\Theta'$ has a solution in $D^{(s+1)}$,
$\rho(x_{1},\ldots,x_{n})$ is crucial in $D^{(s+1)}$.
We get a longer minimal strategy with smaller $\size(D^{(s+1)})$, hence by the inductive assumption
the relation $\rho$ is a critical relation with the parallelogram property.
\end{proof}

\begin{THMParPropertyForSubcontraint}
Suppose $D^{(0)},\ldots,D^{(s)}$ is a minimal strategy for a cycle-consistent irreducible CSP instance
$\Theta$,
$\Upsilon(x_{1},\ldots,x_{n})$ is a subconstraint of $\Theta$,
the solution set of $\Upsilon^{(s)}$ is subdirect,
$k\in\{1,2,\dots,n-1\}$,
$\Var(\Upsilon) = \{x_{1},\ldots,x_{n},u_{1},\ldots,u_{t}\}$,
$$\Omega =
\Upsilon_{x_{1},\ldots,x_{k},u_{1},\ldots,u_{t}}^{y_{1},\ldots,y_{k},v_{1},\ldots,v_{t}}
\wedge
\Upsilon_{x_{k+1},\ldots,x_{n},u_{1},\ldots,u_{t}}^{y_{k+1},\ldots,y_{n},v_{t+1},\ldots,v_{2t}}
\wedge
\Upsilon_{x_{1},\ldots,x_{n},u_{1},\ldots,u_{t}}^{y_{1},\ldots,y_{n},v_{2t+1},\ldots,v_{3t}},$$
and $\Theta^{(s)}$ has no solutions.
Then $(\Theta\setminus\Upsilon)\cup\Omega$ has no solutions in $D^{(s)}$.
\end{THMParPropertyForSubcontraint}

\begin{proof}

Put $\Theta' = (\Theta\setminus\Upsilon)\cup\Omega$.
Since $\Omega$ is a covering of $\Upsilon$, 
Lemma~\ref{ExpandedConsistencyLemma} implies that 
$\Theta\cup\Omega$ is cycle-consistent and irreducible.
Assume that $\Theta'$ has a solution in $D^{(s)}$.

We recursively build a strategy
$D^{(s)},D^{(s+1)},\ldots,D^{(q)}$
for $\Theta\cup\Omega = \Theta'\cup\Upsilon$ 
satisfying the following conditions:
\begin{enumerate}
    \item 
    $D^{(s)},D^{(s+1)},\dots,D^{(q)}$ is a minimal strategy for $\Theta'^{(s)}$;
    \item if $s\le j<q$ and $D^{(j+1)}$ is a linear reduction, then for each $i\in\{1,2,\dots,t\}$
    $$D_{u_{i}}^{(j+1)} = 
    \proj_{n+i}(\rho'\cap 
    (D_{x_{1}}^{(j+1)}\times\dots\times D_{x_{n}}^{(j+1)}
    \times
    D_{u_{1}}^{(j)}\times\dots\times D_{u_{t}}^{(j)})),
    $$
    where $\rho'$ is the relation defined by 
    $\Upsilon(x_{1},\ldots,x_{n},u_1,\ldots,u_t)$;
\item the solution set of $\Upsilon^{(j)}$ is subdirect for $s\le j\le q$;
\item $\Theta'$ has a solution in $D^{(q)}$.
\end{enumerate}

Note that here we allow 
$D^{(j)}$ to be equal to $D^{(j+1)}$ in a strategy, which can happen 
if $D^{(j+1)}$ is a proper reduction for 
$\Upsilon^{(j)}$ but not proper for $\Theta'^{(j)}$.

We will prove that we can make this sequence longer while
$|D^{(q)}_{x_{i}}|>1$ for some $i$.
By Theorem~\ref{NextReduction}, 
there exists a nontrivial one-of-four subuniverse on 
$D^{(q)}_{x}$ if $|D^{(q)}_{x}|>1$.
We consider two cases:

Case 1. There exists a nontrivial binary absorbing subuniverse, or a nontrivial center, or a nontrivial PC congruence on some domain $D_{x}^{(q)}$. 
Then applying Theorems~\ref{NextReductionOne}, \ref{NextReductionTwo}
to the strategy 
$D^{(0)},D^{(1)},\dots,D^{(q)}$ of $\Theta\cup\Omega$, 
we conclude that there exists a minimal 1-consistent nonlinear reduction $D^{(q+1)}$
for $(\Theta\cup\Omega)^{(q)}$.
By Lemma~\ref{ProperReductionPreservesCycleConAndIrreducability},
$\Theta'^{(q)}$ and  
$\Upsilon^{(q)}$ are cycle-consistent and irreducible.
By Theorem~\ref{CannotLooseSolution},
$\Theta'$ has a solution in $D^{(q+1)}$
and
$\Upsilon$ has a solution in $D^{(q+1)}$.
By Lemma~\ref{ProperReductionPreservesSubdirectness},
the solution set of $\Upsilon^{(q+1)}$ is subdirect.
Thus, we made the sequence longer.

Case 2. $\ConLin(D_{x}^{(q)})$ is proper for every $x$ such that $|D_{x}^{(q)}|>1$.
Let $\alpha$ be a solution of $\Theta'$ in $D^{(q)}$.
We define the new linear reduction 
$D^{(q+1)}$ as follows.
For all variables 
but $u_{1},\ldots,u_{t}$,
we choose an equivalence class
of $\ConLin(D_{x}^{(q)})$
containing the corresponding element of 
the solution $\alpha$.
For the variable $u_{i}$ we define 
$D_{u_{i}}^{(q+1)}$ by the formula in (2) 
from the above list for $j=q$.
By Lemma~\ref{ProperReductionPreservesSubdirectness},
$D^{(q+1)}$ is  
1-consistent for $\Theta'$.
Note that it does not follow from the definition 
that 
$D_{u_{i}}^{(q+1)}$ is not empty and we will prove this later.

Let the projection of $\alpha$ onto the variables $x_{1},\ldots,x_n$ be $(a_{1},\ldots,a_{n})$.
Suppose $\Upsilon^{(s)}(x_{1},\ldots,x_{n})$ defines a relation $\rho$.
Since $\alpha$ is a solution of $\Theta'^{(s)}$, there exist
$b_{1},\ldots,b_{n}$ such that
\begin{align*}
(a_{1},\ldots,a_{k},b_{k+1},\ldots,b_{n})&\in\rho,\\
(b_{1},\ldots,b_{k},a_{k+1},\ldots,a_{n})&\in\rho,\\
(b_{1},\ldots,b_{k},b_{k+1},\ldots,b_{n})&\in\rho.
\end{align*}

Since the solution set of $\Upsilon^{(j)}$ is subdirect 
for $s\le j\le q$, we can apply Lemma~\ref{ParProperty}
to $\rho$ and the strategy $D^{(s)},\dots,D^{(q)}$.
Hence, there exists a tuple
$(d_{1},\ldots,d_{n})\in \rho$
such that
$(a_{i},d_{i})\in\LinCon(D_{x_{i}}^{(q)})$ for every $i$.
Therefore, $(\Upsilon^{(s)}(x_{1},\ldots,x_{n}))^{(q+1)}$ is not empty.

Let us show by induction on $j=s,s+1,\dots,q$ that $(\Upsilon^{(j)}(x_{1},\ldots,x_{n}))^{(q+1)}$ is not empty. For $j=s$ we already know this.
Assume that 
$(\Upsilon^{(j)}(x_{1},\ldots,x_{n}))^{(q+1)}$ is not empty.
If the reduction $D^{(j+1)}$ is not linear then 
we apply Theorem~\ref{PreviousReductions}
to $(\Upsilon^{(j)}(x_{1},\ldots,x_{n}))^{(q+1)}$ 
and the strategy $D^{(s)},\dots,D^{(q)}$,
and obtain that
 $(\Upsilon^{(j+1)}(x_{1},\ldots,x_{n}))^{(q+1)}$
is not empty.
If the reduction $D^{(j+1)}$ is linear then it follows from the definition of $D_{u_i}^{(j+1)}$
that $(\Upsilon^{(j+1)}(x_{1},\ldots,x_{n}))^{(q+1)}$ is not empty.
Thus, we can prove that $(\Upsilon^{(q)}(x_{1},\ldots,x_{n}))^{(q+1)}$ is not empty, 
and therefore $D_{u_i}^{(q+1)}$ is not empty
for every $i$. 
Considering the solution set of $\Upsilon$ and 
applying Corollary~\ref{LinearImplies}, 
we derive that 
$D_{u_{i}}^{(q+1)}$ is a linear subuniverse of 
$D_{u_{i}}^{(q)}$. Hence,
the reduction $D^{(q+1)}$ is a 1-consistent linear reduction
for $\Upsilon^{(q)}$.

By Lemma~\ref{ProperReductionPreservesSubdirectness},
$(\Upsilon^{(q)}(x_{1},\ldots,x_{n}))^{(q+1)}$
is subdirect. 
From the definition of 
$D_{u_{i}}^{(q+1)}$ 
we derive that 
the projection of
the solution set of $\Upsilon^{(q+1)}$ onto 
$u_{i}$ is $D_{u_{i}}^{(q+1)}$ 
for every $i$, which means that 
the solution set of $\Upsilon^{(q+1)}$
is subdirect.
Hence, we get a longer strategy having all the necessary properties.

Thus, we showed that we can make the sequence longer 
until  $|D^{(q)}_{x_{i}}|=1$ for every $i$.
Assume that we reached this final state.
Since both $\Upsilon$ and $\Theta'$ have a solution in $D^{(q)}$
and $x_{1},\dots,x_{n}$ are 
their only common variables,
$\Theta$ has a solution in $D^{(q)}$,
which contradicts the fact that
$\Theta$ has no solutions in~$D^{(s)}$.
\end{proof}

\subsection{Theorems from Section~\ref{CorretnessSection}}

In this subsection we assume that the variables of the instance $\Theta$ are $x_{1},\ldots,x_{n}$,
and the domain of $x_{i}$ is $D_{i}$ for every $i$.
The first two theorems are proved together.

\begin{thmAbsorptionCenterStep}
Suppose $\Theta$ is a cycle-consistent irreducible CSP instance, and 
$B$ is a nontrivial binary absorbing subuniverse or a nontrivial center of $D_{i}$.
Then $\Theta$ has a solution if and only if
$\Theta$ has a solution with $x_{i}\in B$.
\end{thmAbsorptionCenterStep}

\begin{thmPCStepThm}
Suppose $\Theta$ is a cycle-consistent irreducible CSP instance, 
there does not exist a nontrivial binary absorbing subuniverse or a nontrivial center on $D_{j}$
for every $j$,
$(D_{i};w)/\sigma$ is a polynomially complete algebra, 
and 
$E$ is an equivalence class of $\sigma$.
Then $\Theta$ has a solution if and only if
$\Theta$ has a solution with $x_{i}\in E$.
\end{thmPCStepThm}
\begin{proof}
By Theorems~\ref{NextReductionOne}, \ref{NextReductionTwo}, there exists a minimal 1-consistent nonlinear reduction 
$D^{(1)}$
such that 
$D_{x_{i}}^{(1)}\subseteq B$ for 
Theorem~\ref{AbsorptionCenterStep}, 
and 
$D_{x_{i}}^{(1)}\subseteq E$ for 
Theorem~\ref{PCStepThm}.
By Theorem~\ref{CannotLooseSolution}, there exists a solution in $D^{(1)}$.
\end{proof}

The next theorem will be used in the proof of Theorem~\ref{LinearStep} 
from Section~\ref{CorretnessSection}.
\begin{thm}\label{LinearStepHelp}
Suppose the following conditions hold:
\begin{enumerate}
\item $\Theta$ is a linked cycle-consistent irreducible CSP instance;
\item there does not exist a nontrivial binary absorbing subuniverse or a nontrivial center on $D_{j}$ for every $j$;
\item if we replace every constraint of $\Theta$ by all weaker constraints then the obtained instance 
has a solution with $x_{i} = b$ for every $i$ and $b\in D_{i}$ (the obtained instance has a subdirect solution set);
\item 
$D^{(1)}$ is a minimal linear reduction for $\Theta$;
\item $\Theta$ is crucial in $D^{(1)}$.
\end{enumerate}
Then there exists a constraint
$\rho(x_{i_1},\ldots,x_{i_s})$ in $\Theta$
and a subuniverse $\zeta$ of $\mathbf{D_{i_1}}\times\dots\times \mathbf{D_{i_s}}\times \mathbf{\mathbb Z_{p}}$
such that
the projection of $\zeta$ onto the first $s$ coordinates
is bigger than $\rho$ but
the projection of $\zeta\cap (D_{i_1}\times\dots\times D_{i_s}\times \{0\})$
onto the first $s$ coordinates is equal to $\rho$.
\end{thm}

\begin{proof}
We consider two cases. 
Case 1. Assume that $\Theta$ contains just one constraint
$\rho(x_1,\ldots,x_{n})$.
By Corollary~\ref{LinearImplies},
$D_{n}'=\proj_{n}(\rho\cap (D_{1}^{(1)}\times\dots\times D_{n-1}^{(1)}
\times D_{n}))$
is a linear subuniverse of $D_{n}$.
By Lemma~\ref{LinearAlgebrasFact}, 
$D_{n}^{(1)}$ and $D_{n}'$ can be viewed as products of affine subspaces and can be defined by linear equations.
Since $D_{n}^{(1)}\cap D_{n}'=\varnothing$ 
and $D_{n}^{(1)}$ is a minimal reduction, 
we can take an equation defining $D_{n}'$ that does not hold on $D_{n}^{(1)}$ to get a maximal linear congruence $\sigma$ on $D_{n}$ such that $D_{n}^{(1)}$ and $D_{n}'$ are in different equivalence classes of $\sigma$.
Note that $D_{n}/\sigma\cong\mathbb Z_{p}$ for some $p$.
Let $\psi$ be the corresponding homomorphism 
from $D_{n}$ to $\mathbb Z_{p}$.
Put
$$\zeta(x_{1},\ldots,x_{n},z) = \exists x_{n}'\;
\rho(x_{1},\ldots,x_{n-1},x_{n}')
\wedge (\psi(x_{n})=\psi(x_{n}') + z),$$
where 
the expression $(\psi(x_{n})=\psi(x_{n}') + z)$
defines ternary subalgebra of $D_{n}\times D_{n}\times \mathbb Z_{p}$. Thus, we have $\rho$ and
$\zeta$ with the required properties.

Case 2. $\Theta$ contains more than one constraint. 
Then by condition (5), every constraint $C^{(1)}$ is not empty,
which by Lemma~\ref{ProperReductionPreservesSubdirectness} implies 
that $C^{(1)}$ is subdirect. Then $D^{(1)}$ is a minimal 1-consistent linear reduction.
By Theorem~\ref{ParPropertyMain}, every constraint in $\Theta$ is critical and has the parallelogram property.
If $\Theta$ is not connected, then by Theorem~\ref{FindPerfectConstraint} there exists an instance $\Theta'\in\Expanded(\Theta)$
that is crucial in $D^{(1)}$ and contains a linked connected component $\Omega$ such that
the solution set of $\Omega$ is not subdirect.
By condition (3), since the solution set of $\Omega$ is not subdirect, $\Omega$ should contain a constraint relation from
the original instance $\Theta$.
If $\Theta$ is connected, then $\Theta$ is a linked connected component itself and we put $\Omega = \Theta$.
Thus, in both cases we have a linked connected instance 
$\Omega$ having a constraint relation $\rho$ from $\Theta$. 
Let $\rho(x_{i_1},\ldots,x_{i_s})$ be a constraint of $\Theta$.

By Lemma~\ref{CriticalMeansIrreducible}, $\ConOne(\rho,1)$ is an irreducible congruence.
By Corollary~\ref{PathInConnectedComponent},
there exists a bridge $\delta$
from $\ConOne(\rho,1)$
to $\ConOne(\rho,1)$ such that $\widetilde\delta$ is a full relation.
By Corollary~\ref{LinkedLink},
there exists a relation
$\xi\subseteq D_{i_{1}}\times D_{i_{1}}\times \mathbb Z_{p}$
such that
$(x_{1},x_{2},0)\in \xi\Leftrightarrow (x_{1},x_{2})\in\ConOne(\rho,1)$
and $\proj_{1,2}(\xi) = \cover{\ConOne(\rho,1)}$.

It remains to put $\zeta(x_{i_1},\ldots,x_{i_s},z) = \exists x_{i_{1}}'\; \rho(x_{i_1}',x_{i_2},\ldots,x_{i_s})\wedge
\xi(x_{i_{1}},x_{i_{1}}',z)$.
\end{proof}

\begin{THMmainLinearStep}
Suppose the following conditions hold:
\begin{enumerate}
\item $\Theta$ is a linked cycle-consistent irreducible CSP instance with domain set
$(D_{1},\ldots,D_{n})$;
\item there does not exist a nontrivial binary absorbing subuniverse or a nontrivial center on $D_{j}$ for every $j$;
\item if we replace every constraint of $\Theta$ by all weaker constraints then the obtained instance
has a solution with $x_{i} = b$ for every $i$ and $b\in D_{i}$ (the obtained instance has a subdirect solution set);
\item $L_{i} = D_{i}/\sigma_{i}$ for every $i$, where $\sigma_{i}$ is the minimal linear congruence on $D_{i}$;
\item $\phi:\mathbb Z_{q_{1}}\times \dots \times \mathbb Z_{q_{k}}
\to L_{1}\times\dots\times L_{n}$ is a homomorphism,
where $q_{1},\dots,q_{k}$ are prime numbers;
\item if we replace any constraint of 
$\Theta$ by all weaker constraints then for every $(a_{1},\ldots,a_{k})\in \mathbb Z_{q_{1}}\times \dots \times \mathbb Z_{q_{k}}$ 
there exists a solution of the obtained instance in 
$\phi(a_{1},\ldots,a_{k})$.
\end{enumerate}
Then 
$\{(a_{1},\dots,a_{k})\mid \Theta \text{ has a solution in }\phi(a_1,\dots,a_{k})\}$ is
either empty, or is full, or is an affine subspace of $\mathbb Z_{q_{1}}\times \dots \times \mathbb Z_{q_{k}}$ of codimension 1 (the solution set of a single linear equation).
\end{THMmainLinearStep}
\begin{proof}
Put 
$B=\{(a_{1},\dots,a_{k})\mid \Theta \text{ has a solution in }\phi(a_1,\dots,a_{k})\}$.
If $B$ is full then there is nothing to prove.
Assume that $B$ is not full, then consider 
$(b_1,\ldots,b_{k})\notin B$. 
It follows from condition (6) that 
$\Theta$ is crucial in 
$\phi(b_{1},\ldots,b_{k})$.
Note that 
$\phi(b_{1},\ldots,b_{k})$ defines a minimal linear reduction for $\Theta$.

By Theorem~\ref{LinearStepHelp} there exists a constraint
$\rho(x_{i_1},\ldots,x_{i_s})$ in $\Theta$
and a subuniverse $\zeta$ of $\mathbf{D_{i_1}}\times\dots\times \mathbf{D_{i_s}}\times \mathbf{\mathbb Z_{p}}$
such that
the projection of $\zeta$ onto the first $s$ coordinates
is bigger than $\rho$ but
the projection of $\zeta\cap (D_{i_1}\times\dots\times D_{i_s}\times \{0\})$
onto the first $s$ coordinates is equal to $\rho$.

Then we add a new variable $z$ with domain $\mathbb Z_{p}$
and replace
$\rho(x_{i_1},\ldots,x_{i_s})$ by
$\zeta(x_{i_1},\ldots,x_{i_s},z)$.
We denote the obtained instance by $\Upsilon$.
Let
$L$ be the set of all tuples $(a_{1},\ldots,a_{k},b)\in
\mathbb Z_{q_{1}}\times \dots \times \mathbb Z_{q_{k}} \times \mathbb Z_{p}$
such that $\Upsilon$ has a solution with $z=b$ in $\phi(a_{1},\ldots,a_{k})$.
We know that the projection of $L$ onto
the first $k$ coordinates is a full relation
and $(b_{1},\dots,b_{k},0)\notin L$.
Therefore $L$ is defined by one linear equation.
If this equation is $z = b$ for some $b\neq 0$, then
$B$ is empty.
Otherwise, we put $z=0$ in this equation and get
an equation describing all $(a_{1},\ldots,a_{k})$ such that
$\Theta$ has a solution in $\phi(a_{1},\ldots,a_{k})$.
\end{proof}

\section{Conclusions}\label{ConclusionsSection}

Even though the main problem has been resolved, 
there are many important questions that are still open.
In this section we will discuss 
some consequences of this result, as well as some open questions 
and generalizations of the CSP.

\subsection{A general algorithm for the CSP}
The algorithm presented in the paper, 
as well as the algorithm of Andrei Bulatov \cite{BulatovProofCSP,BulatovProofCSPFOCS},
uses detailed knowledge of the algebra and depends exponentially on the size of the domain. Is there a ``truly polynomial algorithm''?
By $\CSPWNU$ we denote the following decision 
problem:
given a formula
$$\rho_{1}(v_{1,1},\ldots,v_{1,n_{1}})
\wedge
\dots
\wedge
\rho_{s}(v_{s,1},\ldots,v_{1,n_{s}}),$$
where all relations 
$\rho_{1},\dots,\rho_{s}$ are preserved by 
a WNU (we just know it exists);
decide whether this formula is satisfiable.

\begin{problem}
Does there exist a polynomial algorithm for $\CSPWNU$?
\end{problem}

If the domain is fixed then $\CSPWNU$ can be solved by 
the algorithm presented in this paper. 
In fact, we know from \cite[Theorem 4.2]{cyclicterms}
that from a WNU  on a domain of size $k$ we can always derive 
a WNU (and also a cyclic operation) of any prime arity greater than $k$. Thus, we can find finitely many 
WNU operations on domain of size $k$ such that any constraint language preserved by a WNU is preserved by one of them. It remains to 
apply the algorithm for each WNU 
and return a solution if one of them gave a solution.

\subsection{A simplification of the algorithm.}

We believe that 
the algorithm presented in the paper can be simplified. 
For instance, 
we strongly believe that 
the function $\WeakenEveryConstraint$ 
can be removed from the main function 
$\Solve$ without any consequences.

\begin{problem}
Would the algorithm still work if 
the function $\WeakenEveryConstraint$ 
was removed from the function 
$\Solve$?
 \end{problem}
 
 This would reduce the complexity of the algorithm significantly (the depth of the recursion would be 
 $|A|$ instead of $|A|+|\Gamma|$, see Lemma~\ref{RecursionDepth}).

\subsection{A generalization for the nonWNU case.}

Another important question is whether 
some results and ideas introduced in this paper can be 
applied for constraint languages not preserved by a WNU.
For example, it is not clear what assumptions are sufficient to reduce safely a domain to a binary absorbing subuniverse.

\begin{problem}
What are the weakest assumptions for Theorems~\ref{AbsorptionCenterStep}
and 
\ref{PCStepThm} to hold.
\end{problem}

\subsection{Infinite domain CSP}

If we allow the domain to be infinite, the situation is changing significantly.
As it was shown in \cite{bodirskyInfiniteHell} 
every computational problem is
equivalent (under polynomial-time Turing reductions) to a problem of the form $\CSP(\Gamma)$.
In \cite{InfiniteDomainSurvey} the authors 
gave a nice example of a constraint language $\Gamma$ 
such that $\CSP(\Gamma)$ is undecidable. 
Let
$\Gamma$ consists
of three relations (predicates)
$x+y=z$, $x\cdot y=z$ and $x = 1$
over the set of all integers $\mathbb Z$.
Then 
the Hilbert's 10-th problem can be 
expressed as $\CSP(\Gamma)$, which proves undecidability of 
$\CSP(\Gamma)$.

A reasonable assumption on $\Gamma$ which sends the CSP back to the class NP is that $\Gamma$ is a reduct of a finitely bounded homogeneous structure. 
A nice result for such constraint  languages is the full complexity classification of the CSPs over the reducts of $(\mathbb Q;<)$ \cite{bodirskyforrationals}.
This additional assumption allows to 
formulate a statement of the algebraic dichotomy conjecture
for the complexity of the infinite domain CSP 
\cite{barto2016algebraic}.
For more information about the infinite domain CSP 
and the algebraic approach see \cite{bodirsky2012complexity, InfiniteDomainSurvey}.
For a method of reducing an infinite domain CSP to CSPs over finite domains see \cite{bodirsky2016dichotomy}.

\subsection{Valued CSP}
A natural generalization of 
the Constraint Satisfaction Problem 
is the \emph{Valued Constraint Satisfaction Problem} ($\VCSP$), 
where constraint relations are replaced by mappings
to the set of rational numbers, 
and conjunctions are replaced by sum \cite{VCSPIntroduction}.
For a finite set $A$ and a set $\Gamma$ of mappings 
$A\to \mathbb Q\cup \{\infty\}$ 
by $\VCSP(\Gamma)$
we denote the following problem:
given a formula
$$f(x_{1},\dots,x_{n}) = f_{1}(v_{1,1},\ldots,v_{1,n_{1}})
+
\dots
+
f_{s}(v_{s,1},\ldots,v_{s,n_{s}}),$$
where all the mappings 
$f_{1},\dots,f_{s}$ are from $\Gamma$
and $v_{i,j}\in \{x_{1},\ldots,x_{n}\}$
for every $i,j$;
find an assignment 
$(a_{1},\dots,a_{n})$ that minimizes 
$f(x_{1},\dots,x_{n})$.

In \cite[Theorem  21]{VCSPDichotomy},
the authors 
proved that the dichotomy conjecture for CSP 
would imply the dichotomy conjecture for the Valued CSP,
and described all 
sets of mappings $\Gamma$ such that 
$\VCSP(\Gamma)$ is tractable (modulo the CSP Dichotomy Conjecture).
Thus, the result obtained in this paper 
implies the characterization of the complexity 
of $\VCSP(\Gamma)$ for all $\Gamma$.

\subsection{Quantified CSP}

An equivalent definition of $\CSP(\Gamma)$ is to
evaluate 
a sentence 
$\exists x_1 \dots \exists x_n \ (\rho_{1}(\dots)\wedge\dots
\wedge \rho_{s}(\dots))$, 
where $\rho_1,\dots,\rho_s$ are from the constraint language
$\Gamma$.
Then a natural generalization of CSP is the \emph{Quantified Constraint Satisfaction Problem} ($\QCSP$), 
where we allow to use both existential and 
universal quantifiers.
For a constraint language $\Gamma$,
$\QCSP(\Gamma)$
is the problem 
to evaluate 
a sentence of the form $\forall x_1 \exists y_1 \dots \forall x_n \exists y_n \ (\rho_{1}(\dots)\wedge\dots
\wedge \rho_{s}(\dots))$, 
where $\rho_1,\dots,\rho_s$ are  relations from the constraint language $\Gamma$ (see \cite{BBCJK,hubie-sicomp,Meditations,QC2017}).

It was conjectured by Hubie Chen
 \cite{Meditations,MFCS2017} that 
for any constraint language $\Gamma$
the problem $\QCSP(\Gamma)$ 
is either solvable in polynomial time, 
or NP-complete, or PSpace-complete.
Recently, this conjecture was disproved in \cite{QCSPMonsters},
where the authors 
found constraint languages $\Gamma$ 
such that 
$\QCSP(\Gamma)$ is coNP-complete (on 3-element domain), DP-complete (on 4-element domain),
$\Theta_{2}^{P}$-complete (on 10-element domain).
Also the authors classified the complexity of 
the Quantified Constraint Satisfaction Problem 
for constraint languages on 3-element domain containing all unary singleton relations (so called idempotent case), that is, they showed that for such languages  
$\QCSP(\Gamma)$ is either tractable, or NP-complete, or coNP-complete, 
or PSpace-complete. Nevertheless, for higher domain 
as well as for the nonidempotent case the complexity is not known.

\begin{problem}
What can be the complexity of 
$\QCSP(\Gamma)$? 
\end{problem}

Now it is hard to believe that there will be a simple answer to 
this question, that is why it is interesting to 
start with 3-element domain (nonidempotent case)
and 4-element domain.
Another natural question is how many complexity classes can be expressed by $\QCSP(\Gamma)$ up to polynomial equivalence.
Probably more important problem is to describe all tractable cases.

\begin{problem}
Describe all constraint languages $\Gamma$ such that 
$\QCSP(\Gamma)$ is tractable.
\end{problem}

\subsection{Promise CSP}

Another natural generalization of the CSP is 
\emph{the Promise Constraint Satisfaction Problem}, 
where a promise about the input is given
(see \cite{brakensiek2018promise,PCSPAlgebraicApproach}).
Let 
$\Gamma = \{(\rho_{1},\sigma_{1}), \dots
,(\rho_{t},\sigma_{t})\}$, 
where 
$\rho_{i}$ and $\sigma_{i}$
are relations of the same arity 
over the domains 
$A$ and $B$, respectively.
Then 
$\PCSP(\Gamma)$ is the following decision problem:
given two formulas 
\begin{align*}&\rho_{i_1}(v_{1,1},\ldots,v_{1,n_{1}})
\wedge\dots\wedge \rho_{i_s}(v_{s,1},\ldots,v_{s,n_{s}}),\\
&\sigma_{i_1}(v_{1,1},\ldots,v_{1,n_{1}})
\wedge\dots\wedge \sigma_{i_s}(v_{s,1},\ldots,v_{s,n_{s}}),
\end{align*}
where 
$(\rho_{i_j},\sigma_{i_j})$ are from 
$\Gamma$ for every $i$ and 
$v_{i,j}\in \{x_{1},\ldots,x_{n}\}$
for every $i,j$;
distinguish between 
the case when both of them are satisfiable, 
and when both of them are not satisfiable.
Thus, we are given two CSP instances and a promise that 
if one has a solution then another has a solution.
Usually it is also assumed 
that there exists a mapping (homomorphism) 
$h\colon A\to B$ such that 
$h(\rho_{i})\subseteq \sigma_{i}$ for every $i$.
In this case, 
the satisfiability of the first formula 
implies the satisfiability of the second one.
To make sure that the promise can actually make 
an NP-hard problem tractable, see example 2.8 in \cite{PCSPAlgebraicApproach}.

The most popular example of the Promise CSP
is graph $(k,l)$-colorability, where we need to distinguish between $k$-colorable graphs and not even $l$-colorable, 
where $k\le l$.
This problem can be written as follows.

\begin{problem}
Let $|A|=k$, $|B|=l$, $\Gamma = \{(\neq_{A},\neq_{B})\}$.
What is the complexity of 
$\PCSP(\Gamma)$?
\end{problem}

Recently, it was proved \cite{PCSPAlgebraicApproach} that $(k,l)$-colorability 
is NP-hard for $l = 2k-1$ and $k\ge 3$
but even the complexity of 
$(3,6)$-colorability is still not known.

Even for two element domain the problem is widely open, but 
recently a dichotomy for symmetric Boolean PCSP
was proved \cite{ficak2019dichotomy}.

\begin{problem}
Let $A= B = \{0,1\}$.
Describe the complexity of 
$\PCSP(\Gamma)$ for all $\Gamma$.
\end{problem}

\subsection{Surjective CSP}

Another modification of the CSP is 
\emph{the Surjective Constraint Satisfaction Problem}.
For a constraint language $\Gamma$ over 
a domain $A$, $\SurjCSP(\Gamma)$ is the following decision problem:
given a formula 
$$\rho_{1}(\dots)
\wedge
\dots
\wedge
\rho_{s}(\dots),$$
where all relations 
$\rho_{1},\dots,\rho_{s}$ are from $\Gamma$;
decide whether there exists a surjective solution,
that is a solution with 
$\{x_{1},\dots,x_{n}\} = A$.
Only few results are known about the complexity of 
the Surjective CSP \cite{chen2014algebraic}.
That is why, we suggest to start studying this question with 
a very concrete constraint language on a 3-element domain.

\begin{problem}
Suppose 
$A = \{a,b,c\}$, 
$R = \{(x,y,z)\mid \{x,y,z\}\neq A\}$.
What is the complexity of 
$\SurjCSP(\{R\})$?
\end{problem}

After this problem (called \emph{no-rainbow problem}) we can move to the general question.

\begin{problem}
Describe the complexity 
of $\SurjCSP(\Gamma)$ for all constraint languages $\Gamma$.
\end{problem}

%% file: Transformations.tex
To prove the next theorem we will need additional definitions and few auxiliary lemmas.
First, we assign a characteristic to every variable of an instance, then we introduce a partial order on the set of characteristics.
After that, we define three transformations of the instance 
giving an expanded covering of the original instance.
We will prove that these transformations change 
the characteristics in a good way, so they can be used to 
generate an instance required in Theorem~\ref{FindPerfectConstraint}.

Let us assign a characteristic to every variable of an instance $\Phi$
whose constraints are critical and rectangular.
For a variable $x$ let $\mathfrak C_{1}$ be the set of all minimal congruences among the set $\Congruences(\Phi,x)$.
Then let $\mathfrak C_2$ be the set of all minimal congruences
among the congruences of $\Congruences(\Phi,x)$  that are not adjacent with any congruence from $\mathfrak C_{1}$.
Thus, we assign a pair $(\mathfrak C_1,\mathfrak C_2)$ to every variable $x$, which we denote $\xi(\Phi,x)$ and call \emph{characteristic}.

Let us introduce a partial order on the set of all characteristics.
For two sets of irreducible congruences $\mathfrak C_{1}$ and $\mathfrak C_{2}$
we write $\mathfrak C_{1}\le \mathfrak C_{2}$ if
for every $\sigma\in \mathfrak C_1$
there exists $\delta\in \mathfrak C_2$ such that $\delta\subseteq \sigma$.
We write $\mathfrak C_{1}<\mathfrak C_{2}$ if
$\mathfrak C_{1}\le \mathfrak C_{2}$ and $\mathfrak C_{2}\not\le\mathfrak C_{1}$.
It is easy to see that $\le$ is a transitive relation.

By $\uparrow\Opt(\mathfrak C)$ we denote the set of all congruences
$\sigma$ such that 
$\sigma\supseteq \delta$ for some $\delta\in\Opt(\mathfrak C)$.
We write $(\mathfrak C_1,\mathfrak C_2)\lesssim(\mathfrak C_1',\mathfrak C_2')$ if
one of the following conditions holds:
\begin{enumerate}
\item
$\mathfrak C_{1}<\mathfrak C_{1}'$;
\item
$\mathfrak C_1 = \mathfrak C_1'$
and $\mathfrak C_{2}\le \mathfrak C_{2}'$;
\item
$\mathfrak C_1 = \mathfrak C_1'$,
$\mathfrak C_{2}\not \le \mathfrak C_{2}'$,
$\mathfrak C_{2}'\not \le \mathfrak C_{2}$,
$\mathfrak C_{2}\setminus (\uparrow\Opt(\mathfrak C_{1}))<\mathfrak C_{2}'\setminus (\uparrow\Opt(\mathfrak C_{1}))$.
\end{enumerate}

\begin{lem}
$\lesssim$ is a transitive relation.
\end{lem}
\begin{proof}
Assume that $(\mathfrak C_1,\mathfrak C_2)\lesssim(\mathfrak C_1',\mathfrak C_2')$ 
and 
$(\mathfrak C_1',\mathfrak C_2')\lesssim(\mathfrak C_1'',\mathfrak C_2'')$.

If $\mathfrak C_1<\mathfrak C_{1}'$ or $\mathfrak C_{1}'<\mathfrak C_{1}''$, 
then $\mathfrak C_1<\mathfrak C_1''$, which completes this case.

Thus, we assume that 
$\mathfrak C_1=\mathfrak C_{1}'=\mathfrak C_{1}''$. 
It follows from (2) and (3) that
$$\mathfrak C_{2}\setminus (\uparrow\Opt(\mathfrak C_{1}))\le
\mathfrak C_{2}'\setminus (\uparrow\Opt(\mathfrak C_{1}))\le
\mathfrak C_{2}''\setminus (\uparrow\Opt(\mathfrak C_{1}))
.$$
If $\mathfrak C_2\le\mathfrak C_{2}'\le \mathfrak C_{2}''$, 
then $\mathfrak C_{2}\le\mathfrak C_{2}''$, which completes this case.
Thus, we assume that at least one of the above 
comparisons is strict (comes from (3)).
Hence,
$\mathfrak C_{2}\setminus (\uparrow\Opt(\mathfrak C_{1}))<
\mathfrak C_{2}''\setminus (\uparrow\Opt(\mathfrak C_{1}))$.
Therefore, 
$\mathfrak C_{2}''\not\le\mathfrak C_2$
and (2) or  (3) holds for 
$(\mathfrak C_1,\mathfrak C_2)$ and $(\mathfrak C_1'',\mathfrak C_2'')$, which completes the proof.
\end{proof}

\begin{remark}
Note that $\le$ is not a partial order in general, but it is a partial order on sets of mutually non-inclusive congruences.
Similarly, $\lesssim$ is not a partial order in general, 
but it is a partial order if we consider only 
pairs $(\mathfrak C_{1},\mathfrak C_2)$ such that 
all the congruences of $\mathfrak C_{i}$ are not included into each other
for $i=1,2$.
Thus, as it follows from the definition of the characteristic,
we defined a partial order on the set of all characteristics.
\end{remark}
A variable $x$ of an instance $\Theta$ is called \emph{stable} if all the congruences in $\Congruences(\Theta,x)$ are adjacent.
We say that variables 
$y_1$ and $y_2$ are \emph{friends} in $\Theta$ if they appear in the
scope of some constraint of $\Theta.$

\textbf{Transformation $T_{1}(\Theta)$: make an instance crucial in $D^{(1)}$.}
Using Remark~\ref{GetCrucialInstance},
we replace constraints by all weaker constraints until we get a CSP instance that is crucial in $D^{(1)}$.

Note that $T_{1}(\Theta)\in\Expanded(\Theta)$.

Below we assume that the instance $\Theta$ is crucial in $D^{(1)}$,
which by the inductive assumption for Theorem~\ref{ParPropertyMain} means that  every constraint in $\Theta$ has the parallelogram property and is critical.

\textbf{Transformation $T_{2}(\Theta, \sigma_{1},\sigma_{2},x)$: split a variable.}
Let $\Omega_{i}$ be the set of all constraints
$C\in \Theta$ such that $\ConOne(C,x) = \sigma_{i}$ for $i\in\{1,2\}$.
Let $\Omega_{0}$ be the set of all constraints $C\in\Theta\setminus (\Omega_{1}\cup\Omega_{2})$ containing $x$.
We transform our instance in the following way:
\begin{enumerate}
\item Choose 2 new variables $x_1$ and $x_2$;
\item Rename $x$ by $x_1$ in all constraints from $\Omega_0$ and $\Omega_1$;
\item Rename $x$ by $x_2$ in all constraints from $\Omega_2$;
\item Add the constraints $\sigma_{1}^{*}(x_{1},x_{2})$ and $\sigma_{2}^{*}(x_{1},x_{2})$;
\item For every $\sigma \in \Congruences(\Omega_{0},x)$ add
the constraint $\sigma(x_{1},x_{2})$.
\end{enumerate}

Note that 
$T_{2}(\Theta, \sigma_{1},\sigma_{2},x)$ is an expanded covering 
of $\Theta$, where the parent of $x_{1}$ and $x_{2}$ is $x$.

\begin{lem}\label{transTwo}
Suppose $D^{(1)}$ is a minimal 1-consistent one-of-four reduction of a cycle-consistent irreducible CSP instance $\Theta$,
$\Theta$ is crucial in $D^{(1)}$, 
and congruences $\sigma_{1}, \sigma_{2}\in\Congruences(\Theta,x)$ are not adjacent.
Then the instance $T_{2}(\Theta, \sigma_{1},\sigma_{2},x)$ has no solutions in $D^{(1)}$.
\end{lem}

\begin{proof}
Let $\Theta' = T_{2}(\Theta, \sigma_{1},\sigma_{2},x)$, $\sigma$ be the intersection of all congruences from $\Congruences(\Omega_{0},x)$.

Assume that $\Theta'$ has a solution in $D^{(1)}$.
Suppose $(x_{1},x_{2})=(a_{1},a_{2})$ in this solution.
Put $\Upsilon = \sigma_{1}(x_{1},x)\wedge\sigma_{2}(x_{2},x)\wedge\sigma(x_2,x).$
Consider the instance $\Theta'\wedge\Upsilon$. 
Since $(x_{1},x_{2})\in \sigma$ (by the definition 
of the transformation) and 
$(x_{2},x)\in\sigma$, 
we have $(x,x_{1})\in \sigma$.
Then each solution of $\Theta'\wedge\Upsilon$ can be taken as a 
solution of  $\Theta$ (we just ignore $x_{1}$ and $x_{2}$).
Hence, the instance $\Theta'\wedge \Upsilon$ has no solutions in $D^{(1)}$.
We apply Theorem~\ref{KeyConjunctionMain} to the subconstraint $\Upsilon(x_{1},x_{2})$ to obtain 
a sequence of formulas
$\Omega_{1},\ldots,\Omega_{t}\in \ExpShort(\Upsilon)$
such that 
$\Theta'\cup \Omega_{1}\cup\dots\cup\Omega_{t}$ has no solutions in $D^{(1)}$,
and
$\Omega_{i}^{(1)}(x_{1},x_2)$ defines a subdirect key relation 
$\rho_{i}$ with the parallelogram property for every $i$.
Note that the relation $\rho_{i}$ is reflexive, therefore, 
$\rho_{i}$ is a congruence on $D_{x_{1}}^{(1)}$.
If the reduction $D^{(1)}$ is nonlinear then
by $\omega_{i}$ we denote the relation defined by $\Omega_{i}(x_{1},x_{2})$.
If the reduction $D^{(1)}$ is linear then
by $\omega_{i}$ we denote the relation defined by $\Omega_{i}'(x_{1},x_{2},u_{1},\dots,u_r)$
from Lemma~\ref{AddLinearVariables}.
We know from Lemmas~\ref{AddLinearVariables} and \ref{SameConOneForNonlinear}
that
$\ConOne(\omega_{i},1)^{(1)} = \ConOne(\rho_{i},1)=\rho_{i}$.
Every
constraint in $\Omega_{i}$, 
which contains $x_1$ must have $\sigma_{1}$ for its constraint relation; thus the
first variable of $\omega_{i}$ is stable under $\sigma_{1}$ and $\ConOne(\omega_{i},1)\supseteq\sigma_{1}$.
Consider two cases:

Case 1. Assume that $\rho_{i}\neq \sigma_{1}^{(1)}$ for every $i$,
then $\ConOne(\omega_{i},1)\supseteq\cover{\sigma_{1}}$.
Hence $\rho_{i}\supseteq {(\cover{\sigma_{1}})}^{(1)}$ for every $i$.
Then
we may put $x_{1} = a_{1}$ and $x= x_{2} = a_{2}$ to get a solution
of
$\Theta'\cup\Omega_{1}\cup\dots\cup\Omega_{t}$ in $D^{(1)}$,
which contradicts the properties of the sequence 
$\Omega_{1},\dots,\Omega_{t}$.

Case 2. Assume that $\rho_{i}= \sigma_{1}^{(1)}$ for some $i$.
Since $(a_{1},a_{2})\in (\cover{\sigma_{1}})^{(1)}\setminus \sigma_{1}$
and $\ConOne(\omega_{i},1)^{(1)} = \rho_{i}$,
we have $\ConOne(\omega_{i},1)\not \supseteq\cover{\sigma_{1}}$.
Hence $\ConOne(\omega_{i},1)=\sigma_{1}$.
Suppose $D^{(1)}$ is a nonlinear reduction.
$\Upsilon(x_{1},x_{2})$ contains
$\sigma_{2}\cap\sigma$, and therefore
$\sigma_{2}\cap\sigma\subseteq \ConOne(\omega_{i},1)=\sigma_{1}$.
The symmetric conclusion $\sigma_{1}\cap\sigma\subseteq \sigma_{2}$
can be obtained by a symmetric argument, switching the roles of $\sigma_{1}$ and $\sigma_{2}$.
Since, $(a_{1},a_{2}) \in \sigma\setminus\sigma_{1}$, 
by Lemma~\ref{MinimalsAdjacent}
$\sigma_{1}$ and $\sigma_{2}$ are adjacent, which contradicts our assumptions.
Similarly, if $D^{(1)}$ is a linear reduction,
we can show that 
$\sigma_{2}\cap\sigma\cap \ConLin(D_{x})\subseteq \ConOne(\omega_{i},1) =\sigma_{1}$.
Indeed, suppose $(c,d)\in 
\sigma_{2}\cap\sigma\cap \ConLin(D_{x})$.
To witness that 
 $(c,d) \in \ConOne(\omega_{i},1)$ we need to define two tuples 
 from $\omega_{i}$ that differ only in the first component.
To obtain the first tuple 
we assign $c$ to every variable of $\Omega'_{i}$.
To obtain the second tuple we 
assign $d$ to all variables whose parent is $x_{1}$ or $x$,
and $c$ to the remaining variables (including $u_{1},\ldots,u_{r}$).
Thus, we can show that 
$\sigma_{2}\cap\sigma\cap \ConLin(D_{x})\subseteq\sigma_{1}$
and 
$\sigma_{1}\cap\sigma\cap \ConLin(D_{x})\subseteq\sigma_{2}$.
Since $(a_{1},a_{2}) \in (\sigma\cap\ConLin(D_{x}))\setminus\sigma_{1}$,
Lemma~\ref{MinimalsAdjacent} implies that 
$\sigma_{1}$ and $\sigma_{2}$ are adjacent, 
which contradicts our assumptions.
\end{proof}

Informally speaking, the following lemma
states that 
when we apply $T_{1}(T_{2}(\Theta,\sigma_{1},\sigma_{2},x))$
the characteristic of every new variable 
is less than the characteristic of its parent, 
the characteristic of old variables does not change, and 
if a stable variable gets a new friend then 
the friend's parent is not its friend anymore.

\begin{lem}\label{TransTwoStrength}
Suppose $D^{(1)}$ is a minimal 1-consistent one-of-four reduction of a cycle-consistent irreducible CSP instance $\Theta$,
$\Theta$ is crucial in $D^{(1)}$, congruences $\sigma_{1}, \sigma_{2}$ 
are minimal congruences among $\Congruences(\Theta,x)$, 
$\sigma_{1}$ and $\sigma_{2}$ are not adjacent,
$\Theta' = T_{1}(T_{2}(\Theta))$.
Then 
\begin{enumerate}
    \item $\xi(\Theta',y')< \xi(\Theta,y)$,
    if $y$ is a parent of $y'$ and $y'\neq y$;    
    \item $\xi(\Theta',y) = \xi(\Theta,y)$
    if $y\in \Var(\Theta)\cap\Var(\Theta')$;
\item if $y$ is stable in $\Theta$,
    $y'\in\Var(\Theta')\setminus\Var(\Theta)$,
    then 
    $y$ cannot be a friend of both 
    $y'$ and the parent of $y'$ in $\Theta'$;
\item $\Theta$ and $\Theta'$ have a common variable.    
\end{enumerate}
\end{lem}

\begin{proof}
By Lemma~\ref{transTwo} and the definition of 
$T_{1}$, $\Theta'$ is crucial, then 
by Lemma~\ref{KeepCrucialConstraint} 
for every constraint $C$ 
in $\Theta$ there exists a constraint $C'$ in
$\Theta'$ whose image in $\Theta$ is $C$.
Therefore, when we apply $T_{1}$ 
we weaken only binary constraints we added in $T_{2}$ but not 
the constraints from $\Theta$.
Then Claim (2) follows from the definition of the transformation.

$\ConOne(\Theta',x_{1})$ 
has all the congruences 
of 
$\ConOne(\Theta,x)$ but $\sigma_{2}$.
Additionally, it may contain congruences
$\delta$ such that 
$\delta\supseteq \sigma_{1}^{*}$,
$\delta\supseteq \sigma_{2}^{*}$,
or $\delta\supsetneq \sigma$
for $\sigma\in\ConOne(\Omega_0,x)$.
None of these congruences are minimal, 
so they cannot affect the first 
coordinate of $\xi(\Theta',x_{1})$.
Thus, 
$\xi(\Theta',x_{1})<\xi(\Theta,x)$.
Similarly, we can show that 
$\xi(\Theta',x_{2})<\xi(\Theta,x)$, which completes 
Claim (1).

Claim (3) follows from the fact that 
$x$, which is the only parent of 
variables from $\Var(\Theta')\setminus\Var(\Theta)$, 
is not in $\Theta'$.

Since a crucial instance cannot have just one variable, 
$\Theta$ and $\Theta'$ have a common variable, 
which is Claim (4).
\end{proof}

For an instance $\Omega\subseteq \Theta$ by
$\MinVar(\Omega,\Theta)$ we denote the
set of all variables $x$ such that
there exists $\sigma\in\ConOne(\Omega,x)$ 
that is minimal among $\Congruences(\Theta,x)$.

\textbf{Transformation $T_{3}(\Theta, \Omega)$ for a connected component $\Omega$.}
Let $\MinVar(\Omega,\Theta) = \{x_{1},\ldots,x_{s}\}$, where $s\ge 1$.
Let us define the new instance in the following way:
\begin{enumerate}
\item Choose new variables $x_{1}',\ldots,x_{s}'$;
\item Rename the variables $x_{1},\ldots,x_{s}$ by $x_{1}',\ldots,x_{s}'$ in $\Theta\setminus \Omega$;
\item Add the covers of all constraints from $\Omega$ with
$x_{1}',\ldots,x_{s}'$ instead of $x_{1},\ldots,x_{s}$;
\item For every $j$ and every $\sigma \in \Congruences(\Theta\setminus\Omega,x_{j})$ add
the constraint $\cover{\sigma}(x_{j},x_{j}')$;
\item 
For every $j$ and $\sigma\in\Congruences(\Theta\setminus\Omega,x_{j})$ such that $\LinkedCon(\Omega,x_{j})\not\subseteq \sigma$
add the constraint $\delta_{j}(x_{j},x_{j}')$, where
$\{\delta_{j}\} = \Opt(\Congruences(\Omega,x_{j}))$.
\end{enumerate}
Note that 
by Corollary~\ref{PathInConnectedComponent}
all congruences of $\Congruences(\Omega,x_{j})$
are adjacent. 
Then by Lemma~\ref{OptimalForAdjacent}
$\Opt(\Congruences(\Omega,x_{j}))$ contains just one element
and 
$\delta_{j}$ is well-defined in (5). 
Also, $T_{3}(\Theta, \Omega)$ is an expanded covering 
of $\Theta$, where the parent of every $x_{i}'$ is $x_{i}$.

\begin{lem}\label{transFour}
Suppose $D^{(1)}$ is a minimal 1-consistent one-of-four reduction of a cycle-consistent irreducible CSP instance $\Theta$,
$\Theta$ is crucial in $D^{(1)}$, $\Omega$ is a connected component of $\Theta$,
the solution set of $\Omega$ is subdirect,
$\Omega$ has a solution in $D^{(1)}$,
and for every $x\in\Var(\Omega)$
any two congruences that are minimal among $\ConOne(\Theta, x)$ are adjacent.
Then the instance $T_{3}(\Theta, \Omega)$ has no solutions in $D^{(1)}$.
\end{lem}
\begin{proof}
Suppose $\Var(\Omega)\setminus \MinVar(\Omega,\Theta) = \{z_{1},\ldots,z_{n}\}$.
Since the solution set of $\Omega$ is subdirect,
by Lemma~\ref{ProperReductionPreservesSubdirectness}
we know that the solution set of
$\Omega^{(1)}$ is subdirect.
We consider two cases:

Case 1:
$\LinkedCon(\Omega,y)= \ConOne(C,y)$ for every variable $y$ and every constraint $C\in \Omega$ 
having $y$ in the scope.
This means that 
$\Congruences(\Omega,y)$ contains exactly one congruence 
for every variable $y\in\Var(\Omega)$.
Since 
$\ConOne(C,x_{j})$ is minimal among $\Congruences(\Theta,x_{j})$ for every $j$ and every constraint $C\in \Omega$ containing $x_{j}$,
we have
$\LinkedCon(\Omega,x_{j}) \subsetneq \sigma$
for every $j$ and every $\sigma\in \Congruences(\Theta\setminus\Omega,x_{j})$.
Notice that for any constraint $C\in \Omega$ having $y_{1}$ and $y_{2}$ 
in the scope we have
$\LinkedCon(\Omega,y_1)\supseteq\ConOne(\proj_{y_1,y_{2}}(C),y_{1})$.
Since all constraints  of $\Omega$ are rectangular and critical, 
Lemma~\ref{RectangularCriticalArityTwo}
together with $\LinkedCon(\Omega,y_{1})= \ConOne(C,y_{1})$
imply that
the constraint $C$ should be binary. Thus, 
all the constraint relations are binary.

Assume that $n=0$.
Since $\Theta$ is crucial in $D^{(1)}$, 
the instance $\Omega$, viewed as a graph whose vertexes are variables, cannot have a cycle (otherwise, removing a constraint(edge) from the 
cycle would not affect the solution set, which contradicts 
the fact that $\Theta$ is crucial).
Hence, we can choose a constraint $C\in\Omega$ with a variable $x_{j}$
that appears just once in $\Omega$.
We replace the variable $x_{j}$ in $\Theta\setminus\{C\}$ by $x_{j}'$
and add the constraint 
$\sigma_{0}^{*}(x_{j},x_{j}')$, where $\sigma_{0}=\ConOne(C,x_{j})$. 
The obtained instance we denote by $\Theta'$.
Since the constraint $C$ is crucial in $D^{(1)}$, 
$\Theta'$ has a solution in $D^{(1)}$.
Since $\sigma_{0}^{*}\subseteq\sigma$
for every $\sigma\in \Congruences(\Theta\setminus\Omega,x_{j})$,
the solution of $\Theta'$ gives a solution of $\Theta$ in $D^{(1)}$,
which contradicts our assumptions.

Suppose $n>0$.
By $\Omega'$ we denote the copy of $\Omega$ with covers instead of constraints we introduced in (3).
For every variable $y$ of $\Omega$ by 
$\sigma_{y}$ we denote the minimal congruence 
such that 
$\sigma_{y}\supsetneq \LinkedCon(\Omega,y)$.
Since $\LinkedCon(\Omega,y)$ is an irreducible congruence, $\sigma_{y}$ is well-defined.
For every constraint 
$C = \rho(u,v)$ of $\Omega$, 
by $C'$ we denote the constraint $\rho'(u,v)$, 
where 
$\rho'(u,v) = \exists u' \exists v'\; \rho(u',v')\wedge \sigma_{u}(u,u')\wedge \sigma_{v}(v,v')$.
Let us show that 
$\rho'$ is a rectangular relation such that 
$\ConOne(\rho',1) = \sigma_{u}$ 
and 
$\ConOne(\rho',2) = \sigma_{v}$, 
that is a bijective mapping between equivalence classes of 
$\sigma_{u}$ and $\sigma_{v}$.
Since $\rho$ is rectangular, 
the congruence $\sigma_{u}$ generates a congruence 
on $D_{v}$ that is strictly greater than
$\ConOne(\rho,2)$, and therefore containing $\sigma_{v}$.
Therefore, $\sigma_{v}$ has at least as many equivalence classes 
as $\sigma_{u}$. The same is true for $\sigma_{u}$, 
which means that the congruence 
generated on $D_{v}$ from $\sigma_{u}$ using $\rho$ is equal to $\sigma_{v}$.
Therefore,
$\rho'$ is a rectangular relation such that 
$\ConOne(\rho',1) = \sigma_{u}$ 
and 
$\ConOne(\rho',2) = \sigma_{v}$.
Note that $\rho'\supsetneq \rho$. 

Since $n>0$ and $\Omega$ is not fragmented, 
there exists a path in $\Omega$ connecting 
a variable $z_{i}$ with a variable $x_{j}$ for every $i$ and $j$.
Then we glue a path going from 
$x_{j}$ to $z_{i}$ in $\Omega$ 
with a path going 
from $z_{i}$ to $x_{j}'$ in $\Omega'$.
For every constraint $C\in \Omega$
the constraint $C'$ (defined above) is weaker
or equivalent to its cover in $\Omega'$.
Therefore, every constraint $C$ in the obtained path from 
$x_{j}$ to $x_{j}'$ is not weaker than
$C'$, which means 
(by the properties of $C'$) that 
$x_{j}$ and $x_{j}'$ should be equivalent modulo $\sigma_{x_{j}}$ 
for every $j$ in any solution of $T_{3}(\Theta, \Omega)$.

Assume that $T_{3}(\Theta, \Omega)$ has a solution in $D^{(1)}$ with
$$(x_{1},\ldots,x_{s}, x_{1}',\ldots,x_{s}') = (b_{1},\ldots,b_{s},b_{1}',\ldots,b_{s}').$$
Since $\sigma_{x_{j}}\subseteq\sigma$
for every $\sigma\in \Congruences(\Theta\setminus\Omega)$, 
we have $(b_{i},b_{i}')\in\sigma$.
Therefore, 
we can assign
$$(x_{1},\ldots,x_{s}, x_{1}',\ldots,x_{s}') = (b_{1},\ldots,b_{s},b_{1},\ldots,b_{s}).$$
to get a solution of $\Theta^{(1)}$ (the remaining variables take on the same values).
This contradiction completes this case.

Case 2:
$\LinkedCon(\Omega,z)\neq \ConOne(C_{z},z)$ for some variable $z$ and some constraint $C_{z}\in \Omega$

Assume that $n=0$ and
$\LinkedCon(\Omega,x_{j}) \subseteq \sigma$
for every $j$ and every $\sigma\in \Congruences(\Theta\setminus\Omega,x_{j})$.
We rename the variable 
$z$ in $C_{z}$ by 
$z'$
and add the constraint 
$\sigma_{L}(z,z')$,
where 
$\sigma_{L} = \LinkedCon(\Omega,z)$.
Since $\Theta$ is crucial in $D^{(1)}$, the new instance has a solution $\beta$ in $D^{(1)}$.
Let $z$ be equal to $c$ in $\beta$.
Since the solution set of $\Omega^{(1)}$ is subdirect,
there exists a solution $\gamma$ of $\Omega^{(1)}$ with $z = c$.
Note that the corresponding elements of 
$\beta$ and $\gamma$ are linked in $\Omega$.
Since $\LinkedCon(\Omega,x_{j}) \subseteq \sigma$
for every $j$ and every $\sigma\in \Congruences(\Theta\setminus\Omega,x_{j})$,
we can build a solution of $\Theta^{(1)}$ with the values for $x_{j}$ from
$\gamma$ and the values for the remaining variables from $\beta$, which gives us a contradiction.

Thus, 
we assume that 
$n>0$ or 
$\LinkedCon(\Omega,x_{h}) \not\subseteq\sigma$
for some $h$ and $\sigma\in \Congruences(\Theta\setminus\Omega,x_{h})$.
In this case we consider a different transformation defined as follows:
\begin{enumerate}
\item Choose new variables $x_{1}',\ldots,x_{s}'$ and $x_{1}'',\ldots,x_{s}''$.
\item Add a copy of $\Omega$ to $\Theta$ with all the variables $x_{1},\ldots,x_{s}$ replaced by
$x_{1}',\ldots,x_{s}'$. We denote the copy by $\Omega'$.
\item Rename $x_{1},\ldots,x_{s}$ in $\Theta\setminus\Omega$ by $x_{1}'',\ldots,x_{s}''$.
\item For every $i$ and every $\sigma \in \Congruences(\Theta\setminus\Omega,x_{i})$ add a new variable $y$
and add the constraints $\sigma(x_{i}',y)$ and $\sigma(x_{i}'',y)$.
\item For every $i$ and every $\sigma \in \Congruences(\Theta\setminus\Omega,x_{i})$ add
the constraint $\cover{\sigma}(x_{i},x_{i}'')$.
\item For every $j$ and $\sigma\in\Congruences(\Theta\setminus\Omega,x_{j})$ such that $\LinkedCon(\Omega,x_{j})\not\subseteq \sigma$
add the constraint $\delta_{j}(x_{j},x_{j}')$, where
$\{\delta_{j}\} = \Opt(\Congruences(\Omega,x_{j}))$.
\end{enumerate}
Since here we just copied $\Omega$,
any solution of the obtained instance would give a solution to 
$\Theta$ (we use values of $x_{1}',\dots,x_{s}',z_{1},\dots,z_{n}$ to generate a solution), hence the obtained instance has no solutions in $D^{(1)}$.
We replace constraints from $\Omega'$ containing at least one of the variables $x_{1}',\ldots,x_{s}'$
by their covers step by step.
Thus, in one step we replace just one constraint from $\Omega'$.
We consider two cases.

Assume that after all replacements we get an instance 
$\Theta_{0}$
without solutions in $D^{(1)}$.
Any solution of $T_{3}(\Theta, \Omega)$ gives a solution
of $\Theta_{0}$: if $x_{i}' = a_{i}$ in the solution of $T_{3}(\Theta, \Omega)$,
then we put $x_{i}' = x_{i}'' = y = a_{i}$ in $\Theta_{0}$ for every $i$ and the corresponding $y$'s
(the remaining variables take the same values). 
Therefore, $T_{3}(\Theta, \Omega)$
has no solutions in $D^{(1)}$, which completes this case.

Assume that after some replacement the instance gets a solution in $D^{(1)}$.
Suppose the instance before this replacement is $\Theta'$ and
the corresponding constraint to be replaced is $C$. Choose a variable $x_{l}'\in \Var(C)$.

Let $\delta = \ConOne(C,x_{l}')$, $\rho$ be an optimal bridge from $\delta$ to $\delta$.
Let us define a new bridge by
$$\rho'(u_{1},u_{2},u_{3},u_{4}) = \exists v_{1}\exists v_{2}
\rho(u_1,u_2,v_{1},v_{2})\wedge \rho(u_3,u_4,v_1,v_2)\wedge \rho(u_{1},u_{1},u_{3},u_{3})\wedge
\delta^{*}(u_{3},u_{4}).$$
Since 
$\rho'(x,x,y,y) = \rho(x,x,y,y)$, 
$\rho'$ is also an optimal bridge.
Additionally, $\rho'$ has the following property:
if $(a,b,c,d)\in\rho'$ then $(a,c)\in\widetilde{\rho}$ 
and $(c,d)\in\delta^{*}$.

Then we change $\Theta'$ in the following way.
We add three new variables $u_{1}$, $u_{2}$, $x_{l}'''$, replace $x_{l}'$ in $C$ by $x_{l}'''$,
add the constraint $\rho'(x_{l}',x_{l}''',u_{1},u_{2})$
and the constraint $\delta(u_{1},u_{2})$.
We denote the new instance by $\Theta''$. 
By the definition of a bridge, $\Theta''$ has no solutions in $D^{(1)}$.

By $\Upsilon$ we denote all constraints of $\Theta''$ containing $x_{j}'$ for some $j$ or $x_{l}'''$.
Let $\{y_{1},\ldots, y_{t}\}$ be the set of all variables of $\Upsilon$ except
for $z_{1},\ldots,z_{n}$, $x_{1},\ldots,x_{s}$, $x_{1}',\ldots,x_{s}'$,  $u_{1},u_{2}$, and $x_{l}'''$.
Suppose that the variable $x_{i_{j}}$ is the corresponding variable
and $\sigma_{j}$ is the corresponding congruence for $y_{j}$ (see step (4) of the transformation).

If we remove the constraint
$\delta(u_{1},u_{2})$ from $\Theta''$, 
then it is equivalent to making a constraint $C$ of $\Theta'$ weaker, which means that we get a solution of $\Theta''$ in $D^{(1)}$ after the removal.
Let
\begin{multline*}
(x_{1},\ldots,x_{s},x_{1}',\ldots,x_{s}',x_{1}'',\ldots,x_{s}'', y_{1},\ldots,y_{t}, z_{1},\ldots,z_{n},u_1,u_2)
=\\(a_{1},\ldots,a_{s},a_{1}',\ldots,a_{s}', a_{1}'',\ldots,a_{s}'', d_{1},\ldots,d_{t}, b_{1},\ldots,b_{n},c_1,c_2)
\end{multline*}
in this solution.

First, we want to show that $(a_{l},a_{l},c_{1},c_{1})\in \rho'$. 
By the definition 
of $\rho'$ and $\Theta''$, 
we have $(a_{l}',c_{1})\in \widetilde {\rho}$. We consider two subcases.
Case 2A. Suppose $n>0$.
Gluing a path from $x_{l}$ to $z_{1}$ in $\Omega$ 
and a path from $z_{1}$ to $x_{1}'$ in $\Omega'$, 
we show that $a_{l}$ and $a_{l}'$ are linked in
$\Omega'$.
We apply Theorem~\ref{PathInConnectedComponentThm} to get a bridge 
from $\delta$ to $\delta$ containing $(a_{l},a_{l},a_{l}',a_{l}')$.
Then we compose this bridge with the bridge $\rho$
to obtain a bridge from $\delta$ to $\delta$
containing $(a_{l},a_{l},c_{1},c_{1})$.
Since the bridge $\rho'$ is optimal, we have
$(a_{l},a_{l},c_{1},c_{1})\in \rho'$.

Case 2B. 
$\LinkedCon(\Omega,x_{h}) \not\subseteq\sigma$
for some $h$ and $\sigma\in \Congruences(\Theta\setminus\Omega,x_{h})$.
Let $\zeta\in\Congruences(\Omega,x_{h})$
and $\xi_{0}$ be an optimal bridge from 
$\zeta$ to $\zeta$.
Note that by step (6) of the new transformation
$(a_{h},a_{h}')\in\widetilde\zeta$.
We know that $a_{l}$ and $a_{h}$ are linked in $\Omega$,
$a_{h}'$ and $a_{l}'$ are linked in
$\Omega'$.
We apply Theorem~\ref{PathInConnectedComponentThm} to get a bridge $\xi_{1}$
from $\delta$ to $\zeta$ containing $(a_{l},a_{l},a_{h},a_{h})$
and a bridge $\xi_2$ from $\zeta$ to $\delta$
containing $(a_{h}',a_{h}',a_{l}',a_{l}')$.
Then we compose 
$\xi_{1}$, $\xi_{0}$, $\xi_{2}$ and $\rho$ (in this order) 
to obtain a bridge from $\delta$ to $\delta$
containing $(a_{l},a_{l},c_{1},c_{1})$.
Since the bridge $\rho'$ is optimal, we have
$(a_{l},a_{l},c_{1},c_{1})\in \rho'$.

Consider a subconstraint $\Upsilon(y_{1},\ldots,y_{t},x_{1},\dots,x_{s},z_{1},\ldots,z_{n},u_1,u_2)$.
The constraint $\delta(u_1,u_2)$ is isolated in $\Theta''\setminus\Upsilon$, 
hence 
$\Theta''\setminus\Upsilon$ has a solution in $D^{(1)}$.
Using Theorem~\ref{KeyConjunctionMain}, we find
$\Upsilon_{1},\dots,\Upsilon_{v}\in \ExpShort(\Upsilon)$ such that
$\Upsilon_{i}^{(1)}(y_{1},\ldots,y_{t},x_{1}\dots,x_{s},z_{1},\ldots,z_{n},u_1,u_2)$
defines a key relation $\rho_{i}$ with the parallelogram property for every $i$.
Since 
$(\Theta''\setminus\Upsilon)
\cup\Upsilon_{1}\cup\dots\cup\Upsilon_{v}$ has no solutions in $D^{(1)}$, 
we can choose $k$ such that $\rho_{k}$
omits the tuple
$(d_{1},\ldots,d_{t},a_{1},\dots,a_{s},b_{1},\ldots,b_{n},c_1,c_1)$.
By the definition of $\Upsilon$, 
we can substitute 
the value $c_{1}$ instead of $u_{1}$ and $u_2$ and put 
$x_{i} = x_{i}' = x_{i}''=a_{i}$
and $y_{j} = x_{i_{j}}$ for every $i$ and $j$
to get a solution of $\Upsilon$.
Precisely, for every $j$ we put
$d_{j}' = a_{i_{j}}$,
then
$(d_{1}',\ldots,d_{t}',a_{1},\dots,a_s,b_{1},\ldots,b_{n},c_{1},c_{1})\in \rho_{k}$
(here we used that $(a_{l},a_{l},c_1,c_1)\in\rho'$).
Also we know that 
$$(d_{1},\ldots,d_{t},a_{1},\dots,a_s,b_{1},\ldots,b_{n},c_1,c_1)\notin\rho_{k},$$
$$(d_{1},\ldots,d_{t},a_{1},\dots,a_s,b_{1},\ldots,b_{n},c_1,c_2)\in\rho_{k}.$$
Then we consider the minimal $j$ 
such that 
$(d_{1}',\ldots,d_{j}',d_{j+1},\dots,d_{t},a_{1},\dots,a_s,b_{1},\ldots,b_{n},c_1,c_1)\in\rho_{k}$.
It follows from the definition of $\Theta'$ 
that 
$(d_{j},d_{j}')\in\sigma_{j}^{*}$, 
and from the definition of 
$\rho'$ that $(c_{1},c_{2})\in\delta^{*}$.
Then by Lemma~\ref{SubconstraintConnectivity}
there exists a bridge $\zeta_{1}$ from $\delta$ to $\sigma_{j}$ such that 
$\widetilde{\zeta_{1}}$ contains $\Upsilon(u_2,y_{j})$, 
and therefore it contains 
$\Omega(x_{l},x_{i_{j}})$.

Suppose
$\delta_{0}\in\ConOne(\Omega,x_{i_{j}})$.
Applying Theorem~\ref{PathInConnectedComponentThm} to $\Omega$ and the variables $x_{i_{j}}$ and $x_{l}$,
we get a bridge $\zeta_{2}$  from $\delta_{0}$ to $\delta$
such that $\widetilde\zeta_{2}$ contains 
all elements linked in $\Omega$, and 
therefore it contains $\Omega(x_{i_j},x_{l})$.
Composing the bridges $\zeta_{2}$ and $\zeta_1$ we get a  
bridge from $\delta_{0}$ to $\sigma_{j}$.
Since $\Omega(x_{l},x_{i_j})$ is subdirect, 
the obtained bridge is reflexive.
Hence $\delta_{0}$ and $\sigma_{j}$ are adjacent, which contradicts the fact that $\sigma_{j}\in \Congruences(\Theta\setminus\Omega,x_{i_{j}})$.
\end{proof}


Below we prove a property of the transformation 
$T_{3}$ similar to the property of $T_{2}$ proved in  Lemma~\ref{TransTwoStrength}.

\begin{lem}\label{TransFourStrength}
Suppose $D^{(1)}$ is a minimal 1-consistent one-of-four reduction of a cycle-consistent irreducible CSP instance $\Theta$,
$\Theta$ is crucial in $D^{(1)}$, $\Omega$ is a connected component,
the solution set of $\Omega$ is subdirect,
$\Omega$ has a solution in $D^{(1)}$,
for every $x\in\Var(\Omega)$
any two congruences that are minimal among $\ConOne(\Theta, x)$ are adjacent,
and 
$\Theta' = T_{1}(T_{3}(\Theta,\Omega))$.
Then 
\begin{enumerate}
    \item $\xi(\Theta',y')< \xi(\Theta,y)$,
    if $y$ is a parent of $y'$ and $y'\neq y$;    
    \item $\xi(\Theta',y) \le \xi(\Theta,y)$,
    if $y\in \Var(\Theta)\cap\Var(\Theta')$;
    \item $\xi(\Theta',y) < \xi(\Theta,y)$,
    if $y\in \MinVar(\Omega,\Theta)$
    and $y$ not stable in $\Theta$;
\item if $y$ is stable in $\Theta$,
    $y'\in\Var(\Theta')\setminus\Var(\Theta)$,
    then 
    $y$ cannot be a friend of both 
    $y'$ and the parent of $y'$ in $\Theta'$;
\item $\Theta$ and $\Theta'$ have a common variable.    
    
\end{enumerate}
\end{lem}

\begin{proof}
By Lemma~\ref{transFour} $\Theta'$ is crucial, then 
by Lemma~\ref{KeepCrucialConstraint} 
for every constraint $C$ 
in $\Theta$ there exists a constraint $C'$ in
$\Theta'$ whose image in $\Theta$ is $C$.
Therefore, when we apply $T_{1}$ 
we weaken only binary constraints and covers we added in $T_{3}$ 
but not 
the constraints from $\Theta$.

First, let us show that
$\xi(\Theta,z)= \xi(\Theta',z)$
for every $z\in \Var(\Theta)\setminus \{x_{1},\dots,x_{s}\}$.
The only constraints with $z$ we added are
the constraints we obtained from the covers 
of constraints from $\Omega$ using $T_{1}$.
Let $C''$ be the cover of a constraint 
$C'$ from $\Omega$.
Let $C'''$ be a constraint obtained from 
$C''$ using $T_{1}$.
By Lemma~\ref{RectangularCriticalArityTwo}, 
$\ConOne(C'',z)\supsetneq\ConOne(C',z)$.
By the definition of $T_{1}$, 
$\ConOne(C''',z)\supseteq\ConOne(C'',z)$.
Hence, 
$\ConOne(C''',z)\supsetneq\ConOne(C',z)$.
Since $\ConOne(C',z)$ is not adjacent with a minimal congruence among $\ConOne(\Theta,z)$, 
$\ConOne(C''',z)$ cannot affect 
the characteristic of $z$.
Therefore, $\xi(\Theta,z)= \xi(\Theta',z)$.

Second, let us 
show that $\xi(\Theta',x_{i}) \le \xi(\Theta,x_{i})$
for every $i$.
Since any two congruences that are minimal among 
$\ConOne(\Theta,x_{i})$ are adjacent
and the congruence we add in (5)
cannot be a new minimal congruence 
among $\ConOne(\Theta,x_{i})$, 
the first components of 
$\xi(\Theta',x_{i})$ and $\xi(\Theta,x_{i})$
are equal.
If the variable $x_{i}$ is stable then 
we do not add anything in (4) and (5),
hence the second components of 
$\xi(\Theta',x_{i})$ and $\xi(\Theta,x_{i})$
are empty, which completes Claim (2) for 
this case.
Assume that $x_{i}$ is not stable.
Then the second component 
of $\xi(\Theta',x_{i})$
has congruences appeared in (4) and (5)
instead of congruences 
from $\Congruences(\Theta\setminus\Omega,x_i)$.
The congruences we added in (4) are bigger than 
the corresponding congruences 
from $\ConOne(\Theta\setminus\Omega,x_{i})$.
Hence if we added nothing in (5) then 
$\xi(\Theta',x_{i})<\xi(\Theta,x_{i})$ because of the second components.
Otherwise, consider a minimal congruence 
$\sigma\in\ConOne(\Theta\setminus\Omega,x_{i})$
such that $\LinkedCon(\Omega,x_{i})\not\subseteq\sigma$.
By Corollary~\ref{PathInConnectedComponent}, 
congruences we obtain using (5) are greater than 
$\LinkedCon(\Omega,x_{i})$,
hence they cannot be smaller than $\sigma$.
Therefore, 
$\sigma$ belongs to the second component of 
$\xi(\Theta,x_{i})$ and 
the second component of 
$\xi(\Theta',x_{i})$ does not have a congruence that is equal to or smaller than $\sigma$.
We conclude that either second components of 
$\xi(\Theta,x_{i})$ and $\xi(\Theta',x_{i})$
are incomparable, 
or the second component of $\xi(\Theta',x_{i})$ is smaller.
Note that all the congruences 
we obtain in (5) are from $\uparrow\Opt(\Congruences(\Omega,x_{i}))$, 
which 
means that 
$\xi(\Theta',x_{i}) < \xi(\Theta,x_{i})$ in this case.
Thus we proved Claims (2) and (3).

To prove Claim (1) we need to show that 
$\xi(\Theta',x_{i}')<\xi(\Theta,x_{i})$ for every $i$.
Every congruence 
from $\ConOne(\Theta\setminus\Omega,x_{i})$ 
is bigger than some congruence 
from $\ConOne(\Omega,x_{i})$.
Hence, $\xi(\Theta',x_{i}')<\xi(\Theta,x_{i})$ 
because of the first components.

To prove Claim (4) consider two cases.
Case 1. Suppose $x_{i}$ is a stable variable.
Since $x_{i}$ cannot be a friend of $x_{j}'$, 
we obtain the necessary condition.
Case 2. Suppose $z\in\Var(\Theta)\setminus\{x_{1},\dots,x_{s}\}$ is a stable variable.
By the definition of being stable we conclude that 
$z\notin\Var(\Omega)$. Hence $z$ cannot be a friend 
of $x_{j}$, which proves Claim (4).

The Claim (5) follows from the fact that 
$x_{1}$ should be in both $\Theta$ and $\Theta'$.
\end{proof}

\begin{THMFindPerfectConstraint}
Suppose $D^{(1)}$ is a minimal 1-consistent one-of-four reduction of a cycle-consistent irreducible CSP instance $\Theta$,
$\Theta$ is crucial in $D^{(1)}$ and not connected.
Then there exists an instance $\Theta'\in\Expanded(\Theta)$ that is crucial in $D^{(1)}$
and contains a linked connected component whose solution set is not subdirect.
\end{THMFindPerfectConstraint}

\begin{proof}
We build a sequence 
of instances
$\Theta_{1},\Theta_2,\Theta_{3},\ldots$
such that
$\Theta_{i+1}\in \Expanded(\Theta_{i})$, 
and every $\Theta_{i}$ is crucial in $D^{(1)}$.
Recall that by the inductive assumption for Theorem~\ref{ParPropertyMain}
all constraint relations 
of each $\Theta_{i}$ are critical relations with the parallelogram property.
We start with $\Theta_{1} = \Theta$.
We want the final element of this sequence to 
contain a linked connected component whose solution set is not subdirect.
Suppose we already defined $\Theta_{i}$.

If there exist congruences $\sigma_{1},\sigma_{2}\in\Congruences(\Theta_{i},x)$ for some variable $x$ 
that are not adjacent and 
minimal among $\Congruences(\Theta,x)$,
then put $\Theta_{i+1} := T_{1}(T_{2}(\Theta_{i},\sigma_{1},\sigma_{2},x))$.
By Lemma~\ref{transTwo}, $\Theta_{i+1}$ is crucial in $D^{(1)}$.

Otherwise, we know that any two minimal congruences in $\Congruences(\Theta,x)$ for every variable $x$ are adjacent.
By Lemma~\ref{StayNotConnected}, $\Theta_{i}$ is not connected.
Since $\Theta_{i}$ is crucial, it is also not fragmented.
Then there exist a variable $x$ that is not stable.
Choose ``an oldest'' nonstable variable in $\Theta_{i}$, 
that is a variable $x$ with the minimal number $j$ 
such that $x\in \Var(\Theta_{j})$.
Choose the connected component $\Omega$ containing a minimal congruence of $\Congruences(\Theta,x)$.
Put $\Theta_{i+1} := T_{1}(T_{3}(\Theta_{i},\Omega))$.

If $\Omega$ is not linked, then irreducibility of $\Theta$ implies that the solution set of $\Omega$ is subdirect.
If $\Omega$ is linked and the solution set of $\Omega$ is not subdirect, then the theorem is proved and we stop the process.
Thus, we assume that the solution set of $\Omega$ is subdirect.
Since $\Theta_{i}$ is crucial in $D^{(1)}$ and not connected, $\Omega^{(1)}$ has a solution.
Then by Lemma~\ref{transFour}, $\Theta_{i+1}$ 
is crucial in $D^{(1)}$.

Thus, the next element of the sequence is defined either by 
$T_{1}(T_{2}(\Theta_{i},\sigma_{1},\sigma_{2},x))$,
or by $T_{1}(T_{3}(\Theta_{i},\Omega))$.
Now, we want to prove that the sequence 
$\Theta_{1},\Theta_{2},\Theta_{3},\dots$ cannot be infinite, 
which means that 
the last element with the required property exists.
To prove this we are going to use 
Theorem~\ref{newcolonies}.
First, we extend a partial order $\lesssim$ on
characteristics to a linear order $\le$
such that $(\Omega_1,\Omega_2)\lesssim(\Omega_1',\Omega_2')$ implies
$(\Omega_1,\Omega_2)\le(\Omega_1',\Omega_2')$.
Second, we consider the set of all pairs $(x,\xi(\Theta_{i},x))$, where $x\in\Var(\Theta_{i})$, as the set of organisms.
Two organisms are friends if the corresponding variables 
were friends in $\Theta_{i}$ for some $i$ (if they've ever been friends). 
If $x\in \Var(\Theta_{i})\cap\Var(\Theta_{i+1})$ 
and $\xi(\Theta_{i+1},x)<\xi(\Theta_{i},x)$, 
then we say that 
$(x,\xi(\Theta_{i},x))$ is the parent of 
$(x,\xi(\Theta_{i+1},x))$. 
Also, 
if $x\in\Var(\Theta_{i+1})\setminus \Var(\Theta_{i})$, 
and $x'$ is a parent of $x$ in $\Theta_{i}$, 
then 
$(x',\xi(\Theta_{i},x'))$ is the parent of 
$(x,\xi(\Theta_{i+1},x))$. 
The characteristic $\xi(\Theta_{i},x)$ 
is considered as the strength
of the organism $(x,\xi(\Theta_{i},x))$. 
Then the set of organisms
$X_{i}$ is the set of all pairs $(x,\xi(\Theta_{j},x))$ for $j\le i$.

Let us check all the assumptions we have in Theorem~\ref{newcolonies}.
Condition (1) follows from Lemma~\ref{TransTwoStrength} (claims 1,2)  and Lemma~\ref{TransFourStrength} (claims 1,2).
Conditions (2) and (3) follow from the fact that 
each $\Theta_{i+1}$ is from $\Expanded(\Theta_{i})$.

Since the transformation $T_{2}$ (followed by $T_{1}$) replace a variable by two variables with smaller characteristic and does not change the characteristic of other variables, for the sequence to be infinite, we need to apply 
the transformation $T_{3}$ infinitely many times.
By Lemma~\ref{TransFourStrength} (claim 3)
we always reduce the characteristic of 
the chosen nonstable variable when we apply $T_{3}$, 
which means that every variable will be stable at some moment.

It remains to show that condition (4) holds.
As we noticed above, every variable will be stable at some moment.
It remains to show that a variable $z$ stable in $\Theta_{i}$ cannot get infinitely many friends (here we care only about the variables but not about the organisms).
By Lemma~\ref{TransTwoStrength} (claim 3)  and Lemma~\ref{TransFourStrength} (claim 4), 
the variable $z$ cannot be a friend of some variable $y$ appeared 
in $\Theta_{j}$ for $j> i$ and the parent of $y$.
Thus, if we consider the set of friends of $z$
in $\Theta_{j}$ for $j> i$, then we see that 
going from $\Theta_{j}$ to $\Theta_{j+1}$
we can replace an old friend by new friends (that are weaker) but we cannot add a new friend keeping its parent. Therefore, after getting stable a variable cannot get infinitely many friends and condition (4) holds. 

Since $\Theta_{i}$ is crucial in $D^{(1)}$, it is not fragmented.
By Lemma~\ref{TransTwoStrength} (claim 4)  and Lemma~\ref{TransFourStrength} (claim 5), 
$\Theta_{i}$ and $\Theta_{i+1}$ have at least one common variable for every $i$.
Therefore, the set of all organisms cannot be divided into two disjoint sets
with no friendship between them.
Thus, condition (5) of Theorem~\ref{newcolonies} cannot hold, which proves that the process will stop at some $\Theta_{i}$
having a linked connected component whose solution set is not subdirect.
\end{proof}